\documentclass[aps,prb,twocolumn,nofootinbib,longbibliography,notitlepage,floatfix,10pt]{revtex4-2}
\usepackage[dvipsnames]{xcolor}
\usepackage{amsmath}
\usepackage{amssymb}
\usepackage{tikz}
\usetikzlibrary{3d}
\usetikzlibrary{decorations.pathmorphing}
\usetikzlibrary{arrows.meta}
\usetikzlibrary{svg.path}
\usepackage{tikz-3dplot}
\usepackage[english]{babel}
\usepackage{physics}
\usepackage{bm}
\usepackage{graphicx}
\graphicspath{{./pictures/}}
\usepackage{float}
\usepackage{subcaption}
\usepackage{array}
\usepackage{multirow}
\usepackage{scalerel}
\usepackage{CJKutf8}

\makeatletter
\AddToHook{package/ltcaption/before}{%
	\expandafter\let\csname longtable*\endcsname\undefined
	\expandafter\let\csname endlongtable*\endcsname\undefined
}
\makeatother
\usepackage{caption}
\DeclareCaptionJustification{fulljustified}{\leftskip=0pt\rightskip=0pt\parfillskip=0pt plus 1fil\relax}
\definecolor{darkblue}{RGB}{0,0,150}
\definecolor{darkred}{RGB}{150,0,0}

\definecolor{orcidlogocol}{HTML}{A6CE39}
\tikzset{
	orcidlogo/.pic={
		\fill[orcidlogocol] svg{M256,128c0,70.7-57.3,128-128,128C57.3,256,0,198.7,0,128C0,57.3,57.3,0,128,0C198.7,0,256,57.3,256,128z};
		\fill[white] svg{M86.3,186.2H70.9V79.1h15.4v48.4V186.2z}
		svg{M108.9,79.1h41.6c39.6,0,57,28.3,57,53.6c0,27.5-21.5,53.6-56.8,53.6h-41.8V79.1z M124.3,172.4h24.5c34.9,0,42.9-26.5,42.9-39.7c0-21.5-13.7-39.7-43.7-39.7h-23.7V172.4z}
		svg{M88.7,56.8c0,5.5-4.5,10.1-10.1,10.1c-5.6,0-10.1-4.6-10.1-10.1c0-5.6,4.5-10.1,10.1-10.1C84.2,46.7,88.7,51.3,88.7,56.8z};
	}
}

\newcommand\orcidicon[1]{\href{https://orcid.org/#1}{\mbox{\scalerel*{
				\begin{tikzpicture}[yscale=-1,transform shape]
					\pic{orcidlogo};
				\end{tikzpicture}
			}{|}}}}

\DeclareRobustCommand{\revsubset}{\mathchoice%
	{\rotatebox[origin=c]{180}{$\subset$}}%
	{\rotatebox[origin=c]{180}{$\subset$}}%
	{\rotatebox[origin=c]{180}{$\scriptstyle\subset$}}%
	{\rotatebox[origin=c]{180}{$\scriptscriptstyle\subset$}}%
}
\newcommand{\dual}{\overset{\text{dual}}{\longleftrightarrow}}

\newcommand{\ann}{\operatorname{ann}}
\tikzset{gate/.style={line width=.75pt,->,shorten <=.7pt,shorten >=.7pt}}

\usepackage{mathrsfs}
\newcommand{\bx}{\bar{x}}
\newcommand{\by}{\bar{y}}
\newcommand{\bz}{\bar{z}}
\newcommand{\bF}{\bar{F}}
\newcommand{\bA}{\bar{A}}

\newcommand{\imop}{\operatorname{im}}
\newcommand{\coker}{\operatorname{coker}}
\newcommand{\torsion}{\operatorname{torsion}}
\newcommand{\supp}{\operatorname{supp}}

\newtheorem{theorem}{Theorem}[section]
\newtheorem{lemma}[theorem]{Lemma}
\newtheorem{corollary}[theorem]{Corollary}

\usepackage[
colorlinks=true,
linkcolor=darkblue,
citecolor=darkblue,
urlcolor=darkred,
linktocpage=true
]{hyperref}

\newcommand{\drawsmoothtoplatticegray}{%
	\foreach \j in {0,1,2}{\foreach \k in {0,1}{\draw[gray](0,\j,\k)--(3,\j,\k);}}
	\foreach \i in {0,1,2}{\foreach \j in {0,1,2}{\draw[gray](\i,\j,0)--(\i,\j,2);}}
	\foreach \i in {0,1,2}{\foreach \k in {0,1}{\draw[gray](\i,0,\k)--(\i,2,\k);}}
	
	\foreach \i in {0,1,2,3}{\draw[dashed, gray] (\i,0,2)--(\i,2,2);}
	\foreach \j in {0,1,2}{\draw[dashed, gray] (0,\j,2)--(3,\j,2)   (3,\j,0)--(3,\j,2);}
	\foreach \k in {0,1}{\draw[dashed, gray] (3,0,\k)--(3,2,\k);}
}

\newcommand{\drawroughtoplatticegray}{%
	\foreach \j in {0,1}{\foreach \k in {0,1}{\draw[gray](0,\j,\k)--(3,\j,\k);}}
	\foreach \i in {0,1,2}{\foreach \j in {0,1}{\draw[gray](\i,\j,0)--(\i,\j,2);}}
	\foreach \i in {0,1,2}{\foreach \k in {0,1}{\draw[gray](\i,0,\k)--(\i,2,\k);}}
	
	\foreach \i in {0,1,2,3}{\draw[dashed, gray] (\i,0,2)--(\i,2,2);}
	\foreach \j in {0,1}{\draw[dashed, gray] (0,\j,2)--(3,\j,2)   (3,\j,0)--(3,\j,2);}
	\foreach \k in {0,1}{\draw[dashed, gray] (3,0,\k)--(3,2,\k);}
}

\newcommand{\drawemtoplatticegray}{%
	\foreach \j in {0,1}{\foreach \k in {0,1}{\draw[gray](0,\j,\k)--(3,\j,\k);}}
	\foreach \i in {0,1,2}{\foreach \j in {0,1,2}{\draw[gray](\i,\j,0)--(\i,\j,2);}}
	\foreach \i in {0,1,2}{\foreach \k in {0,1}{\draw[gray](\i,0,\k)--(\i,2,\k);}}
	
	\foreach \i in {0,1,2,3}{\draw[dashed, gray] (\i,0,2)--(\i,2,2);}
	\foreach \j in {0,1}{\draw[dashed, gray] (0,\j,2)--(3,\j,2);}
	\foreach \j in {0,1,2}{\draw[dashed, gray] (3,\j,0)--(3,\j,2);}
	\foreach \k in {0,1}{\draw[dashed, gray] (3,0,\k)--(3,2,\k);}
}

\newcommand{\drawaxes}{%
	\draw[->, blue, line width=1pt] (-2-0.7,0,0)--(-2-0.7,0,1);
	\draw[->, blue, line width=1pt] (-2-0.7,0,0)--(-1.2-0.7,0,0);
	\draw[->, blue, line width=1pt] (-2-0.7,0,0)--(-2-0.7,0.8,0);
	\node[blue] at (-2.18-0.7,0,1) {$x$};
	\node[blue] at (-1.2-0.7,-0.22,0) {$y$};
	\node[blue] at (-1.75-0.7,0.8,0) {$z$};
}

\newcommand{\drawActopEight}[2]{%
	\fill[red, opacity=0.08] (#1,2,#2)--(#1+1,2,#2)--(#1+1,1,#2)--(#1+1,1,#2+1)--(#1,1,#2+1)--(#1,2,#2+1)--cycle;
	\foreach \di/\dk in {0/0, 1/0, 0/1, 1/1}{
		\fill[red, opacity=0.12] (#1+\di-0.12, 1.08, #2+\dk)--(#1+\di+0.12, 1.08, #2+\dk)--(#1+\di+0.12, 1.95, #2+\dk)--(#1+\di-0.12, 1.95, #2+\dk);
	}
	\fill[red, opacity=0.12] (#1+0.12, 1, #2-0.15)--(#1+0.88, 1, #2-0.15)--(#1+0.88, 1, #2+0.15)--(#1+0.12, 1, #2+0.15);
	\fill[red, opacity=0.12] (#1+0.12, 1, #2+1-0.15)--(#1+0.88, 1, #2+1-0.15)--(#1+0.88, 1, #2+1+0.15)--(#1+0.12, 1, #2+1+0.15);
	\fill[red, opacity=0.12] (#1-0.15, 1, #2+0.12)--(#1+0.15, 1, #2+0.12)--(#1+0.15, 1, #2+0.88)--(#1-0.15, 1, #2+0.88);
	\fill[red, opacity=0.12] (#1+1-0.15, 1, #2+0.12)--(#1+1+0.15, 1, #2+0.12)--(#1+1+0.15, 1, #2+0.88)--(#1+1-0.15, 1, #2+0.88);
	\node[red, font=\small] at (#1, 1.55, #2) {$X$};
	\node[red, font=\small] at (#1+1, 1.55, #2) {$X$};
	\node[red, font=\small] at (#1, 1.58, #2+1) {$X$};
	\node[red, font=\small] at (#1+1, 1.62, #2+1) {$X$};
	\node[red, font=\small] at (#1+0.36, 1, #2) {$X$};
	\node[red, font=\small] at (#1+0.5, 1, #2+1) {$X$};
	\node[red, font=\small] at (#1, 1, #2+0.5) {$X$};
	\node[red, font=\small] at (#1+1, 1, #2+0.5) {$X$};
	\node at (#1+0.36, 1.55, #2+0.35) {$c$};
}

\newcommand{\drawActopTenEM}[2]{%
	\fill[red, opacity=0.08] (#1,2,#2)--(#1+1,2,#2)--(#1+1,1,#2)--(#1+1,1,#2+1)--(#1,1,#2+1)--(#1,2,#2+1)--cycle;
	\fill[red, opacity=0.14] (#1-0.14, 2, #2+0.12)--(#1+0.14, 2, #2+0.12)--(#1+0.14, 2, #2+0.88)--(#1-0.14, 2, #2+0.88);
	\fill[red, opacity=0.14] (#1+1-0.14, 2, #2+0.12)--(#1+1+0.14, 2, #2+0.12)--(#1+1+0.14, 2, #2+0.88)--(#1+1-0.14, 2, #2+0.88);
	\node[red, font=\small] at (#1, 2, #2+0.5) {$X$};
	\node[red, font=\small] at (#1+1, 2, #2+0.5) {$X$};
	\node[red, font=\small] at (#1, 1.55, #2) {$X$};
	\node[red, font=\small] at (#1+1, 1.55, #2) {$X$};
	\node[red, font=\small] at (#1, 1.58, #2+1) {$X$};
	\node[red, font=\small] at (#1+1, 1.62, #2+1) {$X$};
	\node[red, font=\small] at (#1+0.36, 1, #2) {$X$};
	\node[red, font=\small] at (#1+0.5, 1, #2+1) {$X$};
	\node[red, font=\small] at (#1, 1, #2+0.5) {$X$};
	\node[red, font=\small] at (#1+1, 1, #2+0.5) {$X$};
	\node at (#1+0.36, 1.55, #2+0.35) {$c$};
}

\begin{document}

\title{Fracton topological holography}

\author{Yu-Tao Hu}\thanks{These authors contributed equally.}
\author{Jie-Yu Zhang}\thanks{These authors contributed equally.}
\author{Peng Ye\orcidicon{0000-0002-6251-677X}}
\email{yepeng5@mail.sysu.edu.cn}
\affiliation{Guangdong Provincial Key Laboratory of Magnetoelectric Physics and Devices, State Key Laboratory of Optoelectronic Materials and Technologies, and School of Physics, Sun Yat-sen University, Guangzhou, 510275, China}

\date{\today}

\begin{abstract}
	Topological Holography (TH) is a framework of representing $(d-1)$-dimensional quantum systems with a $d$-dimensional sandwich, consisting of a topological order bulk, a gapped top boundary and a dynamical bottom boundary.
	Dualities and non-invertible symmetries are represented by the change of top boundary in the TH framework.
	Existing studies have focused on sandwiches with liquid topological bulk, in this paper we examine some fracton bulks, where we call the framework Fracton Topological Holography (FTH).
	In particular, the X-cube FTH induces the transverse field plaquette Ising model (TFPIM) version Kramers--Wannier  duality, the duality between TFPIM and the two-dimensional $90^\circ$ quantum compass model, and the self-duality of quantum compass model.
	Haah's code FTH induces a duality of transverse-field generalized Ising models (TFGIMs).
	The two examples illustrate at least a part of type-I fracton (X-cube) and a part of type-II fracton (Haah's code) models can be used to construct FTH sandwiches.
	Moreover, we construct an explicit Sequential Quantum Circuit (SQC) switching the top boundary of X-cube sandwich, realizing the TFPIM version KW duality.
\end{abstract}

\maketitle
\tableofcontents

%======================================================================
\section{Introduction}\label{sec_introduction}
%======================================================================

Duality matches microscopically distinct systems that are actually alike.
Examples range from the electro-magnetic duality of source-free Maxwell theory~\cite{jackson_2021_classical_electrodynamics}, Kramers--Wannier duality~\cite{kramers_wannier_1941_part1}, and the Jordan--Wigner transformation~\cite{jordan_wigner_1928} to particle-vortex duality~\cite{peskin_1978_mandelstam_t_hooft,dasgupta_halperin_1981_lattice_superconductivity}, boson-fermion duality~\cite{wilczek_1982_magnetic_flux,polyakov_1988_fermi_bose,jain_1989_composite_fermion}, and, in a broader setting, AdS/CFT duality~\cite{maldacena_1999_large_n}.
Topological Holography (TH) / Symmetry Topological Field Theory (SymTFT) / Symmetry Topological Order (SymTO) formalism has emerged as a powerful tool for constructing, representing and understanding dualities and generalized symmetries~\cite{gaiotto_kapustin_seiberg_willett_2015_generalized_global_symmetries,brennan_hong_2023_introduction_generalized_global_symmetries,luo_wang_wang_2024_lecture_notes_generalized_symmetries,kong_wen_zheng_2015_boundary_bulk_relation,ji_wen_2020_categorical_symmetry,kong_lan_wen_zhang_zheng_2020_categorical_symmetry,lichtman_thorngren_lindner_stern_berg_2021_bulk_anyons,chatterjee_wen_2023_shadow,moradi2023topological_holography,freed_moore_teleman_2022_topological_symmetry,kaidi_ohmori_zheng_2023_symmetry_tfts,kaidi_nardoni_zafrir_zheng_2023_symmetry_tfts_anomalies,bhardwaj_schafer_nameki_2025_generalized_charges_part_ii,lin_shao_2023_bootstrapping_noninvertible_symmetries,choi_sanghavi_shao_zheng_2025_lattice_gauge_symmetries,chen_cui_haghighat_wang_2023_symtfts_duality_defects,jia_luo_tian_wang_zhang_2025_flavor_symmetry,lin_okada_seifnashri_tachikawa_2023_density_noninvertible,kong_zheng_2018_gapless_edges,kong_zheng_2020_gapless_edges_part_i,kong_zheng_2021_gapless_edges_part_ii,kong_wen_zheng_2022_1d_gapped_phases,kong_zheng_2022_quantum_liquids_i,kong_zheng_2024_quantum_liquids_ii,xu_zhang_2024_categorical_descriptions,luo_wang_bi_2025_mixed_state_topological_holography,cao_symmetry_2024,apruzzi_etal_2023_symtft_string,apruzzi_bedogna_mancani_2026_symtft_exotic,bedogna_mancani_2026_duality_defect}, and references therein.

In the TH/SymTFT/SymTO picture, a $(d-1)$-dimensional system is realized as a holographic sandwich with a $d$-dimensional topological bulk \cite{kong_wen_zheng_2015_boundary_bulk_relation,freed_moore_teleman_2022_topological_symmetry}. 
The gapped top boundary fixes the symmetry of the sandwich, the bulk labels symmetry sectors, and the bottom boundary carries the dynamics. 
Dualities are induced by changing the top boundary while keeping the bulk and bottom boundary fixed \cite{kong_lan_wen_zhang_zheng_2020_categorical_symmetry,chatterjee_wen_2023_gaplessness,huang_cheng_2025_topological_holography}. 
Microscopically, this can be realized in lattice models~\cite{you_bi_rasmussen_slagle_xu_2014_strange_correlator,aasen_mong_fendley_2016_ising_defects,aasen_fendley_mong_2020_dualities_degeneracies,vanhove_bal_williamson_bultinck_haegeman_verstraete_2018_strange_correlators,vanhove_lootens_tu_verstraete_2022_potts}, and the boundary-switch can be implemented via linear-depth local unitary sequential quantum circuit (SQC)~\cite{chen_dua_hermele_stephen_tantivasadakarn_vanhove_zhao_2024_sqc,dual_via_SQC}. 
Closely related versions of the topological-holography picture have also appeared in the mathematical-physics literature under the name ``Topological Wick rotation'' \cite{kong_zheng_2018_gapless_edges,kong_zheng_2020_gapless_edges_part_i,kong_zheng_2021_gapless_edges_part_ii}. 

However, this conventional TH framework fundamentally relies on the mobility of topological excitations and standard string-like Wilson lines/membranes.
When applying this paradigm to fracton topological orders~\cite{chamon_2005_quantum_glassiness,fontana_pereira_2023_boundary_chamon,haah_2011_no_strings,haah_2013_modules,vijay_haah_fu_2016_fracton,shirley_slagle_wang_chen_2018_general_manifolds,hyt_1_prepare_TD_model_via_SQC,li2020fracton,li_ye_2021_fracton_extended_gsd,li_2023_gERG,vijay_haah_fu_2015_new_kind,song_schonmeier_kromer_liu_viyuela_pollet_martin_delgado_2022_fracton_thresholds,ma_lake_chen_hermele_2017_coupled_layers,shirley_slagle_chen_2019_gauging_subsystem,prem_haah_nandkishore_2017_glassy_dynamics,dua_kim_cheng_williamson_2019_sorting_stabilizer,nandkishore_hermele_2019_fractons,bulmash_barkeshli_2019_gauging_fractons,prem_huang_song_hermele_2019_cage_net,slagle_2021_foliated_qft,zhou_li_yan_ye_meng_2022_dynamical_signature,zhu_chen_ye_trebst_2023_fracton_phase_transitions,canossa_pollet_martin_delgado_song_liu_2024_self_dual_fracton,li_zhou_ye_2024_boundary_toeplitz}, the traditional TH framework faces fundamental challenges.
Fracton phases are characterized by excitations that are strictly immobile or restricted to sub-dimensional regions, meaning the requisite string-like operators to move a single fracton simply do not exist.  

To bridge this fundamental gap, in this work we propose and formulate \textit{Fracton Topological Holography (FTH)}.
To achieve this, we introduce the concept of \textit{transport operators} for fracton excitations, serving as the fracton generalization of Wilson lines to connect boundaries and induce duality.
Together with the bottom boundary gauge operators, transport operators generate the low-energy-preserving operators of the sandwich.
Under each gapped top boundary, we identify these operators, together with the symmetry, relation, and twist data they carry, as operators of a $(d-1)$-dimensional boundary system; switching among top boundaries then induces dualities between the identified boundary systems.
An illustrative figure of FTH is shown in Fig.~\ref{fig:intro_fth}.
In the language of the TH/SymTFT sandwich, bulk topological operators running parallel to the boundaries act on the boundary theory as its symmetry and twist-sector operators, while those stretched between the two boundaries act as charged and disorder operators; FTH generalizes both, and the fracton bulk in addition hosts parallel operators that cannot be pushed by stabilizers along the vertical direction.

\begin{figure*}[htbp]
	\centering
	% Subfigure a: FTH Sandwich
		\centering
		\tdplotsetmaincoords{70}{110}
		\begin{tikzpicture}[tdplot_main_coords, scale=1.2]
			
			% Custom perspective command for text with reset cm to avoid double projection
			\def\persp#1#2#3{\tikz[baseline=(T.base)]{\node[reset cm, cm={#1}, transform shape, scale=1.2, inner sep=0pt, #2] (T) {#3};}}
			\def\drawsandwich#1#2#3#4#5#6{
				% #1 = Anyon text (e.g. \mathcal{A})
				% #2 = Code to draw Anyon shape in 3D
				% #3 = Code for bottom excitations in 3D
				% #4 = Arrow out angle
				% #5 = Pillar color
				% #6 = Bottom Boundary color
				
				% Back bulk faces (Faint Color)
				\draw[draw=none, fill=cyan!10, fill opacity=0.4] (0,0,0) -- (3,0,0) -- (3,0,2.5) -- (0,0,2.5) -- cycle;
				\draw[draw=none, fill=cyan!10, fill opacity=0.4] (0,0,0) -- (0,3,0) -- (0,3,2.5) -- (0,0,2.5) -- cycle;
				
				% Bottom Boundary Fill
				\fill[#6, opacity=0.8] (0,0,0) -- (3,0,0) -- (3,3,0) -- (0,3,0) -- cycle;
				% Bottom Boundary Edges (Invisible)
				\draw[dashed, thick, gray, draw opacity=0.6] (3,0,0) -- (0,0,0) -- (0,3,0);
				
				% Vertical bulk lines (Invisible)
				\draw[dashed, thick, gray, draw opacity=0.6] (0,0,0) -- (0,0,2.5);
				
				\node at (1.5, 1.5, 1.25) {\persp{0.940, -0.117, 0, 1, (0,0)}{font=\footnotesize\bfseries, text=black!70, align=center}{Transport\\Operator}};
				% \node at (1.5, 0, 1.25) {\persp{0.342, 0.321, 0, 1, (0,0)}{font=\large\bfseries, text=black!70, align=center}{Fracton\\Bulk}};
				\node at (2.6, 1.4, 0) {\persp{0.940, -0.117, 0.513, 0.482, (0,0)}{font=\normalsize\bfseries, text=black!70, anchor=west}{Excitations}};
				
				% Bottom Excitations
				#3
				
				% Transport Pillar (Semi-transparent)
				% Back faces
				\draw[fill=#5, fill opacity=0.15, draw=#5, draw opacity=0.3] (1.2,1.2,0) -- (1.2,1.8,0) -- (1.2,1.8,2.5) -- (1.2,1.2,2.5) -- cycle;
				\draw[fill=#5, fill opacity=0.15, draw=#5, draw opacity=0.3] (1.2,1.8,0) -- (1.8,1.8,0) -- (1.8,1.8,2.5) -- (1.2,1.8,2.5) -- cycle;
				% Front faces
				\draw[fill=#5, fill opacity=0.25, draw=#5, draw opacity=0.4] (1.8,1.8,0) -- (1.8,1.2,0) -- (1.8,1.2,2.5) -- (1.8,1.8,2.5) -- cycle;
				\draw[fill=#5, fill opacity=0.25, draw=#5, draw opacity=0.4] (1.8,1.2,0) -- (1.2,1.2,0) -- (1.2,1.2,2.5) -- (1.8,1.2,2.5) -- cycle;
				
				% Top Boundary
				\draw[fill=red!15, opacity=0.8, thick] (0,0,2.5) -- (3,0,2.5) -- (3,3,2.5) -- (0,3,2.5) -- cycle;
				
				% Bottom Boundary Edges (Visible)
				\draw[thick, draw opacity=0.8] (3,0,0) -- (3,3,0) -- (0,3,0);
				
				% Vertical bulk lines (Visible)
				\draw[thick, gray] (3,0,0) -- (3,0,2.5);
				\draw[thick, gray] (3,3,0) -- (3,3,2.5);
				\draw[thick, gray] (0,3,0) -- (0,3,2.5);
				
				% Anyon Shape and Label
				#2
				\node at (1.5, 1.5, 2.5) {\persp{0.940, -0.117, 0.342, 0.321, (0,0)}{font=\normalsize, align=center}{#1}};
			}
			
			% Macro for a small square excitation with perspective
			\def\drawexcite#1#2{
				\draw[fill=yellow, draw=orange, thick] (#1-0.1, #2-0.1, 0) -- (#1+0.1, #2-0.1, 0) -- (#1+0.1, #2+0.1, 0) -- (#1-0.1, #2+0.1, 0) -- cycle;
			}
			
			% --- Left Sandwich ---
			\begin{scope}[xshift=0cm]
				\drawsandwich{$\mathcal{A}$}{
					% Perspective Hexagon
					\draw[fill=white, draw=black, thick, join=round] 
					({1.5 + 0.55*cos(0)},   {1.5 + 0.55*sin(0)},   2.5) --
					({1.5 + 0.55*cos(60)},  {1.5 + 0.55*sin(60)},  2.5) --
					({1.5 + 0.55*cos(120)}, {1.5 + 0.55*sin(120)}, 2.5) --
					({1.5 + 0.55*cos(180)}, {1.5 + 0.55*sin(180)}, 2.5) --
					({1.5 + 0.55*cos(240)}, {1.5 + 0.55*sin(240)}, 2.5) --
					({1.5 + 0.55*cos(300)}, {1.5 + 0.55*sin(300)}, 2.5) -- cycle;
				}{
					% Excitations A: 4 edge centers
					\drawexcite{1.5}{0.9}
					\drawexcite{1.5}{2.1}
					\drawexcite{0.9}{1.5}
					\drawexcite{2.1}{1.5}
				}{-20}{Magenta}{blue!15}
				
				% Left sandwich labels
				\node at (-0.5, 1.2, 2.5) {\persp{0.940, -0.117, 0, 1, (0,0)}{font=\small\bfseries, anchor=west}{Condense Excitation $\mathcal{A}$}};

				\node at (3.6, 1.7, 0) {\persp{0.940, -0.117, 0, 1, (0,0)}{font=\small\bfseries, anchor=west}{Boundary Phase A}};
				
				\coordinate (C1) at (1.5, 1.5, 1.25);
				
			\end{scope}
			
			% --- Right Sandwich ---
			\begin{scope}[xshift=6cm]
				\drawsandwich{$\mathcal{B}$}{
					% Perspective Hexagon
					\draw[fill=white, draw=black, thick, join=round]
					({1.5 + 0.55*cos(0)},   {1.5 + 0.55*sin(0)},   2.5) --
					({1.5 + 0.55*cos(60)},  {1.5 + 0.55*sin(60)},  2.5) --
					({1.5 + 0.55*cos(120)}, {1.5 + 0.55*sin(120)}, 2.5) --
					({1.5 + 0.55*cos(180)}, {1.5 + 0.55*sin(180)}, 2.5) --
					({1.5 + 0.55*cos(240)}, {1.5 + 0.55*sin(240)}, 2.5) --
					({1.5 + 0.55*cos(300)}, {1.5 + 0.55*sin(300)}, 2.5) -- cycle;
				}{
					% Excitations B: 4 corners
					\drawexcite{1.0}{1.0}
					\drawexcite{1.0}{2.0}
					\drawexcite{2.0}{1.0}
					\drawexcite{2.0}{2.0}
				}{-20}{Cyan}{green!15}
				
				% Right sandwich labels
				\node at (-0.5, 1.2, 2.5) {\persp{0.940, -0.117, 0, 1, (0,0)}{font=\small\bfseries, anchor=west}{Condense Excitation $\mathcal{B}$}};
				\node at (3.6, 1.7, 0) {\persp{0.940, -0.117, 0, 1, (0,0)}{font=\small\bfseries, anchor=west}{Boundary Phase B}};
				\coordinate (C2) at (1.5, 1.5, 1.25);
			\end{scope}
			
			% Duality Arrow connecting the two bulk systems
			% Using 2D shifts from the center coordinates to ensure it's perfectly horizontal
			\draw[<->, thick, line width=1.5pt] ([xshift=2.4cm]C1) -- ([xshift=-2.4cm]C2) node[midway, above, font=\small\bfseries] {Duality};
			
		\end{tikzpicture}
	\caption{Overview of the Fracton Topological Holography (FTH) framework.
	The topological holography sandwich structure.
	A $d$-dimensional fracton bulk is terminated by a gapped top boundary and a dynamical bottom boundary.
	Transport operators $\mathcal{T}$ connect the two boundaries.
	Changing the top boundary condition (e.g., between $\mathcal{A}$ condensed and $\mathcal{B}$ condensed) induces an exact duality ($\tilde{H}_A \protect\dual \tilde{H}_B$) on the identified $(d-1)$-dimensional boundary system.}
	\label{fig:intro_fth}
\end{figure*}

As proof of concept, we apply this holographic construction to two representative models hosting type-I and type-II fracton orders~\cite{vijay_haah_fu_2016_fracton}.
First, we consider the X-cube model, a canonical type-I fracton model where individual fractons are strictly immobile, but their composites are mobile along sub-dimensional submanifolds~\cite{vijay_haah_fu_2016_fracton}.
Second, we test FTH on Haah's cubic code~\cite{haah_2011_no_strings}, a canonical type-II fracton model where all excitations (including composite excitations) are strictly immobile and no string-like operators exist.
We analyze Haah's cubic code using the translation-invariant stabilizer module formalism over Laurent polynomial rings, demonstrating that FTH is robust enough to capture the complex, scale-dependent features of at least a part of type-II fractons. 
Our results demonstrate that FTH provides a powerful tool for characterizing dualities involving fractal and subsystem symmetries, as well as boundary physics of fracton phases.

Type-I fracton refers to those fracton phases that host strictly immobile fractons together with excitations---either composites of fractons, or other elementary excitations such as lineons---that are mobile only within sub-dimensional submanifolds; equivalently, they admit string logical operators, although fractal operators may be present as well.
The X-cube model is a canonical type-I fracton model, with line-like subsystem symmetries and can be viewed as the coupled toric code layers~\cite{vijay_haah_fu_2016_fracton,ma_lake_chen_hermele_2017_coupled_layers,shirley_slagle_wang_chen_2018_general_manifolds}.
It can also be viewed as a special member of the broader TD stabilizer-code family, which contains the toric code and X-cube model as special cases and includes higher-dimensional fracton models with spatially extended excitations~\cite{hyt_1_prepare_TD_model_via_SQC,li2020fracton,li_ye_2021_fracton_extended_gsd,li_2023_gERG}.
Type-II fracton refers to those fracton phases whose excitations are immobile anyway, even when all composite excitations are considered.
Haah's cubic code is a canonical type-II fracton model, with no string logical operators and is convenient to study in terms of translation-invariant stabilizer modules over Laurent polynomial rings~\cite{haah_2011_no_strings,haah_2013_modules,williamson_2016_fractal_symmetries}.
Complementary holographic viewpoints on stabilizer codes and fracton orders have also appeared recently, including a holographic view of topological stabilizer codes with fracton examples and emergent boundary symmetries~\cite{schuster_tantivasadakarn_vishwanath_yao_2023_holographic_stabilizer}, a mixed-state holographic viewpoint on subdimensional entanglement and fracton orders~\cite{li_ye_2025_subdimensional_entanglement_entropy}, and a bulk--boundary correspondence between generalized Ising (GI) models and TO/type-I fracton bulks constructed from the GI and its dual~\cite{liu_ji_2023_noninvertible_anomalies}. 

For the X-cube FTH, we consider four gapped top boundaries, $(ee)$, $(mm)$, $(em)$, and $(me)$~\cite{bulmash_iadecola_2019_braiding_boundaries,luo_spieler_sun_karch_2022_xcube_boundary}, where $l_z$-lineons are condensed on the $(ee)$ top boundary and fractons are condensed on the $(mm)$ top boundary.
For each choice, we identify the low-energy-preserving operator algebra, the minimal low-energy effective Hamiltonian, and the subsystem symmetries and twist sectors encoded by bulk logical operators and top boundary defects.
Changing the top boundary from $(ee)$ to $(mm)$ gives an exact duality between the two identified boundary systems.
With a minimal low-energy effective Hamiltonian, the identified 2d boundary systems of the X-cube FTH under the $(ee)$ and $(mm)$ top boundaries are both transverse-field plaquette Ising models (TFPIMs); thus the duality is the TFPIM version of KW duality~\cite{xu_moore_2004_self_duality,TFPIM_KW_duality}---a self-duality, as shown in Fig.~\ref{fig_X_cube_4_boundaries_dualities}.
Before lifting the relations of plaquette Ising terms to twist DOFs (enlarging the Hilbert space), the identified TFPIMs under the two top boundaries have Hilbert spaces with different dimension, thus the duality is non-invertible, in agreement with the discussions in Refs.~\cite{cao_li_yamazaki_zheng_2023_subsystem_noninvertible,2024_TFPIM_non_invertible_operator}.
We construct the duality operator $\mathsf{D}_\square$ explicitly, and $\mathsf{D}_\square^\dagger\mathsf{D}_\square$ is proportional to the projector onto the $+1$ common eigenspace of the line-like symmetry indicators.
Adding defects on the top boundary realizes the twist of the identified form of the minimal low-energy effective Hamiltonian.  
By deleting specific top stabilizer generators, adding defects on the top boundary costs no energy anymore. 
Consequently, the low-energy subspace is enlarged.
The duality matches the generators of the two identified operator algebras while preserving their commutation relations, and after the enlargement the Hilbert spaces under the two top boundaries have the same dimension; the TFPIM version KW duality is therefore a unitary operator between the two enlarged Hilbert spaces, in agreement with the discussion in Ref.~\cite{2025_TFPIM_duality_unitary_in_enlarged_Hilbert_space}.
We also write down an explicit linear-depth local-unitary sequential quantum circuit (SQC) that implements the top-boundary change, making the fracton-holographic duality operational at the circuit level.

Under the mixed $(em)$ and $(me)$ top boundaries, the identified 2d boundary systems are both $90^\circ$ quantum compass models (QCMs)~\cite{kugel_khomskii_1982_jahn_teller,nussinov_vandenbrink_2015_compass_review}, with the $\tilde{Z}\tilde{Z}$ bonds along the $x$- and $y$-directions, respectively.
The QCM is known to be self-dual and dual to the TFPIM~\cite{nussinov_fradkin_2005_compass_dualities}.
Within the X-cube FTH, the four top boundaries of the same sandwich give a duality web among the TFPIM and the QCM, consisting of the TFPIM self-duality, the QCM self-duality, and the QCM--TFPIM duality, with the local operator dictionaries illustrated in Fig.~\ref{fig_X_cube_4_boundaries_dualities}.
The QCM self-duality can be realized invertibly by a diagonal reflection combined with a transversal Hadamard gate.
However, the QCM self-duality induced by switching the top boundary from $(em)$ to $(me)$ is a different one: it maps each local bond term to a local bond term labeled by the same vertex or plaquette, and it is non-invertible.
We construct both QCM duality operators explicitly; as in the TFPIM case, each of them absorbs a set of line-like symmetries, which is the origin of the non-invertibility.
Deleting specific top stabilizer generators again introduces twist DOFs; upon switching the top boundary from $(em)$ to $(me)$, the symmetry and twist operators exchange their identities, and the QCM self-duality becomes invertible.
The QCM--TFPIM duality, induced by switching the top boundary between $(em)$ and $(ee)$, behaves differently: it maps the line-like symmetries of the QCM along the $x$-direction to relations of the TFPIM, while those along the $y$-direction are mapped to the line-like symmetries of the TFPIM along the $y$-direction.

Haah's cubic code FTH serves as a type-II fracton stress test.
This example is important because Haah's code is not built from a simple foliation of two-dimensional topological orders and does not admit string logical operators \cite{haah_2011_no_strings,Haah_code_boundary_2024}.
We formulate Haah's code FTH using the Laurent-polynomial stabilizer formalism and the local boundary-gauge-operator framework \cite{haah_2013_modules,Chen_Yu_An_2024}.
Under the two natural top boundaries~\cite{Haah_code_boundary_2024}, which we call the $(Z)$ and $(X)$ top boundaries, the general forms of low-energy-preserving Pauli operators are constructed, and the generators are derived.  
Then, under each top boundary, the low-energy-preserving Pauli operators are identified as Pauli operators of the identified 2d boundary system with qubit DOFs. 
Consequently, the minimal low-energy effective Hamiltonian $\tilde H$ on the bottom boundary is identified as a transverse-field generalized Ising model (TFGIM) under both the $(Z)$ and $(X)$ top boundary [see the left and right sides of Fig.~\ref{fig_Haah_code_FTH_duality}, respectively], differing by an inversion of relative field-strength parameter and a spatial inversion.
By changing the top boundaries between $(Z)$ and $(X)$, a duality is induced, which swaps the generalized Ising term and the transverse field term [shown in Fig.~\ref{fig_Haah_code_FTH_duality}], as well as the fractal symmetries of the 2d Hamiltonian and relations of the generalized Ising terms.
As in the X-cube case, the duality operator absorbs these fractal symmetries and is therefore non-invertible.
Moreover, under the $(Z)$ top boundary, the number of global constraints of the identified 2d system grows with the thickness $L_z$ of the sandwich, a geometric information inherited from the type-II fracton bulk.

In all, FTH is a generalization of TH, enabling the TH framework to contain geometric information, inherited from the fracton bulk.  
With transport operators, switching the gapped top boundary of a single fracton sandwich systematically produces dualities between the identified boundary systems, together with their local operator dictionaries.
The X-cube and Haah's code examples show the effectiveness of this construction, for both type-I and type-II fracton bulks.
The bottom boundary Hamiltonian $H_{\text{bot}}$ [see Eqs.~(\ref{eq_X_cube_H_bot}) and (\ref{eq_Haah_H_bot})], kept fixed while the top boundary is switched, serves as the common pivot that facilitates the dualities.
The rest of the paper is organized as follows.
Section~\ref{sec_TH_review} reviews the KW duality of the TFIM and its TH representation (the 2d toric code TH).
Section~\ref{sec_X_cube_FTH} constructs the X-cube FTH: Sec.~\ref{subsec_X_cube_ee_mm} identifies the TFPIMs under the $(ee)$ and $(mm)$ top boundaries, derives the non-invertible TFPIM self-duality and its invertible version with twist DOFs, and constructs a linear-depth local unitary SQC that realizes the duality; Sec.~\ref{subsec_X_cube_em_me} identifies the QCMs under the $(em)$ and $(me)$ top boundaries and completes the duality web.
Section~\ref{sec_Haah_code_FTH} constructs Haah's code FTH and the induced duality.
Section~\ref{sec_summary_outlook} summarizes the results and discusses future directions.
Some technical details are collected in the appendices.

%======================================================================
\section{Review: TH, 2d toric code TH and KW duality}\label{sec_TH_review}
%======================================================================

In this section, we review the KW duality of transverse field Ising model (TFIM), the TH formalism in the example of 2d toric code (2dTC) TH, and how the TFIM, KW duality are represented by the 2dTC TH.

Consider a length-$L$ spin ring, with each vertex hosting a $\frac{1}{2}$-spin.
The TFIM refers to such a spin system with Hamiltonian
\begin{equation}
	H_{\text{TFIM}} = -\sum_{i=0}^{L-1} Z_i Z_{i+1} - g\sum_{i=0}^{L-1} X_i.
\end{equation}
The KW duality refers to the linear map that maps
\begin{equation}
	Z_i Z_{i+1} \mapsto X'_{i+1}, \qquad X_i \mapsto Z'_i Z'_{i+1},
\end{equation}
where $X'_i, Z'_i$ are Pauli operators of edge spins on the ring.
It is well-known that the KW duality is non-invertible.
A convenient way of seeing that is the KW duality maps the global symmetry indicator $\prod_{i=0}^{L-1} X_i$ to $1$.
More explicitly, we can write the KW duality as the following map on states:
\begin{equation}\label{eq_KW_explicit_form}
	\mathsf{D}_{\text{KW}} = 2^{-L/2}\sum_{\mathbf{s},\mathbf{t}} (-1)^{\sum_i t_{i+1}(s_i+s_{i+1})} |\mathbf{t}\rangle\langle\mathbf{s}|,
\end{equation}
where $\mathbf{s}, \mathbf{t}$ are Pauli $Z$ configurations of vertex spins and edge spins, respectively.
It can be checked that
\begin{align}
	\mathsf{D}_{\text{KW}} Z_i Z_{i+1} &= X'_{i+1} \mathsf{D}_{\text{KW}}, \notag\\
	\mathsf{D}_{\text{KW}} X_i &= Z'_i Z'_{i+1} \mathsf{D}_{\text{KW}}.\label{eq_KW_intertwining}
\end{align}
Also,
\begin{equation}\label{eq_KW_Gram}
	\mathsf{D}_{\text{KW}}^\dagger \mathsf{D}_{\text{KW}} = 1 + U, \qquad
	\mathsf{D}_{\text{KW}} \mathsf{D}_{\text{KW}}^\dagger = 1 + U',
\end{equation}
where $U\equiv\prod_i X_i$ and $U'\equiv\prod_i X'_i$ are the global symmetry indicators of vertex spin TFIM and edge spin TFIM, respectively.
For a non-zero state $|\psi\rangle$ with $U|\psi\rangle = -|\psi\rangle$, $\|\mathsf{D}_{\text{KW}}|\psi\rangle\|^2 = \langle\psi|\mathsf{D}_{\text{KW}}^\dagger\mathsf{D}_{\text{KW}}|\psi\rangle = \langle\psi|1+U|\psi\rangle = 0$, so $\mathsf{D}_{\text{KW}}|\psi\rangle = 0$, $\mathsf{D}_{\text{KW}}$ is non-invertible.

The reason $\mathsf{D}_{\text{KW}}$ is non-invertible is the two symmetry sectors $U=\pm 1$ are mapped to one ($U'=+1$).
If the twist sectors of the preimage side are considered, and an extra twist degree of freedom (DOF) is introduced to record the twist sectors, the KW duality becomes invertible.
This can be naturally represented in the TH formalism, by the 2dTC TH.

Consider a $L_x\times L_y$ square lattice with $x$-direction under periodic boundary condition (PBC) and $y$-direction under open boundary condition (OBC).
Keep the bottom boundary smooth and first consider the rough top boundary, as illustrated in Fig.~\ref{fig_2dTC_TH}(a).
The 2dTC TH is settled on such a lattice by arranging a qubit on each edge, and let the Hamiltonian of the $L_x\times L_y$ sandwich to be
\begin{equation}\label{eq_2dTC_TH_Hamiltonian}
	H_{\text{TC}} = -\eta\sum_a S_a + H_{\text{bot}},
\end{equation}
where $S_a$ takes through the set of local stabilizer generators, consisting of 2dTC bulk stabilizers and top boundary truncated three-leg plaquette terms [see Fig.~\ref{fig_2dTC_TH}(a)].
$\eta\gg \|H_{\text{bot}}\|$, so that $H_{\text{bot}}$ is the effective low-energy Hamiltonian of the system.
$H_{\text{bot}}$ is a function of local bottom boundary gauge operators\footnote{Here bottom boundary gauge operators refer to the operators commuting with bulk stabilizers but not generated by bulk stabilizers near the bottom boundary.}.
$H_{\text{bot}}$ could be a sum of mutually non-commuting operators, which is why the bottom boundary is called dynamical.
The low-energy subspace of the 2dTC TH is identified as the Hilbert space of a vertex spin system on a length-$L_x$ ring.
Taking $H_{\text{bot}}$ to be the sum of bottom boundary gauge generators,
\begin{equation}\label{eq_2dTC_H_bot}
	H_{\text{bot}} = -\sum_{e\subset\text{b.b.}} X_e - g\sum_{v\subset\text{b.b.}} B_v^{\text{bot}},
\end{equation}
the identified spin ring is a TFIM, with the dictionary shown in Fig.~\ref{fig_2dTC_TH}(a).
Here $B_v^{\text{bot}}$ denotes the bottom boundary truncated $B_v$ terms.
Similarly, the top boundary truncated $A_p, B_v$ terms are denoted as $A_p^{\text{top}}, B_v^{\text{top}}$, respectively.
On the rough top boundary, $e$ particles ($B_v=-1$) are condensed, while $m$ particles ($A_p=-1$) are not allowed in the low-energy subspace.
The operator identification is settled by requiring preserving multiplication and commutation relation.
See Appendix~\ref{details_2dTC_TH_Hilbert_space_structure_and_operator_identification} for the details of the low-energy subspace structure and the operator identification.

\begin{figure*}[htbp]
	\centering
	\begin{subfigure}{0.45\textwidth}
		\centering
		\begin{tikzpicture}
			\foreach \i in {0,1.2,2.4,3.6,4.8} {\draw (\i,0)--(\i,4.8);}
			\foreach \j in {0,1.2,2.4,3.6} {\draw (0,\j)--(6,\j);}
			\draw[dashed] (6,0)--(6,4.8);
			
			\fill[red, opacity=0.1] (2.55,0.2)--(2.55,-0.2)--(3.45,-0.2)--(3.45,0.2);
			\node[red] at (3,0) {$X$};
			\node[red] at (3,-0.45) {$\sim\tilde{Z}_i\tilde{Z}_{i+1}$};
			
			\fill[ForestGreen, opacity=0.1] (3.9,0.2)--(3.9,-0.2)--(5.7,-0.2)--(5.7,0.2);
			\fill[ForestGreen, opacity=0.1] (4.6,-0.2)--(5,-0.2)--(5,0.9)--(4.6,0.9);
			\node[ForestGreen] at (4.2,0) {$Z$};
			\node[ForestGreen] at (5.4,0) {$Z$};
			\node[ForestGreen] at (4.8,0.6) {$Z$};
			\node[ForestGreen] at (4.8,-0.45) {$\sim\tilde{X}_i$};
			
			\fill[ForestGreen, opacity=0.1] (-0.2,1.6)--(-0.2,2)--(6,2)--(6,1.6);
			\foreach \i in {0,1.2,2.4,3.6,4.8} {\node[ForestGreen] at (\i,1.8) {$Z$};}
			\node[ForestGreen] at (6.3,1.8) {$\mathcal{Z}_x$};
			\node[ForestGreen] at (4.15,1.4) {$\sim\prod_i\tilde{X}_i$};
			
			\fill[red, opacity=0.1] (-0.2,2.2)--(-0.2,2.6)--(6,2.6)--(6,2.2);
			\foreach \i in {0.6,1.8,3,4.2,5.4} {\node[red] at (\i,2.4) {$X$};}
			\node[red] at (6.3,2.4) {$\mathcal{X}_x$};
			\node[red] at (4.75,2.8) {$\sim\tilde{Z}_{\text{twist}}$};
			
			\fill[red, opacity=0.1] (-0.2,4.8)--(0.2,4.8)--(0.2,0)--(-0.2,0);
			\foreach \j in {0.6,1.8,3,4.2} {\node[red] at (0,\j) {$X$};}
			\node[red] at (0,-0.3) {$\tilde{Z}_0$};
			
			\fill[ForestGreen, opacity=0.1] (0.4,-0.2)--(0.8,-0.2)--(0.8,4.2)--(0.4,4.2);
			\foreach \j in {0,1.2,2.4,3.6} {\node[ForestGreen] at (0.6,\j) {$Z$};}
			\node[ForestGreen] at (0.8,-0.45) {$\tilde{X}_{\text{twist}}$};
			
			\fill[ForestGreen, opacity=0.1] (2.7,3.4)--(2.7,3.8)--(4.5,3.8)--(4.5,3.4);
			\fill[ForestGreen, opacity=0.1] (3.4,4.5)--(3.8,4.5)--(3.8,2.7)--(3.4,2.7);
			\foreach \i/\j in {3/3.6, 4.2/3.6, 3.6/3, 3.6/4.2} {\node[ForestGreen] at (\i,\j) {$Z$};}
			\node[ForestGreen] at (3.2,4) {$B_v$};
			
			\foreach \i in {4.8,6} {\fill[red, opacity=0.1] (\i-0.2,4.65)--(\i+0.2,4.65)--(\i+0.2,3.75)--(\i-0.2,3.75);}
			\foreach \j in {3.6} {\fill[red, opacity=0.1] (4.95,\j-0.2)--(4.95,\j+0.2)--(5.85,\j+0.2)--(5.85,\j-0.2);}
			\foreach \i/\j in {5.4/3.6, 4.8/4.2, 6/4.2} {\node[red] at (\i,\j) {$X$};}
			\node[red] at (5.4,4.25) {$A_p^{\text{top}}$};
			
			\draw[red, line width=1.5pt] (0.6,4.2) circle (0.15);
			\node[red] at (0.6,4.65) {$A_{p_0}^{\text{top}}$};
			
			\node at (0,-1) {0};
			\node at (1.2,-1) {1};
			\node at (2.4,-1) {2};
			\node at (3.6,-1) {3};
			\node at (4.8,-1) {4};
			\node at (6,-1) {5};
			\draw[->] (0,-1.25)--(6.2,-1.25);
			\node at (6.4,-1.25) {$x$};
			
			\node at (3,-2) {(a)};
		\end{tikzpicture}
	\end{subfigure}
	\begin{subfigure}{0.45\textwidth}
		\centering
		\begin{tikzpicture}
			\foreach \i in {0,1.2,2.4,3.6,4.8} {\draw (\i,0)--(\i,4.8);}
			\foreach \j in {0,1.2,2.4,3.6,4.8} {\draw (0,\j)--(6,\j);}
			\draw[dashed] (6,0)--(6,4.8);
			
			\fill[red, opacity=0.1] (2.55,0.2)--(2.55,-0.2)--(3.45,-0.2)--(3.45,0.2);
			\node[red] at (3,0) {$X$};
			\node[red] at (3,-0.45) {$\sim\tilde{X}'_i$};
			
			\fill[ForestGreen, opacity=0.1] (3.9,0.2)--(3.9,-0.2)--(5.7,-0.2)--(5.7,0.2);
			\fill[ForestGreen, opacity=0.1] (4.6,-0.2)--(5,-0.2)--(5,0.9)--(4.6,0.9);
			\node[ForestGreen] at (4.2,0) {$Z$};
			\node[ForestGreen] at (5.4,0) {$Z$};
			\node[ForestGreen] at (4.8,0.6) {$Z$};
			\node[ForestGreen] at (4.8,-0.45) {$\sim\tilde{Z}'_i\tilde{Z}'_{i+1}$};
			
			\fill[ForestGreen, opacity=0.1] (-0.2,1.6)--(-0.2,2)--(6,2)--(6,1.6);
			\foreach \i in {0,1.2,2.4,3.6,4.8} {\node[ForestGreen] at (\i,1.8) {$Z$};}
			\node[ForestGreen] at (6.3,1.8) {$\mathcal{Z}_x$};
			\node[ForestGreen] at (4.2,1.4) {$\sim\tilde{Z}'_{\text{twist}}$};
			
			\fill[red, opacity=0.1] (-0.2,2.2)--(-0.2,2.6)--(6,2.6)--(6,2.2);
			\foreach \i in {0.6,1.8,3,4.2,5.4} {\node[red] at (\i,2.4) {$X$};}
			\node[red] at (6.3,2.4) {$\mathcal{X}_x$};
			\node[red] at (4.8,2.8) {$\sim\prod_i\tilde{X}'_i$};
			
			\draw[ForestGreen, line width=1.2pt] (0,4.8) circle (0.15);
			\fill[red, opacity=0.1] (-0.2,4.8)--(0.2,4.8)--(0.2,0)--(-0.2,0);
			\foreach \j in {0.6,1.8,3,4.2} {\node[red] at (0,\j) {$X$};}
			\node[red] at (0,-0.3) {$\tilde{X}'_{\text{twist}}$};
			
			\fill[ForestGreen, opacity=0.1] (0.4,-0.2)--(0.8,-0.2)--(0.8,5)--(0.4,5);
			\foreach \j in {0,1.2,2.4,3.6,4.8} {\node[ForestGreen] at (0.6,\j) {$Z$};}
			\node[ForestGreen] at (0.8,-0.45) {$\tilde{Z}'_i$};
			
			\fill[ForestGreen, opacity=0.1] (2.7,4.6)--(2.7,5)--(4.5,5)--(4.5,4.6);
			\fill[ForestGreen, opacity=0.1] (3.4,5)--(3.8,5)--(3.8,3.9)--(3.4,3.9);
			\foreach \i/\j in {3/4.8, 4.2/4.8, 3.6/4.2} {\node[ForestGreen] at (\i,\j) {$Z$};}
			\node[ForestGreen] at (3.6,5.25) {$B_v^{\text{top}}$};
			
			\foreach \i in {4.8,6} {\fill[red, opacity=0.1] (\i-0.2,4.65)--(\i+0.2,4.65)--(\i+0.2,3.75)--(\i-0.2,3.75);}
			\foreach \j in {3.6,4.8} {\fill[red, opacity=0.1] (4.95,\j-0.2)--(4.95,\j+0.2)--(5.85,\j+0.2)--(5.85,\j-0.2);}
			\foreach \i/\j in {5.4/3.6, 5.4/4.8, 4.8/4.2, 6/4.2} {\node[red] at (\i,\j) {$X$};}
			\node[red] at (5.4,4.2) {$A_p$};
			
			\node[ForestGreen] at (0,5.25) {$B_{v_0}^{\text{top}}$};
			
			\node at (0,-1) {0};
			\node at (1.2,-1) {1};
			\node at (2.4,-1) {2};
			\node at (3.6,-1) {3};
			\node at (4.8,-1) {4};
			\node at (6,-1) {5};
			\draw[->] (0,-1.25)--(6.2,-1.25);
			\node at (6.4,-1.25) {$x$};
			
			\draw[->] (7.3,0)--(7.3,4.5);
			\node at (7.3,4.7) {$y$};
			\node at (7.3,-0.2) {$0$};
			
			\node at (3,-2) {(b)};
		\end{tikzpicture}
	\end{subfigure}
	\caption{2dTC TH, and the operator identification/dictionary from 2dTC TH low-energy preserving operators to 1d TFIM.
	(a) Under the rough top boundary: the stabilizer generators are complete bulk $A_p, B_v$ terms and truncated $A_p^{\text{top}}$ terms.
	The bottom boundary Pauli $X$ and $B_v^{\text{bot}}$ terms are identified as the Ising terms $\tilde{Z}_i\tilde{Z}_{i+1}$ and transverse field terms $\tilde{X}_i$ of the TFIM, respectively.
	The dual $Z$ Wilson loop along $x$-direction is then naturally identified as $\prod_i\tilde{X}_i$, the symmetry indicator of TFIM; the symmetry toggler $\tilde{Z}_i$, mapping between $\prod_i\tilde{X}_i=\pm 1$ sectors, is identified from the $X$ Wilson line connecting top and bottom boundaries.
	The $X$ Wilson loop along $x$-direction is identified as the twist indicator $\tilde{Z}_{\text{twist}}$; its anti-commuting pair $\tilde{X}_{\text{twist}}$ is identified from the dual $Z$ Wilson line connecting top and bottom boundaries.
	(b) Under the smooth top boundary: the stabilizer generators are bulk $A_p, B_v$ terms and $B_v^{\text{top}}$ terms.
	The $X$ Wilson loop along $x$-direction, dual $Z$ Wilson line connecting top and bottom boundaries are identified as symmetry indicator $\prod_i\tilde{X}'_i$ and symmetry toggler $\tilde{Z}'_i$, respectively.
	The dual $Z$ Wilson loop along $x$-direction and $X$ Wilson line connecting top and bottom boundaries are identified as twist indicator $\tilde{Z}'_{\text{twist}}$ and twist toggler $\tilde{X}'_{\text{twist}}$, respectively.}
	\label{fig_2dTC_TH}
\end{figure*}

As stated before, by adding an extra twist DOF, the KW duality can be made invertible.
The addition of the extra twist DOF can be concretely represented by deleting a specific top boundary stabilizer, i.e., allowing putting an $m$ particle on the top boundary, in the 2dTC TH representation of 1d TFIM, see Fig.~\ref{fig_2dTC_TH}(a) for illustration.
The enlarged Hilbert space with twist DOF has the dimension $2^{L_x+1}$, which can be divided into four symmetry and twist sectors,
\begin{equation}
	\Big(\prod_i\tilde{X}_i, \tilde{Z}_{\text{twist}}\Big) = (1,1),\ (1,-1),\ (-1,1),\ (-1,-1).
\end{equation}
When the top boundary is smooth, the low-energy effective theory is also a TFIM, but with the identity of Ising term and transverse field term exchanged, and the identity of symmetry operators and twist operators exchanged [see Fig.~\ref{fig_2dTC_TH}(b) for the stabilizer setting and operator identification].
So, the change of top boundary from rough to smooth realizes the KW duality.
A linear-depth local unitary SQC can be constructed to change the top boundary from rough to smooth (see Ref.~\cite{dual_via_SQC}), in agreement with the theory that SQC maps between gapped phases ($g>1$ paramagnetic phase and $g<1$ ferromagnetic phase).

%======================================================================
\section{X-cube FTH and TFPIM, quantum compass model, dualities}\label{sec_X_cube_FTH}
%======================================================================

In this section, we introduce X-cube FTH, illustrate how it represents the TFPIMs and $90^\circ$ quantum compass models under four different gapped top boundaries, and show the dualities it induces.
We also give a linear-depth SQC that changes the top boundary from rough to smooth, which realizes the TFPIM version KW duality.

The X-cube model is usually defined on a cubic lattice, with each edge hosting a qubit, and the Hamiltonian being the sum of cube terms and cross terms [see Fig.~\ref{X-cube-stabilizers} for definition],
\begin{equation}
	H_{\text{X-cube}} = -\sum_c A_c - \sum_v \sum_{l=xy,yz,zx} B_{v,l}.
\end{equation}

\begin{figure}[htbp]
	\centering
	\begin{tikzpicture}[
		cube3d/.style={x={(1cm,0cm)}, y={(0cm,1cm)}, z={(-0.48cm,-0.31cm)}},
		latline/.style={gray!55, line width=0.6pt, line cap=round},
		xedge/.style={red, line width=1.35pt, line cap=round, line join=round},
		xedgeback/.style={red!75, line width=1.1pt, densely dashed, line cap=round},
		zedge/.style={ForestGreen, line width=1.45pt, line cap=round},
		xlab/.style={red, font=\footnotesize, inner sep=1pt},
		zlab/.style={ForestGreen, font=\footnotesize, inner sep=1pt},
		optitle/.style={font=\small, inner sep=1pt}
	]

		% ==================== Coordinate triad (x, y, z) ====================
		\begin{scope}[shift={(-0.74, 2.10)}, cube3d]
			\draw[->, blue, line width=0.95pt] (0,0,0) -- (0,0,0.68)
				node[blue, font=\footnotesize, below left=-0.5pt] {$x$};
			\draw[->, blue, line width=0.95pt] (0,0,0) -- (0.58,0,0)
				node[blue, font=\footnotesize, right=1pt] {$y$};
			\draw[->, blue, line width=0.95pt] (0,0,0) -- (0,0.58,0)
				node[blue, font=\footnotesize, above=1pt] {$z$};
		\end{scope}

		% ==================== (a) Cube stabilizer A_c ====================
		\begin{scope}[shift={(0.32, 0.18)}, cube3d, scale=1.62]
			% Soft cube shading
			\fill[red, opacity=0.06] (0,1,0)--(1,1,0)--(1,0,0)--(1,0,1)--(0,0,1)--(0,1,1)--cycle;
			\fill[red, opacity=0.04] (0,1,0)--(1,1,0)--(1,1,1)--(0,1,1)--cycle;

			% 3 hidden back edges of cube c
			\draw[xedgeback] (0,0,0)--(1,0,0);
			\draw[xedgeback] (0,0,0)--(0,1,0);
			\draw[xedgeback] (0,0,0)--(0,0,1);

			% 9 visible front/outer edges of cube c
			\draw[xedge]
				(0,1,0)--(1,1,0)
				(0,0,1)--(1,0,1)
				(0,1,1)--(1,1,1)
				(1,0,0)--(1,1,0)
				(0,0,1)--(0,1,1)
				(1,0,1)--(1,1,1)
				(1,0,0)--(1,0,1)
				(0,1,0)--(0,1,1)
				(1,1,0)--(1,1,1);

			% Cube center label and operator title
			\node[red!85!black, font=\small] at (0.52, 0.45, 0.48) {$c$};
			\node[red, optitle] at (0.50, 1.38, 0) {$A_c$};

			% 12 Pauli X labels on the 12 edges
			% 4 y-edges
			\node[xlab, above=1.5pt] at (0.50, 1, 0) {$X$};
			\node[xlab, below=1.5pt] at (0.50, 0, 1) {$X$};
			\node[xlab, above=1.2pt] at (0.68, 1, 1) {$X$};
			\node[xlab, above=1.2pt] at (0.27, 0, 0) {$X$};
			% 4 z-edges
			\node[xlab, left=1.5pt]  at (0, 0.48, 1) {$X$};
			\node[xlab, right=1.5pt] at (1, 0.50, 0) {$X$};
			\node[xlab, left=1.5pt]  at (0, 0.43, 0) {$X$};
			\node[xlab, right=1.5pt] at (1, 0.56, 1) {$X$};
			% 4 x-edges
			\node[xlab, above left=1.2pt]  at (0, 1, 0.52) {$X$};
			\node[xlab, below right=1.2pt] at (1, 0, 0.52) {$X$};
			\node[xlab] at (0.77, 0.70, 0)   {$X$};
			\node[xlab] at (-0.23, -0.02, 0) {$X$};
		\end{scope}

		% ==================== (b) Three Cross Stabilizers B_{v,l} ====================
		% --- 1. B_{v,xy} (top-left of panel b) ---
		\begin{scope}[shift={(3.72, 2.06)}, cube3d, scale=0.76]
			\draw[latline] (0,-0.92,0)--(0,0,0);
			\fill[ForestGreen, opacity=0.10, draw=ForestGreen!35, line width=0.3pt]
				(-0.82,0,-0.92)--(0.82,0,-0.92)--(0.82,0,0.92)--(-0.82,0,0.92)--cycle;
			\draw[zedge] (-0.98,0,0)--(0.98,0,0) (0,0,-1.12)--(0,0,1.12);
			\draw[latline] (0,0,0)--(0,0.92,0);
			\fill[blue] (0,0,0) circle (1.7pt);
			\node[font=\footnotesize, inner sep=1pt] at (-0.23, 0.22, 0) {$v$};
			\node[zlab, above=1.5pt] at (-0.72,0,0) {$Z$};
			\node[zlab, below=1.5pt] at (0.72,0,0)  {$Z$};
			\node[zlab] at (0.35, 0.50, 0)   {$Z$};
			\node[zlab] at (-0.35, -0.50, 0) {$Z$};
			\node[ForestGreen, optitle] at (0, 1.26, 0) {$B_{v,xy}$};
		\end{scope}

		% --- 2. B_{v,yz} (top-right of panel b) ---
		\begin{scope}[shift={(5.96, 2.06)}, cube3d, scale=0.76]
			\draw[latline] (0,0,-1.05)--(0,0,0);
			\fill[ForestGreen, opacity=0.10, draw=ForestGreen!35, line width=0.3pt]
				(-0.80,-0.80,0)--(0.80,-0.80,0)--(0.80,0.80,0)--(-0.80,0.80,0)--cycle;
			\draw[zedge] (-0.98,0,0)--(0.98,0,0) (0,-0.98,0)--(0,0.98,0);
			\draw[latline] (0,0,0)--(0,0,1.05);
			\fill[blue] (0,0,0) circle (1.7pt);
			\node[font=\footnotesize, inner sep=1pt] at (0.23, -0.22, 0) {$v$};
			\node[zlab, above=1.5pt] at (-0.65,0,0) {$Z$};
			\node[zlab, below=1.5pt] at (0.67,0,0)  {$Z$};
			\node[zlab, left=1.5pt]  at (0,0.65,0)  {$Z$};
			\node[zlab, right=1.5pt] at (0,-0.67,0) {$Z$};
			\node[ForestGreen, optitle] at (0, 1.26, 0) {$B_{v,yz}$};
		\end{scope}

		% --- 3. B_{v,zx} (bottom-center of panel b) ---
		\begin{scope}[shift={(4.84, 0.00)}, cube3d, scale=0.76]
			\draw[latline] (-0.95,0,0)--(0,0,0);
			\fill[ForestGreen, opacity=0.10, draw=ForestGreen!35, line width=0.3pt]
				(0,-0.80,-0.92)--(0,0.80,-0.92)--(0,0.80,0.92)--(0,-0.80,0.92)--cycle;
			\draw[zedge] (0,-0.98,0)--(0,0.98,0) (0,0,-1.12)--(0,0,1.12);
			\draw[latline] (0,0,0)--(0.95,0,0);
			\fill[blue] (0,0,0) circle (1.7pt);
			\node[font=\footnotesize, inner sep=1pt] at (-0.23, 0.22, 0) {$v$};
			\node[zlab, left=1.5pt]  at (0,0.66,0)  {$Z$};
			\node[zlab, right=1.5pt] at (0,-0.67,0) {$Z$};
			\node[zlab] at (0.35, 0.50, 0)   {$Z$};
			\node[zlab] at (-0.35, -0.50, 0) {$Z$};
			\node[ForestGreen, optitle] at (0, 1.32, 0) {$B_{v,zx}$};
		\end{scope}

		% ==================== Subfigure labels (a) and (b) ====================
		\node[font=\small] at (0.74, -1.08) {(a)};
		\node[font=\small] at (4.84, -1.08) {(b)};

	\end{tikzpicture}
	\caption{Stabilizer generators of the X-cube model.
	(a) Cube term $A_c = \prod_{e\subset c} X_e$, the product of Pauli $X_e$ on the 12 edges of cube $c$.
	(b) Cross terms $B_{v,xy}$, $B_{v,yz}$, and $B_{v,zx}$, the products of the four Pauli $Z_e$ on edges attached to vertex $v$ in the corresponding plane.
	One of $B_{v,xy},B_{v,yz},B_{v,zx}$ is redundant since $B_{v,xy}B_{v,yz}B_{v,zx}=1$.}
	\label{X-cube-stabilizers}
\end{figure}

The X-cube model has two kinds of elementary topological excitations, fracton and lineon [see illustration in Fig.~\ref{fig_X_cube_excitation_content}].
The representative syndrome of fracton is $A_c=-1$.
A pair of nearby fractons in a straight line form a planon, which can be moved by dual $Z$ Wilson lines in a plane.
There are three types of lineons, $l_x, l_y, l_z$, which are moveable along $x, y, z$-directions by $X$ Wilson lines in the specific lines, respectively.
The representative syndrome of $l_x$ lineon at $v$ is $B_{v,zx}=B_{v,xy}=-1$.
Similarly, the representative syndromes of $l_y, l_z$ are $B_{v,xy}=B_{v,yz}=-1$ and $B_{v,zx}=B_{v,yz}=-1$.
An $l_x$ lineon and an $l_y$ lineon fuse to an $l_z$ lineon; similarly, $l_y\times l_z = l_x$, $l_z\times l_x = l_y$.
Excitation condensation on the top boundary lies in the core of constructing FTH and dualities.

\begin{figure}[htbp]
	\centering
	\begin{tikzpicture}[
		cube3d/.style={x={(1cm,0cm)}, y={(0cm,1cm)}, z={(-0.43cm,-0.32cm)}},
		latline/.style={gray!55, line width=0.6pt, line cap=round},
		xedge/.style={red, line width=1.45pt, line cap=round},
		zedge/.style={ForestGreen, line width=1.45pt, line cap=round},
		xlab/.style={red, font=\scriptsize, inner sep=0pt},
		zlab/.style={ForestGreen, font=\scriptsize, inner sep=0pt},
		fracmark/.style={fill=red!88, draw=darkred, line width=0.55pt},
		linemark/.style={fill=ForestGreen!90, draw=ForestGreen!35!black, line width=0.6pt}
	]

		% ==================== Shared coordinate triad (x, y, z) ====================
		\begin{scope}[shift={(3.66, 2.58)}, cube3d]
			\draw[->, blue, line width=0.95pt] (0,0,0) -- (0,0,0.65)
				node[blue, font=\footnotesize, below left=-0.5pt] {$x$};
			\draw[->, blue, line width=0.95pt] (0,0,0) -- (0.52,0,0)
				node[blue, font=\footnotesize, right=1pt] {$y$};
			\draw[->, blue, line width=0.95pt] (0,0,0) -- (0,0.52,0)
				node[blue, font=\footnotesize, above=1pt] {$z$};
		\end{scope}

		% ==================== (a) Fracton excitations (A_c = -1) ====================
		\begin{scope}[shift={(0.00, 0.00)}, cube3d, scale=0.97]
			% Back yz grid (k = 0)
			\foreach \j in {0,1,2,3} {\draw[latline] (0,\j,0)--(3,\j,0);}
			\foreach \i in {0,1,2,3} {\draw[latline] (\i,0,0)--(\i,3,0);}

			% Connecting x-edges (k: 0 -> 1)
			\foreach \i in {0,1,2,3} {
				\foreach \j in {0,1,2,3} {
					\draw[latline] (\i,\j,0)--(\i,\j,1);
				}
			}

			% Subtle red 3D shading on the 4 excited corner cubes
			\foreach \ic/\jc in {0/0, 2/0, 0/2, 2/2} {
				\fill[red, opacity=0.14] (\ic,\jc+1,0)--(\ic+1,\jc+1,0)--(\ic+1,\jc+1,1)--(\ic,\jc+1,1)--cycle;
				\fill[red, opacity=0.10] (\ic,\jc,1)--(\ic+1,\jc,1)--(\ic+1,\jc+1,1)--(\ic,\jc+1,1)--cycle;
				\fill[red, opacity=0.07] (\ic+1,\jc,0)--(\ic+1,\jc+1,0)--(\ic+1,\jc+1,1)--(\ic+1,\jc,1)--cycle;
			}

			% Front yz grid (k = 1)
			\foreach \j in {0,1,2,3} {\draw[latline] (0,\j,1)--(3,\j,1);}
			\foreach \i in {0,1,2,3} {\draw[latline] (\i,0,1)--(\i,3,1);}

			% Visible red outline on the 4 excited corner cubes
			\foreach \ic/\jc in {0/0, 2/0, 0/2, 2/2} {
				\draw[red!65, line width=0.5pt, line join=round]
					(\ic,\jc+1,0)--(\ic+1,\jc+1,0)--(\ic+1,\jc+1,1)--(\ic,\jc+1,1)--cycle
					(\ic,\jc,1)--(\ic+1,\jc,1)--(\ic+1,\jc+1,1)
					(\ic,\jc,1)--(\ic,\jc+1,1)
					(\ic+1,\jc,0)--(\ic+1,\jc+1,0)
					(\ic+1,\jc,0)--(\ic+1,\jc,1);
			}

			% 2x2 array of Pauli Z operators on x-edges
			\foreach \i in {1,2} {
				\foreach \j in {1,2} {
					\draw[zedge] (\i,\j,0)--(\i,\j,1);
					\node[zlab] at (\i-0.106, \j+0.07, 0.5) {$Z$};
				}
			}

			% 4 corner fracton markers at cube centers
			\foreach \ic/\jc in {0.5/0.5, 2.5/0.5, 0.5/2.5, 2.5/2.5} {
				\filldraw[fracmark] (\ic, \jc, 0.5) circle (2.2pt);
			}

			% Syndrome label pointing to top-left fracton
			\node[red!85!black, font=\footnotesize, inner sep=1pt] (aclab) at (0.52, 3.24, 0) {$A_c=-1$};
			\draw[->, red!75!black, line width=0.55pt, >={Stealth[length=1.4mm,width=1.0mm]}, shorten >=1.6pt]
				(aclab.south) -- (0.5, 2.5, 0.5);
		\end{scope}

		% ==================== (b) Lineon excitations (l_x) ====================
		\begin{scope}[shift={(5.14, 0.36)}, cube3d, scale=1.22]
			% 2x2x2 cubic lattice
			\foreach \k in {0,1,2} {
				\foreach \j in {0,1,2} {\draw[latline] (0,\j,\k)--(2,\j,\k);}
				\foreach \i in {0,1,2} {\draw[latline] (\i,0,\k)--(\i,2,\k);}
			}
			\foreach \i in {0,1,2} {
				\foreach \j in {0,1,2} {\draw[latline] (\i,\j,0)--(\i,\j,2);}
			}

			% Background-only white halo for X labels (keeps bold red X string uncut)
			\node[font=\scriptsize, text=white, fill=white, inner sep=0.45pt] at (0.885, 1.060, 0.5) {$X$};
			\node[font=\scriptsize, text=white, fill=white, inner sep=0.45pt] at (0.885, 1.060, 1.5) {$X$};

			% Pauli X string on two consecutive x-edges along (y=1, z=1)
			\draw[xedge] (1,1,0)--(1,1,1) (1,1,1)--(1,1,2);

			% Pauli X labels cleanly above-left of the two consecutive x-edges
			\node[xlab] at (0.885, 1.060, 0.5) {$X$};
			\node[xlab] at (0.885, 1.060, 1.5) {$X$};

			% Unexcited middle vertex (hollow) vs two excited end vertices (l_x)
			\filldraw[fill=white, draw=gray!80, line width=0.55pt] (1,1,1) circle (1.55pt);
			\filldraw[linemark] (1,1,0) circle (2.4pt);
			\filldraw[linemark] (1,1,2) circle (2.4pt);

			% Lineon labels in open windows next to the two end vertices
			\node[ForestGreen!80!black, font=\footnotesize, inner sep=1pt] at (0.78, 1.19, 0) {$l_x$};
			\node[ForestGreen!80!black, font=\footnotesize, inner sep=1pt] at (1.19, 0.78, 2) {$l_x$};
		\end{scope}

		% ==================== Subfigure labels (a) and (b) ====================
		\node[font=\small] at (1.24, -0.70) {(a)};
		\node[font=\small] at (5.84, -0.70) {(b)};

	\end{tikzpicture}
	\caption{Elementary excitations of the X-cube model.
	(a) Pauli $Z$ on a dual lattice square create four fractons ($A_c=-1$) on the four corners (the shaded cubes).
	(b) Pauli $X$ on a line along $x$-direction create a pair of $l_x$ lineons ($B_{v,zx}=B_{v,xy}=-1$) at the two ends of the line. Similarly, $y,z$-lines of Pauli $X$ create $l_y,l_z$ lineons on the two ends.}
	\label{fig_X_cube_excitation_content}
\end{figure}

Next, we put the X-cube model on a sandwich and construct the X-cube FTH.
Consider an $L_x\times L_y\times L_z$ cubic lattice, with $x,y$-directions under PBC, $z$-direction under OBC.
Keep the bottom boundary to be smooth, and we consider the top boundary to be $(ee), (mm), (em), (me)$ one-by-one (Refs.~\cite{bulmash_iadecola_2019_braiding_boundaries,luo_spieler_sun_karch_2022_xcube_boundary}).

\subsection{Under \texorpdfstring{$(ee)$}{(ee)} and \texorpdfstring{$(mm)$}{(mm)} top boundaries, TFPIM}\label{subsec_X_cube_ee_mm}

\begin{figure*}[!t]
	\centering
	\begin{subfigure}{0.48\textwidth}
		\centering
		\begin{tikzpicture}[baseline=4ex, scale=1.14]
			\drawroughtoplatticegray
			
			% (3) T_{l_z,z} at v = (0,0,1)
			\fill[red, opacity=0.1] (-0.2,0,1)--(0.2,0,1)--(0.2,2,1)--(-0.2,2,1);
			\foreach \j in {0.5,1.5}{\node[red] at (0,\j,1) {$X$};}
			\node[red] at (0,-0.35,1) {$\sim\tilde{Z}_{v}$};
			\node at (-0.22,0,1) {$v$};
			
			% (1) 8-edge truncated top stabilizer A_c^{top} at cube c = [1,2] x [1,2] x [0,1]
			\drawActopEight{1}{0}
			\node[red] at (1.5, 2.28, 0.25) {$A_c^{\text{top}}$};
			
			\fill[blue] (0,0,0) circle[radius=1.5pt];
			\drawaxes
			\node at (1.5,-1.25,1) {(a)};
		\end{tikzpicture}
	\end{subfigure}
	\hfill
	\begin{subfigure}{0.48\textwidth}
		\centering
		\begin{tikzpicture}[baseline=4ex, scale=1.14]
			\drawroughtoplatticegray
			
			% (2b) Bottom truncated B_{v,zx} at v = (0,0,1) (left face \i=0)
			\fill[ForestGreen, opacity=0.12] (0,-0.15,0.15)--(0,0.15,0.15)--(0,0.15,1.85)--(0,-0.15,1.85);
			\fill[ForestGreen, opacity=0.12] (-0.15,0,1)--(0.15,0,1)--(0.15,0.75,1)--(-0.15,0.75,1);
			\node[ForestGreen] at (0,0,0.5) {$Z$};
			\node[ForestGreen] at (0,0,1.5) {$Z$};
			\node[ForestGreen] at (0,0.5,1) {$Z$};
			\node at (-0.22,0,1) {$v$};
			\node[ForestGreen] at (0,-0.4,1.2) {$\sim\tilde{X}_{v}$};
			
			% (2a) Bottom truncated A_c at p = [1,2] x {0} x [0,1]
			\fill[red, opacity=0.1] (0.8,0,0.1)--(1.2,0,0.1)--(1.2,0,0.9)--(0.8,0,0.9);
			\fill[red, opacity=0.1] (1.8,0,0.1)--(2.2,0,0.1)--(2.2,0,0.9)--(1.8,0,0.9);
			\fill[red, opacity=0.1] (1.125,0,-0.4)--(1.125,0,0.4)--(1.875,0,0.4)--(1.875,0,-0.4);
			\fill[red, opacity=0.1] (1.125,0,0.6)--(1.125,0,1.4)--(1.875,0,1.4)--(1.875,0,0.6);
			\foreach \i/\k in {1/0.5, 2/0.5, 1.5/0, 1.5/1}{\node[red] at (\i,0,\k) {$X$};}
			
			\node at (1.5,0,0.5) {$p$};
			\node at (1,0,0) {$v_1$};
			\node at (2,0,0) {$v_2$};
			\node at (2,0,1) {$v_3$};
			\node at (1,0,1) {$v_4$};
			\node[red] at (1.8,-0.55,0.5) {$\sim\tilde{Z}_{v_1}\tilde{Z}_{v_2}\tilde{Z}_{v_3}\tilde{Z}_{v_4}$};
			
			\fill[blue] (0,0,0) circle[radius=1.5pt];
			\node at (1.5,-1.25,1) {(b)};
		\end{tikzpicture}
	\end{subfigure}
	\vspace{-0.15cm}
	\caption{X-cube FTH.
	Top boundary setting and local operator identification under $(ee)$ top boundary.
	(a) Under $(ee)$ top boundary, the stabilizer generators consist of bulk $A_c, B_{v,l}$ terms, and the truncated 8-edge $A_c^{\text{top}}$ terms.
	The straight $X$ Wilson line connecting top and bottom boundaries are identified as single $\tilde{Z}_v$; it is low-energy preserving since $l_z$ lineon is condensed on the $(ee)$ top boundary.
	(b) The bottom boundary truncated $A_c^{\text{bot}}$, $B_{v,zx}^{\text{bot}}\sim B_{v,yz}^{\text{bot}}$ terms are identified as plaquette Ising terms and transverse field terms, respectively.}
	\label{fig_ee_scheme2A}
\end{figure*}

\begin{figure*}[!t]
	\centering
	\begin{subfigure}{0.48\textwidth}
		\centering
		\begin{tikzpicture}[baseline=4ex, scale=1.14]
			\drawsmoothtoplatticegray
			
			% (3) Condensing fracton-dipole transport T_{f,z} along y-edges at \i=0.5, \k=1
			\fill[ForestGreen, opacity=0.12] (0.3,-0.2,1)--(0.7,-0.2,1)--(0.7,2.2,1)--(0.3,2.2,1);
			\foreach \j in {0,1,2}{
				\node[ForestGreen] at (0.5,\j,1) {$Z$};
			}
			\node at (0.5,0,0.5) {$p$};
			\node at (0.5,0,1.5) {$p'$};
			\node[ForestGreen] at (0.5,-0.52,1.35) {$\sim\tilde{Z}_p\tilde{Z}_{p'}$};
			
			% (1) Top boundary 3-leg stabilizer B_{v',yz}^{top} at v' = (2,2,1)
			\fill[ForestGreen, opacity=0.14] (1.25, 2, 1-0.15)--(2.75, 2, 1-0.15)--(2.75, 2, 1+0.15)--(1.25, 2, 1+0.15);
			\fill[ForestGreen, opacity=0.14] (1.82, 1.18, 1)--(2.18, 1.18, 1)--(2.18, 2.0, 1)--(1.82, 2.0, 1);
			\node[ForestGreen] at (1.5, 2, 1) {$Z$};
			\node[ForestGreen] at (2.5, 2, 1) {$Z$};
			\node[ForestGreen] at (2, 1.5, 1) {$Z$};
			\node at (2, 2.25, 1) {$v'$};
			\node[ForestGreen] at (2.15, 2.55, 0.8) {$B_{v',yz}^{\text{top}}$};
			
			\fill[blue] (0,0,0) circle[radius=1.5pt];
			\drawaxes
			\node at (1.5,-1.25,1) {(a)};
		\end{tikzpicture}
	\end{subfigure}
	\hfill
	\begin{subfigure}{0.48\textwidth}
		\centering
		\begin{tikzpicture}[baseline=4ex, scale=1.14]
			\drawsmoothtoplatticegray
			
			% (2a) Bottom truncated A_c at p = [0,1] x {0} x [0,1] (left side)
			\fill[red, opacity=0.1] (-0.2,0,0.1)--(0.2,0,0.1)--(0.2,0,0.9)--(-0.2,0,0.9);
			\fill[red, opacity=0.1] (0.8,0,0.1)--(1.2,0,0.1)--(1.2,0,0.9)--(0.8,0,0.9);
			\fill[red, opacity=0.1] (0.125,0,-0.4)--(0.125,0,0.4)--(0.875,0,0.4)--(0.875,0,-0.4);
			\fill[red, opacity=0.1] (0.125,0,0.6)--(0.125,0,1.4)--(0.875,0,1.4)--(0.875,0,0.6);
			\foreach \i/\k in {0/0.5, 1/0.5, 0.5/0, 0.5/1}{\node[red] at (\i,0,\k) {$X$};}
			\node at (0.5,0,0.5) {$p$};
			\node[red] at (0.5,-0.52,0.5) {$\sim\tilde{X}_p$};
			
			% (2b) Bottom truncated B_{v,yz} ~ B_{v,zx} at v = (2,0,1) (right side)
			\fill[ForestGreen, opacity=0.12] (1.25,0.2,1)--(2.75,0.2,1)--(2.75,-0.2,1)--(1.25,-0.2,1);
			\fill[ForestGreen, opacity=0.12] (1.8,0.75,1)--(2.2,0.75,1)--(2.2,0,1)--(1.8,0,1);
			\foreach \i/\j in {1.5/0, 2.5/0, 2/0.5}{\node[ForestGreen] at (\i,\j,1) {$Z$};}
			\node at (1.5,0,0.5) {$p_1$};
			\node at (2.5,0,0.5) {$p_2$};
			\node at (1.5,0,1.5) {$p_3$};
			\node at (2.5,0,1.5) {$p_4$};
			\node at (2,0,1) {$v$};
			\node[ForestGreen] at (2.05,-0.32,2) {$\sim\tilde{Z}_{p_1}\tilde{Z}_{p_2}\tilde{Z}_{p_3}\tilde{Z}_{p_4}$};
			\draw[->, ForestGreen, >={Triangle[round,length=2mm,width=1.33mm]}] (2,-0.1,1)--(2,-0.15,2);
			
			\fill[blue] (0,0,0) circle[radius=1.5pt];
			\node at (1.5,-1.25,1) {(b)};
		\end{tikzpicture}
	\end{subfigure}
	\vspace{-0.15cm}
	\caption{X-cube FTH.
	Top boundary setting and local operator identification under $(mm)$ top boundary.
	(a) Under $(mm)$ top boundary, the stabilizer generators consist of bulk $A_c, B_{v,l}$ terms, and the truncated 3-leg $B_{v,zx}^{\text{top}}, B_{v,yz}^{\text{top}}$ terms.
	The dual $Z$ Wilson line connecting top and bottom boundaries are identified as Ising terms $\tilde{Z}_p\tilde{Z}_{p'}$; it is low-energy preserving since $f$ is condensed on the $(mm)$ top boundary.
	(b) The bottom boundary truncated $A_c^{\text{bot}}$, $B_{v,zx}^{\text{bot}}\sim B_{v,yz}^{\text{bot}}$ terms are identified as transverse field terms and plaquette Ising terms, respectively.}
	\label{fig_mm_scheme2A}
\end{figure*}

The $(ee)$ top boundary refers to the boundary that $l_z$-lineons are condensed.
It is naturally described by shaping the top boundary rough and putting top boundary truncated $A_c^{\text{top}}$ terms in the stabilizer generator set [see Fig.~\ref{fig_ee_scheme2A}(a)].
The Hamiltonian of the X-cube FTH under $(ee)$ condensed top boundary is
\begin{equation}
	H_{\text{X-cube}}^{(ee)} = -\eta\sum_a S_a + H_{\text{bot}},
\end{equation}
where $S_a$ takes through bulk $A_c, B_{v,l}$ terms and $A_c^{\text{top}}$ terms.
$\eta\gg \|H_{\text{bot}}\|$ separates the low-energy subspace from the rest of the Hilbert space by a large enough gap.
$H_{\text{bot}}$ is a function of local bottom boundary gauge operators, and is the low-energy effective Hamiltonian of the sandwich.
Take the simplest form of $H_{\text{bot}}$,
\begin{equation}\label{eq_X_cube_H_bot}
	H_{\text{bot}} = -\sum_c A_c^{\text{bot}} - \frac{1}{2}g\sum_v \sum_{l=zx,yz} B_{v,l}^{\text{bot}},
\end{equation}
the operator identification shown in Fig.~\ref{fig_ee_scheme2A}(a)(b) leads to the 2d TFPIM.
Again, like in the 2dTC TH case, the operator identification is settled by requiring preserving multiplication and commutation relations.
Similarly, when the top boundary is $(mm)$ [see Fig.~\ref{fig_mm_scheme2A}], the operator identification also leads to the 2d TFPIM, but with the identity of plaquette Ising term and transverse field term exchanged ($H_{\text{bot}}$ always keeps the same, the operator identification of bottom boundary gauge operator changes because different Wilson lines connecting top and bottom boundaries are present under different top boundaries).

\begin{figure*}[!t]
	\centering
	\begin{subfigure}{0.50\textwidth}
		\centering
		\begin{tikzpicture}[scale=1.14]
			\drawroughtoplatticegray
			
			% Uncontractible dual Z-loop \mathcal{Z}_y(1) at x = 1, z = 1/2
			\fill[ForestGreen, opacity=0.14] (-0.2,0.3,1)--(-0.2,0.7,1)--(3,0.7,1)--(3,0.3,1);
			\foreach \i in {0,1,2}{\node[ForestGreen] at (\i,0.5,1) {$Z$};}
			\node[ForestGreen] at (3.5,0.5,1) {$\mathcal{Z}_y(1)$};
			
			% Bottom boundary vertices v_j = [1, j, 0] hosting \tilde{X}_{[1,j]}
			\foreach \i/\lbl in {0/v_1, 1/v_2, 2/v_3}{
				\fill[ForestGreen] (\i,0,1) circle[radius=1.3pt];
				\node at (\i,-0.25,1) {$\lbl$};
			}
			
			\node[ForestGreen] at (1.5,-0.72,1.15) {$\mathcal{Z}_y(1)\sim\prod_{j=0}^{L_y-1}\tilde{X}_{[1,j]}$};
			
			\fill[blue] (0,0,0) circle[radius=1.5pt];
			\drawaxes
			\node at (1.5,-1.22,1.15) {(a)};
		\end{tikzpicture}
	\end{subfigure}
	\hfill
	\begin{subfigure}{0.46\textwidth}
		\centering
		\begin{tikzpicture}[scale=1.14]
			\drawsmoothtoplatticegray
			
			% Light red shading on the strip of bottom plaquettes p_j = [1/2, j+1/2, 0]
			\fill[red, opacity=0.06] (0,0,0)--(3,0,0)--(3,0,1)--(0,0,1)--cycle;
			\foreach \i/\lbl in {0.5/p_1, 1.5/p_2, 2.5/p_3}{
				\node at (\i,0,0.5) {$\lbl$};
			}
			
			% Uncontractible X-loop \mathcal{X}_y(0) at x = 0, z = 0
			\fill[red, opacity=0.14] (0,0,-0.32)--(0,0,0.32)--(3,0,0.32)--(3,0,-0.32);
			\foreach \i in {0.5, 1.5, 2.5}{\node[red] at (\i,0,0) {$X$};}
			\node[red] at (3.55,0,0) {$\mathcal{X}_y(0)$};
			
			% Uncontractible X-loop \mathcal{X}_y(1) at x = 1, z = 0
			\fill[red, opacity=0.14] (0,0,0.68)--(0,0,1.32)--(3,0,1.32)--(3,0,0.68);
			\foreach \i in {0.5, 1.5, 2.5}{\node[red] at (\i,0,1) {$X$};}
			\node[red] at (3.55,0,1) {$\mathcal{X}_y(1)$};
			
			\node[red] at (1.5,-0.72,1.15) {$\mathcal{X}_y(0)\mathcal{X}_y(1)\sim\prod_{j=0}^{L_y-1}\tilde{X}_{[1/2,j+1/2]}$};
			
			\fill[blue] (0,0,0) circle[radius=1.5pt];
			\node at (1.5,-1.22,1.15) {(b)};
		\end{tikzpicture}
	\end{subfigure}
	\vspace{-0.15cm}
	\caption{Line-like subsystem symmetries of TFPIM, identified from the logical operators of the X-cube sandwich.
	(a) Under $(ee)$ top boundary.
	(b) Under $(mm)$ top boundary.}
	\label{fig_mm_ee_symmetries_option_A}
\end{figure*}

\begin{figure*}[!t]
	\centering
	\begin{subfigure}{0.50\textwidth}
		\centering
		\begin{tikzpicture}[scale=1.14]
			\drawroughtoplatticegray
			
			% Other deleted top 8-edge A_c^{top} stabilizers at (x(c)=1/2, y(c)) in blue
			\foreach \i/\k in {0/0, 1/0, 2/0}{
				\fill[blue, opacity=0.1] (\i,2,\k)--(\i+1,2,\k)--(\i+1,1,\k)--(\i+1,1,\k+1)--(\i,1,\k+1)--(\i,2,\k+1);
				\draw[blue, line width=1pt] (\i+0.5,1.5,\k+0.5) circle (0.15);
			}
			
			% Highlighted deleted top cube c = [3/2, 1/2, L_z-1/2] at (\i=0, \k=1) in red
			\fill[red, opacity=0.14] (0,2,1)--(1,2,1)--(1,1,1)--(1,1,2)--(0,1,2)--(0,2,2);
			\draw[red, line width=1pt] (0.5,1.5,1.5) circle (0.18);
			\node[red, font=\small] at (0.5,1.5,1.5) {$c$};
			
			% Equivalent composite uncontractible X-loop \mathcal{X}_y(1)\mathcal{X}_y(2) at z = 0
			\fill[red, opacity=0.14] (0,0,0.68)--(0,0,1.32)--(3,0,1.32)--(3,0,0.68);
			\foreach \i in {0.5, 1.5, 2.5}{\node[red] at (\i,0,1) {$X$};}
			\node[red] at (3.55,0,1) {$\mathcal{X}_y(1)$};
			
			\fill[red, opacity=0.14] (0,0,1.68)--(0,0,2.32)--(3,0,2.32)--(3,0,1.68);
			\foreach \i in {0.5, 1.5, 2.5}{\node[red] at (\i,0,2) {$X$};}
			\node[red] at (3.55,0,2) {$\mathcal{X}_y(2)$};
			
			\node[red] at (1.5,-0.72,1.15) {$\tilde{Z}^{\text{twist}}(\tfrac32,\tfrac12)\equiv A_c^{\text{top}}\sim\mathcal{X}_y(1)\mathcal{X}_y(2)$};
			
			\fill[blue] (0,0,0) circle[radius=1.5pt];
			\drawaxes
			\node at (1.5,-1.22,1.15) {(a)};
		\end{tikzpicture}
	\end{subfigure}
	\hfill
	\begin{subfigure}{0.46\textwidth}
		\centering
		\begin{tikzpicture}[scale=1.14]
			\drawsmoothtoplatticegray
			
			% Other deleted top 3-leg B_{v,zx}^{top} ~ B_{v,yz}^{top} stabilizers at x(v)=0 in blue
			\foreach \i/\k in {0/0, 1/0, 2/0}{
				\draw[blue, line width=1pt] (\i-0.2,2,\k) arc [start angle=180, end angle=360, radius=0.2];
				\begin{scope}[canvas is yz plane at x=\i]
					\draw[blue, line width=1pt] (2,\k+0.2) arc [start angle=90, end angle=270, radius=0.2];
				\end{scope}
			}
			% Highlighted deleted top stabilizer at v = [1,0,L_z] (\i=0, \j=2, \k=1) in ForestGreen
			\draw[ForestGreen, line width=1.1pt] (0-0.2,2,1) arc [start angle=180, end angle=360, radius=0.2];
			\begin{scope}[canvas is yz plane at x=0]
				\draw[ForestGreen, line width=1.1pt] (2,1+0.2) arc [start angle=90, end angle=270, radius=0.2];
			\end{scope}
			\fill[ForestGreen] (0,2,1) circle[radius=1.4pt];
			\node[ForestGreen] at (-0.28,2.18,1) {$v$};
			
			% Uncontractible dual Z-loop \mathcal{Z}_y(1) at x = 1, z = 1/2 equivalent to B_{v,yz}^{top} at v = [1,0,L_z]
			\fill[ForestGreen, opacity=0.14] (-0.2,0.3,1)--(-0.2,0.7,1)--(3,0.7,1)--(3,0.3,1);
			\foreach \i in {0,1,2}{\node[ForestGreen] at (\i,0.5,1) {$Z$};}
			\node[ForestGreen] at (3.5,0.5,1) {$\mathcal{Z}_y(1)$};
			
			\node[ForestGreen] at (1.5,-0.72,1.15) {$\tilde{Z}^{\text{twist}}(1,0)\equiv B_{v,yz}^{\text{top}}\sim\mathcal{Z}_y(1)$};
			
			\fill[blue] (0,0,0) circle[radius=1.5pt];
			\node at (1.5,-1.22,1.15) {(b)};
		\end{tikzpicture}
	\end{subfigure}
	\vspace{-0.15cm}
	\caption{\textbf{Deleted top stabilizers and twist logical operators (twist indicators) under $(ee)$ and $(mm)$ top boundaries}: (a) Under the $(ee)$ rough top boundary, the $L_x+L_y-1$ top 8-edge stabilizers $A_c^{\text{top}}$ with $\frac12\in\{x(c),y(c)\}$ (marked by shaded top cubes with circles) are deleted, leaving $L_x+L_y-2$ independent twist DOFs; for example, the deleted top cube stabilizer at $c=[\frac32,\frac12,L_z-\frac12]$ (highlighted in red) is identified as the twist indicator $\tilde{Z}^{\text{twist}}(\frac32,\frac12)\equiv A_c^{\text{top}}\sim\mathcal{X}_y(1)\mathcal{X}_y(2)$.
	(b) Under the $(mm)$ smooth top boundary, the $L_x+L_y-1$ top 3-leg stabilizers $B_{v,zx}^{\text{top}}\sim B_{v,yz}^{\text{top}}$ with $0\in\{x(v),y(v)\}$ (marked by half-circles at $z=L_z$) are deleted to lift the bottom plaquette-Ising relations to twist DOFs; for example, the deleted top stabilizer at $v=[1,0,L_z]$ (highlighted in green) is identified as the twist indicator $\tilde{Z}^{\text{twist}}(1,0)\equiv B_{v,yz}^{\text{top}}\sim\mathcal{Z}_y(1)$.
	Comparing Fig.~\ref{fig_mm_ee_symmetries_option_A} and Fig.~\ref{fig_mm_ee_twists} directly exhibits the symmetry--twist exchange between $(ee)$ and $(mm)$.}
	\label{fig_mm_ee_twists}
\end{figure*}

Thus the change of top boundary from $(ee)$ to $(mm)$ realizes the TFPIM version KW duality (Refs.~\cite{xu_moore_2004_self_duality,TFPIM_KW_duality}), which is non-invertible.
An explicit linear map can be constructed to represent the duality,
\begin{equation}
	\mathsf{D}_\square = 2^{-N/2}\sum_{\mathbf{s},\mathbf{t}} (-1)^{\mathbf{t}^T M\mathbf{s}} |\mathbf{t}\rangle\langle\mathbf{s}|,
\end{equation}
which maps the vertex-qubit TFPIM (with $\tilde{Z}$ configurations $\mathbf{s}$) to plaquette-qubit TFPIM (with $\tilde{Z}'$ configurations $\mathbf{t}$).
Here $N\equiv L_x L_y$ is the qubit number of TFPIM, $M$ is the matrix with $M_{pv}=1$ when $v\subset p$ and $M_{pv}=0$ otherwise.
$\mathsf{D}_\square$ is usually called an intertwiner as it intertwines the plaquette Ising terms and transverse field terms
\begin{align}
	\mathsf{D}_\square \prod_{v\subset p} \tilde{Z}_v &= \tilde{X}'_p \mathsf{D}_\square, \notag\\
	\mathsf{D}_\square \tilde{X}_v &= \prod_{p\revsubset v} \tilde{Z}'_p \mathsf{D}_\square.\label{eq_TFPIM_KW_intertwining}
\end{align}

The TFPIM has line-like subsystem symmetries, which are identified from the logical operators of the X-cube sandwich [see Fig.~\ref{fig_mm_ee_symmetries_option_A}].
Suppose a set of independent line-like subsystem symmetry indicators of vertex TFPIM is $\{S_\ell\}$, then
\begin{equation}\label{eq_TFPIM_KW_composition}
	\mathsf{D}_\square^\dagger \mathsf{D}_\square = \prod_\ell (1+S_\ell) \propto P_{\text{lines},+},
\end{equation}
i.e., $\mathsf{D}_\square^\dagger \mathsf{D}_\square$ projects onto the common $+1$ eigenspace of subsystem line-like symmetries.
This implies that there exists a state $|\psi\rangle$, s.t.\ $\mathsf{D}_\square|\psi\rangle=0$, thus the duality is non-invertible.
Another well-known property of $\mathsf{D}_\square$ is that it absorbs symmetry indicators from both sides, i.e., $\mathsf{D}_\square S_\ell=\mathsf{D}_\square$; while on the other side, $S'_\ell \mathsf{D}_\square=\mathsf{D}_\square$ for the subsystem indicators $S'_\ell$ of the plaquette TFPIM.
Physically, this means that the symmetry indicators are intertwined by $\mathsf{D}_\square$ to relations (which are just 1).
See step-by-step details in Appendix~\ref{app_subsec_D_square}. 

Similar to the canonical KW duality of TFIM, the TFPIM version KW duality can also be made invertible by introducing twist DOFs to record the twist information.
This can be well-described with the X-cube FTH, by allowing adding defects on the top boundary, which can be realized by deleting specific stabilizer generators on the top boundary [see Fig.~\ref{fig_mm_ee_twists}].
After the twist DOFs introduced, the TFPIM KW duality exchanges the identities of symmetry operators and twist operators, as well as the symmetry sectors and twist sectors.
It is now invertible since no information is dropped by the duality any more.
The vertex TFPIM has $2^{L_x+L_y-1}$ symmetry sectors and $2^{L_x+L_y-2}$ twist sectors, while the plaquette TFPIM has $2^{L_x+L_y-2}$ symmetry sectors and $2^{L_x+L_y-1}$ twist sectors.
The sector numbers match exactly, as well as the low-energy subspace of X-cube FTH under $(ee)$ and $(mm)$ top boundaries [see counting details in Appendix~\ref{app_X_cube_ee_mm_counting}].
This makes a unitary circuit possible, and we construct a linear-depth local unitary SQC that maps the top boundary from $(ee)$ to $(mm)$ as following:
\begin{widetext}
	\begin{equation}\label{X_cube_rough_to_smooth_top_SQC}
		\hspace*{3.5cm}
		\begin{tikzpicture}[>={Triangle[round,length=2mm,width=1.33mm]}]
			\foreach \i in {0,1,2,3,4}{\foreach \k in {0,1,2,3}{
					\draw[line width=1pt] (\i,0,\k)--(\i,1,\k);
					\draw (\i,0,\k)--(\i,-0.5,\k);
			}}
			\foreach \i in {0,1,2,3,4}{
				\draw[line width=1pt, dashed] (\i,0,4)--(\i,1,4);
				\draw[dashed] (\i,0,4)--(\i,-0.5,4);
			}
			\foreach \k in {0,1,2,3}{
				\draw[line width=1pt, dashed] (5,0,\k)--(5,1,\k);
				\draw[dashed] (5,0,\k)--(5,-0.5,\k);
			}
			\draw[line width=1pt, dashed] (5,0,4)--(5,1,4);
			\draw[dashed] (5,0,4)--(5,-0.5,4);
			\foreach \i in {0,1,2,3,4}{\draw[line width=1pt] (\i,0,0)--(\i,0,4);}
			\draw[line width=1pt, dashed] (5,0,0)--(5,0,4);
			\foreach \k in {0,1,2,3}{\draw[line width=1pt] (0,0,\k)--(5,0,\k);}
			\draw[line width=1pt, dashed] (0,0,4)--(5,0,4);
			
			\fill[blue, opacity=0.2] (0,0,4)--(1,0,4)--(1,0,1)--(5,0,1)--(5,0,0)--(5,1,0)--(0,1,0)--(0,1,4);
			
			\draw[red, line width=1pt, ->] (1,0.5,1)--(1.5,0,1);
			\draw[red, line width=1pt, ->] (1,0.5,1)--(1,0,1.5);
			\draw[red, line width=1pt, ->] (1,0.5,1)--(2,0,1.5);
			\draw[red, line width=1pt, ->] (1,0.5,1)--(1.5,0,2);
			\draw[red, line width=1pt, ->] (1,0.5,1)--(2,0.5,1);
			\draw[red, line width=1pt, ->] (1,0.5,1)--(2,0.5,2);
			\draw[red, line width=1pt, ->] (1,0.5,1)--(1,0.5,2);
			
			\draw[gray, line width=1pt, ->] (2,0.5,1)--(2.5,0,1);
			\draw[gray, line width=1pt, ->] (2,0.5,1)--(2,0,1.5);
			\draw[gray, line width=1pt, ->] (2,0.5,1)--(3,0,1.5);
			\draw[gray, line width=1pt, ->] (2,0.5,1)--(2.5,0,2);
			\draw[gray, line width=1pt, ->] (2,0.5,1)--(3,0.5,1);
			\draw[gray, line width=1pt, ->] (2,0.5,1)--(3,0.5,2);
			\draw[gray, line width=1pt, ->] (2,0.5,1)--(2,0.5,2);
			
			\draw[gray, line width=1pt, ->] (1,0.5,2)--(1.5,0,2);
			\draw[gray, line width=1pt, ->] (1,0.5,2)--(1,0,2.5);
			\draw[gray, line width=1pt, ->] (1,0.5,2)--(2,0,2.5);
			\draw[gray, line width=1pt, ->] (1,0.5,2)--(1.5,0,3);
			\draw[gray, line width=1pt, ->] (1,0.5,2)--(2,0.5,2);
			\draw[gray, line width=1pt, ->] (1,0.5,2)--(2,0.5,3);
			\draw[gray, line width=1pt, ->] (1,0.5,2)--(1,0.5,3);
			
			\draw[blue, line width=1pt, ->] (1,0.5,3)--(1.5,0,3);
			\draw[blue, line width=1pt, ->] (1,0.5,3)--(1,0,3.5);
			\draw[blue, line width=1pt, ->] (1,0.5,3)--(2,0,3.5);
			\draw[blue, line width=1pt, ->] (1,0.5,3)--(1.5,0,4);
			\draw[blue, line width=1pt, ->] (1,0.5,3)--(2,0.5,3);
			\draw[blue, line width=1pt, ->] (1,0.5,3)--(2,0.5,4);
			\draw[blue, line width=1pt, ->] (1,0.5,3)--(1,0.5,4);
			
			\draw[blue, line width=1pt, ->] (2,0.5,2)--(2.5,0,2);
			\draw[blue, line width=1pt, ->] (2,0.5,2)--(2,0,2.5);
			\draw[blue, line width=1pt, ->] (2,0.5,2)--(3,0,2.5);
			\draw[blue, line width=1pt, ->] (2,0.5,2)--(2.5,0,3);
			\draw[blue, line width=1pt, ->] (2,0.5,2)--(3,0.5,2);
			\draw[blue, line width=1pt, ->] (2,0.5,2)--(3,0.5,3);
			\draw[blue, line width=1pt, ->] (2,0.5,2)--(2,0.5,3);
			
			\draw[blue, line width=1pt, ->] (3,0.5,1)--(3.5,0,1);
			\draw[blue, line width=1pt, ->] (3,0.5,1)--(3,0,1.5);
			\draw[blue, line width=1pt, ->] (3,0.5,1)--(4,0,1.5);
			\draw[blue, line width=1pt, ->] (3,0.5,1)--(3.5,0,2);
			\draw[blue, line width=1pt, ->] (3,0.5,1)--(4,0.5,1);
			\draw[blue, line width=1pt, ->] (3,0.5,1)--(4,0.5,2);
			\draw[blue, line width=1pt, ->] (3,0.5,1)--(3,0.5,2);
			
			\draw[cyan, line width=1pt, ->] (2,0.5,3)--(2.5,0,3);
			\draw[cyan, line width=1pt, ->] (2,0.5,3)--(2,0,3.5);
			\draw[cyan, line width=1pt, ->] (2,0.5,3)--(3,0,3.5);
			\draw[cyan, line width=1pt, ->] (2,0.5,3)--(2.5,0,4);
			\draw[cyan, line width=1pt, ->] (2,0.5,3)--(3,0.5,3);
			\draw[cyan, line width=1pt, ->] (2,0.5,3)--(3,0.5,4);
			\draw[cyan, line width=1pt, ->] (2,0.5,3)--(2,0.5,4);
			
			\draw[cyan, line width=1pt, ->] (3,0.5,2)--(3.5,0,2);
			\draw[cyan, line width=1pt, ->] (3,0.5,2)--(3,0,2.5);
			\draw[cyan, line width=1pt, ->] (3,0.5,2)--(4,0,2.5);
			\draw[cyan, line width=1pt, ->] (3,0.5,2)--(3.5,0,3);
			\draw[cyan, line width=1pt, ->] (3,0.5,2)--(4,0.5,2);
			\draw[cyan, line width=1pt, ->] (3,0.5,2)--(4,0.5,3);
			\draw[cyan, line width=1pt, ->] (3,0.5,2)--(3,0.5,3);
			
			\draw[cyan, line width=1pt, ->] (4,0.5,1)--(4.5,0,1);
			\draw[cyan, line width=1pt, ->] (4,0.5,1)--(4,0,1.5);
			\draw[cyan, line width=1pt, ->] (4,0.5,1)--(5,0,1.5);
			\draw[cyan, line width=1pt, ->] (4,0.5,1)--(4.5,0,2);
			\draw[cyan, line width=1pt, ->] (4,0.5,1)--(5,0.5,1);
			\draw[cyan, line width=1pt, ->] (4,0.5,1)--(5,0.5,2);
			\draw[cyan, line width=1pt, ->] (4,0.5,1)--(4,0.5,2);
			
			\draw[magenta, line width=1pt, ->] (3,0.5,3)--(3.5,0,3);
			\draw[magenta, line width=1pt, ->] (3,0.5,3)--(3,0,3.5);
			\draw[magenta, line width=1pt, ->] (3,0.5,3)--(4,0,3.5);
			\draw[magenta, line width=1pt, ->] (3,0.5,3)--(3.5,0,4);
			\draw[magenta, line width=1pt, ->] (3,0.5,3)--(4,0.5,3);
			\draw[magenta, line width=1pt, ->] (3,0.5,3)--(4,0.5,4);
			\draw[magenta, line width=1pt, ->] (3,0.5,3)--(3,0.5,4);
			
			\draw[magenta, line width=1pt, ->] (4,0.5,2)--(4.5,0,2);
			\draw[magenta, line width=1pt, ->] (4,0.5,2)--(4,0,2.5);
			\draw[magenta, line width=1pt, ->] (4,0.5,2)--(5,0,2.5);
			\draw[magenta, line width=1pt, ->] (4,0.5,2)--(4.5,0,3);
			\draw[magenta, line width=1pt, ->] (4,0.5,2)--(5,0.5,2);
			\draw[magenta, line width=1pt, ->] (4,0.5,2)--(5,0.5,3);
			\draw[magenta, line width=1pt, ->] (4,0.5,2)--(4,0.5,3);
			
			\draw[ForestGreen, line width=1pt, ->] (4,0.5,3)--(4.5,0,3);
			\draw[ForestGreen, line width=1pt, ->] (4,0.5,3)--(4,0,3.5);
			\draw[ForestGreen, line width=1pt, ->] (4,0.5,3)--(5,0,3.5);
			\draw[ForestGreen, line width=1pt, ->] (4,0.5,3)--(4.5,0,4);
			\draw[ForestGreen, line width=1pt, ->] (4,0.5,3)--(5,0.5,3);
			\draw[ForestGreen, line width=1pt, ->] (4,0.5,3)--(5,0.5,4);
			\draw[ForestGreen, line width=1pt, ->] (4,0.5,3)--(4,0.5,4);
			
			\foreach \i/\k in {1/0, 2/0, 3/0, 4/0}{
				\draw[orange, line width=1pt, ->] (\i,-0.47,\k)--(\i,0.47,\k);
				\draw[orange, line width=1pt, ->] (\i-0.5,0,\k)--(\i-0.05,0.47,\k);
				\draw[orange, line width=1pt, ->] (\i+0.5,0,\k)--(\i+0.05,0.47,\k);
			}
			\draw[orange, line width=1pt, ->] (0,-0.47,0)--(0,0.47,0);
			\draw[orange, line width=1pt, ->] (4.5,0,0)--(4.95,0.47,0);
			\draw[orange, line width=1pt, ->] (0.5,0,0)--(0.05,0.47,0);
			
			\foreach \i/\k in {0/1, 0/2, 0/3}{
				\draw[orange, line width=1pt, ->] (\i,-0.47,\k)--(\i,0.47,\k);
				\draw[orange, line width=1pt, ->] (\i,0,\k-0.5)--(\i,0.47,\k-0.15);
				\draw[orange, line width=1pt, ->] (\i,0,\k+0.5)--(\i,0.47,\k+0.15);
			}
			
			\node at (7,0.5,0) {\textbf{layer}};
			\node[red] at (6,0,0) {\textbf{1}};
			\node[gray] at (7,0,0) {\textbf{2}};
			\node[blue] at (8,0,0) {\textbf{3}};
			\node[cyan] at (6,-0.5,0) {\textbf{4}};
			\node[magenta] at (7,-0.5,0) {\textbf{5}};
			\node[ForestGreen] at (8,-0.5,0) {\textbf{6}};
			\node[ForestGreen] at (9.35,-0.54,0) {$=L_x+L_y-3$};
			\node[orange] at (7,-1,0) {\textbf{7}};
		\end{tikzpicture}
	\end{equation}
\end{widetext}
The arrows represent CNOT gates here.
A step-by-step explanation of the SQC is given in Appendix~\ref{app_X_cube_SQC}.

%======================================================================
\subsection{Under \texorpdfstring{$(em)$}{(em)} and \texorpdfstring{$(me)$}{(me)} boundaries, quantum compass model, duality web}\label{subsec_X_cube_em_me}
%======================================================================

\begin{figure*}[!t]
	\centering
	\begin{subfigure}{0.48\textwidth}
		\centering
		\begin{tikzpicture}[baseline=4ex, scale=1.14]
			\drawemtoplatticegray
			
			% Highlight bottom x-edge e = (0,0,[0,1])
			\draw[black, line width=0.75pt] (0,0,0)--(0,0,1);
			\node[circle, fill=white, inner sep=1.7pt] at (0,0,0.5) {};
			
			% (3a) Condensing fracton-dipole transport T_{f,z} on x-edges at \i=0, \k=0.5
			\fill[ForestGreen, opacity=0.12] (0,-0.2,0.25)--(0,-0.2,0.75)--(0,2.2,0.75)--(0,2.2,0.25);
			\foreach \j in {0,1,2}{
				\node[ForestGreen] at (0,\j,0.5) {$Z$};
			}
			\node[black, font=\small] at (-0.42,0,0.78) {$e$};
			\draw[->, black, thin, >={Triangle[round,length=1.4mm,width=1mm]}] (-0.28,0,0.78)--(-0.04,0,0.78);
			\node[ForestGreen] at (-0.22,-0.52,0.5) {$\sim\tilde{Z}_{e}$};
			
			% (1a) 10-edge top cube stabilizer A_c^{top} at c = [1,2] x [1,2] x [0,1]
			\drawActopTenEM{1}{0}
			\node[red] at (1.45, 2.38, 0.25) {$A_c^{\text{top}}$};
			
			% (1b) 3-leg top cross stabilizer B_{v',zx}^{top} at v' = (3,2,1) (on the right dashed face \i=3)
			\fill[ForestGreen, opacity=0.14] (3, 1.84, 0.15)--(3, 2.16, 0.15)--(3, 2.16, 1.85)--(3, 1.84, 1.85);
			\fill[ForestGreen, opacity=0.14] (2.84, 1.18, 1)--(3.16, 1.18, 1)--(3.16, 2.0, 1)--(2.84, 2.0, 1);
			\node[ForestGreen] at (3, 2, 0.5) {$Z$};
			\node[ForestGreen] at (3, 2, 1.5) {$Z$};
			\node[ForestGreen] at (3, 1.5, 1) {$Z$};
			\node at (3.32, 1.92, 1) {$v'$};
			\node[ForestGreen] at (3.25, 2.42, 0.75) {$B_{v',zx}^{\text{top}}$};
			
			\fill[blue] (0,0,0) circle[radius=1.5pt];
			\drawaxes
			\node at (1.5,-1.15,1) {(a)};
		\end{tikzpicture}
	\end{subfigure}
	\hfill
	\begin{subfigure}{0.48\textwidth}
		\centering
		\begin{tikzpicture}[baseline=4ex, scale=1.14]
			\drawemtoplatticegray
			
			% Highlight bottom x-edges e, e_1, e_2, e_3, e_4 hosting the effective QCM qubits
			\draw[black, line width=0.75pt] (0,0,0)--(0,0,1);
			\draw[black, line width=0.75pt] (1,0,0)--(1,0,1);
			\draw[black, line width=0.75pt] (1,0,1)--(1,0,2);
			\draw[black, line width=0.75pt] (2,0,0)--(2,0,1);
			\draw[black, line width=0.75pt] (3,0,0)--(3,0,1);
			\foreach \i/\k in {1/0.5, 1/1.5, 2/0.5, 3/0.5}{
				\node[circle, fill=white, inner sep=1.7pt] at (\i,0,\k) {};
			}
			
			% (3b) Condensing inverted-U lineon-dipole transport T_{l_z,z}^{\sqcap} in xz-plane at \i=0, \k\in[0,1]
			\fill[red, opacity=0.12] (-0.16, 0, 0)--(0.16, 0, 0)--(0.16, 2.15, 0)--(-0.16, 2.15, 0);
			\fill[red, opacity=0.12] (-0.16, 0, 1)--(0.16, 0, 1)--(0.16, 2.15, 1)--(-0.16, 2.15, 1);
			\fill[red, opacity=0.14] (0, 1.84, 0)--(0, 2.16, 0)--(0, 2.16, 1)--(0, 1.84, 1);
			\foreach \j in {0.5, 1.5}{
				\node[red] at (0, \j, 0) {$X$};
				\node[red] at (0, \j, 1) {$X$};
			}
			\node[red] at (0, 2, 0.5) {$X$};
			\node[black, font=\small] at (-0.42, 0, 0.55) {$e$};
			\draw[->, black, thin, >={Triangle[round,length=1.4mm,width=1mm]}] (-0.28, 0, 0.55)--(-0.04, 0, 0.55);
			\node[red] at (-0.38, -0.48, 0.45) {$\sim\tilde{X}_{e}$};
			
			% (2a) Bottom truncated B_{v,zx} ~ B_{v,yz} at v = (1,0,1) (middle face \i=1)
			\fill[ForestGreen, opacity=0.12] (1,-0.15,0.15)--(1,0.15,0.15)--(1,0.15,1.85)--(1,-0.15,1.85);
			\fill[ForestGreen, opacity=0.12] (0.85,0,1)--(1.15,0,1)--(1.15,0.75,1)--(0.85,0.75,1);
			\node[ForestGreen] at (1,0,0.5) {$Z$};
			\node[ForestGreen] at (1,0,1.5) {$Z$};
			\node[ForestGreen] at (1,0.5,1) {$Z$};
			\node at (0.8,0.18,1) {$v$};
			\node[black, font=\footnotesize] at (1.42,-0.04,0.25) {$e_1$};
			\draw[->, black, thin, >={Triangle[round,length=1.4mm,width=1mm]}] (1.24,-0.03,0.25)--(1.04,0,0.25);
			\node[black, font=\footnotesize] at (1.42,-0.04,1.9) {$e_2$};
			\draw[->, black, thin, >={Triangle[round,length=1.4mm,width=1mm]}] (1.24,-0.03,1.9)--(1.04,0,1.9);
			\node[ForestGreen] at (1.05,-0.58,1.55) {$\sim\tilde{Z}_{e_1}\tilde{Z}_{e_2}$};
			
			% (2b) Bottom truncated A_c at p = [2,3] x {0} x [0,1] (right side \i\in[2,3], \k\in[0,1])
			\fill[red, opacity=0.1] (1.8,0,0.1)--(2.2,0,0.1)--(2.2,0,0.9)--(1.8,0,0.9);
			\fill[red, opacity=0.1] (2.8,0,0.1)--(3.2,0,0.1)--(3.2,0,0.9)--(2.8,0,0.9);
			\fill[red, opacity=0.1] (2.125,0,-0.4)--(2.125,0,0.4)--(2.875,0,0.4)--(2.875,0,-0.4);
			\fill[red, opacity=0.1] (2.125,0,0.6)--(2.125,0,1.4)--(2.875,0,1.4)--(2.875,0,0.6);
			\foreach \i/\k in {2/0.5, 3/0.5, 2.5/0, 2.5/1}{\node[red] at (\i,0,\k) {$X$};}
			\node at (2.5,0,0.5) {$p$};
			\node[black, font=\footnotesize] at (1.62,-0.22,0.78) {$e_3$};
			\draw[->, black, thin, >={Triangle[round,length=1.4mm,width=1mm]}] (1.78,-0.14,0.78)--(1.96,0,0.78);
			\node[black, font=\footnotesize] at (3.38,-0.22,0.78) {$e_4$};
			\draw[->, black, thin, >={Triangle[round,length=1.4mm,width=1mm]}] (3.22,-0.14,0.78)--(3.04,0,0.78);
			\node[red] at (2.55,-0.55,0.5) {$\sim\tilde{X}_{e_3}\tilde{X}_{e_4}$};
			
			\fill[blue] (0,0,0) circle[radius=1.5pt];
			\node at (1.5,-1.15,1) {(b)};
		\end{tikzpicture}
	\end{subfigure}
	\vspace{-0.15cm}
	\caption{X-cube FTH.
	Top boundary setting and local operator identification under $(em)$ top boundary, where the effective QCM spins live on the bottom $x$-directed edges (highlighted in bold black for $e, e_1, e_2, e_3, e_4$).
	(a) The 10-edge top cube stabilizer $A_c^{\text{top}}$, the 3-leg top cross stabilizer $B_{v',zx}^{\text{top}}$, and the condensing fracton-dipole transport operator $\mathcal{T}_{f,z}\sim\tilde{Z}_e$ on $x$-edges.
	(b) The condensing inverted-U lineon-dipole transport operator $\mathcal{T}_{l_z,z}^{\sqcap}\sim\tilde{X}_e$ in the $zx$-plane and the bottom boundary gauge operators $B_{v,zx}^{\text{bot}}\sim B_{v,yz}^{\text{bot}}\sim\tilde{Z}_{e_1}\tilde{Z}_{e_2}$ ($x$-directed bond) and $A_c^{\text{bot}}\sim\tilde{X}_{e_3}\tilde{X}_{e_4}$ ($y$-directed bond).
	The corresponding figures for $(me)$ are obtained by exchanging $x\leftrightarrow y$.}
	\label{fig_em_scheme2A}
\end{figure*}

\begin{figure*}[!t]
	\centering
	\begin{subfigure}{0.48\textwidth}
		\centering
		\begin{tikzpicture}[scale=1.12]
			\drawemtoplatticegray
			
			% Top single-row X-loop \mathcal{X}_x^{\text{top}}(0) at \i=0 (y=0), \j=2 (z=L_z)
			\fill[red, opacity=0.14] (-0.15,2,0)--(0.15,2,0)--(0.15,2,2)--(-0.15,2,2);
			\foreach \k in {0.5, 1.5}{\node[red] at (0,2,\k) {$X$};}
			\node[red] at (-0.1, 2.35, 0.3) {$\mathcal{X}_x^{\text{top}}(0)\sim\prod_{i=0}^{L_x-1}\tilde{X}_{[i+1/2,0]}$};
			
			% Bottom composite X-loop \mathcal{X}_x(1)\mathcal{X}_x(2) at \i=1,2 (y=1,2), \j=0 (z=0)
			\fill[red, opacity=0.06] (1,0,0)--(2,0,0)--(2,0,2)--(1,0,2)--cycle;
			\fill[red, opacity=0.14] (0.85,0,0)--(1.15,0,0)--(1.15,0,2)--(0.85,0,2);
			\fill[red, opacity=0.14] (1.85,0,0)--(2.15,0,0)--(2.15,0,2)--(1.85,0,2);
			\foreach \k in {0.5, 1.5}{
				\node[red] at (1,0,\k) {$X$};
				\node[red] at (2,0,\k) {$X$};
			}
			\node[red] at (0.95, -0.32, 2) {$\mathcal{X}_x(1)$};
			\node[red] at (2.05, -0.32, 2) {$\mathcal{X}_x(2)$};
			\node[red] at (1.5, -0.75, 1.6) {$\mathcal{X}_x(1)\mathcal{X}_x(2)\sim\prod_{i=0}^{L_x-1}\tilde{X}_{[i+1/2,1]}\tilde{X}_{[i+1/2,2]}$};
			
			\fill[blue] (0,0,0) circle[radius=1.5pt];
			\drawaxes
			\node at (1.5,-1.2,1.15) {(a)};
		\end{tikzpicture}
	\end{subfigure}
	\hfill
	\begin{subfigure}{0.48\textwidth}
		\centering
		\begin{tikzpicture}[scale=1.12]
			\drawemtoplatticegray
			
			% y-direction symmetry indicator \mathcal{Z}_y(1) at \k=1 (x=1), \j=0.5 (z=1/2)
			\fill[ForestGreen, opacity=0.14] (-0.2,0.32,1)--(-0.2,0.68,1)--(3,0.68,1)--(3,0.32,1);
			\foreach \i in {0,1,2}{\node[ForestGreen] at (\i,0.5,1) {$Z$};}
			\node[ForestGreen] at (3.5,0.5,1) {$\mathcal{Z}_y(1)$};
			
			\node[ForestGreen] at (1.5, -0.75, 1.2) {$\mathcal{Z}_y(1)\sim \prod_{j=0}^{L_y-1}\tilde{Z}_{[1/2,j]}\tilde{Z}_{[3/2,j]}$};
			
			\fill[blue] (0,0,0) circle[radius=1.5pt];
			\node at (1.5,-1.2,1.15) {(b)};
		\end{tikzpicture}
	\end{subfigure}
	\vspace{-0.15cm}
	\caption{Line-like subsystem symmetry indicators of the 2d $x$-edge QCM under the $(em)$ top boundary:
	(a) $x$-direction row symmetries ($L_y$ independent generators).
	The X Wilson line along $x$-direction on the top boundary is identified as the subsystem symmetry indicator, e.g., $\mathcal X_x^{\rm top}(j)\sim\prod_{i=0}^{L_x-1}\tilde X_{[i+1/2,j]}\equiv S_x^{(X)}(j)$, consistent with the identification of $\mathcal T_{l_z}^\sqcap$.
	The product of neighboring $\mathcal X_x(j),\mathcal X_x(j+1)$, identified as $S_x^{(X)}(j)S_x^{(X)}(j+1)$, can be moved to the bottom boundary by $A_c^{\rm top}$ and $A_c$ stabilizers, consistent with the identification of $A_c^{\rm bot}$.
	Single-qubit operators $\tilde{Z}_{[i+1/2,j]}$ are identified from the transport operators $\mathcal{T}_{f,z}(i+\frac12,j)$, anticommuting with one $S_x^{(X)}(j)$, thus acting as a single-row symmetry toggler.
	(b) $y$-direction column symmetries ($L_x-1$ independent generators).
	The dual $Z$ Wilson line along $y$-direction $\mathcal Z_y(i)$, being the product of $B_{v,yz}^{\rm bot}$ terms on the $y$-line below, is identified as the subsystem symmetry indicator $\prod_{j=0}^{L_y-1}\tilde Z_{[i-1/2,j]}\tilde Z_{[i+1/2,j]}\equiv S_y^{(Z)}(i)$.
	Since every column of $\tilde{Z}$ appears twice in the product over all $i$, the $L_x$ indicators satisfy the global relation $\prod_{i=0}^{L_x-1}S_y^{(Z)}(i)=1$, leaving $L_x-1$ independent column symmetries.
	Each single-qubit operator $\mathcal{T}_{l_z,z}^{\sqcap}(i+\frac12,j)\sim\tilde{X}_{[i+1/2,j]}$ flips the two adjacent indicators $S_y^{(Z)}(i)$ and $S_y^{(Z)}(i+1)$ simultaneously, acting as a double-column symmetry toggler that preserves $\prod_{i=0}^{L_x-1}S_y^{(Z)}(i)=1$.}
	\label{fig_em_symmetries_2fig}
\end{figure*}

\begin{figure*}[!t]
	\centering
	\begin{subfigure}{0.48\textwidth}
		\centering
		\begin{tikzpicture}[scale=1.12]
			\drawemtoplatticegray
			
			% Deleted top B_{v,zx}^{top} at x(v)=0 (\k=0): \i=0,2 in blue, \i=1 highlighted in ForestGreen
			\foreach \i in {0,2}{
				\begin{scope}[canvas is yz plane at x=\i]
					\draw[blue, line width=1pt] (2,0.2) arc [start angle=90, end angle=270, radius=0.2];
				\end{scope}
			}
			\begin{scope}[canvas is yz plane at x=1]
				\draw[ForestGreen, line width=1.1pt] (2,0.2) arc [start angle=90, end angle=270, radius=0.2];
			\end{scope}
			\fill[ForestGreen] (1,2,0) circle[radius=1.4pt];
			\node[ForestGreen] at (1,2.28,0) {$v$};
			
			% Twist toggler \mathcal{T}_{l_z,z}(0,0) \sim \tilde{X}_B^{\text{twist}}(0,0) on vertical z-edges at \i=0, \k=0
			\fill[red, opacity=0.14] (-0.15, 0, 0)--(0.15, 0, 0)--(0.15, 2.0, 0)--(-0.15, 2.0, 0);
			\foreach \j in {0.5, 1.5}{\node[red] at (0, \j, 0) {$X$};}
			\node[red, font=\small] at (-0.65, 1.0, 0) {$\sim\tilde{X}_B^{\text{twist}}$};
			
			% Equivalent dual Z-loop \mathcal{Z}_x(1) at \i=1 (y=1), \j=0.5 (z=1/2)
			\fill[ForestGreen, opacity=0.14] (1, 0.32, -0.15)--(1, 0.68, -0.15)--(1, 0.68, 2.0)--(1, 0.32, 2.0);
			\foreach \k in {0,1}{\node[ForestGreen] at (1, 0.5, \k) {$Z$};}
			\node[ForestGreen] at (1.42, 0.5, 1.85) {$\mathcal{Z}_x(1)$};
			\node[ForestGreen] at (1.5, -0.72, 1.15) {$\tilde{Z}_B^{\text{twist}}(0,1)\equiv B_{v,zx}^{\text{top}}\sim\mathcal{Z}_x(1)$};
			
			\fill[blue] (0,0,0) circle[radius=1.5pt];
			\drawaxes
			\node at (1.5,-1.18,1.15) {(a)};
		\end{tikzpicture}
	\end{subfigure}
	\hfill
	\begin{subfigure}{0.48\textwidth}
		\centering
		\begin{tikzpicture}[scale=1.12]
			\drawemtoplatticegray
			
			% Deleted top 10-edge A_c^{top} at y(c)=1/2 (\i\in[0,1]): \k=0 in blue, \k=1 highlighted in red
			\fill[blue, opacity=0.1] (0,2,0)--(1,2,0)--(1,1,0)--(1,1,1)--(0,1,1)--(0,2,1);
			\draw[blue, line width=1pt] (0.5,1.5,0.5) circle (0.15);
			
			\fill[red, opacity=0.14] (0,2,1)--(1,2,1)--(1,1,1)--(1,1,2)--(0,1,2)--(0,2,2);
			\draw[red, line width=1pt] (0.5,1.5,1.5) circle (0.18);
			\node[red, font=\small] at (0.5,1.5,1.5) {$c$};
			
			% Twist toggler \mathcal{T}_{f,z}(0,1/2) on y-edges at \k=0, \i\in[0,1], \j=0,1 (missing top \j=2!)
			\fill[ForestGreen, opacity=0.14] (0.22, -0.15, 0)--(0.78, -0.15, 0)--(0.78, 1.25, 0)--(0.22, 1.25, 0);
			\foreach \j in {0,1}{\node[ForestGreen] at (0.5, \j, 0) {$Z$};}
			\node[ForestGreen, font=\small] at (-0.65, 0.45, 1.15) {$\sim\tilde{X}_A^{\text{twist}}\tilde{X}_A^{\text{twist}}$};
			
			% Equivalent bottom X-loop pair \mathcal{X}_y(1)\mathcal{X}_y(2) at \j=0 (z=0), \k=1,2 (x=1,2)
			\fill[red, opacity=0.14] (0,0,0.68)--(0,0,1.32)--(3,0,1.32)--(3,0,0.68);
			\foreach \i in {0.5, 1.5, 2.5}{\node[red] at (\i,0,1) {$X$};}
			\node[red] at (3.55,0,1) {$\mathcal{X}_y(1)$};
			
			\fill[red, opacity=0.14] (0,0,1.68)--(0,0,2.32)--(3,0,2.32)--(3,0,1.68);
			\foreach \i in {0.5, 1.5, 2.5}{\node[red] at (\i,0,2) {$X$};}
			\node[red] at (3.55,0,2) {$\mathcal{X}_y(2)$};
			
			\node[red] at (1.5,-0.72,1.15) {$\tilde{Z}_A^{\text{twist}}(\tfrac32,\tfrac12)\equiv A_c^{\text{top}}\sim\mathcal{X}_y(1)\mathcal{X}_y(2)$};
			
			\fill[blue] (0,0,0) circle[radius=1.5pt];
			\node at (1.5,-1.18,1.15) {(b)};
		\end{tikzpicture}
	\end{subfigure}
	\vspace{-0.15cm}
	\caption{Deleted top stabilizers, twist indicators, and twist togglers under the $(em)$ top boundary:
	(a) $x$-direction $B$-twists ($L_y$ independent twist DOFs): on the untwisted $(em)$ sandwich, the product of bottom $\tilde Z\tilde Z$ bonds on a uncontractible $x$-line is fixed to $+1$.
	Deleting the $L_y$ top 3-leg stabilizers $B_{[0,j,L_z],zx}^{\text{top}}$ along the column $x(v)=0$ ($j=0,\dots,L_y-1$, marked by half-circles) lifts these $L_y$ row relations to independent twist indicators $\tilde{Z}_B^{\text{twist}}(0,j)\equiv B_{[0,j,L_z],zx}^{\text{top}}\sim\mathcal{Z}_x(j)$ (green, shown at $j=1$).
	A single vertical $X$-string $\mathcal{T}_{l_z,z}(0,j)=\prod_{k=0}^{L_z-1}X_{[0,j,k+1/2]}$ at $x=0$ commutes with all remaining stabilizers and anticommutes only with $B_{[0,j,L_z],zx}^{\text{top}}$, acting as the single twist toggler $\mathcal{T}_{l_z,z}(0,j)\sim\tilde{X}_B^{\text{twist}}(0,j)$ (red, shown at $j=0$).
	(b) $y$-direction $A$-twists ($L_x-1$ independent twist DOFs): similarly, the product of a column of bottom $\tilde X\tilde X$ bonds $\prod_{j=0}^{L_y-1}A_{[i+1/2,j+1/2,0]}^{\text{bot}}\sim\mathcal{X}_y(i)\mathcal{X}_y(i+1)$ is fixed to $+1$ by $\prod_{j=0}^{L_y-1}A_{[i+1/2,j+1/2,L_z-1/2]}^{\text{top}}=1$.
	Deleting the $L_x$ top 10-edge cube stabilizers $A_{[i+1/2,1/2,L_z-1/2]}^{\text{top}}$ along the row $y(c)=\frac12$ ($i=0,\dots,L_x-1$, marked by shaded cubes with circles) promotes them to twist indicators $\tilde{Z}_A^{\text{twist}}(i+\frac12,\frac12)\equiv A_{[i+1/2,1/2,L_z-1/2]}^{\text{top}}\sim\mathcal{X}_y(i)\mathcal{X}_y(i+1)$ (red, shown at $i=1$).
	Because the product of all cube stabilizers in the top layer equals $+1$, these $L_x$ indicators satisfy the global constraint $\prod_{i=0}^{L_x-1}\tilde{Z}_A^{\text{twist}}(i+\frac12,\frac12)=1$, leaving $L_x-1$ independent $A$-twist DOFs (and $L_x+L_y-1$ independent twist DOFs in total, matching the $2^{L_x+L_y-1}$ symmetry sectors).
	The shortened vertical $Z$-string $\mathcal{T}_{f,z}(i,\frac12)=\prod_{k=0}^{L_z-1}Z_{[i,1/2,k]}$ on $y$-edges (green, shown at $i=0$, containing $L_z$ operators $Z$ since the top $y$-edge at $z=L_z$ is truncated) anticommutes with the two adjacent deleted cubes at $i-1$ and $i$, and is therefore identified for every $i=0,\dots,L_x-1$ as the double twist toggler $\mathcal{T}_{f,z}(i,\frac12)\sim\tilde{X}_A^{\text{twist}}(i-\frac12,\frac12)\tilde{X}_A^{\text{twist}}(i+\frac12,\frac12)$, which automatically preserves $\prod_{i=0}^{L_x-1}\tilde{Z}_A^{\text{twist}}(i+\frac12,\frac12)=1$.}
	\label{fig_em_twists_2fig}
\end{figure*}

Now consider the X-cube FTH under $(em)$ and $(me)$ top boundaries.
The $(em)$ top boundary can be naturally realized on a mixed arrangement of top boundary shape: rough along $x$-direction while smooth along $y$-direction [see Fig.~\ref{fig_em_scheme2A}].
It is realized by putting top boundary truncated $A_c^{\text{top}}$ and $B_{v,zx}^{\text{top}}$ terms on such a boundary into the stabilizer set.
Like for the $(ee)$ and $(mm)$ top boundaries, take the dynamical bottom Hamiltonian to be $H_{\text{bot}}$ [see Eq.~(\ref{eq_X_cube_H_bot})], the identified 2d system is a $90^\circ$ quantum compass model (QCM) [see Fig.~\ref{fig_em_scheme2A}(b)].
QCM is a simplified model of 3d partially-filled transition metal oxides (Refs.~\cite{kugel_khomskii_1982_jahn_teller,nussinov_vandenbrink_2015_compass_review}), whose Hamiltonian consists of $ZZ$ terms along $x$-direction and $XX$ terms along $y$-direction.
Similarly, under the $(me)$ top boundary, the X-cube FTH is also identified as a QCM, but with the $ZZ$ terms along $y$-direction and $XX$ terms along $x$-direction.
The operator identification and top stabilizers are all the same as shown in Fig.~\ref{fig_em_scheme2A} but with the $x,y$-directions exchanged.
In all, the local term dictionary of the models identified from X-cube FTH is shown in Fig.~\ref{fig_X_cube_4_boundaries_dualities}.

The self-duality of QCM can be realized invertibly, by a diagonal reflection and a transversal Hadamard gate (Refs.~\cite{nussinov_fradkin_2005_compass_dualities,nussinov_vandenbrink_2015_compass_review}), which is non-local because the diagonal reflection is involved.
However, the QCM self-duality induced from the X-cube FTH is a different one, which is non-invertible and local.
The non-invertible duality operator is
\begin{equation}
	\mathsf{D}_{\text{QCM}} = 2^{-N}\sum_{\mathbf{s},\mathbf{t}}\sum_{\mathbf{u}\in\mathbb{F}_2^V} (-1)^{\mathbf{u}^T(M_x\mathbf{s}+M_y\mathbf{t})} |\mathbf{t}\rangle\langle\mathbf{s}|,
\end{equation}
where $\mathbf{s},\mathbf{t}$ are the $Z$-configurations of the $x$-edge QCM (identified under $(em)$ top) and $y$-edge QCM (identified under $(me)$ top), respectively.
$M_x$ and $M_y$ are incidence matrices with
\begin{equation}
	(M_x)_{v,e_x} =
	\begin{cases}
		1, & v\subset e_x,\\
		0, & v\not\subset e_x,
	\end{cases}
	\qquad
	(M_y)_{v,e_y} =
	\begin{cases}
		1, & v\subset e_y,\\
		0, & v\not\subset e_y.
	\end{cases}
	\label{eq_Mx_My_main}
\end{equation}
The intertwining relation is realized by $\mathsf{D}_{\text{QCM}}$,
\begin{align}
	\mathsf{D}_{\text{QCM}}\prod_{e_x\revsubset v}\tilde{Z}_{e_x} &= \prod_{e_y\revsubset v}\tilde{Z}'_{e_y}\mathsf{D}_{\text{QCM}}, \notag\\
	\mathsf{D}_{\text{QCM}}\prod_{e_x\subset p}\tilde{X}_{e_x} &= \prod_{e_y\subset p}\tilde{X}'_{e_y}\mathsf{D}_{\text{QCM}}.
\end{align}

\begin{figure*}[!t]
	\centering
	\begin{tikzpicture}[scale=0.72, every node/.style={transform shape}]
		% Macro to draw a 2D square grid
		\def\drawlattice#1#2#3#4{
			\begin{scope}[shift={(#1,#2)}]
				\draw[gray!40, step=1] (#3,#3) grid (#4,#4);
				\foreach \x in {#3,...,#4} {
					\foreach \y in {#3,...,#4} {
						\fill[gray!60] (\x,\y) circle (1.5pt);
					}
				}
			\end{scope}
		}
		
		% Common styles aligned with Fig. 1(b) of fracton_topological_holography.tex
		\tikzstyle{zblob} = [preaction={transform canvas={shift={(0.08,-0.08)}}, fill=black, fill opacity=0.12, draw=none}, fill=blue, fill opacity=0.1, draw=blue, draw opacity=0.4, thick, rounded corners=8pt]
		\tikzstyle{xblob} = [preaction={transform canvas={shift={(0.08,-0.08)}}, fill=black, fill opacity=0.12, draw=none}, fill=red, fill opacity=0.1, draw=red, draw opacity=0.4, thick, rounded corners=8pt]
		\tikzstyle{ztext} = [blue, fill=white, circle, inner sep=0.5pt]
		\tikzstyle{xtext} = [red, fill=white, circle, inner sep=0.5pt]
		
		% Background box
		\draw[fill=blue!2, draw=blue!10, thick, rounded corners=15pt] (0.25, -1.3) rectangle (20.65, 5.1);
		\node[font=\bfseries, align=center] at (10.45, 4.6) {X-cube Model FTH Dualities};
		
		% ================= 1. (ee) Top Boundary (Vertex TFPIM) =================
		\drawlattice{0}{0}{1}{4}
		\node[align=center, font=\small] at (2.5, -0.6) {$(ee)$ Top Boundary (TFPIM)\\[0.5ex]$\prod_{v\subset p}\tilde{Z}_v$ and $\tilde{X}_v$};
		\begin{scope}[shift={(0,0)}]
			% A_c^{bot} -> 4-body Z on vertices of plaquette (red)
			\draw[xblob] (2.6, 2.6) rectangle (4.4, 4.4);
			\node[xtext] at (3, 3) {$Z$};
			\node[xtext] at (4, 3) {$Z$};
			\node[xtext] at (3, 4) {$Z$};
			\node[xtext] at (4, 4) {$Z$};
			
			% B_{v,l}^{bot} -> 1-body X on vertex (blue)
			\draw[zblob] (1.6, 1.6) rectangle (2.4, 2.4);
			\node[ztext] at (2, 2) {$X$};
		\end{scope}
		
		\draw[<->, thick] (4.7, 2.45) -- (5.6, 2.45) node[midway, above, font=\footnotesize] {Dual};
		
		% ================= 2. (mm) Top Boundary (Plaquette TFPIM) =================
		\drawlattice{5.3}{0}{1}{4}
		\node[align=center, font=\small] at (7.8, -0.6) {$(mm)$ Top Boundary (TFPIM)\\[0.5ex]$\tilde{X}_p$ and $\prod_{p\revsubset v}\tilde{Z}_p$};
		\begin{scope}[shift={(5.3,0)}]
			% A_c^{bot} -> 1-body X on plaquette (red)
			\draw[xblob] (3.1, 3.1) rectangle (3.9, 3.9);
			\node[xtext] at (3.5, 3.5) {$X$};
			
			% B_{v,l}^{bot} -> 4-body Z on plaquettes around vertex (blue)
			\draw[zblob] (1.1, 1.1) rectangle (2.9, 2.9);
			\node[ztext] at (1.5, 1.5) {$Z$};
			\node[ztext] at (2.5, 1.5) {$Z$};
			\node[ztext] at (1.5, 2.5) {$Z$};
			\node[ztext] at (2.5, 2.5) {$Z$};
		\end{scope}
		
		\draw[<->, thick] (10.0, 2.45) -- (10.9, 2.45) node[midway, above, font=\footnotesize] {Dual};
		
		% ================= 3. (em) Top Boundary (QCM on x-edges) =================
		\drawlattice{10.6}{0}{1}{4}
		\node[align=center, font=\small] at (13.1, -0.6) {$(em)$ Top Boundary (QCM)\\[0.5ex]$\prod_{e_x\subset p}\tilde X_{e_x}$ and $\prod_{e_x\revsubset v}\tilde Z_{e_x}$};
		\begin{scope}[shift={(10.6,0)}]
			% A_c^{bot} -> 2-body XX along y on x-edges of plaquette (red)
			\draw[xblob] (3.1, 2.65) rectangle (3.9, 4.35);
			\node[xtext] at (3.5, 3) {$X$};
			\node[xtext] at (3.5, 4) {$X$};
			
			% B_{v,l}^{bot} -> 2-body ZZ along x on x-edges sharing vertex (blue)
			\draw[zblob] (1.1, 1.6) rectangle (2.9, 2.4);
			\node[ztext] at (1.5, 2) {$Z$};
			\node[ztext] at (2.5, 2) {$Z$};
		\end{scope}
		
		\draw[<->, thick] (15.3, 2.45) -- (16.2, 2.45) node[midway, above, font=\footnotesize] {Dual};
		
		% ================= 4. (me) Top Boundary (QCM on y-edges) =================
		\drawlattice{15.9}{0}{1}{4}
		\node[align=center, font=\small] at (18.4, -0.6) {$(me)$ Top Boundary (QCM)\\[0.5ex]$\prod_{e_y\subset p}\tilde X_{e_y}$ and $\prod_{e_y\revsubset v}\tilde Z_{e_y}$};
		\begin{scope}[shift={(15.9,0)}]
			% A_c^{bot} -> 2-body XX along x on y-edges of plaquette (red)
			\draw[xblob] (2.65, 3.1) rectangle (4.35, 3.9);
			\node[xtext] at (3, 3.5) {$X$};
			\node[xtext] at (4, 3.5) {$X$};
			
			% B_{v,l}^{bot} -> 2-body ZZ along y on y-edges sharing vertex (blue)
			\draw[zblob] (1.6, 1.1) rectangle (2.4, 2.9);
			\node[ztext] at (2, 1.5) {$Z$};
			\node[ztext] at (2, 2.5) {$Z$};
		\end{scope}
		
		% ================= Connecting mapping arrows =================
		% Red dashed arrows tracking A_c^{bot} across (ee) <-> (mm) <-> (em) <-> (me)
		\draw[<->, thick, red, densely dashed] (4.4, 3.5) -- (8.4, 3.5);
		\draw[<->, thick, red, densely dashed] (9.2, 3.5) -- (13.7, 3.5);
		\draw[<->, thick, red, densely dashed] (14.5, 3.5) -- (18.55, 3.5);
		
		% Blue dashed arrows tracking B_{v,l}^{bot} across (ee) <-> (mm) <-> (em) <-> (me)
		\draw[<->, thick, blue, densely dashed] (2.4, 2.0) -- (6.4, 2.0);
		\draw[<->, thick, blue, densely dashed] (8.2, 2.0) -- (11.7, 2.0);
		\draw[<->, thick, blue, densely dashed] (13.5, 2.0) -- (17.5, 2.0);
	\end{tikzpicture}
	\caption{Dualities between the 2d boundary models (TFPIM and QCM) induced by the X-cube FTH under the $(ee)$, $(mm)$, $(em)$, and $(me)$ top boundaries.
	Red and blue shaded blocks trace the operator identifications of the bottom boundary gauge operators $A_c^{\text{bot}}$ and $B_{v,zx}^{\text{bot}}\sim B_{v,yz}^{\text{bot}}$, respectively.}
	\label{fig_X_cube_4_boundaries_dualities}
\end{figure*}

Like the TFPIM self-duality case, calculation shows that $\mathsf{D}_{\text{QCM}}^\dagger\mathsf{D}_{\text{QCM}}$ is proportional to the projector onto the $+1$ common eigenspace of the line-like symmetry indicators of the $x$-edge QCM.
The difference is that the line-like symmetry indicators, identified from the X-cube logical operators (see Fig.~\ref{fig_em_symmetries_2fig}), consist of both $\tilde{X}$ and $\tilde{Z}$ products now:
\begin{align}
	S_x^{(X)}(j) &\equiv \prod_{i=0}^{L_x-1}\tilde{X}_{[i+1/2,j]}, \quad j=0,1,\dots,L_y-1, \notag\\
	S_y^{(Z)}(i) &\equiv \prod_{j=0}^{L_y-1}\tilde{Z}_{[i-1/2,j]}\tilde{Z}_{[i+1/2,j]}, \quad i=0,1,\dots,L_x-1.\label{eq_QCM_line_symmetry_indicators}
\end{align}
The $-1$ eigenstate of any symmetry indicator is mapped by $\mathsf{D}_{\text{QCM}}$ to zero, so $\mathsf{D}_{\text{QCM}}$ is non-invertible.
Also, the symmetry absorption is present here.
$\mathsf{D}_{\text{QCM}}S_x^{(X)}(j) = \mathsf{D}_{\text{QCM}}S_y^{(Z)}(i) = \mathsf{D}_{\text{QCM}}$.
The line-like subsystem symmetries of $y$-edge QCM can be absorbed on the left side of $\mathsf{D}_{\text{QCM}}$.
See step-by-step details in Appendix~\ref{app_subsec_D_QCM}.

Similar to the $(ee)\leftrightarrow(mm)$ case, twist DOFs can be introduced by deleting specific top stabilizer generators (allowing top boundary defects) [see Fig.~\ref{fig_em_twists_2fig}, and GSD counting details in Appendix~\ref{app_X_cube_em_me_counting}].
When switching the top boundary from $(em)$ to $(me)$, the identities of the symmetry operators and twist operators are exchanged, and the QCM self-duality becomes invertible.

The single-row symmetry indicators $S_x^{(X)}(j)$ are not on an equal footing.
The products $S_x^{(X)}(j)S_x^{(X)}(j+1)$ can be moved to the bottom boundary by multiplying stabilizer [see Fig.~\ref{fig_em_symmetries_2fig}(a)], while $S_x^{(X)}(0)\sim\prod_{i=0}^{L_x-1}X_{[i+1/2,0,L_z]}$ cannot.
Under the $(em)$ top boundary, it is identified as a QCM symmetry operator, together with its anticommuting pair, the single-column string $\prod_{j=0}^{L_y-1}\tilde Z_{[1/2,j]}$.
This pair commutes with all the QCM bonds and twist operators, so every energy level of the QCM is (at least) two-fold degenerate, a known property of the QCM~\cite{doucot_feigelman_ioffe_ioselevich_2005_protected_qubits}.
Consequently, with twist DOFs included, the QCM has $L_xL_y+L_x+L_y-1$ qubits, twice the Hilbert space dimension of the TFPIMs, which have $L_xL_y+L_x+L_y-2$ qubits.
The two match if removing this pair; see Appendices~\ref{app_X_cube_em_me_counting} and~\ref{appendix_X_cube_details_em} for details.

The duality between TFPIM and QCM (Ref.~\cite{nussinov_fradkin_2005_compass_dualities}) is now understood as the switch between $(ee)$ and $(em)$ top boundaries of the X-cube FTH.
Together with the self-dualities of TFPIM and QCM, a duality web is given by the X-cube FTH, with the local operator dictionary given in Fig.~\ref{fig_X_cube_4_boundaries_dualities}.
The duality between TFPIM and QCM induced by the X-cube FTH is also non-invertible, with the duality operator
\begin{equation}
	\mathsf{D}_{\text{QC}\to\text{TFPI}} = 2^{-N/2}\sum_{\mathbf{s},\mathbf{t}} (-1)^{\mathbf{t}^T M_x \mathbf{s}} |\mathbf{t}\rangle\langle\mathbf{s}|,
\end{equation}
where $\mathbf{s}\in\mathbb{F}_2^{E_x}$, $\mathbf{t}\in\mathbb{F}_2^V$ are $Z$ configurations of $x$-edge QCM and vertex TFPIM, respectively.
$M_x$ is the incidence matrix defined in Eq.~(\ref{eq_Mx_My_main}).
$\mathsf{D}_{\text{QC}\to\text{TFPI}}$ realizes the intertwining relation shown in Fig.~\ref{fig_X_cube_4_boundaries_dualities}, and $\mathsf{D}_{\text{QC}\to\text{TFPI}}^\dagger\mathsf{D}_{\text{QC}\to\text{TFPI}}$ is proportional to the projector onto the $+1$ symmetry subspace of $x$-edge QCM, $\mathsf{D}_{\text{QC}\to\text{TFPI}}$ is non-invertible.
Unlike the TFPIM KW duality and QCM self-duality, $\mathsf{D}_{\text{QC}\to\text{TFPI}}$ does not exchange the identities of all symmetry operators and twist operators.
Instead, it maps the line-like symmetries on the $x$-lines of the QCM to relations of TFPIM, and maps the line-like symmetries on the $y$-lines of the QCM to the line-like symmetries on the $y$-lines of the TFPIM, see Appendix~\ref{app_subsec_D_QCM_TFPIM} for details.

%======================================================================
\section{Haah's code FTH, generalized Ising model and duality}\label{sec_Haah_code_FTH}
%======================================================================

%======================================================================
% Haah's cubic code single-Pauli excitation configurations
%======================================================================
\newcommand{\drawHaahGridOblique}{%
	\foreach \j in {0,1,2}{\foreach \k in {0,1,2}{\draw[gray!40, line width=0.35pt] (0,\j,\k)--(2,\j,\k);}}
	\foreach \i in {0,1,2}{\foreach \j in {0,1,2}{\draw[gray!40, line width=0.35pt] (\i,\j,0)--(\i,\j,2);}}
	\foreach \i in {0,1,2}{\foreach \k in {0,1,2}{\draw[gray!40, line width=0.35pt] (\i,0,\k)--(\i,2,\k);}}
}

\newcommand{\drawVoxelCube}[4]{%
	% #1 = \i (y), #2 = \j (z), #3 = \k (x), #4 = color
	% Top face (\j+1)
	\fill[#4, opacity=0.26] (#1,#2+1,#3)--(#1+1,#2+1,#3)--(#1+1,#2+1,#3+1)--(#1,#2+1,#3+1)--cycle;
	% Front-left face (\k+1)
	\fill[#4, opacity=0.18] (#1,#2,#3+1)--(#1+1,#2,#3+1)--(#1+1,#2+1,#3+1)--(#1,#2+1,#3+1)--cycle;
	% Front-right face (\i+1)
	\fill[#4, opacity=0.12] (#1+1,#2,#3)--(#1+1,#2+1,#3)--(#1+1,#2+1,#3+1)--(#1+1,#2,#3+1)--cycle;
	% Visible edges of the excited voxel
	\draw[#4!85!black, line width=0.55pt]
		(#1,#2+1,#3)--(#1+1,#2+1,#3)--(#1+1,#2+1,#3+1)--(#1,#2+1,#3+1)--cycle
		(#1,#2,#3+1)--(#1+1,#2,#3+1)--(#1+1,#2+1,#3+1)
		(#1,#2,#3+1)--(#1,#2+1,#3+1)
		(#1+1,#2,#3)--(#1+1,#2+1,#3)
		(#1+1,#2,#3)--(#1+1,#2,#3+1);
}

% Half-circle vertex operator markers aligned with 261015_hyt.tex
\newcommand{\drawQubitOneOp}[2]{%
	% #1 = 3D coord, #2 = color (left semicircle filled)
	\fill[white] #1 circle (0.1cm);
	\fill[#2] #1 ++(90:0.1cm) arc[start angle=90,end angle=270,radius=0.1cm] -- #1 -- cycle;
	\draw[black, line width=0.6pt] #1 circle (0.1cm);
}

\newcommand{\drawQubitTwoOp}[2]{%
	% #1 = 3D coord, #2 = color (right semicircle filled)
	\fill[white] #1 circle (0.1cm);
	\fill[#2] #1 ++(90:0.1cm) arc[start angle=90,end angle=-90,radius=0.1cm] -- #1 -- cycle;
	\draw[black, line width=0.6pt] #1 circle (0.1cm);
}

In this section, we show that FTH can also be constructed with type-II fracton bulk, Haah's cubic code.
The difficulty of constructing type-II FTH is that excitations cannot move, no Wilson line exists.
The difficulty is overcome by replacing the Wilson lines connecting top and bottom boundaries by transport operators, which transport excitations along a direction, while allowing the change of particle number during transportation.

Haah's code 1\footnote{There are 18 Haah's codes in Haah's original paper (Ref.~\cite{haah_2011_no_strings}).
Haah's code refers to Haah's code 1 in this paper by default.} is defined on a cubic lattice, with each vertex hosting two qubits.
The Hamiltonian is the sum of $-S_X, -S_Z$ and their translations, where
\begin{equation}\label{eq_Haah_SX_SZ}
	S_X = \begin{tikzpicture}[scale=1.18, baseline=1.6ex, font=\footnotesize]
		\foreach \i/\j in {0/0, 1/0, 1/1, 0/1}{\draw (\i,\j,0.32)--(\i,\j,0.68);}
		\foreach \i/\k in {0/0, 1/0, 1/1, 0/1}{\draw (\i,0.18,\k)--(\i,0.82,\k);}
		\foreach \j/\k in {0/0, 1/0, 1/1, 0/1}{\draw (0.24,\j,\k)--(0.76,\j,\k);}
		\node[red] at (0,0,0) {$XX$};
		\node[red] at (1,0,0) {$IX$};
		\node[red] at (1,0,1) {$XI$};
		\node[red] at (0,0,1) {$IX$};
		\node[red] at (0,1,0) {$IX$};
		\node[red] at (1,1,0) {$XI$};
		\node at (1,1,1) {$II$};
		\node[red] at (0,1,1) {$XI$};
	\end{tikzpicture},\quad
	S_Z = \begin{tikzpicture}[scale=1.18, baseline=1.6ex, font=\footnotesize]
		\foreach \i/\j in {0/0, 1/0, 1/1, 0/1}{\draw (\i,\j,0.32)--(\i,\j,0.68);}
		\foreach \i/\k in {0/0, 1/0, 1/1, 0/1}{\draw (\i,0.18,\k)--(\i,0.82,\k);}
		\foreach \j/\k in {0/0, 1/0, 1/1, 0/1}{\draw (0.24,\j,\k)--(0.76,\j,\k);}
		\node at (0,0,0) {$II$};
		\node[ForestGreen] at (1,0,0) {$IZ$};
		\node[ForestGreen] at (1,0,1) {$ZI$};
		\node[ForestGreen] at (0,0,1) {$IZ$};
		\node[ForestGreen] at (0,1,0) {$IZ$};
		\node[ForestGreen] at (1,1,0) {$ZI$};
		\node[ForestGreen] at (1,1,1) {$ZZ$};
		\node[ForestGreen] at (0,1,1) {$ZI$};
	\end{tikzpicture}.
\end{equation}
$S_X$ commutes with $S_Z$ and all its translations, so the Hamiltonian dictates a stabilizer code.
Haah's code satisfies the TO/fracton condition that no non-trivial local logical operator exists, where logical operator here means the operator commuting with all stabilizers (i.e., $S_X, S_Z$ and all their translations) but not generated by stabilizers.
There are two elementary excitations, which we denote by $e_X, e_Z$, with representative syndromes $S_X=-1$ and $S_Z=-1$, respectively.
Single $e_X$ or $e_Z$ cannot be created or annihilated by finite-support operators, but are created four at once by single Pauli operators, as shown in Fig.~\ref{fig_Haah_single_Pauli_excitations}.

\begin{figure}[!t]
	\centering
	% Top coordinate axes and legend bar
	\begin{tikzpicture}[baseline=1ex, x={(0.82cm,-0.07cm)}, y={(0cm,0.82cm)}, z={(-0.46cm,-0.28cm)}]
		\draw[->, blue, line width=0.8pt] (0,0,0)--(0,0,0.8) node[below left, inner sep=1pt, font=\footnotesize] {$x$};
		\draw[->, blue, line width=0.8pt] (0,0,0)--(0.8,0,0) node[right, inner sep=1pt, font=\footnotesize] {$y$};
		\draw[->, blue, line width=0.8pt] (0,0,0)--(0,0.8,0) node[above, inner sep=1pt, font=\footnotesize] {$z$};
	\end{tikzpicture}%
	\hspace{1.6em}%
	\begin{tikzpicture}[baseline=-0.5ex]
		\drawQubitOneOp{(0,0)}{red}
		\node[font=\small, right] at (0.12,0) {$:X_1$};
		\drawQubitTwoOp{(1.45,0)}{red}
		\node[font=\small, right] at (1.57,0) {$:X_2$};
		\drawQubitOneOp{(2.9,0)}{ForestGreen}
		\node[font=\small, right] at (3.02,0) {$:Z_1$};
		\drawQubitTwoOp{(4.35,0)}{ForestGreen}
		\node[font=\small, right] at (4.47,0) {$:Z_2$};
	\end{tikzpicture}

	\vspace{1ex}
	% (a) Z_1 -> d S_X
	\begin{subfigure}{0.48\columnwidth}
		\centering
		\begin{tikzpicture}[baseline=4ex, x={(0.82cm,-0.07cm)}, y={(0cm,0.82cm)}, z={(-0.46cm,-0.28cm)}]
			\drawHaahGridOblique
			\foreach \i/\j/\k in {1/0/0, 0/1/0, 0/0/1, 1/1/1}{
				\drawVoxelCube{\i}{\j}{\k}{red}
			}
			\drawQubitOneOp{(1,1,1)}{ForestGreen}
			\node[font=\small] at (1,-0.8,1) {(a)};
		\end{tikzpicture}
	\end{subfigure}
	\hfill
	% (b) Z_2 -> c S_X
	\begin{subfigure}{0.48\columnwidth}
		\centering
		\begin{tikzpicture}[baseline=4ex, x={(0.82cm,-0.07cm)}, y={(0cm,0.82cm)}, z={(-0.46cm,-0.28cm)}]
			\drawHaahGridOblique
			\foreach \i/\j/\k in {1/1/0, 1/0/1, 0/1/1, 1/1/1}{
				\drawVoxelCube{\i}{\j}{\k}{red}
			}
			\drawQubitTwoOp{(1,1,1)}{ForestGreen}
			\node[font=\small] at (1,-0.8,1) {(b)};
		\end{tikzpicture}
	\end{subfigure}

	\vspace{1.2ex}
	% (c) X_1 -> b S_Z
	\begin{subfigure}{0.48\columnwidth}
		\centering
		\begin{tikzpicture}[baseline=4ex, x={(0.82cm,-0.07cm)}, y={(0cm,0.82cm)}, z={(-0.46cm,-0.28cm)}]
			\drawHaahGridOblique
			\foreach \i/\j/\k in {0/0/0, 1/0/0, 0/1/0, 0/0/1}{
				\drawVoxelCube{\i}{\j}{\k}{ForestGreen}
			}
			\drawQubitOneOp{(1,1,1)}{red}
			\node[font=\small] at (1,-0.8,1) {(c)};
		\end{tikzpicture}
	\end{subfigure}
	\hfill
	% (d) X_2 -> a S_Z
	\begin{subfigure}{0.48\columnwidth}
		\centering
		\begin{tikzpicture}[baseline=4ex, x={(0.82cm,-0.07cm)}, y={(0cm,0.82cm)}, z={(-0.46cm,-0.28cm)}]
			\drawHaahGridOblique
			\foreach \i/\j/\k in {0/0/0, 1/1/0, 1/0/1, 0/1/1}{
				\drawVoxelCube{\i}{\j}{\k}{ForestGreen}
			}
			\drawQubitTwoOp{(1,1,1)}{red}
			\node[font=\small] at (1,-0.8,1) {(d)};
		\end{tikzpicture}
	\end{subfigure}
	\caption{Excitation configurations created by single Pauli operators.
	Left- and right-filled semicircles represent single-Pauli operators acting on qubit $1$ and qubit $2$ (green for $Z_{1,2}$ and red for $X_{1,2}$), respectively.
	(a) $Z_1$ excites four $e_X$ fractons ($S_X=-1$, red shaded cubes).
	(b) $Z_2$ excites four $e_X$ fractons.
	(c) $X_1$ excites four $e_Z$ fractons ($S_Z=-1$, green shaded cubes).
	(d) $X_2$ excites four $e_Z$ fractons.}
	\label{fig_Haah_single_Pauli_excitations}
\end{figure}
\begin{figure*}[!t]
	\centering
	\begin{tikzpicture}
		% graph explanation
		\begin{scope}[yshift=8.6cm]
			\fill[white] (-3,0,0) circle (0.1cm);
			\fill[red] (-3,0,0) ++(90:0.1cm) arc[start angle=90,end angle=270,radius=0.1cm] -- (-3,0,0) -- cycle;
			\draw[black, line width=0.6pt] (-3,0,0) circle (0.1cm);
			\node at (-2.45,0,0) {$:X_1$};
		\end{scope}

		\begin{scope}[xshift=2cm, yshift=8.6cm]
			\fill[white] (-3,0,0) circle (0.1cm);
			\fill[red] (-3,0,0) ++(90:0.1cm) arc[start angle=90,end angle=-90,radius=0.1cm] -- (-3,0,0) -- cycle;
			\draw[black, line width=0.6pt] (-3,0,0) circle (0.1cm);
			\node at (-2.45,0,0) {$:X_2$};
		\end{scope}

		\begin{scope}[xshift=4cm, yshift=8.6cm]
			\fill[white] (-3,0,0) circle (0.1cm);
			\fill[red] (-3,0,0) circle (0.1cm);
			\draw[black, line width=0.6pt] (-3,0,0) circle (0.1cm);
			\node at (-2.25,0,0) {$:X_1X_2$};
		\end{scope}

		\begin{scope}[xshift=7cm, yshift=8.6cm]
			\fill[white] (-3,0,0) circle (0.1cm);
			\fill[ForestGreen] (-3,0,0) ++(90:0.1cm) arc[start angle=90,end angle=270,radius=0.1cm] -- (-3,0,0) -- cycle;
			\draw[black, line width=0.6pt] (-3,0,0) circle (0.1cm);
			\node at (-2.45,0,0) {$:Z_1$};
		\end{scope}

		\begin{scope}[xshift=9cm, yshift=8.6cm]
			\fill[white] (-3,0,0) circle (0.1cm);
			\fill[ForestGreen] (-3,0,0) ++(90:0.1cm) arc[start angle=90,end angle=-90,radius=0.1cm] -- (-3,0,0) -- cycle;
			\draw[black, line width=0.6pt] (-3,0,0) circle (0.1cm);
			\node at (-2.45,0,0) {$:Z_2$};
		\end{scope}

		\begin{scope}[xshift=11cm, yshift=8.6cm]
			\fill[white] (-3,0,0) circle (0.1cm);
			\fill[ForestGreen] (-3,0,0) circle (0.1cm);
			\draw[black, line width=0.6pt] (-3,0,0) circle (0.1cm);
			\node at (-2.25,0,0) {$:Z_1Z_2$};
		\end{scope}

		% boundary gauge operator 1
		\begin{scope}[
			x={(-0.6cm,-0.2cm)},
			y={(0.6cm,-0.15cm)},
			z={(0cm,1cm)},
			yshift=6.35cm]
			% 1. 底面网格 (z=0 平面)
			\foreach \i in {0,...,4} \draw[gray, opacity=0.55, line width=0.4pt] (\i,-1,0) -- (\i,5,0);
			\foreach \j in {0,...,4} \draw[gray, opacity=0.55, line width=0.4pt] (-1,\j,0) -- (5,\j,0);

			% 4. 显式的坐标轴箭头
			\draw[-latex, thick, black!70] (0,0,0) -- (5.2,0,0) node[below left] {$x$};
			\draw[-latex, thick, black!70] (0,0,0) -- (0,5.2,0) node[below right] {$y$};
			\draw[-latex, thick, black!70] (0,0,0) -- (0,0,1.2) node[above] {$z$};

			% bottom boundary gauge syndrome
			\foreach \i/\j in {0/0, 1/0, 2/0, 0/1, 0/2, 1/1}{
				\fill[ForestGreen, opacity=0.25] (\i,\j,0)--(\i-1,\j,0)--(\i-1,\j-1,0)--(\i,\j-1,0)--cycle;
			}

			% boundary gauge operator 1 times xy
			\foreach \i/\j in {1/0, 0/1}{
				\fill[white] (\i,\j,0) circle (0.1cm);
				\fill[red] (\i,\j,0) ++(90:0.1cm) arc[start angle=90,end angle=270,radius=0.1cm] -- (\i,\j,0) -- cycle;
				\draw[black, line width=0.6pt] (\i,\j,0) circle (0.1cm);
			}
			\foreach \i/\j in {0/0}{
				\fill[white] (\i,\j,0) circle (0.1cm);
				\fill[red] (\i,\j,0) ++(90:0.1cm) arc[start angle=90,end angle=-90,radius=0.1cm] -- (\i,\j,0) -- cycle;
				\draw[black, line width=0.6pt] (\i,\j,0) circle (0.1cm);
			}
		\end{scope}

		\begin{scope}[yshift=6.35cm]
			\draw[->] (3.5+0.4,0.3,0)--(5+0.4,0.3,0);
			\node at (4.25+0.4,0.6,0) {$\Phi_Z$};
			\node at (-4,1.2,0) {(a)};
		\end{scope}

		\begin{scope}[xshift=7cm, yshift=5.15cm, scale=0.6]
			\draw[-latex, thick, black!70] (0,0,0) -- (5.6,0) node[below left] {$x$};
			\draw[-latex, thick, black!70] (0,0,0) -- (0,5.6) node[below right] {$y$};
			\foreach \i in {0,...,4} \draw[gray, opacity=0.55, line width=0.4pt] (\i,-1,0) -- (\i,5,0);
			\foreach \j in {0,...,4} \draw[gray, opacity=0.55, line width=0.4pt] (-1,\j,0) -- (5,\j,0);

			\foreach \i/\j in {0/0, 1/0, 2/0, 0/1, 0/2, 1/1}{
				\node at (\i-0.5,\j-0.5) {$\tilde Z$};
				\fill[ForestGreen, opacity=0.25] (\i,\j,0)--(\i-1,\j,0)--(\i-1,\j-1,0)--(\i,\j-1,0)--cycle;
			}
		\end{scope}

		% e_Z transport operator
		\begin{scope}[
			x={(-0.6cm,-0.2cm)},
			y={(0.6cm,-0.15cm)},
			z={(0cm,0.9cm)}]
			% 1. 底面网格 (z=0 平面)
			\foreach \i in {0,...,4} \draw[gray, opacity=0.55, line width=0.4pt] (\i,-1,0) -- (\i,5,0);
			\foreach \j in {0,...,4} \draw[gray, opacity=0.55, line width=0.4pt] (-1,\j,0) -- (5,\j,0);

			% 2. 左背墙网格 (y=0 平面)
			\foreach \i in {0,...,4} \draw[gray, opacity=0.25, line width=0.4pt] (\i,0,0) -- (\i,0,4);
			\foreach \k in {0,...,4} \draw[gray, opacity=0.25, line width=0.4pt] (0,0,\k) -- (5,0,\k);

			% 3. 右背墙网格 (x=0 平面)
			\foreach \j in {0,...,4} \draw[gray, opacity=0.25, line width=0.4pt] (0,\j,0) -- (0,\j,4);
			\foreach \k in {0,...,4} \draw[gray, opacity=0.25, line width=0.4pt] (0,0,\k) -- (0,5,\k);

			% 4. 显式的坐标轴箭头
			\draw[-latex, thick, black!70] (0,0,0) -- (5.2,0,0) node[below left] {$x$};
			\draw[-latex, thick, black!70] (0,0,0) -- (0,5.2,0) node[below right] {$y$};
			\draw[-latex, thick, black!70] (0,0,0) -- (0,0,4.6) node[above] {$z$};

			% layer indicator
			\fill[gray, opacity=0.25] (0,0,1)--(3,0,1)--(0,3,1)--cycle;
			\fill[gray, opacity=0.25] (0,0,2)--(2,0,2)--(0,2,2)--cycle;
			\fill[gray, opacity=0.25] (0,0,3)--(1,0,3)--(0,1,3)--cycle;

			% bottom boundary gauge syndrome
			\foreach \i/\j in {0/0, 1/0, 0/1, 4/0, 5/0, 4/1, 0/4, 1/4, 0/5}{
				\fill[ForestGreen, opacity=0.25] (\i,\j,0)--(\i-1,\j,0)--(\i-1,\j-1,0)--(\i,\j-1,0)--cycle;
			}

			% z=4 layer
			\foreach \i/\j in {0/0}{
				\fill[white] (\i,\j,4) circle (0.1cm);
				\fill[red] (\i,\j,4) ++(90:0.1cm) arc[start angle=90,end angle=270,radius=0.1cm] -- (\i,\j,4) -- cycle;
				\draw[black, line width=0.6pt] (\i,\j,4) circle (0.1cm);
			}

			% z=3 layer
			\foreach \i/\j in {0/0, 1/0, 0/1}{
				\fill[white] (\i,\j,3) circle (0.1cm);
				\fill[red] (\i,\j,3) ++(90:0.1cm) arc[start angle=90,end angle=270,radius=0.1cm] -- (\i,\j,3) -- cycle;
				\draw[black, line width=0.6pt] (\i,\j,3) circle (0.1cm);
			}

			% z=2 layer
			\foreach \i/\j in {0/0, 2/0, 0/2}{
				\fill[white] (\i,\j,2) circle (0.1cm);
				\fill[red] (\i,\j,2) ++(90:0.1cm) arc[start angle=90,end angle=270,radius=0.1cm] -- (\i,\j,2) -- cycle;
				\draw[black, line width=0.6pt] (\i,\j,2) circle (0.1cm);
			}

			% z=1 layer
			\foreach \i/\j in {0/0, 1/0, 0/1, 2/0, 3/0, 2/1, 0/2, 1/2, 0/3}{
				\fill[white] (\i,\j,1) circle (0.1cm);
				\fill[red] (\i,\j,1) ++(90:0.1cm) arc[start angle=90,end angle=270,radius=0.1cm] -- (\i,\j,1) -- cycle;
				\draw[black, line width=0.6pt] (\i,\j,1) circle (0.1cm);
			}

			% z=0 layer
			\foreach \i/\j in {0/0, 4/0, 0/4}{
				\fill[white] (\i,\j,0) circle (0.1cm);
				\fill[red] (\i,\j,0) ++(90:0.1cm) arc[start angle=90,end angle=270,radius=0.1cm] -- (\i,\j,0) -- cycle;
				\draw[black, line width=0.6pt] (\i,\j,0) circle (0.1cm);
			}
		\end{scope}

		\draw[->] (3.5+0.4,1.3,0)--(5+0.4,1.3,0);
		\node at (4.25+0.4,1.6,0) {$\Phi_Z$};
		\node at (-4,4.2,0) {(b)};

		\begin{scope}[xshift=7cm, yshift=0.1cm, scale=0.6]
			\draw[-latex, thick, black!70] (0,0,0) -- (5.6,0) node[below left] {$x$};
			\draw[-latex, thick, black!70] (0,0,0) -- (0,5.6) node[below right] {$y$};
			\foreach \i in {0,...,4} \draw[gray, opacity=0.55, line width=0.4pt] (\i,-1,0) -- (\i,5,0);
			\foreach \j in {0,...,4} \draw[gray, opacity=0.55, line width=0.4pt] (-1,\j,0) -- (5,\j,0);

			\foreach \i/\j in {0/0, 1/0, 0/1, 4/0, 5/0, 4/1, 0/4, 1/4, 0/5}{
				\node at (\i-0.5,\j-0.5) {$\tilde Z$};
				\fill[ForestGreen, opacity=0.25] (\i,\j,0)--(\i-1,\j,0)--(\i-1,\j-1,0)--(\i,\j-1,0)--cycle;
			}
		\end{scope}

		% boundary gauge operator 2
		\begin{scope}[
			x={(-0.6cm,-0.2cm)},
			y={(0.6cm,-0.15cm)},
			z={(0cm,1cm)}, yshift=-3.2cm]
			% 1. 底面网格 (z=0 平面)
			\foreach \i in {0,...,4} \draw[gray, opacity=0.55, line width=0.4pt] (\i,-1,0) -- (\i,5,0);
			\foreach \j in {0,...,4} \draw[gray, opacity=0.55, line width=0.4pt] (-1,\j,0) -- (5,\j,0);

			% 4. 显式的坐标轴箭头
			\draw[-latex, thick, black!70] (0,0,0) -- (5.2,0,0) node[below left] {$x$};
			\draw[-latex, thick, black!70] (0,0,0) -- (0,5.2,0) node[below right] {$y$};
			\draw[-latex, thick, black!70] (0,0,0) -- (0,0,1.2) node[above] {$z$};

			% bottom boundary gauge syndrome
			\foreach \i/\j in {2/2, 1/2, 0/2, 2/1, 2/0, 1/1}{
				\fill[red, opacity=0.25] (\i,\j,0)--(\i-1,\j,0)--(\i-1,\j-1,0)--(\i,\j-1,0)--cycle;
			}

			% boundary gauge operator 2 times xy
			\foreach \i/\j in {1/0, 0/1}{
				\fill[white] (\i,\j,0) circle (0.1cm);
				\fill[ForestGreen] (\i,\j,0) ++(90:0.1cm) arc[start angle=90,end angle=270,radius=0.1cm] -- (\i,\j,0) -- cycle;
				\draw[black, line width=0.6pt] (\i,\j,0) circle (0.1cm);
			}
			\foreach \i/\j in {0/0}{
				\fill[white] (\i,\j,0) circle (0.1cm);
				\fill[ForestGreen] (\i,\j,0) ++(90:0.1cm) arc[start angle=90,end angle=-90,radius=0.1cm] -- (\i,\j,0) -- cycle;
				\draw[black, line width=0.6pt] (\i,\j,0) circle (0.1cm);
			}
			\foreach \i/\j in {1/1}{
				\fill[ForestGreen] (\i,\j,0) circle (0.1cm);
				\draw[black, line width=0.6pt] (\i,\j,0) circle (0.1cm);
			}
		\end{scope}

		\begin{scope}[yshift=-3.3cm]
			\draw[->] (3.5+0.4,0.3,0)--(5+0.4,0.3,0);
			\node at (4.25+0.4,0.6,0) {$\Phi_Z$};
			\node at (-4,1.2,0) {(c)};
		\end{scope}

		\begin{scope}[xshift=7cm, yshift=-4.5cm, scale=0.6]
			\draw[-latex, thick, black!70] (0,0,0) -- (5.6,0) node[below left] {$x$};
			\draw[-latex, thick, black!70] (0,0,0) -- (0,5.6) node[below right] {$y$};
			\foreach \i in {0,...,4} \draw[gray, opacity=0.55, line width=0.4pt] (\i,-1,0) -- (\i,5,0);
			\foreach \j in {0,...,4} \draw[gray, opacity=0.55, line width=0.4pt] (-1,\j,0) -- (5,\j,0);

			\foreach \i/\j in {1/1}{
				\node at (\i-0.5,\j-0.5) {$\tilde X$};
				\fill[red, opacity=0.25] (\i,\j,0)--(\i-1,\j,0)--(\i-1,\j-1,0)--(\i,\j-1,0)--cycle;
			}
		\end{scope}
	\end{tikzpicture}
	\vspace{-0.15cm}
	\caption{The illustration of operator identification of Haah's code FTH under $(Z)$ top boundary.
	The red circles and green circles represent Pauli $X$ and $Z$, respectively, as shown on the top of this figure.
	On the left side, generators of low-energy-preserving operators and their bottom boundary gauge syndromes are illustrated.
	Green squares represent $S_Z^{\text{bot}}$ violation, red squares represent $S_X^{\text{bot}}$ violation.
	The right side illustrates the identified operators in the 2d boundary system under $(Z)$ top boundary.
	(a) The bottom boundary gauge operator $S_X^{\text{bot}}$, which has the bottom boundary gauge syndrome $\bar x\bar y F S_Z^{\text{bot}}$, is identified as $\bar x\bar y F\tilde Z$.
	(b) The $e_Z$-transport operator along $z$-direction $\bar x\bar y\mathcal T_{e_Z,z}$, which has the bottom boundary gauge syndrome $\bar x\bar y A^{L_z} S_Z^{\text{bot}}$, is identified as $\bar x\bar y A^{L_z}\tilde Z$.
	Note that $A^{L_z}$ ($A\equiv 1+x+y$) is $L_z$ dependent, here $L_z=5$.
	(c) The bottom boundary gauge operator $S_Z^{\text{bot}}$, which has the bottom boundary gauge syndrome $xy\bar F S_X^{\text{bot}}$, is identified as $\tilde X$.}
	\label{fig_Haah_code_FTH_operator_illustration_Z_top}
\end{figure*}

It is useful to introduce the polynomial representation of Pauli operators, in order to describe the excitation configurations and boundary gauge operators, transport operators (which we introduce soon).
Given the syndrome $S_X$, we use $x^i y^j S_X$ to represent the syndrome obtained by moving $S_X$ $i$ times the unit distance to the $x$-direction, and $j$ times the unit distance to the $y$-direction.
Polynomial like $(1+x)S_X$ represents the combination of syndrome $S_X$ and $xS_X$.
In this language\footnote{See a pedagogical review in Appendix~\ref{app_poly_review}, and a more detailed review in e.g.\ Refs.~\cite{haah_2013_modules,haah_2016_algebraic_methods}.}, we can write the syndromes of single Pauli in Fig.~\ref{fig_Haah_single_Pauli_excitations} as
\begin{align}\label{eq_Haah_single_Pauli_syndromes}
	\epsilon(Z_1) &= (x+y+z+xyz)S_X, \notag\\
	\epsilon(Z_2) &= (xy+yz+zx+xyz)S_X, \notag\\
	\epsilon(X_1) &= (1+x+y+z)S_Z, \notag\\
	\epsilon(X_2) &= (1+xy+yz+zx)S_Z.
\end{align}

Put Haah's code on a sandwich truncated by $z=0$ and $z=L_z-1$ planes, with $x,y$-directions infinite.
There are two different gapped top boundaries (Ref.~\cite{Haah_code_boundary_2024}), which we dub $(Z)$ and $(X)$ top boundary:
\begin{itemize}
	\item $(Z)$ top boundary: top boundary truncated $S_X^{\text{top}}$ terms (i.e.\ the bottom layer of $S_X$, on the $z=L_z-1$ top layer) are present in the stabilizer group, while the top boundary truncated $S_Z^{\text{top}}$ terms (i.e.\ the bottom layer of $S_Z$, on the $z=L_z-1$ top layer) are not.
	On the $(Z)$ top boundary, $e_Z$ are condensed.
	\item $(X)$ top boundary: $S_Z^{\text{top}}$ terms are present in the stabilizer group, while $S_X^{\text{top}}$ are not.
	On the $(X)$ top boundary, $e_X$ are condensed.
\end{itemize}
The boundary elementary excitations of Haah's code FTH are $e_X^{\text{top}}, e_Z^{\text{top}}, e_X^{\text{bot}}, e_Z^{\text{bot}}$, with representative syndromes $S_X^{\text{top}}, S_Z^{\text{top}}, S_X^{\text{bot}}, S_Z^{\text{bot}}$, respectively.
Single $e_{X(Z)}^{\text{top}(\text{bot})}$ can be created by infinite-support Pauli operators, see Appendix~\ref{app_Haah_infinite_support}.

\begin{figure*}[!t]
	\centering
	\begin{tikzpicture}
		% graph explanation
		\begin{scope}[yshift=8.6cm]
			\fill[white] (-3,0,0) circle (0.1cm);
			\fill[red] (-3,0,0) ++(90:0.1cm) arc[start angle=90,end angle=270,radius=0.1cm] -- (-3,0,0) -- cycle;
			\draw[black, line width=0.6pt] (-3,0,0) circle (0.1cm);
			\node at (-2.45,0,0) {$:X_1$};
		\end{scope}

		\begin{scope}[xshift=2cm, yshift=8.6cm]
			\fill[white] (-3,0,0) circle (0.1cm);
			\fill[red] (-3,0,0) ++(90:0.1cm) arc[start angle=90,end angle=-90,radius=0.1cm] -- (-3,0,0) -- cycle;
			\draw[black, line width=0.6pt] (-3,0,0) circle (0.1cm);
			\node at (-2.45,0,0) {$:X_2$};
		\end{scope}

		\begin{scope}[xshift=4cm, yshift=8.6cm]
			\fill[white] (-3,0,0) circle (0.1cm);
			\fill[red] (-3,0,0) circle (0.1cm);
			\draw[black, line width=0.6pt] (-3,0,0) circle (0.1cm);
			\node at (-2.25,0,0) {$:X_1X_2$};
		\end{scope}

		\begin{scope}[xshift=7cm, yshift=8.6cm]
			\fill[white] (-3,0,0) circle (0.1cm);
			\fill[ForestGreen] (-3,0,0) ++(90:0.1cm) arc[start angle=90,end angle=270,radius=0.1cm] -- (-3,0,0) -- cycle;
			\draw[black, line width=0.6pt] (-3,0,0) circle (0.1cm);
			\node at (-2.45,0,0) {$:Z_1$};
		\end{scope}

		\begin{scope}[xshift=9cm, yshift=8.6cm]
			\fill[white] (-3,0,0) circle (0.1cm);
			\fill[ForestGreen] (-3,0,0) ++(90:0.1cm) arc[start angle=90,end angle=-90,radius=0.1cm] -- (-3,0,0) -- cycle;
			\draw[black, line width=0.6pt] (-3,0,0) circle (0.1cm);
			\node at (-2.45,0,0) {$:Z_2$};
		\end{scope}

		\begin{scope}[xshift=11cm, yshift=8.6cm]
			\fill[white] (-3,0,0) circle (0.1cm);
			\fill[ForestGreen] (-3,0,0) circle (0.1cm);
			\draw[black, line width=0.6pt] (-3,0,0) circle (0.1cm);
			\node at (-2.25,0,0) {$:Z_1Z_2$};
		\end{scope}

		% boundary gauge operator 2
		\begin{scope}[
			x={(-0.6cm,-0.2cm)},
			y={(0.6cm,-0.15cm)},
			z={(0cm,1cm)}, yshift=6.35cm]
			% 1. 底面网格 (z=0 平面)
			\foreach \i in {0,...,4} \draw[gray, opacity=0.55, line width=0.4pt] (\i,-1,0) -- (\i,5,0);
			\foreach \j in {0,...,4} \draw[gray, opacity=0.55, line width=0.4pt] (-1,\j,0) -- (5,\j,0);

			% 4. 显式的坐标轴箭头
			\draw[-latex, thick, black!70] (0,0,0) -- (5.2,0,0) node[below left] {$x$};
			\draw[-latex, thick, black!70] (0,0,0) -- (0,5.2,0) node[below right] {$y$};
			\draw[-latex, thick, black!70] (0,0,0) -- (0,0,1.2) node[above] {$z$};

			% bottom boundary gauge syndrome
			\foreach \i/\j in {2/2, 1/2, 0/2, 2/1, 2/0, 1/1}{
				\fill[red, opacity=0.25] (\i,\j,0)--(\i-1,\j,0)--(\i-1,\j-1,0)--(\i,\j-1,0)--cycle;
			}

			% boundary gauge operator 2 times xy
			\foreach \i/\j in {1/0, 0/1}{
				\fill[white] (\i,\j,0) circle (0.1cm);
				\fill[ForestGreen] (\i,\j,0) ++(90:0.1cm) arc[start angle=90,end angle=270,radius=0.1cm] -- (\i,\j,0) -- cycle;
				\draw[black, line width=0.6pt] (\i,\j,0) circle (0.1cm);
			}
			\foreach \i/\j in {0/0}{
				\fill[white] (\i,\j,0) circle (0.1cm);
				\fill[ForestGreen] (\i,\j,0) ++(90:0.1cm) arc[start angle=90,end angle=-90,radius=0.1cm] -- (\i,\j,0) -- cycle;
				\draw[black, line width=0.6pt] (\i,\j,0) circle (0.1cm);
			}
			\foreach \i/\j in {1/1}{
				\fill[ForestGreen] (\i,\j,0) circle (0.1cm);
				\draw[black, line width=0.6pt] (\i,\j,0) circle (0.1cm);
			}
		\end{scope}

		\begin{scope}[yshift=6.35cm]
			\draw[->] (3.5+0.4,0.3,0)--(5+0.4,0.3,0);
			\node at (4.25+0.4,0.6,0) {$\Phi_X$};
			\node at (-4,1.2,0) {(a)};
		\end{scope}

		\begin{scope}[xshift=7cm, yshift=5.15cm, scale=0.6]
			\draw[-latex, thick, black!70] (0,0,0) -- (5.6,0) node[below left] {$x$};
			\draw[-latex, thick, black!70] (0,0,0) -- (0,5.6) node[below right] {$y$};
			\foreach \i in {0,...,4} \draw[gray, opacity=0.55, line width=0.4pt] (\i,-1,0) -- (\i,5,0);
			\foreach \j in {0,...,4} \draw[gray, opacity=0.55, line width=0.4pt] (-1,\j,0) -- (5,\j,0);

			\foreach \i/\j in {1/1, 2/2, 2/1, 2/0, 1/2, 0/2}{
				\node at (\i-0.5,\j-0.5) {$\tilde Z$};
				\fill[red, opacity=0.25] (\i,\j,0)--(\i-1,\j,0)--(\i-1,\j-1,0)--(\i,\j-1,0)--cycle;
			}
		\end{scope}

		% e_X transport operator
		\begin{scope}[
			x={(-0.6cm,-0.2cm)},
			y={(0.6cm,-0.15cm)},
			z={(0cm,0.9cm)}]
			% 1. 底面网格 (z=0 平面)
			\foreach \i in {0,...,4} \draw[gray, opacity=0.55, line width=0.4pt] (\i,-1,0) -- (\i,5,0);
			\foreach \j in {0,...,4} \draw[gray, opacity=0.55, line width=0.4pt] (-1,\j,0) -- (5,\j,0);

			% 2. 左背墙网格 (y=0 平面)
			\foreach \i in {0,...,4} \draw[gray, opacity=0.15, line width=0.4pt] (\i,0,0) -- (\i,0,4);
			\foreach \k in {0,...,4} \draw[gray, opacity=0.15, line width=0.4pt] (0,0,\k) -- (5,0,\k);

			% 3. 右背墙网格 (x=0 平面)
			\foreach \j in {0,...,4} \draw[gray, opacity=0.15, line width=0.4pt] (0,\j,0) -- (0,\j,4);
			\foreach \k in {0,...,4} \draw[gray, opacity=0.15, line width=0.4pt] (0,0,\k) -- (0,5,\k);

			% 4. 显式的坐标轴箭头
			\draw[-latex, thick, black!70] (0,0,0) -- (5.2,0,0) node[below left] {$\bar x$};
			\draw[-latex, thick, black!70] (0,0,0) -- (0,5.2,0) node[below right] {$\bar y$};
			\draw[-latex, thick, black!70] (0,0,0) -- (0,0,4.6) node[above] {$z$};

			% layer indicator
			\fill[gray, opacity=0.25] (0,0,1)--(1,0,1)--(0,1,1)--cycle;
			\fill[gray, opacity=0.25] (0,0,2)--(2,0,2)--(0,2,2)--cycle;
			\fill[gray, opacity=0.25] (0,0,3)--(3,0,3)--(0,3,3)--cycle;
			\fill[gray, opacity=0.25] (0,0,4)--(4,0,4)--(0,4,4)--cycle;

			% bottom boundary gauge syndrome
			\foreach \i/\j in {0/0}{
				\fill[red, opacity=0.25] (\i,\j,0)--(\i-1,\j,0)--(\i-1,\j-1,0)--(\i,\j-1,0)--cycle;
			}

			% z=4 layer
			\foreach \i/\j in {0/0, 0/4, 4/0}{
				\fill[white] (\i,\j,4) circle (0.1cm);
				\fill[ForestGreen] (\i,\j,4) ++(90:0.1cm) arc[start angle=90,end angle=-90,radius=0.1cm] -- (\i,\j,4) -- cycle;
				\draw[black, line width=0.6pt] (\i,\j,4) circle (0.1cm);
			}

			% z=3 layer
			\foreach \i/\j in {0/0, 1/0, 0/1, 2/0, 0/2, 3/0, 0/3, 2/1, 1/2}{
				\fill[white] (\i,\j,3) circle (0.1cm);
				\fill[ForestGreen] (\i,\j,3) ++(90:0.1cm) arc[start angle=90,end angle=-90,radius=0.1cm] -- (\i,\j,3) -- cycle;
				\draw[black, line width=0.6pt] (\i,\j,3) circle (0.1cm);
			}

			% z=2 layer
			\foreach \i/\j in {0/0, 2/0, 0/2}{
				\fill[white] (\i,\j,2) circle (0.1cm);
				\fill[ForestGreen] (\i,\j,2) ++(90:0.1cm) arc[start angle=90,end angle=-90,radius=0.1cm] -- (\i,\j,2) -- cycle;
				\draw[black, line width=0.6pt] (\i,\j,2) circle (0.1cm);
			}

			% z=1 layer
			\foreach \i/\j in {0/0, 1/0, 0/1}{
				\fill[white] (\i,\j,1) circle (0.1cm);
				\fill[ForestGreen] (\i,\j,1) ++(90:0.1cm) arc[start angle=90,end angle=-90,radius=0.1cm] -- (\i,\j,1) -- cycle;
				\draw[black, line width=0.6pt] (\i,\j,1) circle (0.1cm);
			}

			% z=0 layer
			\foreach \i/\j in {0/0}{
				\fill[white] (\i,\j,0) circle (0.1cm);
				\fill[ForestGreen] (\i,\j,0) ++(90:0.1cm) arc[start angle=90,end angle=-90,radius=0.1cm] -- (\i,\j,0) -- cycle;
				\draw[black, line width=0.6pt] (\i,\j,0) circle (0.1cm);
			}
		\end{scope}

		\draw[->] (3.5+0.4,1.3,0)--(5+0.4,1.3,0);
		\node at (4.25+0.4,1.6,0) {$\Phi_X$};
		\node at (-4,4.2,0) {(b)};

		\begin{scope}[xshift=7cm, yshift=0.1cm, scale=0.6]
			\draw[-latex, thick, black!70] (0,0,0) -- (5.6,0) node[below left] {$x$};
			\draw[-latex, thick, black!70] (0,0,0) -- (0,5.6) node[below right] {$y$};
			\foreach \i in {0,...,4} \draw[gray, opacity=0.55, line width=0.4pt] (\i,-1,0) -- (\i,5,0);
			\foreach \j in {0,...,4} \draw[gray, opacity=0.55, line width=0.4pt] (-1,\j,0) -- (5,\j,0);

			\foreach \i/\j in {1/1}{
				\node at (\i-0.5,\j-0.5) {$\tilde Z$};
				\fill[red, opacity=0.25] (\i,\j,0)--(\i-1,\j,0)--(\i-1,\j-1,0)--(\i,\j-1,0)--cycle;
			}
		\end{scope}

		% boundary gauge operator 1
		\begin{scope}[
			x={(-0.6cm,-0.2cm)},
			y={(0.6cm,-0.15cm)},
			z={(0cm,1cm)}, yshift=-3.2cm]
			% 1. 底面网格 (z=0 平面)
			\foreach \i in {0,...,4} \draw[gray, opacity=0.55, line width=0.4pt] (\i,-1,0) -- (\i,5,0);
			\foreach \j in {0,...,4} \draw[gray, opacity=0.55, line width=0.4pt] (-1,\j,0) -- (5,\j,0);

			% 4. 显式的坐标轴箭头
			\draw[-latex, thick, black!70] (0,0,0) -- (5.2,0,0) node[below left] {$x$};
			\draw[-latex, thick, black!70] (0,0,0) -- (0,5.2,0) node[below right] {$y$};
			\draw[-latex, thick, black!70] (0,0,0) -- (0,0,1.2) node[above] {$z$};

			% bottom boundary gauge syndrome
			\foreach \i/\j in {0/0, 1/0, 2/0, 0/1, 0/2, 1/1}{
				\fill[ForestGreen, opacity=0.25] (\i,\j,0)--(\i-1,\j,0)--(\i-1,\j-1,0)--(\i,\j-1,0)--cycle;
			}

			% boundary gauge operator 1 times xy
			\foreach \i/\j in {1/0, 0/1}{
				\fill[white] (\i,\j,0) circle (0.1cm);
				\fill[red] (\i,\j,0) ++(90:0.1cm) arc[start angle=90,end angle=270,radius=0.1cm] -- (\i,\j,0) -- cycle;
				\draw[black, line width=0.6pt] (\i,\j,0) circle (0.1cm);
			}
			\foreach \i/\j in {0/0}{
				\fill[white] (\i,\j,0) circle (0.1cm);
				\fill[red] (\i,\j,0) ++(90:0.1cm) arc[start angle=90,end angle=-90,radius=0.1cm] -- (\i,\j,0) -- cycle;
				\draw[black, line width=0.6pt] (\i,\j,0) circle (0.1cm);
			}
		\end{scope}

		\begin{scope}[yshift=-3.3cm]
			\draw[->] (3.5+0.4,0.3,0)--(5+0.4,0.3,0);
			\node at (4.25+0.4,0.6,0) {$\Phi_X$};
			\node at (-4,1.2,0) {(c)};
		\end{scope}

		\begin{scope}[xshift=7cm, yshift=-4.5cm, scale=0.6]
			\draw[-latex, thick, black!70] (0,0,0) -- (5.6,0) node[below left] {$x$};
			\draw[-latex, thick, black!70] (0,0,0) -- (0,5.6) node[below right] {$y$};
			\foreach \i in {0,...,4} \draw[gray, opacity=0.55, line width=0.4pt] (\i,-1,0) -- (\i,5,0);
			\foreach \j in {0,...,4} \draw[gray, opacity=0.55, line width=0.4pt] (-1,\j,0) -- (5,\j,0);

			\foreach \i/\j in {1/1}{
				\node at (\i-0.5,\j-0.5) {$\tilde X$};
				\fill[ForestGreen, opacity=0.25] (\i,\j,0)--(\i-1,\j,0)--(\i-1,\j-1,0)--(\i,\j-1,0)--cycle;
			}
		\end{scope}
	\end{tikzpicture}
	\vspace{-0.15cm}
	\caption{The illustration of operator identification of Haah's code FTH under $(X)$ top boundary.
	The red circles and green circles represent Pauli $X$ and $Z$, respectively, as shown on the top of this figure.
	On the left side, generators of low-energy-preserving operators and their bottom boundary gauge syndromes are illustrated.
	Green squares represent $S_Z^{\text{bot}}$ violation, red squares represent $S_X^{\text{bot}}$ violation.
	The right side illustrates the identified operators in the 2d boundary system under $(X)$ top boundary.
	(a) The bottom boundary gauge operator $S_Z^{\text{bot}}$, which has the bottom boundary gauge syndrome $xy\bar F S_X^{\text{bot}}$, is identified as $xy\bar F\tilde Z$.
	(b) The $e_X$-transport operator along $z$-direction $\mathcal T_{e_X,z}$, which has the bottom boundary gauge syndrome $S_X^{\text{bot}}$, is identified as $\tilde Z$.
	Note that here we use the axis $\bar x$ and $\bar y$ instead of $x$ and $y$ for visual clearness.
	(c) The bottom boundary gauge operator $S_X^{\text{bot}}$, which has the bottom boundary gauge syndrome $\bar x\bar y F S_Z^{\text{bot}}$, is identified as $\tilde X$.}
	\label{fig_Haah_code_FTH_operator_illustration_X_top}
\end{figure*}

Like the 2dTC TH and X-cube FTH, we identify the low-energy subspace of Haah's code FTH as the Hilbert space of a 2d system.
The low-energy preserving operators are generated by the bottom boundary gauge operators\footnote{Bottom boundary gauge operators refer to the local Pauli operators near the bottom boundary that commute with all stabilizers but are not generated by stabilizers.
Boundary gauge operators can be classified into primary boundary gauge operators---truncated stabilizers, and secondary boundary gauge operators---those that are not truncated stabilizers (Ref.~\cite{Chen_Yu_An_2024}).
Here we only meet primary boundary gauge operators, both on top and bottom boundaries.} and transport operators.
Under $(Z)$ top boundary, the low-energy preserving operator generators and their identification are shown in Fig.~\ref{fig_Haah_code_FTH_operator_illustration_Z_top}; while those under $(X)$ top boundary are shown in Fig.~\ref{fig_Haah_code_FTH_operator_illustration_X_top}.
The operator identification is inspired by the bottom boundary syndrome of the operators.
Under such operator identifications, the minimal bottom boundary dynamical Hamiltonian
\begin{equation}\label{eq_Haah_H_bot}
	H_{\text{bot}} = -\mathop{\sum\nolimits_{\mathbb{C}}}\limits_{\text{unit } m\in R_\partial} m\left(S_X^{\text{bot}} +_{\mathbb{C}} g S_Z^{\text{bot}}\right)
\end{equation}
is identified as transverse field generalized Ising models (TFGIM) under $(Z)$ and $(X)$ top boundaries, where $R_\partial \equiv \mathbb{F}_2[x^{\pm 1}, y^{\pm 1}]$ is the group ring of translation group along the boundary, unit $m\in R_\partial$ are the monomials in $R_\partial$.
$\sum_{\mathbb{C}}$ and $+_{\mathbb{C}}$ are additions on $\mathbb{C}$, the subscript $\mathbb{C}$ is added to distinguish from the addition $+$ in $R_\partial$, which means the Pauli product.
$g\in\mathbb{R}$ is a strength factor.
Specifically, $H_{\text{bot}}$ is identified as
\begin{equation}\label{eq_Haah_H_Z}
	\tilde{H}_{(Z)} = -\mathop{\sum\nolimits_{\mathbb{C}}}\limits_{\text{unit } m\in R_\partial} m\left(F\tilde{Z} +_{\mathbb{C}} g\tilde{X}\right)
\end{equation}
under $(Z)$ top boundary, where $F = 1+x+x^2+y+y^2+xy$; while identified as
\begin{equation}\label{eq_Haah_H_X}
	\tilde{H}_{(X)} = -\mathop{\sum\nolimits_{\mathbb{C}}}\limits_{\text{unit } m\in R_\partial} m\left(\tilde{X}' +_{\mathbb{C}} g\bar{F}\tilde{Z}'\right)
\end{equation}
under $(X)$ top boundary, where $\bar{F} = 1+\bar{x}+\bar{x}^2+\bar{y}+\bar{y}^2+\bar{x}\bar{y}$, $\bar{x}\equiv x^{-1}$, $\bar{y}\equiv y^{-1}$, prime are added for the operators identified under $(X)$ top boundary to distinguish with those identified under $(Z)$ top boundary.
Switching the top boundary between $(Z)$ and $(X)$ induces the duality between $\tilde{H}_{(Z)}$ and $\tilde{H}_{(X)}$, as shown in Fig.~\ref{fig_Haah_code_FTH_duality}.
\begin{figure}[!t]
	\centering
	\begin{tikzpicture}[scale=0.68, every node/.style={transform shape}]
		\def\drawlattice#1#2#3#4{
			\begin{scope}[shift={(#1,#2)}]
				\draw[gray!40, step=1] (#3,#3) grid (#4,#4);
				\foreach \x in {#3,...,#4} {
					\foreach \y in {#3,...,#4} {
						\fill[gray!60] (\x,\y) circle (1.5pt);
					}
				}
			\end{scope}
		}

		\tikzstyle{zblob} = [preaction={transform canvas={shift={(0.08,-0.08)}}, fill=black, fill opacity=0.12, draw=none}, fill=blue, fill opacity=0.1, draw=blue, draw opacity=0.4, thick, rounded corners=8pt]
		\tikzstyle{xblob} = [preaction={transform canvas={shift={(0.08,-0.08)}}, fill=black, fill opacity=0.12, draw=none}, fill=red, fill opacity=0.1, draw=red, draw opacity=0.4, thick, rounded corners=8pt]
		\tikzstyle{ztext} = [blue, fill=white, circle, inner sep=0.5pt]
		\tikzstyle{xtext} = [red, fill=white, circle, inner sep=0.5pt]

		% Background box
		\draw[fill=red!2, draw=red!10, thick, rounded corners=15pt] (11.25, -1.15) rectangle (21.25, 5.15);

		\node[font=\bfseries, align=center] at (16.25, 4.65) {Haah's Code FTH Duality};

		% (Z) Top Boundary
		\drawlattice{11}{0}{1}{4}
		\node[align=center, font=\small] at (13.5, -0.5) {$(Z)$ Top Boundary\\[0.5ex]$F\tilde{Z}$ and $\tilde{X}$};

		\begin{scope}[shift={(11,0)}]
			% Haah's generalized Ising term F = 1+x+x^2+y+y^2+xy
			\draw[zblob] (1.1, 1.1) -- (3.9, 1.1) -- (3.9, 1.9) -- (2.9, 1.9) -- (2.9, 2.9) -- (1.9, 2.9) -- (1.9, 3.9) -- (1.1, 3.9) -- cycle;

			\node[ztext] at (1.5, 1.5) {$Z$};
			\node[ztext] at (2.5, 1.5) {$Z$};
			\node[ztext] at (3.5, 1.5) {$Z$};
			\node[ztext] at (1.5, 2.5) {$Z$};
			\node[ztext] at (2.5, 2.5) {$Z$};
			\node[ztext] at (1.5, 3.5) {$Z$};

			% X field
			\draw[xblob] (3.1, 3.1) rectangle (3.9, 3.9);
			\node[xtext] at (3.5, 3.5) {$X$};
		\end{scope}

		% Main dual arrow
		\draw[<->, thick] (15.7, 2.5) -- (16.8, 2.5) node[midway, above, font=\footnotesize] {Dual};

		% (X) Top Boundary
		\drawlattice{16.5}{0}{1}{4}
		\node[align=center, font=\small] at (19.0, -0.5) {$(X)$ Top Boundary\\[0.5ex]$\bar{F}\tilde{Z}$ and $\tilde{X}$};

		\begin{scope}[shift={(16.5,0)}]
			% Haah's generalized Ising term \bar{F}\tilde{Z}
			\draw[xblob] (1.1, 3.9) -- (3.9, 3.9) -- (3.9, 1.1) -- (3.1, 1.1) -- (3.1, 2.1) -- (2.1, 2.1) -- (2.1, 3.1) -- (1.1, 3.1) -- cycle;

			\node[xtext] at (3.5, 3.5) {$Z$};
			\node[xtext] at (2.5, 3.5) {$Z$};
			\node[xtext] at (1.5, 3.5) {$Z$};
			\node[xtext] at (3.5, 2.5) {$Z$};
			\node[xtext] at (2.5, 2.5) {$Z$};
			\node[xtext] at (3.5, 1.5) {$Z$};

			% X field \tilde{X}
			\draw[zblob] (1.1, 1.1) rectangle (1.9, 1.9);
			\node[ztext] at (1.5, 1.5) {$X$};
		\end{scope}

		% Connecting mapping arrows for Haah's Code
		\draw[<->, thick, blue, densely dashed] (14.9, 1.5) -- (17.6, 1.5);
		\draw[<->, thick, red, densely dashed] (14.9, 3.5) -- (17.6, 3.5);
	\end{tikzpicture}
	\vspace{-0.15cm}
	\caption{Duality between the 2d boundary models ($\tilde{H}_{(Z)}$ and $\tilde{H}_{(X)}$, TFGIMs) induced by Haah's code FTH under the $(Z)$ and $(X)$ top boundaries.
	Blue and red shaded blocks trace the operator identifications of the bottom boundary gauge operators $S_X^{\text{bot}}$ and $S_Z^{\text{bot}}$, respectively.}
	\label{fig_Haah_code_FTH_duality}
\end{figure}%
The identifications in Fig.~\ref{fig_Haah_code_FTH_operator_illustration_Z_top} and Fig.~\ref{fig_Haah_code_FTH_operator_illustration_X_top} dictate the dictionaries
\begin{equation}\label{eq_Haah_local_dictionary}
	mf\tilde{Z}\dual m\tilde{X}',\qquad m\tilde{X}\dual m\bar{f}\tilde{Z}'
\end{equation}
for any unit $m\in R_\partial$, where $f\equiv\bar{x}\bar{y}F$.
$mf\tilde{Z}, m\tilde{X}'$ are identified from $mS_X^{\text{bot}}$ and $m\tilde{X}, m\bar{f}\tilde{Z}'$ are identified from $mS_Z^{\text{bot}}$.

As in the X-cube FTH (untwisted), the duality also exchanges symmetries and relations.
Let $\widehat{R}_\partial$ be the set of formal series $u=\sum_{\text{unit } m}u_m m$ ($u_m\in\mathbb{F}_2$), which may have infinite support, and let $u\tilde{X}$ represent the product $\prod_m (m\tilde{X})^{u_m}$, similarly for $u\tilde{Z}, u\tilde{X}', u\tilde{Z}'$.
The symmetry indicators of $\tilde{H}_{(Z)}$ are $u\tilde{X}$ with $u\in\ann_{\widehat{R}_\partial}(\bar{F})$, i.e., $u\bar{F}=0$.
The same $u$ gives the relations of the Ising terms of $\tilde{H}_{(X)}$, since $\prod_m(m\bar{f}\tilde{Z}')^{u_m}$ corresponds to $u\bar{f}\tilde{Z}'=0$, i.e., $\prod_m(m\bar{f}\tilde{Z}')^{u_m}=1$.
By Eq.~(\ref{eq_Haah_local_dictionary}), both the symmetry $u\tilde{X}$ and the relation $u\bar{f}\tilde{Z}'$ are identified from the product $uS_Z^{\text{bot}}$.
Similarly, $u\in\ann_{\widehat{R}_\partial}(F)$ gives the symmetry indicators $u\tilde{X}'$ of $\tilde{H}_{(X)}$ and the relations $uf\tilde{Z}=0$ of $\tilde{H}_{(Z)}$, both identified from the product $uS_X^{\text{bot}}$.
Since $R_\partial$ has no zero-divisor, any non-zero $u$ that $u\bar{F}=0$ or $uF=0$ is infinite support; these are the fractal symmetries of the TFGIMs, generated by a linear cellular-automaton rule~\cite{zhang_li_ye_2024_hoca_spt_strange_correlators,huang_zhang_ye_2026_nonuniform_ca,zhang-ye-fusion-mobility2026}.
The simplest one is the global symmetry $\prod_{\text{all}}\tilde{X}$, i.e., $u=\sum_m m$, which works since $F$ has an even number of terms.
See Appendix~\ref{app_Haah_constraints} for details.

The induced duality is non-invertible.
On the infinite plane, we state this in representations generated from reference states by finite-support operators.
Under the $(Z)$ top boundary, the finite-support low-energy-preserving operators are identified as products of $m\tilde{X}, mf\tilde{Z}, mA^{L_z}\tilde{Z}$ (see Fig.~\ref{fig_Haah_code_FTH_operator_illustration_Z_top}).
Acting on the reference state $|+\cdots+\rangle_{(Z)}$, they span
\begin{equation}
	\tilde{\mathcal{H}}_{(Z)}(L_z)=\operatorname{span}_{\mathbb{C}}\big\{|s\rangle_{(Z)}\equiv \tilde{Z}(s)|+\cdots+\rangle_{(Z)}:\ s\in I_{L_z}\big\},
\end{equation}
where $I_{L_z}\equiv(F,A^{L_z})$ is the ideal of $R_\partial$ generated by $F$ and $A^{L_z}$.
Under the $(X)$ top boundary, single $\tilde{X}'$ and single $\tilde{Z}'$ are both identified from finite-support operators (see Fig.~\ref{fig_Haah_code_FTH_operator_illustration_X_top}).
Acting on the reference state $|0\cdots0\rangle_{(X)}$, they span
\begin{equation}
	\tilde{\mathcal{H}}_{(X)}=\operatorname{span}_{\mathbb{C}}\big\{|t\rangle_{(X)}\equiv \tilde{X}'(t)|0\cdots0\rangle_{(X)}:\ t\in R_\partial\big\}.
\end{equation}
The two reference states $|+\cdots+\rangle_{(Z)}$ and $|0\cdots0\rangle_{(X)}$ are ground states of $\tilde{H}_{(Z)}$ and $\tilde{H}_{(X)}$ in the limit $g\to\infty$, respectively.
The duality operator
\begin{equation}\label{eq_D_Haah}
	\mathsf{D}_{\text{Haah}}=\sum_{t\in R_\partial}|t\rangle_{(X)(Z)}\langle ft|
\end{equation}
realizes the local term dictionary in Eq.~(\ref{eq_Haah_local_dictionary}),
\begin{equation}
	\mathsf{D}_{\text{Haah}}mf\tilde{Z}=m\tilde{X}'\mathsf{D}_{\text{Haah}},\qquad \mathsf{D}_{\text{Haah}}m\tilde{X}=m\bar{f}\tilde{Z}'\mathsf{D}_{\text{Haah}},
\end{equation}
and absorbs every symmetry indicator of $\tilde{H}_{(Z)}$,
\begin{equation}
	\mathsf{D}_{\text{Haah}}u\tilde{X}=\mathsf{D}_{\text{Haah}},\qquad \forall u\in\ann_{\widehat{R}_\partial}(\bar{F}).
\end{equation}
These symmetries are mapped to relations $u\bar{F}\tilde{Z}'=0$.
Also, $\mathsf{D}_{\text{Haah}}^\dagger\mathsf{D}_{\text{Haah}}$ is proportional to the projector onto the common $+1$ eigenspace of symmetry indicators $u\tilde{X}$ of $\tilde{\mathcal{H}}_{(Z)}(L_z)$.
The state created by the image of transport operator $\mathcal{T}_{e_Z,z}$, $|A^{L_z}\rangle_{(Z)}$, has $\prod_{\text{all}}\tilde{X}=-1$ since $A^{L_z}$ has an odd number of terms, so
\begin{gather}
	\big|\mathsf{D}_{\text{Haah}}|A^{L_z}\rangle_{(Z)}\big|^2 = {}_{(Z)}\langle A^{L_z}|\mathsf{D}_{\text{Haah}}^\dagger\mathsf{D}_{\text{Haah}}|A^{L_z}\rangle_{(Z)}=0, \notag\\
	\mathsf{D}_{\text{Haah}}|A^{L_z}\rangle_{(Z)} = 0,
\end{gather}
$\mathsf{D}_{\text{Haah}}$ is non-invertible.
The symmetry sectors accessible in $\tilde{\mathcal{H}}_{(Z)}(L_z)$ also depend on the thickness $L_z$: the symmetry indicators $u\tilde{X}$ with $u\in\ann_{\widehat{R}_\partial}(\bar{F})\cap\ann_{\widehat{R}_\partial}(\bar{A}^{L_z})$, which form a $2L_z$-dimensional space, cannot be flipped by any finite-support operator and equal $1$ on $\tilde{\mathcal{H}}_{(Z)}(L_z)$, i.e., they become global constraints.
See Appendices~\ref{app_Haah_duality_op} and~\ref{appendix_Haah_infinite_support_quotients} for details.

\section{Summary and outlook}\label{sec_summary_outlook}

\emph{Summary.}---We generalized topological holography to fracton bulks, which we call fracton topological holography (FTH).
The setup is a sandwich consisting of a fracton bulk, a gapped top boundary, and a dynamical bottom boundary.
Transport operators play the role of the Wilson lines in TH: together with the bottom boundary gauge operators, they generate the low-energy-preserving operators of the sandwich.
Under each gapped top boundary, these operators, together with the symmetry, relation, and twist data they carry, are identified as operators of a $(d-1)$-dimensional boundary system.
Switching the top boundary while keeping the bottom boundary Hamiltonian $H_{\text{bot}}$ fixed then induces a duality between the identified boundary systems, together with its local operator dictionary.

For the X-cube FTH, under the four gapped top boundaries $(ee)$, $(mm)$, $(em)$, $(me)$, $H_{\text{bot}}$ is identified as two TFPIMs and two QCMs, and the same sandwich gives a duality web consisting of the TFPIM self-duality (the TFPIM version KW duality), the QCM self-duality, and the QCM--TFPIM duality.
We constructed the duality operators explicitly.
Each of them absorbs a set of line-like symmetries and is therefore non-invertible.
Deleting specific top stabilizer generators introduces twist DOFs; the symmetry and twist operators then exchange their identities under the top boundary switching, and the dualities become invertible.
For the TFPIM self-duality, we further constructed a linear-depth local unitary SQC that switches the top boundary from $(ee)$ to $(mm)$ and realizes the duality.
The QCM--TFPIM duality behaves differently: the line-like symmetries of the QCM along the $x$-direction are mapped to relations/twists of the TFPIM, while those along the $y$-direction are mapped to line-like symmetries of the TFPIM.

For Haah's code FTH, we used the Laurent-polynomial formalism and the same identification rule to identify $H_{\text{bot}}$ as TFGIMs under the $(Z)$ and $(X)$ top boundaries, which differ by an inversion of the relative field strength and a spatial inversion.
The induced duality exchanges the generalized Ising terms and the transverse-field terms, as well as the fractal symmetries and the relations of the generalized Ising terms.
The duality operator $\mathsf{D}_{\text{Haah}}$ absorbs the fractal symmetries, and the state created by the image of transport operator lies in its kernel, so the duality is non-invertible.
Under the $(Z)$ top boundary, the number of global constraints grows with the thickness $L_z$ of the sandwich, a geometric information inherited from the type-II fracton bulk.
Together, the two examples show that FTH applies to both type-I and type-II fracton bulks, including those without movable excitations and string operators.

\emph{Outlook.}---First, a natural generalization is FTH with general abelian DOFs, e.g., $\mathbb{Z}_N$ qudit fracton models, whose translation-invariant stabilizer codes can be formulated with Laurent polynomials over $\mathbb{Z}_N$~\cite{haah_2013_modules,ruba_yang_2024_homological}, as well as fracton orders with extensive fracton excitations~\cite{li2020fracton,li_ye_2021_fracton_extended_gsd,li_2023_gERG,hyt_1_prepare_TD_model_via_SQC}.
The identified boundary systems are expected to be models with more general classes of Pauli-type Hamiltonians, and the dualities would carry $\mathbb{Z}_N$-valued symmetry, relation, and twist data.

Second, it is interesting to construct FTH for non-abelian fracton orders, such as the cage-net models~\cite{prem_huang_song_hermele_2019_cage_net} and fracton models obtained by gauging~\cite{bulmash_barkeshli_2019_gauging_fractons}.
There, the transport operators are no longer Pauli operators, and the condensed excitations on the gapped top boundaries can have non-abelian fusion rules, so both the identified boundary systems and the dualities between them are expected to go beyond the stabilizer setting.

Third, the polynomial representation makes FTH algorithmic: the boundary gauge operators, transport operators, operator identifications, and symmetry indicators (as annihilators of polynomials) are all computed by polynomial operations.
This suggests a systematic search over translation-invariant stabilizer codes~\cite{haah_2013_modules,ruba_yang_2024_homological} and their gapped top boundaries, to generate dualities, find generalized symmetries such as fractal and subsystem symmetries, and locate self-dual points as candidates of phase transitions.

Fourth, since $\tilde{H}_{(Z)}$ and $\tilde{H}_{(X)}$ differ only by an inversion of the relative field strength and a spatial inversion, composing $\mathsf{D}_{\text{Haah}}$ with the spatial inversion is expected to give a non-invertible fractal symmetry of a single TFGIM at the self-dual point, in parallel with the non-invertible subsystem symmetry of the TFPIM at its self-dual point~\cite{cao_li_yamazaki_zheng_2023_subsystem_noninvertible}.
Making this precise requires a compatible identification of the representations on the two sides, and its implications for the phase diagram of the TFGIM deserve further study.

Fifth, for the X-cube FTH, deleting top stabilizer generators lifts relations to twist DOFs and makes the dualities invertible.
Whether the relations of the type-II example can be lifted in the same way, and whether the switching between the $(Z)$ and $(X)$ top boundaries of Haah's code can be realized by an SQC of finite or linear depth, remain open.
More generally, classifying the gapped top boundaries of a given fracton order would produce the complete duality web induced by a single fracton sandwich.

\acknowledgments
We thank Yu-An Chen, Zongyuan Wang, Hank Chen, Yidun Wan, Yifei Wang for useful discussions.
This work was supported by the National Natural Science Foundation of China (NSFC) under Grant No.~12474149, the Research Center for Magnetoelectric Physics of Guangdong Province under Grant No.~2024B0303390001, and the Guangdong Provincial Key Laboratory of Magnetoelectric Physics and Devices under Grant No.~2026B1212170049.

%\bibliography{references}
 %apsrev4-2.bst 2019-01-14 (MD) hand-edited version of apsrev4-1.bst
%Control: key (0)
%Control: author (8) initials jnrlst
%Control: editor formatted (1) identically to author
%Control: production of article title (0) allowed
%Control: page (0) single
%Control: year (1) truncated
%Control: production of eprint (0) enabled
%

	\appendix

	\onecolumngrid
	% one-column footnotes in the appendices (revtex keeps two-column footnotes after \onecolumngrid otherwise)
	\makeatletter
	\let\set@footnotewidth\set@footnotewidth@one
	\let\compose@footnotes\compose@footnotes@one
	\expandafter\let\csname combine@insert@1\endcsname\combine@insert@@ne
	\makeatother
	
	\section{Details of tensor product structure of low-energy subspace and operator identification of 2dTC TH}\label{details_2dTC_TH_Hilbert_space_structure_and_operator_identification}

	We follow the conventions of Sec.~\ref{sec_TH_review} and Fig.~\ref{fig_2dTC_TH}.
	For brevity, under the smooth top boundary we drop the primes of the dual spin operators $\tilde{X}'_i,\tilde{Z}'_i$ in Fig.~\ref{fig_2dTC_TH}(b) and label the edge spins by half-integers $i+1/2$, and we write the twist operators as $\tilde{X}^{\text{twist}},\tilde{Z}^{\text{twist}}$.

	\subsection{With smooth top boundary (\texorpdfstring{$m$}{m} particle condensed)}\label{appendix_2dTC_smooth_top}

	Let us start from Eq.~(\ref{eq_2dTC_TH_Hamiltonian}), i.e.
	\begin{equation*}
		H_{\text{TC}}=-\eta\sum_a S_a + H_{\text{bot}}\,.
	\end{equation*}
	When the top boundary is smooth, the stabilizers are (1) the $B_v$ terms associated with all the crossed vertices in Fig.~\ref{details_2dTC_TH_1}, i.e. all the $B_v$ terms except those associated with the vertices on the bottom boundary and the hollowed vertex in Fig.~\ref{details_2dTC_TH_1}; (2) all the complete four-edge $A_p$ terms.
	\begin{figure}[H]
		\centering
		\begin{tikzpicture}
			\foreach \i in {0,1,2,3}{\draw (\i,0)--(\i,3);}
			\foreach \j in {0,1,2,3}{\draw (0,\j)--(4,\j);}
			\draw[dashed] (4,0)--(4,3);
			
			\foreach \i in {0,1,2,3}{
				\foreach \j in {1,2}{
					\draw[blue, line width=1.2pt] 
					(\i+0.12,\j+0.12)--(\i-0.12,\j-0.12)
					(\i+0.12,\j-0.12)--(\i-0.12,\j+0.12);
			}}
			\foreach \i in {1,2,3}{
				\draw[blue, line width=1.2pt] 
				(\i+0.12,3+0.12)--(\i-0.12,3-0.12)
				(\i+0.12,3-0.12)--(\i-0.12,3+0.12);
			}
			
			\draw[blue, line width=1.5pt] (0,3) circle (0.15); 
			
			\draw[blue, line width=1pt, ->] (-2,0)--(-1,0);
			\draw[blue, line width=1pt, ->] (-2,0)--(-2,1);
			\node[blue] at (-0.8,0) {$x$};
			\node[blue] at (-2,1.2) {$y$};
			
			\fill[blue] (0,0) circle[radius=1.5pt];
			\node[blue] at (0,-0.3) {$(0,0)$};
		\end{tikzpicture}
		\caption{Blue crosses mark $B_v$ stabilizer generators, and the blue circle marks the omitted $B_{v_0}^{\text{top}}$.}
		\label{details_2dTC_TH_1}
	\end{figure}
	Since $\eta\gg\|H_{\text{bot}}\|$, at low temperature ($kT\ll\eta$), the system stays in the common $+1$ eigenspace of stabilizers, i.e. the low-energy subspace, denoted by $\mathcal{H}_{\text{low}}$.
	The dimension of $\mathcal{H}_{\text{low}}$ for 2dTC TH when the top boundary is smooth is $2^{L_x+1}$: (since the $B_v$ and $A_p$ terms are independent)
	\begin{align*}
		&\#\,B_v=L_xL_y-1\quad,\quad\#\,A_p=L_xL_y\\
		&\#\,\text{spin}=L_y\cdot L_x+L_x\cdot(L_y+1)\,,\\
		&\log_2\,\dim\,\mathcal{H}_{\text{low}}=\#\,\text{spin}-(\#\,B_v+\#\,A_p)=L_x+1\,,
	\end{align*}
	where $\# B_v$, $\# A_p$ stand for the number of $B_v,A_p$ terms in the stabilizer generator set, respectively.
	
	Next, we show that $\mathcal{H}_{\text{low}}$ can be decomposed into $\tilde{\mathcal{H}}^{\text{twist}}\bigotimes_{i=0}^{L_x-1}\tilde{\mathcal{H}}_{i+1/2}$, where $\tilde{\mathcal{H}}^{\text{twist}}\cong\tilde{\mathcal{H}}_{i+1/2}\cong\mathbb{C}^2$, for each $\tilde{\mathcal{H}}_{i+1/2}$, the operator algebra $\mathcal{B}(\tilde{\mathcal{H}}_{i+1/2})$ is local with respect to the $x$-direction.
	With such a decomposition, $\mathcal{H}_{\text{low}}$ can be identified as a 1$d$ length-$L_x$ spin ring along $x$-direction (a $\frac12$-spin on each edge, with Hilbert space $\tilde{\mathcal{H}}_{i+1/2}$), together with an ancillary twist qubit (with the Hilbert space $\tilde{\mathcal{H}}^{\text{twist}}$) encoding twist information.
	
	We define $\tilde{\mathcal{H}}_{i+1/2}$ through its operator algebra $\mathcal{B}(\tilde{\mathcal{H}}_{i+1/2})$.
	Let $\mathcal{B}(\tilde{\mathcal{H}}_{i+1/2})$ be the operator algebra generated by $\tilde{X}_{i+1/2}$ and $\tilde{Z}_{i+1/2}$ under addition and multiplication, where $\tilde{X}_{i+1/2}$ and $\tilde{Z}_{i+1/2}$ are identified as follows (we show $\tilde{X}_{5/2}$ and $\tilde{Z}_{5/2}$ as examples):
	\begin{equation}
		\begin{tikzpicture}[baseline=7ex]
			\foreach \i in {0,1,2,3,4}{\draw (\i,0)--(\i,3);}
			\foreach \j in {0,1,2,3}{\draw (0,\j)--(5,\j);}
			\draw[dashed] (5,0)--(5,3);
			
			\fill[red, opacity=0.1] (2,0.2)--(3,0.2)--(3,-0.2)--(2,-0.2);
			\node[red] at (2.4,0) {$X$};
			
			\fill[ForestGreen, opacity=0.1] (2.45,-0.2)--(2.75,-0.2)--(2.75,3.2)--(2.45,3.2);
			\node[ForestGreen] at (2.6,0) {$Z$};
			\foreach \j in {1,2,3}{\node[ForestGreen] at (2.6,\j) {$Z$};}
			
			\node[red] at (2.4,-0.5) {$\sim\tilde{X}_{5/2}$};
			
			\node[ForestGreen] at (2.6,3.5) {$\mathcal Z_y(5/2)\sim\tilde{Z}_{5/2}$};
			
			\node at (0,-1) {0};
			\node at (1,-1) {1};
			\node at (2,-1) {2};
			\node at (3,-1) {3};
			\node at (4,-1) {4};
			\node at (5,-1) {5};
			\draw[->] (0,-1.25)--(5.2,-1.25);
			\node at (5.4,-1.25) {$x$};
			
			\node at (6,-1) {$,$};
		\end{tikzpicture}
		\label{details_2dTC_TH_2}
	\end{equation}
	where $\mathcal Z_y(5/2)$ is the extensive $Z$ dual Wilson line along $y$ direction at $x=5/2$.
	Together with the dimension count $\log_2\dim\mathcal{H}_{\text{low}}=L_x+1$ obtained above, the following four conditions ensure that the above decomposition is legal, and each $\tilde{\mathcal{H}}_{i+1/2}$ can be identified as the Hilbert space of a local spin:
	\begin{enumerate}
		\item $\tilde{X}_{i+1/2},\tilde{Z}_{i+1/2}$ both commute with all stabilizers, so that $\tilde{X}_{i+1/2},\tilde{Z}_{i+1/2}$ map $\mathcal{H}_{\text{low}}$ to itself;
		\item $\tilde{X}_{i+1/2},\tilde{Z}_{i+1/2}$ generate the full $2\times 2$ matrix algebra $\mathcal{M}_2(\mathbb{C})\cong\mathcal{B}(\tilde{\mathcal{H}}_{i+1/2})$;
		\item $\tilde{X}_{i+1/2},\tilde{Z}_{i+1/2}$ are local along $x$-direction, so each $\mathcal{B}(\tilde{\mathcal{H}}_{i+1/2})$ is local along $x$-direction.
		\item $\forall i\neq j$, $\left[\mathcal{B}(\tilde{\mathcal{H}}_{i+1/2}),\mathcal{B}(\tilde{\mathcal{H}}_{j+1/2})\right]=0$, i.e., operators acting on $\tilde{\mathcal{H}}_{i+1/2}$ and $\tilde{\mathcal{H}}_{j+1/2}$ commute.
	\end{enumerate}
	From Eq.~(\ref{details_2dTC_TH_2}) it can be seen that conditions 1,3 and 4 are true.
	As for condition 2, $\tilde{X}_{i+1/2},\tilde{Z}_{i+1/2}$ indeed generate the full $2\times 2$ matrix algebra $\mathcal{M}_2(\mathbb{C})\cong\mathcal{B}(\tilde{\mathcal{H}}_{i+1/2})$, since (1) $\tilde{X}_{i+1/2},\tilde{Z}_{i+1/2}$ generate $\mathbb{I}$ and $\tilde{X}_{i+1/2}\tilde{Z}_{i+1/2}$ under multiplication; (2) $\mathbb{I},\tilde{X}_{i+1/2},\tilde{Z}_{i+1/2},\tilde{X}_{i+1/2}\tilde{Z}_{i+1/2}$ are linearly independent under addition and $\dim_\mathbb{C}\mathcal{M}_2(\mathbb{C})=4$.
	Moreover, since $\left\{\tilde{X}_{i+1/2},\tilde{Z}_{i+1/2}\right\}=0$, $\tilde{X}_{i+1/2}^2=\tilde{Z}_{i+1/2}^2=1$, and $\tilde{X}_{i+1/2},\tilde{Z}_{i+1/2}$ are hermitian, we can identify $\tilde{X}_{i+1/2},\tilde{Z}_{i+1/2}$ as the Pauli $X$ and Pauli $Z$ of the spin associated with $\tilde{\mathcal{H}}_{i+1/2}$, respectively.
	
	This identification is related to that $m$ particles are condensed on the top boundary: $m$ particles (i.e. $A_p=-1$ excitations) are created by dual $Z$ Wilson lines, so dual $Z$ Wilson lines attached to the top boundary does not violate any stabilizer, i.e. commute with all stabilizers.
	On the other hand, this identification is not unique, a pointwise basis change unitary can be applied, so that the Pauli $X$ on the bottom boundary at $i+1/2$ and $\mathcal Z_y(i+1/2)$ are identified as
	\begin{equation}
		\tilde{X}^{U}_{i+1/2}=\tilde{U}_{i+1/2}\tilde{X}_{i+1/2}\tilde{U}^\dagger_{i+1/2}\ \ \text{and}\ \ \tilde{Z}^{U}_{i+1/2}=\tilde{U}_{i+1/2}\tilde{Z}_{i+1/2}\tilde{U}^\dagger_{i+1/2}
	\end{equation} 
	respectively, where $\tilde{U}_{i+1/2}$ is any unitary in $U(2)\leq\mathcal{B}(\tilde{\mathcal{H}}_{i+1/2})$.
	At the same time, other operators in $\mathcal{B}(\tilde{\mathcal{H}}_{i+1/2})$ will be mapped as well,
	\begin{equation}
		\text{poly}\left(\tilde{X}_{i+1/2},\tilde{Z}_{i+1/2}\right)\mapsto\text{poly}\left(\tilde{X}^{U}_{i+1/2},\tilde{Z}^{U}_{i+1/2}\right)\,.
	\end{equation}
	
	Within the low-energy subspace $\mathcal{H}_{\text{low}}$, we have $\text{stabilizer}=+1$, i.e. we can regard stabilizers as identity operators.
	If $A=B\cdot\text{stabilizers}$, we write $A\sim B$.
	Then we can get the $\tilde{Z}_i\tilde{Z}_{i+1}$ in Fig.~\ref{fig_2dTC_TH}(b) as follows,
	\begin{equation}
		\begin{tikzpicture}[baseline=8.5ex]
			\foreach \i in {0,1,2,3}{\draw (\i,0)--(\i,3);}
			\foreach \j in {0,1,2,3}{\draw (0,\j)--(4,\j);}
			\draw[dashed] (4,0)--(4,3);
			
			\foreach \i in {1.5, 2.5}{
				\fill[ForestGreen, opacity=0.1] (\i-0.15,-0.2)--(\i+0.15,-0.2)--(\i+0.15,3.2)--(\i-0.15,3.2);
				\foreach \j in {0,1,2,3}{
					\node[ForestGreen] at (\i,\j) {$Z$};
				}
			}
			
			\node[ForestGreen] at (1.4,-0.5) {$\sim\tilde{Z}_i$};
			\node[ForestGreen] at (2.7,-0.5) {$\sim\tilde{Z}_{i+1}$};
		\end{tikzpicture}
		\quad\sim\quad
		\begin{tikzpicture}[baseline=8.5ex]
			\foreach \i in {0,1,2,3}{\draw (\i,0)--(\i,3);}
			\foreach \j in {0,1,2,3}{\draw (0,\j)--(4,\j);}
			\draw[dashed] (4,0)--(4,3);
			
			\fill[ForestGreen, opacity=0.1] (1,-0.2)--(1,0.2)--(3,0.2)--(3,-0.2);
			\fill[ForestGreen, opacity=0.1] (1.8,0)--(2.2,0)--(2.2,1)--(1.8,1);
			\node[ForestGreen] at (1.5,0) {$Z$};
			\node[ForestGreen] at (2.5,0) {$Z$};
			\node[ForestGreen] at (2,0.5) {$Z$};
			\node[ForestGreen] at (2,-0.5) {$\sim\tilde{Z}_i\tilde{Z}_{i+1}$};
			
			\node at (4.5,0.5) {$,$};
		\end{tikzpicture}
	\end{equation}
	where $i=\frac12,\frac32,\cdots,L_x-\frac32$.
	The exception is that when $i=-\frac12\ \text{mod }L_x$, since $B_{v_0}^{\text{top}}\sim\tilde{Z}^{\text{twist}}$ [see Eq.~(\ref{details_2dTC_TH_3}) or Fig.~\ref{fig_2dTC_TH}(b)] is not a stabilizer, we have the following operator identification,
	\begin{equation}\label{details_2dTC_TH_3}
		\begin{tikzpicture}[baseline=8.5ex, >={Triangle[round,length=2mm,width=1.33mm]}]
			\foreach \i in {0,1,2,3}{\draw (\i,0)--(\i,3);}
			\foreach \j in {0,1,2,3}{\draw (0,\j)--(4,\j);}
			\draw[dashed] (4,0)--(4,3);
			
			\draw[ForestGreen, line width=1.5pt] (0,3) circle (0.15);
			\node[ForestGreen] at (0.5,3.5) {$B_{v_0}^{\text{top}}\sim\tilde{Z}^{\text{twist}}$};
			
			\foreach \j in {0,1,2,3}{\node[ForestGreen] at (0.5,\j) {$Z$};}
			\fill[ForestGreen, opacity=0.1] (0.3,-0.2)--(0.7,-0.2)--(0.7,3.2)--(0.3,3.2);
			\node[ForestGreen] at (0.5,-0.6) {$\sim\tilde{Z}_{1/2}$};
			\foreach \j in {0,1,2,3}{\node[ForestGreen] at (3.5,\j) {$Z$};}
			\fill[ForestGreen, opacity=0.1] (3.3,-0.2)--(3.7,-0.2)--(3.7,3.2)--(3.3,3.2);
			\node[ForestGreen] at (3.5,-0.6) {$\sim\tilde{Z}_{-1/2}$};
		\end{tikzpicture}
		\quad\sim\quad
		\begin{tikzpicture}[baseline=8.5ex]
			\foreach \i in {0,1,2,3}{\draw (\i,0)--(\i,3);}
			\foreach \j in {0,1,2,3}{\draw (0,\j)--(4,\j);}
			\draw[dashed] (4,0)--(4,3);
			
			\node[ForestGreen] at (0,0.5) {$Z$};
			\node[ForestGreen] at (0.5,0) {$Z$};
			\node[ForestGreen] at (3.5,0) {$Z$};
			\fill[ForestGreen, opacity=0.1] (-0.2,1)--(0.2,1)--(0.2,0)--(-0.2,0);
			\fill[ForestGreen, opacity=0.1] (0,0.2)--(0,-0.2)--(1,-0.2)--(1,0.2);
			\fill[ForestGreen, opacity=0.1] (3,0.2)--(3,-0.2)--(4,-0.2)--(4,0.2);
			
			\node[ForestGreen] at (2,-0.6) {$\sim\tilde{Z}_{1/2}\tilde{Z}_{-1/2}\tilde{Z}^{\text{twist}}$};
			
			\node at (4.5,0.5) {$,$};
		\end{tikzpicture}
	\end{equation}
	Twist can be discussed only after we choose the specific low-energy Hamiltonian $H_{\text{bot}}$ in Eq.~(\ref{eq_2dTC_H_bot}), so that the eigenvalue of $B_{v_0}^{\text{top}}$ determines the sign of the term $\tilde{Z}_{-1/2}\tilde{Z}_{1/2}$.
	From Eq.~(\ref{details_2dTC_TH_3}) we can clearly see that the appearance of $\tilde{Z}^{\text{twist}}$ in $B_{[0,0]}=\prod_{e\revsubset [0,0]}Z_e$ is because $x(v_0)=0$.
	If $v_0=[i,L_y]$ instead, then $B_{[i,0]}\sim\tilde{Z}_{i+1/2}\tilde{Z}_{i-1/2}\tilde{Z}^{\text{twist}}$.
	
	We say there is a defect at $v_0$ when $B_{v_0}^{\text{top}}=-1$.
	Whether there is a defect at $v_0$ determines whether the 1d TFIM is twisted or not, since $B_{v_0}^{\text{top}}\sim\tilde{Z}^{\text{twist}}\sim\mathcal Z_x$, as shown in Eq.~(\ref{details_2dTC_TH_4}) (a), where $\mathcal Z_x$ is the uncontractible $Z$ dual Wilson loop along $x$-direction.
	\begin{equation}\label{details_2dTC_TH_4}
		\begin{tikzpicture}[baseline=8.5ex]
			\foreach \i in {0,1,2,3}{\draw (\i,0)--(\i,3);}
			\foreach \j in {0,1,2,3}{\draw (0,\j)--(4,\j);}
			\draw[dashed] (4,0)--(4,3);
			
			\draw[ForestGreen, line width=1.5pt] (0,3) circle (0.15);
			\node[ForestGreen] at (0.5,3.5) {$B_{v_0}^{\text{top}}\sim\tilde{Z}^{\text{twist}}$};
			
			\foreach \i in {0,1,2,3}{\node[ForestGreen] at (\i,0.5) {$Z$};}
			\fill[ForestGreen, opacity=0.1] (-0.2,0.3)--(-0.2,0.7)--(4.2,0.7)--(4.2,0.3);
			
			\node[ForestGreen] at (4.8,0.5) {$\mathcal Z_x\sim\tilde{Z}^{\text{twist}}$};
			
			\node at (2,-0.5) {(a)};
		\end{tikzpicture}
		\quad\quad
		\begin{tikzpicture}[baseline=8.5ex]
			\foreach \i in {0,1,2,3}{\draw (\i,0)--(\i,3);}
			\foreach \j in {0,1,2,3}{\draw (0,\j)--(4,\j);}
			\draw[dashed] (4,0)--(4,3);
			
			\draw[ForestGreen, line width=1.5pt] (0,3) circle (0.15);
			\node[ForestGreen] at (0.5,3.5) {$B_{v_0}^{\text{top}}\sim\tilde{Z}^{\text{twist}}$};
			
			\foreach \j in {0.5,1.5,2.5}{\node[red] at (0,\j) {$X$};}
			\fill[red, opacity=0.1] (-0.2,0)--(0.2,0)--(0.2,3)--(-0.2,3);
			
			\node at (2,-0.5) {(b)};
		\end{tikzpicture}
	\end{equation}
	The operator that creates a defect at $v_0$ (or equivalently, changes the boundary condition between twisted and untwisted), and commutes with all stabilizers is the extensive $X$ Wilson line along $y$-direction at $x=0$, denoted by $\mathcal{X}_y(0)$, as shown in Eq.~(\ref{details_2dTC_TH_4}) (b).
	Since $\mathbb{I},\tilde{X}^{\text{twist}},\tilde{Z}^{\text{twist}},\tilde{X}^{\text{twist}}\tilde{Z}^{\text{twist}}$ are linearly independent and $\dim_\mathbb{C}\mathcal{B}(\tilde{\mathcal{H}^{\text{twist}}})=4$, $\tilde{X}^{\text{twist}},\tilde{Z}^{\text{twist}}$ generate the whole $2\times 2$ matrix algebra under multiplication and addition.
	And since $\forall\ i\in\{0,1,\cdots,L_x-1\}$, $\tilde{X}_{i+1/2},\tilde{Z}_{i+1/2}$ commute with $\tilde{Z}^{\text{twist}},\tilde{X}^{\text{twist}}$, these $L_x+1$ mutually commuting pairs generate a subalgebra of $\mathcal B(\mathcal{H}_{\text{low}})$ isomorphic to $\bigotimes_{k=1}^{L_x+1}\mathcal M_2(\mathbb C)$, of complex dimension $4^{L_x+1}$.
	Combined with the dimension count $\log_2\dim\mathcal{H}_{\text{low}}=L_x+1$ obtained above, i.e. $\dim_\mathbb{C}\mathcal B(\mathcal{H}_{\text{low}})=4^{L_x+1}$, the whole $\mathcal{H}_{\text{low}}$ indeed decomposes to $\tilde{\mathcal{H}}^{\text{twist}}\bigotimes_{i=0}^{L_x-1}\tilde{\mathcal{H}}_{i+1/2}$.
	
	Finally, the global symmetry $\prod_i\tilde{X}_i$ of the 1d TFIM is naturally equivalent to the 1-form uncontractible $X$ Wilson loop along $x$-direction $\mathcal{X}_x\sim\prod_i\tilde{X}_i$.
	It is consistent that $\tilde{Z}_i$ toggles $\prod_i\tilde{X}_i$ and $\tilde{Z}_i\tilde{Z}_{i+1}$ preserves $\prod_i\tilde{X}_i$.

	\subsection{With rough top boundary (\texorpdfstring{$e$}{e} particle condensed)}\label{appendix_2dTC_rough_top}

	On the other hand, it can be calculated that the dimension of low-energy subspace $\mathcal{H}_{\text{low}}$ when the top boundary is rough [with an omitted $A_p^{\text{top}}$ term as shown in Fig.~\ref{fig_2dTC_TH}(a)] is also given by $\log_2\dim\mathcal{H}_{\text{low}}=L_x+1$.
	Then, it can be shown that $\mathcal{H}_{\text{low}}$ can be decomposed into $\tilde{\mathcal{H}}^{\text{twist}}\bigotimes_{i=1}^{L_x}\tilde{\mathcal{H}}_{i}$, where $\tilde{\mathcal{H}}^{\text{twist}}\cong\tilde{\mathcal{H}}_i\cong\mathbb{C}^2$.
	Each $\tilde{\mathcal{H}}_i$ has a local (along $x$-direction) operator algebra $\mathcal{B}(\tilde{\mathcal{H}}_i)$, thus can be identified as the Hilbert space of a local $\frac12$-spin.
	The whole $\mathcal{H}_{\text{low}}$ is identified as a length-$L_x$ ring with each vertex hosting a $\frac12$-spin, with an extra ancilla qubit with Hilbert space $\tilde{\mathcal{H}}^{\text{twist}}$, encoding twist information.
	We quickly go through the identification progress here.
	Define $\tilde{\mathcal{H}}_i$ through its operator algebra $\mathcal{B}(\tilde{\mathcal{H}}_i)$, where $\mathcal{B}(\tilde{\mathcal{H}}_i)$ is generated by $\tilde{X}_i,\tilde{Z}_i$, which are identified as follows (we show $\tilde{X}_3$ and $\tilde{Z}_3$ as examples):
	\begin{equation}
		\begin{tikzpicture}[baseline=7ex]
			\foreach \i in {0,1,2,3,4}{\draw (\i,0)--(\i,3);}
			\foreach \j in {0,1,2}{\draw (0,\j)--(5,\j);}
			\draw[dashed] (5,0)--(5,3);
			
			\fill[ForestGreen, opacity=0.1] (2,-0.2)--(4,-0.2)--(4,0.2)--(2,0.2);
			\fill[ForestGreen, opacity=0.1] (2.9,0)--(3.3,0)--(3.3,1)--(2.9,1);
			\node[ForestGreen] at (2.5,0) {$Z$};
			\node[ForestGreen] at (3.5,0) {$Z$};
			\node[ForestGreen] at (3.1,0.5) {$Z$};
			\node[ForestGreen] at (3,-0.5) {$\sim\tilde{X}_3$};
			
			\fill[red, opacity=0.1] (2.7,0)--(3.1,0)--(3.1,3)--(2.7,3);
			\foreach \j in {0.5,1.5,2.5}{\node[red] at (2.9,\j) {$X$};}
			\node[red] at (3,3.3) {$\sim\tilde{Z}_3$};
			
			\node at (0,-1) {0};
			\node at (1,-1) {1};
			\node at (2,-1) {2};
			\node at (3,-1) {3};
			\node at (4,-1) {4};
			\node at (5,-1) {5};
			\draw[->] (0,-1.25)--(5.2,-1.25);
			\node at (5.4,-1.25) {$x$};
			
			\node at (6,-1) {$,$};
		\end{tikzpicture}
	\end{equation}
	Each $\tilde{\mathcal{H}}_i$ is identified as the Hilbert space of a local $\frac12$-spin (on $x=i$ vertex) in the length-$L_x$ spin ring.
	$\tilde{X}_i,\tilde{Z}_i$ are identified as the Pauli $X,Z$ of the spin associated with the Hilbert space $\tilde{\mathcal{H}}_i$, respectively.
	This identification is related to that $e$ particles (i.e. $B_v=-1$ excitations) are condensed on the top boundary: $e$ particles are created by $X$ Wilson lines, so $X$ Wilson lines attached to the top boundary commute with all stabilizers.
	Within the low-energy subspace $\mathcal{H}_{\text{low}}$, we have $\text{stabilizer}=+1$, the stabilizers can be regarded as identity operators, we can get the $\tilde{Z}_i\tilde{Z}_{i+1}$ in Fig.~\ref{fig_2dTC_TH}(a) as follows,
	\begin{equation}
		\begin{tikzpicture}[baseline=8.5ex]
			\foreach \i in {0,1,2,3}{\draw (\i,0)--(\i,3);}
			\foreach \j in {0,1,2}{\draw (0,\j)--(4,\j);}
			\draw[dashed] (4,0)--(4,3);
			
			\foreach \i in {2,3}{
				\fill[red, opacity=0.1] (\i-0.2,0)--(\i+0.2,0)--(\i+0.2,3)--(\i-0.2,3);		
				\foreach \j in {0.5,1.5,2.5}{
					\node[red] at (\i,\j) {$X$};
				}
			}
			
			\node[red] at (1.9,-0.4) {$\sim\tilde{Z}_i$};
			\node[red] at (3.2,-0.4) {$\sim\tilde{Z}_{i+1}$};
		\end{tikzpicture}
		\quad\sim\quad
		\begin{tikzpicture}[baseline=8.5ex]
			\foreach \i in {0,1,2,3}{\draw (\i,0)--(\i,3);}
			\foreach \j in {0,1,2}{\draw (0,\j)--(4,\j);}
			\draw[dashed] (4,0)--(4,3);
			
			\fill[red, opacity=0.1] (2,-0.2)--(3,-0.2)--(3,0.2)--(2,0.2);
			\node[red] at (2.5,0) {$X$};
			
			\node[red] at (2.5,-0.5) {$\sim\tilde{Z}_i\tilde{Z}_{i+1}$};
			
			\node at (4.5,0.5) {$,$};
		\end{tikzpicture}
	\end{equation}
	where $i=1,2,\cdots,L_x-1$.
	The exception is that when $i=0\text{ mod }L_x$, since $A_{p_0}^{\text{top}}\sim\tilde{Z}^{\text{twist}}$ [see Eq.~(\ref{details_2dTC_TH_5}) or Fig.~\ref{fig_2dTC_TH}(a)] is not a stabilizer, we have the following operator identification,
	\begin{equation}\label{details_2dTC_TH_5}
		\begin{tikzpicture}[baseline=8.5ex]
			\foreach \i in {0,1,2,3}{\draw (\i,0)--(\i,3);}
			\foreach \j in {0,1,2}{\draw (0,\j)--(4,\j);}
			\draw[dashed] (4,0)--(4,3);
			
			\foreach \i in {0,1}{
				\fill[red, opacity=0.1] (\i-0.2,0)--(\i+0.2,0)--(\i+0.2,3)--(\i-0.2,3);		
				\foreach \j in {0.5,1.5,2.5}{
					\node[red] at (\i,\j) {$X$};
				}
			}
			
			\node[red] at (0,-0.4) {$\sim\tilde{Z}_0$};
			\node[red] at (1,-0.4) {$\sim\tilde{Z}_1$};
			
			\draw[red, line width=1.5pt] (0.5,2.5) circle (0.15);
			\node[red] at (0.5,2.9) {$A_{p_0}^{\text{top}}$};
			\node[red] at (1.4,2.9) {$\sim\tilde{Z}^{\text{twist}}$};
		\end{tikzpicture}
		\quad\sim\quad
		\begin{tikzpicture}[baseline=8.5ex]
			\foreach \i in {0,1,2,3}{\draw (\i,0)--(\i,3);}
			\foreach \j in {0,1,2}{\draw (0,\j)--(4,\j);}
			\draw[dashed] (4,0)--(4,3);
			
			\fill[red, opacity=0.1] (0,-0.2)--(1,-0.2)--(1,0.2)--(0,0.2);
			\node[red] at (0.5,0) {$X$};
			
			\node[red] at (1,-0.5) {$\sim\tilde{Z}_0\tilde{Z}_{1}\tilde{Z}^{\text{twist}}$};
			
			\node at (4.5,0.5) {$,$};
		\end{tikzpicture}
	\end{equation}
	From Eq.~(\ref{details_2dTC_TH_5}) we can clearly see that the appearance of $\tilde{Z}^{\text{twist}}$ in $X_{[1/2,0]}$ is because $x(p_0)=1/2$.
	If $p_0=[i,L_y]$ instead, then $X_{[i,0]}\sim\tilde{Z}_{i-1/2}\tilde{Z}_{i+1/2}\tilde{Z}^{\text{twist}}$.
	
	The rest of the operator identification is similar to the case when the top boundary is smooth, we do not illustrate any more [see Fig.~\ref{fig_2dTC_TH}(a)].

	\section{Higher-order stabilizer redundancy: a tool of counting GSD}\label{section_redundancy_formalism_theory}
	
	A usual way of counting GSD of $\mathbb{Z}_2$ stabilizer code is counting the number of spins and the number of independent stabilizer generators, then using the formula
	\begin{equation}\label{red_formal_theory_5}
		\log_2\text{GSD}=\#\,\text{spin}\ -\ \#\,\text{independent stabilizer generators}
	\end{equation}
	to obtain GSD.
	For stabilizer code with special structure (e.g. translation invariant), sometimes it is convenient to count a set of dependent stabilizer generators (with some redundant generators), and then count the redundancy of stabilizer generators.
	This can be formalized as the following exact chain:
	\begin{equation}
		\mathbb{F}_2^{\mathcal{R}\mathbb{S}}\xrightarrow{V_s}\mathbb{F}_2^{\mathbb{S}}\xrightarrow{V_P}\mathbb{F}_2^{2n}\,,
	\end{equation}
	where the vector space $\mathbb{F}_2^{2n}$ (representing the Pauli group modulo phase factors) is spanned by the symbols $X_1,\cdots,X_n,Z_1,\cdots,Z_n$\footnote{The commutation relation can be given by adding a symplectic structure to this vector space.}, $\mathbb{S}=\left\{s_1,\cdots,s_{|\mathbb{S}|}\right\}$ is a set of (possibly redundant) stabilizer generators, $\mathcal{R}\mathbb{S}=\left\{r^1_1,\cdots,r^1_{|\mathcal{R}\mathbb{S}|}\right\}$ is an independent set of stabilizer redundancies, i.e. (i) $\left\{V_s(r^1_i):i=1,\cdots,|\mathcal{R}\mathbb{S}|\right\}$ is linearly independent, and (ii) $\forall r^1\in\mathbb{F}_2^{\mathcal{R}\mathbb{S}}$, $V_P\circ V_s(r^1)=0$.
	$r^1$ being a stabilizer redundancy is encoded into the chain being exact here.
	We go through details of this formalism now.
	
	Suppose the set $\mathbb{S}$ generates the stabilizer group $\mathcal{S}$.
	We represent each $s\in\mathbb{S}$ as a unit $|\mathbb{S}|$-dimensional vector over $\mathbb{F}_2$, these vectors span $\mathbb{F}_2^{\mathbb{S}}=\mathbb{F}_2\left[s_1,\cdots,s_{|\mathbb{S}|}\right]$.
	With the abuse of notation, we do not distinguish $s_i$ and their corresponding unit/basis vectors in $\mathbb{F}_2^{\mathbb{S}}$.
	$s_i$ are Pauli operators on the $n$ physical qubits, and can be written as the products of single qubit $X$ and $Z$, the information of $s_i$ can be encoded into a linear map $V_P$ from $\mathbb{F}_2^{\mathbb{S}}$ to $\mathbb{F}_2^{2n}$, s.t.
	\begin{equation}
		s_i=\prod_{j=1}^n X_j^{(V_P(s_i))_j} Z_j^{(V_P(s_i))_{j+n}}\ ,\quad i=1,\cdots,|\mathbb{S}|\,.
	\end{equation}
	A stabilizer redundancy is a linear combination of $s_i$, e.g. $\xi=\sum_{i=1}^{|\mathbb{S}|}c_is_i\in\mathbb{F}_2^{\mathbb{S}}$, s.t.
	$V_P(\xi)=0$.
	If $\xi_1,\xi_2\in\mathbb{F}_2^{\mathbb{S}}$ are stabilizer redundancies, $\xi_1+\xi_2$ is also a stabilizer redundancy.
	The stabilizer redundancies form an abelian group $\mathcal{RS}$, with an independent generator set $\mathcal{R}\mathbb{S}$.
	Analogous to encoding stabilizer generators $\mathbb{S}$ into the linear map $V_P$, the generator set $\mathcal{R}\mathbb{S}$ can be encoded into a linear map $V_s:\mathbb{F}_2^{\mathcal{R}\mathbb{S}}\to\mathbb{F}_2^{\mathbb{S}}$, s.t.
	$\forall r^1\in\mathcal{R}\mathbb{S}$,
	\begin{equation}\label{red_formal_theory_1}
		V_P\circ V_s(r^1)\overset{P}{=}0\,,
	\end{equation}
	where $\overset{P}{=}$ is the equality in $\mathbb{F}_2^{2n}$, i.e. for any two Pauli operators $P_1,P_2$,
	\begin{equation}
		P_1+P_2\overset{P}{=}0\quad\Longleftrightarrow\quad P_1\cdot P_2\propto\mathbb{I}\,.
	\end{equation}
	The condition in Eq.~(\ref{red_formal_theory_1}) is equivalent to
	\begin{equation}
		\text{im}\, V_s\subseteq\text{ker}\, V_P\,.
	\end{equation}
	$\ker\,V_P$ is by definition the space of stabilizer redundancies.
	Then since $\mathcal R\mathbb S$ generates all the stabilizer redundancies by definition, we have
	\begin{equation}\label{red_formal_theory_2}
		\text{im}\, V_s\overset{s}{=}\text{ker}\, V_P\,,
	\end{equation}
	where $\overset{s}{=}$ is the equivalence in $\mathbb{F}_2^{\mathbb{S}}$.
	According to rank-nullity theorem,
	\begin{gather}
		\dim(\text{im}\,V_P)+\dim(\text{ker}\,V_P)=\dim(\text{domain}\,V_P)=|\mathbb{S}|\,,\label{rank_nullity_theorem_1}\\
		\dim(\text{im}\,V_s)+\dim(\text{ker}\,V_s)=\dim(\text{domain}\,V_s)=|\mathcal{R}\mathbb{S}|\,.\label{rank_nullity_theorem_2}
	\end{gather}
	We have supposed that $\mathcal{R}\mathbb{S}$ is an independent set of stabilizer redundancies, which means $\text{ker}\,V_s=\{0\}$, $\dim(\text{ker}\,V_s)=0$, substituting it into Eq.~(\ref{rank_nullity_theorem_2}), we get
	\begin{equation}\label{red_formal_theory_3}
		\dim(\text{im}\,V_s)=|\mathcal{R}\mathbb{S}|\,.
	\end{equation}
	Substituting Eq.~(\ref{red_formal_theory_2}) into Eq.~(\ref{rank_nullity_theorem_1}), we get
	\begin{equation}\label{red_formal_theory_4}
		\dim(\text{im}\,V_P)+\dim(\text{im}\,V_s)=|\mathbb{S}|\,.
	\end{equation}
	Substituting Eq.~(\ref{red_formal_theory_3}) into Eq.~(\ref{red_formal_theory_4}), and notice that $\#\text{independent stabilizer}=\dim(\text{im}\,V_P)$, we get
	\begin{equation}
		\#\text{independent stabilizer}=|\mathbb{S}|-|\mathcal{R}\mathbb{S}|\,.
	\end{equation}
	Take it back to Eq.~(\ref{red_formal_theory_5}), we get
	\begin{equation}
		\log_2\text{GSD}=n-\left(|\mathbb{S}|-|\mathcal{R}\mathbb{S}|\right)\,.
	\end{equation}
	
	However, sometimes we meet the situation that $\mathcal{R}\mathbb{S}$ is dependent, we have to further consider its dependency.
	For stabilizer codes with indefinite size [e.g. 3-dimensional toric code, X-cube, tetra-digit (TD) models], this situation is quite common.
	We discuss this case now.
	Suppose $\mathcal{R}\mathbb{S}=\left\{\left.r^1_i\right|i=1,\cdots,|\mathcal{R}\mathbb{S}|\right\}$.
	Define an order-2 stabilizer redundancy as a linear combination of $\{r^1_i\}$, say, $\eta\in\mathbb{F}_2^{\mathcal{R}\mathbb{S}}$, s.t.
	$V_s(\eta)\overset{s}{=}0$.
	The order-2 stabilizer redundancies form an abelian group $\mathcal{R}^2\mathcal{S}$, with a generator set $\mathcal{R}^2\mathbb{S}$.
	Analogous to encoding stabilizer generators $\mathbb{S}$ into the linear map $V_P$ and encoding order-1 stabilizer redundancy generators $\mathcal{R}\mathbb{S}$ into the linear map $V_s$, the generator set $\mathcal{R}^2\mathbb{S}$ can be encoded into a linear map $V_1:\mathbb{F}_2^{\mathcal{R}^2\mathbb{S}}\to\mathbb{F}_2^{\mathcal{R}\mathbb{S}}$, s.t.
	$\forall r^2\in\mathbb{F}_2^{\mathcal{R}^2\mathbb{S}}$,
	\begin{equation}
		V_s\circ V_1(r^2)\overset{s}{=}0\,,
	\end{equation}
	which implies
	\begin{equation}
		\text{im}\, V_1\overset{1}{=}\text{ker}\,V_s\,.
	\end{equation}
	This process can be iteratively done, ending with an infinite exact chain
	\begin{equation}
		\cdots\xrightarrow{V_2}\mathbb{F}_2^{\mathcal{R}^2\mathbb{S}}\xrightarrow{V_1}\mathbb{F}_2^{\mathcal{R}\mathbb{S}}\xrightarrow{V_s}\mathbb{F}_2^{\mathbb{S}}\xrightarrow{V_P}\mathbb{F}_2^{2n}\,.
	\end{equation}
	For clarity, we denote the equality in $\mathbb{F}_2^{\mathcal{R}^m\mathbb{S}}$ as $\overset{m}{=}$.
	The exact chain gives a formula for computing GSD
	\begin{equation}
		\log_2\text{GSD}=n-|\mathbb{S}|+|\mathcal{R}\mathbb{S}|-|\mathcal{R}^2\mathbb{S}|+\cdots
	\end{equation}
	When $\mathcal{R}^m\mathbb{S}$ is linearly independent, $|\mathcal{R}^M\mathbb{S}|=0$ for all $M>m$, the series is cut off.

	\section{Computing the dimension of low-energy subspace of X-cube FTH}\label{appendix_X_cube_low_energy_subspace}\label{app_X_cube_ee_mm_counting}

	\subsection{Under \texorpdfstring{$(mm)$}{(mm)} top boundary}\label{appendix_X_cube_FTH_low_energy_subspace_planeon_condensed_top}

	In this subsection, we use the higher-order redundancy theory introduced in Appendix~\ref{section_redundancy_formalism_theory} to compute the dimension of the low-energy subspace $\mathcal{H}_{\text{low}}$, i.e. the common $+1$ eigenspace of all the stabilizers\footnote{Here ground space is replaced by the low-energy subspace $\mathcal{H}_{\text{low}}$, the underlying math is the same.} of the X-cube FTH under the $(mm)$ top boundary.
	The lattice and stabilizer setting follow Sec.~\ref{subsec_X_cube_ee_mm} and Fig.~\ref{fig_mm_scheme2A}(a).
	For clarity, we start by counting the number of spins and the number of independent stabilizers when all $B_{v,l}$ terms are present: denote the number of vertices, edges, plaquettes, cubes as $V,E,P,C$, respectively,
	\begin{gather}
		V=L_xL_y(L_z+1)\quad,\quad E=L_xL_y(L_z+1)+L_yL_x(L_z+1)+L_zL_xL_y\notag\\
		P=L_xL_y(L_z+1)+L_xL_zL_y+L_yL_zL_x\quad,\quad C=L_xL_yL_z
	\end{gather}
	The number of spin is
	\begin{equation}
		\#\,\text{spin}=E=2L_xL_y(L_z+1)+L_xL_yL_z\,.
	\end{equation}
	The number of $A_c$ term is
	\begin{equation}
		\#\,A_c=C=L_xL_yL_z\,.
	\end{equation}
	The incomplete $B_{v,l}$ terms with three edges on the top and bottom boundaries are present, so the number of $B_{v,l}$ for any $l\in\{xy,yz,zx\}$ equals the number of vertices, i.e.
	\begin{gather}
		\#\,B_{v,xy}=\#\,B_{v,yz}=\#\,B_{v,zx}=V=L_xL_y(L_z+1)\,,
	\end{gather}
	the number of $B_{v,l}$ terms is
	\begin{equation}
		\#\,B_{v,l}=3V=3L_xL_y(L_z+1)\,.
	\end{equation}
	Let us stick to the following symbol defined in Ref.~\cite{hyt_1_prepare_TD_model_via_SQC}: for any set-valued function 
	\begin{equation}
		g:X\to Y,\ x\mapsto y=g(x)\,,
	\end{equation}
	define 
	\begin{align}
		g(*)\equiv\bigcup_{x\in X}g(x)\,.
	\end{align}
	With the notation $*$, cubes or plaquettes in an extended line or plane can be conveniently represented, as shown in Fig.~\ref{*-intuition}.
	\begin{figure}[t]
		\centering
		\begin{subfigure}{0.27\textwidth}
			\centering
			\begin{tikzpicture}[baseline=-3ex]
				\foreach \j in {0,0.8,1.6}{\foreach \k in {0,0.8}{\draw(0,\j,\k)--(2.4,\j,\k);}}
				\foreach \i in {0,0.8,1.6}{\foreach \j in {0,0.8,1.6}{\draw(\i,\j,0)--(\i,\j,1.6);}}
				\foreach \i in {0,0.8,1.6}{\foreach \k in {0,0.8}{\draw(\i,0,\k)--(\i,2.4,\k);}}

				\draw[->, blue, line width=1pt] (-1.5,0,0)--(-1.5,0,1);
				\draw[->, blue, line width=1pt] (-1.5,0,0)--(-0.7,0,0);
				\draw[->, blue, line width=1pt] (-1.5,0,0)--(-1.5,0.8,0);
				\node[blue] at (-1.68,0,1) {$x$};
				\node[blue] at (-0.7,-0.22,0) {$y$};
				\node[blue] at (-1.25,0.8,0) {$z$};

				\draw[ForestGreen, line width=1.5pt] (0,0,0.8)--(0,2.4,0.8);
				\fill[red, opacity=0.1] (0.8,0,0.8)--(0.8,0,0)--(0.8,2.4,0)--(0.8,2.4,0.8);
				\draw[red, dashed] (0.8,2.4,0)--(0.8,2.4,0.8);
				\fill[blue, opacity=0.1] (1.6,0,0.8)--(2.4,0,0.8)--(2.4,0,0)--(2.4,2.4,0)--(1.6,2.4,0)--(1.6,2.4,0.8);
				\draw[blue, dashed]
				(1.6,2.4,0)--(1.6,2.4,0.8)
				(1.6,2.4,0)--(2.4,2.4,0)
				(2.4,2.4,0)--(2.4,2.4,0.8)
				(1.6,2.4,0.8)--(2.4,2.4,0.8)
				(2.4,2.4,0.8)--(2.4,0,0.8)
				(2.4,0,0)--(2.4,0,0.8)
				(2.4,0,0)--(2.4,2.4,0);

				\fill[blue] (0,0) circle (2pt);

				\node[anchor=north west, xshift=35] at (current bounding box.north west) {(a)};
			\end{tikzpicture}
		\end{subfigure}
		\begin{subfigure}{0.21\textwidth}
			\centering
			\begin{tikzpicture}[baseline=-3ex]
				\foreach \j in {0,0.8,1.6}{\foreach \k in {0,0.8}{\draw(0,\j,\k)--(2.4,\j,\k);}}
				\foreach \i in {0,0.8,1.6}{\foreach \j in {0,0.8,1.6}{\draw(\i,\j,0)--(\i,\j,1.6);}}
				\foreach \i in {0,0.8,1.6}{\foreach \k in {0,0.8}{\draw(\i,0,\k)--(\i,2.4,\k);}}

				\fill[ForestGreen, opacity=0.1] (0,0.8,0.8)--(0.8,0.8,0.8)--(0.8,0,0.8)--(0.8,0,1.6)--(0,0,1.6)--(0,0.8,1.6);
				\draw[ForestGreen, dashed]
				(0,0.8,1.6)--(0.8,0.8,1.6)--(0.8,0,1.6)--(0,0,1.6)--cycle;
				\fill[red, opacity=0.1] (0.8,2.4,0.8)--(1.6,2.4,0.8)--(1.6,0,0.8)--(1.6,0,1.6)--(0.8,0,1.6)--(0.8,2.4,1.6);
				\draw[red, dashed]
				(0.8,2.4,0.8)--(1.6,2.4,0.8)
				(0.8,2.4,1.6)--(1.6,2.4,1.6)
				(0.8,2.4,0.8)--(0.8,2.4,1.6)
				(1.6,2.4,0.8)--(1.6,2.4,1.6)
				(0.8,2.4,1.6)--(0.8,0,1.6)
				(1.6,2.4,1.6)--(1.6,0,1.6)
				(0.8,0,1.6)--(1.6,0,1.6);
				\fill[blue, opacity=0.1] (1.6,2.4,0)--(2.4,2.4,0)--(2.4,0,0)--(2.4,0,1.6)--(1.6,0,1.6)--(1.6,2.4,1.6);
				\draw[blue, dashed]
				(1.6,2.4,0)--(2.4,2.4,0)
				(2.4,2.4,0)--(2.4,2.4,1.6)
				(2.4,2.4,0)--(2.4,0,0)
				(2.4,2.4,1.6)--(2.4,0,1.6)
				(1.6,2.4,0)--(1.6,2.4,1.6)
				(1.6,2.4,1.6)--(1.6,0,1.6)
				(2.4,0,0)--(2.4,0,1.6)
				(1.6,0,1.6)--(2.4,0,1.6)
				(1.6,2.4,1.6)--(2.4,2.4,1.6);

				\fill[blue] (0,0) circle (2pt);

				\node[anchor=north west, xshift=-15] at (current bounding box.north west) {(b)};
			\end{tikzpicture}
		\end{subfigure}
		\caption{Illustration of $*$ notation.
		The origin is marked by a blue dot in the two figures, and the lattice constant is taken to be 1.
		(a) The green line is a circle $S^1$, which can be denoted by $\left[1,0,*\right]$, and viewed as the trajectory (uncontractible loop) of dragging the vertex $[1,0,0]$ along $x_3$-direction.
		The red surface is a spanned circle $S^1\times B^1$, which can be denoted by $\left[\frac12,1,*\right]$, and viewed as the trajectory (uncontractible loop) of dragging the edge $\left[\frac12,1,0\right]$ along $x_3$-direction.
		Similarly, the blue volume is a spanned circle $S^1\times B^2$, which can be denoted by $\left[\frac12,\frac52,*\right]$, and viewed as the trajectory (uncontractible loop) of dragging the plaquette $\left[\frac12,\frac52,0\right]$ along $x_3$-direction.
		(b) The green cube, red volume and blue volume can be denoted by $\left[\frac32,\frac12,\frac12\right],\ \left[\frac32,\frac32,*\right]$ and $\left[*,\frac52,*\right]$, respectively.
		Topologically, they are $B^3$, $T^1\times B^2$ and $T^2\times B^1$, respectively.
		$T^1=S^1$.
		For a $D$-dimensional cubic lattice, when $m$ inputs are half-integers and $n$ inputs are taken to be $*$, the output is topologically $T^n\times B^m$.}
		\label{*-intuition}
	\end{figure}
	Let the stabilizer generator set
	\begin{equation}
		\mathbb{S}=\left\{A_c\right\}\cup\left\{B_{v,l}\right\}\,,
	\end{equation}
	we use the higher-order redundancy theory introduced in Appendix~\ref{section_redundancy_formalism_theory} to count the number of independent stabilizer generators.
	All the $A_c$ terms redundancies can be generated by the following redundancies,
	\begin{equation}
		V_s\left(r^1\Big(A,\Big[*,*,z+\frac12\Big]\Big)\right)\overset{s}{=}\sum_{x\in\mathbb{Z}_{L_x},y\in\mathbb{Z}_{L_y}}A_{[x+\frac12,y+\frac12,z+\frac12]}\ ,\quad z\in\mathbb{Z}_{L_z}\,,
	\end{equation}
	so we have
	\begin{equation}
		\#\,r^1\left(A,\Big[*,*,z+\frac12\Big]\right)=L_z\,.
	\end{equation}
	For the $B_{v,l}$ terms, first we notice that they have the following local redundancies,
	\begin{equation}
		V_s\left(r^1(B,v)\right)\overset{s}{=}\sum_{l=xy,yz,zx}B_{v,l}\ \ ,\ \ \text{for any vertex }v\ \ ,\ \ \#=V\,.
	\end{equation}
	Then, within any $xy$ plane $[*,*,z]$,
	\begin{equation}
		V_s\left(r^1(B,[*,*,z])\right)\overset{s}{=}\sum_{v\subset[*,*,z]}B_{v,xy}\ \ ,\ \ z=0,1,\cdots,L_z\ \ ,\ \ \#=L_z+1\,.
	\end{equation}
	Similarly, within any $yz$ plane $[x,*,*]$ or $zx$ plane $[*,y,*]$,
	\begin{gather}
		V_s\left(r^1(B,[x,*,*])\right)\overset{s}{=}\sum_{v\subset[x,*,*]}B_{v,yz}\ \ ,\ \ x=0,1\cdots,L_x-1\ \ ,\ \ \#=L_x\,,\\
		V_s\left(r^1(B,[*,y,*])\right)\overset{s}{=}\sum_{v\subset[*,y,*]}B_{v,zx}\ \ ,\ \ y=0,1\cdots,L_y-1\ \ ,\ \ \#=L_y\,.
	\end{gather}
	Then, there is a second order redundancy
	\begin{equation}\label{second_order_redundancy_X_cube}
		V_1\left(r^2\right)\overset{1}{=}\sum_{v}r^1(B,v)+\sum_{z=0}^{L_z}r^1(B,[*,*,z])+\sum_{x=0}^{L_x-1}r^1(B,[x,*,*])+\sum_{y=0}^{L_y-1}r^1(B,[*,y,*])\ \ ,\ \ \#=1\,.
	\end{equation}
	Applying $V_s$ to the right hand side counts every $B_{v,l}$ twice, so it gives zero, as an order-2 redundancy requires.
	In all, 
	\begin{align}
		\log_2\dim\,\mathcal{H}_{\text{low}}=\#\text{spin}-(\# A_c + \# B_{v,l})+ \# r^1 - \# r^2 = L_x+L_y+2L_z\,.
	\end{align}
	
	Next, delete the $B_{v,zx},B_{v,yz}$ terms on the bottom boundary from the stabilizer generator set.
	We count the change of $\log_2\dim\,\mathcal{H}_{\text{low}}$ when deleting the $B_{v,l}$ terms on the bottom boundary from the stabilizer generator set.
	
	First, there are 2 $B_{v,zx},B_{v,yz}$ terms associated with each vertex on the bottom boundary, in total there are $2L_xL_y$ of them.
	Denote the bottom boundary (subcomplex) as b.b..
	Deleting the $B_{v,zx},B_{v,yz}$ terms on the bottom boundary makes the following 1st order stabilizer redundancies disappear:
	\begin{gather}
		V_s\big(r^1(B,v)\big)\overset{s}{=}\sum_{l=xy,yz,zx}B_{v,l}\ \ ,\ \ v\in\text{b.b.}\ \ ,\ \ \#=L_xL_y\\
		V_s\big(r^1(B,[x,*,*])\big)\overset{s}{=}\sum_{v\subset[x,*,*]}B_{v,yz}\ \ ,\ \ x=0,1,\cdots,L_x-1\ \ ,\ \ \#=L_x\label{X_cube_redundancy_1}\\
		V_s\big(r^1(B,[*,y,*])\big)\overset{s}{=}\sum_{v\subset[*,y,*]}B_{v,zx}\ \ ,\ \ y=0,1,\cdots,L_y-1\ \ ,\ \ \#=L_y\,.\label{X_cube_redundancy_2}
	\end{gather}
	The second order redundancy in Eq.~(\ref{second_order_redundancy_X_cube}) also disappears after the deletion, $\#=1$.
	Therefore, deleting the $B_{v,zx},B_{v,yz}$ terms on the bottom boundary increases $\log_2\dim\,\mathcal{H}_{\text{low}}$ by 
	\begin{equation}
		2L_xL_y-(L_xL_y+L_x+L_y)+1=L_xL_y-L_x-L_y+1\,,
	\end{equation}
	changes $\log_2\dim\,\mathcal{H}_{\text{low}}$ from $L_x+L_y+2L_z$ to $L_xL_y+2L_z+1$.
	
	Finally, let us count the $\log_2\dim\,\mathcal{H}_{\text{low}}$ after deleting the $B_{v,zx},B_{v,yz}$ terms on the top boundary (the ones marked by half-circles in Fig.~\ref{fig_mm_ee_twists}(b)).
	There are $2(L_x+L_y-1)$ $B_{v,l}$ terms marked by half-circles in Fig.~\ref{fig_mm_ee_twists}(b), and deleting them makes the following $L_x+L_y-1$ order-1 redundancies disappear:
	\begin{gather}
		V_s\big(r^1(B,[i,0,L_z])\big)\overset{s}{=}\sum_{l=xy,yz,zx}B_{[i,0,L_z],l}\ \ ,\ \ i=0,1,\cdots,L_x-1\ \ ,\ \ \#=L_x\,,\\
		V_s\big(r^1(B,[0,j,L_z])\big)\overset{s}{=}\sum_{l=xy,yz,zx}B_{[0,j,L_z],l}\ \ ,\ \ j=1,2,\cdots,L_y-1\ \ ,\ \ \#=L_y-1\,.
	\end{gather}
	So, deleting the $B_{v,l}$ terms marked by half-circles in Fig.~\ref{fig_mm_ee_twists}(b) increases $\log_2\dim\,\mathcal{H}_{\text{low}}$ by $L_x+L_y-1$, deleting these $B_{v,l}$ terms changes $\log_2\dim\,\mathcal{H}_{\text{low}}$ from $L_xL_y+2L_z+1$ to $L_xL_y+L_x+L_y+2L_z$, i.e.
	\begin{equation}\label{eq::X-cube_FTH_smooth_GSD}
		\log_2\dim\,\mathcal{H}_{\text{low}}^{(mm)}=L_xL_y+L_x+L_y+2L_z\,,
	\end{equation} 
	which is the final result of $\log_2\dim\,\mathcal{H}_{\text{low}}^{(mm)}$ in our construction.

	\subsection{Under \texorpdfstring{$(ee)$}{(ee)} top boundary}\label{appendix_low_energy_subspace_dim_under_lineon_condensed_top}

	In this subsection, we use the higher-order redundancy theory introduced in Appendix~\ref{section_redundancy_formalism_theory} to compute the dimension of the low-energy subspace $\mathcal{H}_{\text{low}}$, i.e. the common $+1$ eigenspace of all the stabilizers, under the $(ee)$ top boundary.
	The stabilizer setting follows Sec.~\ref{subsec_X_cube_ee_mm} and Fig.~\ref{fig_ee_scheme2A}(a).
	To compare with the $(mm)$ sandwich on the $L_x\times L_y\times L_z$ lattice, we take the $(ee)$ sandwich to have $L_z+1$ layers of cubes, as the SQC in Eq.~(\ref{X_cube_rough_to_smooth_top_SQC}) maps the extra layer of vertical edges to product states.
	For clarity, we start by counting the number of spins and the number of independent stabilizers when all $A_c$ and $B_{v,l}$ terms are present\footnote{Again, the $B_{v,l}$ terms on vertices with $z=L_z+1$ are excluded, since these single spin Pauli $Z$ would break the stabilizer code condition.}:
	
	Denote $L'_z\equiv L_z+1$.
	The number of cells are
	\begin{gather}
		V=L_xL_y(L'_z+1)\quad,\quad E=L_xL_yL'_z+L_yL_xL'_z+L'_zL_xL_y=3L_xL_yL'_z\notag\\
		P=L_xL_yL'_z+L_xL'_zL_y+L_yL'_zL_x=3L_xL_yL'_z\quad,\quad C=L_xL_yL'_z
	\end{gather}
	The number of spin is
	\begin{equation}
		\#\text{spin}=E=3L_xL_yL'_z
	\end{equation}
	Each cube hosts an $A_c$ term,
	\begin{equation}
		\# A_c=C=L_xL_yL'_z\,.
	\end{equation}
	Not all vertices host $B_{v,l}$ terms.
	The $L_xL_y$ vertices on the top boundary do not host $B_{v,l}$ terms, so
	\begin{equation}
		\# B_{v,l}=3(V-L_xL_y)=3L_xL_yL'_z\,.
	\end{equation}
	All the $A_c$ terms redundancies can be generated by the following redundancies,
	\begin{equation}\label{X_cube_redundancy_3}
		V_s\left(r^1\Big(A,\Big[*,*,z+\frac12\Big]\Big)\right)\overset{s}{=}\sum_{x\in\mathbb{Z}_{L_x},y\in\mathbb{Z}_{L_y}}A_{[x+\frac12,y+\frac12,z+\frac12]}\ \ ,\ \ z\in\mathbb{Z}_{L'_z}\ \ ,\ \ \#=L'_z
	\end{equation}
	For the $B_{v,l}$ terms, there is a redundancy on each vertex,
	\begin{equation}\label{X_cube_redundancy_4}
		V_s\left(r^1(B,v)\right)\overset{s}{=}\sum_{l=xy,yz,zx}B_{v,l}\ \ ,\ \ \text{for any vertex }v\text{ hosting }B_{v,l}\text{ terms}\ \ ,\ \ \#=L_xL_yL'_z\,.
	\end{equation}
	Within each $xy$ plane $[*,*,z]$,
	\begin{equation}\label{X_cube_redundancy_5}
		V_s\left(r^1(B,[*,*,z])\right)\overset{s}{=}\sum_{v\subset[*,*,z]}B_{v,xy}\ \ ,\ \ z=0,1,\cdots,L'_z-1\ \ ,\ \ \#=L'_z\,.
	\end{equation}
	The stabilizer redundancies in Eqs.~(\ref{X_cube_redundancy_3},\ref{X_cube_redundancy_4},\ref{X_cube_redundancy_5}) are independent, so
	\begin{equation}
		\log_2\dim\,\mathcal{H}_{\text{low}}=\#\text{spin}-(\# A_c+\# B_{v,l})+\# r^1 = 2L'_z=2(L_z+1)\,.
	\end{equation}
	
	Next, delete the $B_{v,zx},B_{v,yz}$ terms on the bottom boundary.
	There are $2L_xL_y$ $B_{v,zx},B_{v,yz}$ terms on the bottom boundary, and $L_xL_y$ 1st order redundancies disappear because of the deletion of $B_{v,l}$ terms on the bottom boundary, which are
	\begin{gather}
		V_s\left(r^1(B,v)\right)\overset{s}{=}\sum_{l=xy,yz,zx}B_{v,l}\ \ ,\ \ v\in\text{b.b.}\ \ ,\ \ \#=L_xL_y\,,
	\end{gather}
	where $\text{b.b.}\equiv[*,*,0]$ is the bottom boundary.
	So, deleting the $B_{v,l}$ terms on the bottom boundary changes $\log_2\dim\,\mathcal{H}_{\text{low}}$ from $2(L_z+1)$ to $2(L_z+1)+(2L_xL_y-L_xL_y)=L_xL_y+2(L_z+1)$.
	
	Then, delete the $A_c^{\text{top}}$ terms with $\frac12\in\{x(c),y(c)\}$ (marked by shaded top cubes with circles in Fig.~\ref{fig_mm_ee_twists}(a)).
	There are $L_x+L_y-1$ such deleted $A_c$ terms on the top boundary, and deleting these $A_c$ terms makes 1 stabilizer redundancy disappear, i.e.
	\begin{equation}\label{X_cube_redundancy_6}
		V_s\left(r^1\Big(A,\Big[*,*,L_z+\frac12\Big]\Big)\right)\overset{s}{=}\sum_{x\in\mathbb{Z}_{L_x},y\in\mathbb{Z}_{L_y}}A_{[x+\frac12,y+\frac12,L_z+\frac12]}\ \ ,\ \ \#=1\,.
	\end{equation}
	So, deleting these $A_c$ terms on the top boundary changes $\log_2\dim\,\mathcal{H}_{\text{low}}^{(ee)}$ from $L_xL_y+2(L_z+1)$ to $L_xL_y+2(L_z+1)+(L_x+L_y-1-1)=L_xL_y+L_x+L_y+2L_z$, which is the same as the $\dim\,\mathcal{H}_{\text{low}}^{(mm)}$ of X-cube FTH on $L_x\times L_y\times L_z$ lattice with the $(mm)$ top boundary.

	\subsection{\texorpdfstring{Under $(em)$ and $(me)$ top boundaries}{Under (em) and (me) top boundaries}}\label{app_X_cube_em_me_counting}
	In this subsection, we use the higher-order redundancy theory introduced in Appendix~\ref{section_redundancy_formalism_theory} to compute the dimension of the low-energy subspace $\mathcal{H}_{\text{low}}$ under the $(em)$ top boundary.
	The stabilizer setting follows Sec.~\ref{subsec_X_cube_em_me} and Fig.~\ref{fig_em_scheme2A}(a): the horizontal layers $z=0,1,\cdots,L_z-1$ are complete, while on the top layer $z=L_z$ only the $x$-edges are kept and the $y$-edges are truncated.
	Accordingly, the cube terms in the top layer $[*,*,L_z-\frac12]$ are the 10-edge $A_c^{\text{top}}$, and each top vertex $[i,j,L_z]$ hosts only the 3-leg $B_{[i,j,L_z],zx}^{\text{top}}$.
	For clarity, we start by counting the number of spins and the number of independent stabilizers when all $A_c$ and $B_{v,l}$ terms are present, including the $B_{v,zx},B_{v,yz}$ terms on the bottom boundary.

	Counting the $x$-, $y$- and $z$-edges, the number of spin is
	\begin{equation}
		\#\,\text{spin}=E=L_xL_y(L_z+1)+L_xL_yL_z+L_xL_yL_z=L_xL_y(3L_z+1)\,.
	\end{equation}
	Each cube hosts an $A_c$ term,
	\begin{equation}
		\#\,A_c=C=L_xL_yL_z\,.
	\end{equation}
	Each vertex with $z<L_z$ hosts three $B_{v,l}$ terms, while each top vertex hosts only $B_{v,zx}^{\text{top}}$, so
	\begin{equation}
		\#\,B_{v,l}=3L_xL_yL_z+L_xL_y=L_xL_y(3L_z+1)\,.
	\end{equation}
	All the $A_c$ terms redundancies can be generated by the planar redundancies
	\begin{equation}
		V_s\left(r^1\Big(A,\Big[*,*,z+\frac12\Big]\Big)\right)\overset{s}{=}\sum_{x\in\mathbb{Z}_{L_x},y\in\mathbb{Z}_{L_y}}A_{[x+\frac12,y+\frac12,z+\frac12]}\ \ ,\ \ z\in\mathbb{Z}_{L_z}\ \ ,\ \ \#=L_z\,,
	\end{equation}
	including the top layer, where every top $x$-edge is shared by two 10-edge $A_c^{\text{top}}$.
	For the $B_{v,l}$ terms, there is a local redundancy on each vertex hosting three $B_{v,l}$ terms,
	\begin{equation}
		V_s\left(r^1(B,v)\right)\overset{s}{=}\sum_{l=xy,yz,zx}B_{v,l}\ \ ,\ \ z(v)=0,1,\cdots,L_z-1\ \ ,\ \ \#=L_xL_yL_z\,.
	\end{equation}
	Within each $xy$ plane $[*,*,z]$ with $z<L_z$, and within each $zx$ plane $[*,y,*]$ (closed on the top by the 3-leg $B_{v,zx}^{\text{top}}$),
	\begin{gather}
		V_s\left(r^1(B,[*,*,z])\right)\overset{s}{=}\sum_{v\subset[*,*,z]}B_{v,xy}\ \ ,\ \ z=0,1,\cdots,L_z-1\ \ ,\ \ \#=L_z\,,\\
		V_s\left(r^1(B,[*,y,*])\right)\overset{s}{=}\sum_{v\subset[*,y,*]}B_{v,zx}\ \ ,\ \ y=0,1,\cdots,L_y-1\ \ ,\ \ \#=L_y\,.
	\end{gather}
	Compared with the $(mm)$ case, there is no $yz$ plane redundancy, since the top vertices do not host $B_{v,yz}$.
	Moreover, $B_{v,yz}$ appears only in the local redundancies $r^1(B,v)$, so the order-1 redundancies above are independent, and there is no second order redundancy.
	In all,
	\begin{equation}
		\log_2\dim\,\mathcal{H}_{\text{low}}=\#\text{spin}-(\# A_c+\# B_{v,l})+\# r^1 = L_xL_y(3L_z+1)-L_xL_y(4L_z+1)+(L_xL_yL_z+2L_z+L_y)=L_y+2L_z\,.
	\end{equation}

	Next, delete the $2L_xL_y$ $B_{v,zx},B_{v,yz}$ terms on the bottom boundary.
	This makes the $L_xL_y$ local redundancies $r^1(B,v)$ with $v\in\text{b.b.}$ and the $L_y$ $zx$ plane redundancies $r^1(B,[*,y,*])$ disappear, so $\log_2\dim\,\mathcal{H}_{\text{low}}$ changes from $L_y+2L_z$ to $L_y+2L_z+(2L_xL_y-L_xL_y-L_y)=L_xL_y+2L_z$.

	Finally, delete the top stabilizers marked in Fig.~\ref{fig_em_twists_2fig}, i.e., the $L_y$ terms $B_{[0,j,L_z],zx}^{\text{top}}$ and the $L_x$ terms $A_{[i+1/2,1/2,L_z-1/2]}^{\text{top}}$.
	Deleting the former makes no redundancy disappear, since the $zx$ plane redundancies have already disappeared and the top vertices host no local redundancy.
	Deleting the latter makes one redundancy, $r^1(A,[*,*,L_z-\frac12])$, disappear.
	So $\log_2\dim\,\mathcal{H}_{\text{low}}$ increases by $L_y+(L_x-1)$, i.e.
	\begin{equation}\label{eq_X_cube_FTH_em_GSD}
		\log_2\dim\,\mathcal{H}_{\text{low}}^{(em)}=L_xL_y+L_x+L_y+2L_z-1\,.
	\end{equation}
	Exchanging $x\leftrightarrow y$, the same result holds under the $(me)$ top boundary, $\log_2\dim\,\mathcal{H}_{\text{low}}^{(me)}=\log_2\dim\,\mathcal{H}_{\text{low}}^{(em)}$.

	Eq.~(\ref{eq_X_cube_FTH_em_GSD}) is written for the $(em)$ sandwich with $L_z$ layers of cubes.
	As for the $(ee)$ sandwich in Appendix~\ref{appendix_low_energy_subspace_dim_under_lineon_condensed_top}, to compare with the $(mm)$ sandwich on the $L_x\times L_y\times L_z$ lattice, we take the $(em)$ sandwich to have $L'_z\equiv L_z+1$ layers of cubes.
	Then the complete horizontal layers $z=0,1,\cdots,L_z$ are the same as those of the $(mm)$ sandwich, and the lattice is obtained from that of the $(ee)$ sandwich by adding the $x$-edges on its top layer $z=L'_z$.
	Replacing $L_z$ by $L'_z$ in Eq.~(\ref{eq_X_cube_FTH_em_GSD}),
	\begin{equation}\label{eq_X_cube_FTH_em_GSD_Lz1}
		\log_2\dim\,\mathcal{H}_{\text{low}}^{(em)}=L_xL_y+L_x+L_y+2L_z+1\,,
	\end{equation}
	which is larger by $1$ than Eq.~(\ref{eq::X-cube_FTH_smooth_GSD}) under the $(mm)$ and $(ee)$ top boundaries, i.e., $\dim\mathcal{H}_{\text{low}}^{(em)}$ is exactly twice $\dim\mathcal{H}_{\text{low}}^{(mm)}=\dim\mathcal{H}_{\text{low}}^{(ee)}$.
	This can be understood from the redundant logical operators.
	Both $\mathcal{H}_{\text{low}}^{(mm)}$ and $\mathcal{H}_{\text{low}}^{(ee)}$ consist of the identified 2d system with twist DOFs ($L_xL_y+L_x+L_y-2$ qubits) and $2(L_z+1)$ redundant qubits, generated by the redundant logical operators $\tilde X^{\text{red.}}_{x,y}(\tau),\tilde Z^{\text{red.}}_{x,y}(\tau)$ with $\tau=0,1,\cdots,L_z$ in Eqs.~(\ref{eq_X_cube_redundant_intrinsic_logical_smooth_top}) and (\ref{eq_X_cube_redundant_intrinsic_logical_rough_top}).
	Under the $(em)$ top boundary, Appendix~\ref{appendix_X_cube_details_em} shows that the $2L'_z=2(L_z+1)$ pairs with $\tau=0,1,\cdots,L_z$ remain redundant, exactly as under the $(mm)$ and $(ee)$ top boundaries.
	The top layer $z=L'_z$, which hosts only $x$-edges, contributes one more pair, which is not redundant: the single-row string $\prod_iX_{[i+1/2,0,L'_z]}$ on this layer is identified as the single-row symmetry $S_x^{(X)}(0)$ of the QCM, and its anti-commuting pair, the single-column string $\prod_{j}\tilde Z_{[1/2,j]}$ is another symmetry of the QCM, see Eq.~(\ref{eq_X_cube_em_global_pair}).
	The $y$-type pair has no counterpart on the top layer, since the top $y$-edges are truncated.
	So $\mathcal{H}_{\text{low}}^{(em)}$ consists of the $L_xL_y$ QCM spins, the $L_x+L_y-1$ twist DOFs, and $2(L_z+1)$ redundant qubits.
	The redundant qubits match those under the $(mm)$ and $(ee)$ top boundaries, the factor of two comes from the extra string pair in the $z=L_z+1$ layer.
	The pair $S_x^{(X)}(0)$ and $\prod_{j}\tilde Z_{[1/2,j]}$ anticommute with each other, but commute with all the QCM bonds, as well as all the twist indicators and twist togglers in Fig.~\ref{fig_em_twists_2fig}.
	So every energy level of the QCM is (at least) two-fold degenerate, with or without twist DOFs.
	This is why the QCM with twist DOFs ($L_xL_y+L_x+L_y-1$ qubits) has twice the Hilbert space dimension of the TFPIMs with twist DOFs ($L_xL_y+L_x+L_y-2$ qubits).
	The two match if removing this pair.

	\section{Details of operator identification and redundant DOFs of X-cube FTH}\label{appendix_X_cube_details}
	
	\subsection{Under \texorpdfstring{$(mm)$}{(mm)} top boundary}\label{appendix_X_cube_details_planeon}
	
	\paragraph{Twist logical operator and twist indicator. }The identification of twist logical operators are:
	\begin{equation}
		B_{[i,j,L_z],zx}\sim B_{[i,j,L_z],yz}\sim\tilde{Z}^{\text{twist}}(i,j)\ \ ,\ \ i\in\mathbb{Z}_{L_x},\ j\in\mathbb{Z}_{L_y},\ 0\in\{i,j\}\,.
	\end{equation}
	Here $B_{[i,j,L_z],zx}\sim B_{[i,j,L_z],yz}$ means differ by a multiplication of stabilizer, while $B_{[i,j,L_z],yz}\sim\tilde Z^{\text{twist}}(i,j)$ means the twist logical operator $B_{[i,j,L_z],yz}$ is identified  as the twist indicator $\tilde Z^{\text{twist}}(i,j)$.
	The twist logical operators are equivalent to the line-like subsystem symmetries of X-cube [e.g. $B_{[1,0,L_z],yz}\sim \mathcal Z_y(1)$, as shown in Fig.~\ref{fig_mm_ee_twists}(b)], except $B_{[0,0,L_z],zx}\sim B_{[0,0,L_z],yz}$.
	More precisely,
	\begin{gather}
		\tilde{Z}^{\text{twist}}(i,0) \sim B_{[i,0,L_z],yz} \sim \mathcal Z_y(i) \equiv \prod_{j=0}^{L_y-1} Z_{[i,j,1/2]}\ \ ,\ \ i=1,\cdots,L_x-1\,,\notag\\
		\tilde{Z}^{\text{twist}}(0,j) \sim B_{[0,j,L_z],zx} \sim \mathcal Z_x(j) \equiv \prod_{i=0}^{L_x-1} Z_{[i,j,1/2]}\ \ ,\ \ j=1,\cdots,L_y-1\,,\label{eq_X_cube_FTH_dual_Z_loop_terminology}
	\end{gather}
	where $\mathcal Z_{y(x)}(i)$ is the uncontractible dual $Z$ loop along $y(x)$-direction in the $x(y)=i,\ z=1/2$ section.
	By calculation,
	\begin{equation}\label{eq_X_cube_FTH_twist_00}
		\tilde Z^{\text{twist}}(0,0) = B_{[0,0,L_z],yz} \sim \mathcal Z_x(0)\prod_{i=1}^{L_x-1}\mathcal Z_y(i) \sim \mathcal Z_y(0)\prod_{j=1}^{L_y-1}\mathcal Z_x(j),
	\end{equation}
	which is a composite of uncontractible dual $Z$ loops.
	Combining Eq.~(\ref{eq_X_cube_FTH_twist_00}) with Eq.~(\ref{eq_X_cube_FTH_dual_Z_loop_terminology}) gives
	\begin{equation}\label{eq_X_cube_FTH_dual_Z_loop_boundary_index}
		\mathcal Z_x(0)\sim\prod_{i=0}^{L_x-1}\tilde Z^{\text{twist}}(i,0)\ \ ,\ \ \mathcal Z_y(0)\sim\prod_{j=0}^{L_y-1}\tilde Z^{\text{twist}}(0,j)\,;
	\end{equation}
	unlike all the other uncontractible dual $Z$ loops, $\mathcal Z_x(0)$ and $\mathcal Z_y(0)$ are identified as products of twist indicators rather than as single ones.

	\paragraph{$f$-transport operator along $z$-direction. }While a specific $f$-transport operator along $z$-direction is illustrated in Fig.~\ref{fig_mm_scheme2A}(a), a set of generators of pure $f$-transport operators along $z$-direction is as follows:
	\begin{gather}
		\mathcal T_{f,z}(i+1/2,j)\equiv\prod_{k=0}^{L_z}Z_{[i+1/2,j,k]}\sim\tilde{Z}_{[i+1/2,j-1/2]}\tilde{Z}_{[i+1/2,j+1/2]}\ \ ,\ \ i\in\mathbb{Z}_{L_x},\ j=1,\cdots,L_y-1\,,\label{eq_X_cube_pure_f_transport_x}\\
		\mathcal T_{f,z}(i,j+1/2)\equiv\prod_{k=0}^{L_z}Z_{[i,j+1/2,k]}\sim\tilde{Z}_{[i-1/2,j+1/2]}\tilde{Z}_{[i+1/2,j+1/2]}\ \ ,\ \ i=1,\cdots,L_x-1,\ j\in\mathbb{Z}_{L_y}\,.\label{eq_X_cube_pure_f_transport_y}
	\end{gather}
	Here $\mathcal T_{f,z}$ represents pure $f$-transport operator along $z$-direction, i.e., an $f$-transport operator along $z$-direction that contains no nontrivial logical operator; they are identified as neighboring pairs of $\tilde Z$, without ambiguity of logical operator.
	The values $j=0$ in Eq.~(\ref{eq_X_cube_pure_f_transport_x}) and $i=0$ in Eq.~(\ref{eq_X_cube_pure_f_transport_y}) are not included because the full families are not independent up to redundant logical operators: at $x=1/2$,
	\begin{equation}\label{eq_X_cube_f_transport_redundant_relation}
		\prod_{j=0}^{L_y-1}\mathcal T_{f,z}(1/2,j)=\prod_{\tau=0}^{L_z}\tilde{Z}^{\text{red.}}_x(\tau)\neq 1\,,
	\end{equation}
	and likewise $\prod_{i=0}^{L_x-1}\mathcal T_{f,z}(i,1/2)=\prod_{\tau=0}^{L_z}\tilde{Z}^{\text{red.}}_y(\tau)\neq1$, where $\tilde Z^{\text{red.}}_{x},\tilde Z^{\text{red.}}_{y}$ are given in Eq.~(\ref{eq_X_cube_redundant_intrinsic_logical_smooth_top}).
	\paragraph{Seam $f$-transport operators. }The $f$-transport operators on the seams, $\mathcal T_{f,z}(i+1/2,0)$ and $\mathcal T_{f,z}(0,j+1/2)$, are not included in Eqs.~(\ref{eq_X_cube_pure_f_transport_x}) and (\ref{eq_X_cube_pure_f_transport_y}).
	They are also low-energy-preserving, but they are not purely TFPIM operators.
	We derive the identification of $\mathcal T_{f,z}(i+1/2,0)$; that of $\mathcal T_{f,z}(0,j+1/2)$ follows by exchanging $x\leftrightarrow y$.
	First, multiply all the $f$-transport operators at the same $x=i+1/2$.
	Regrouping the $Z$'s layer by layer gives a $Z$-string along $y$ on each layer,
	\begin{equation*}
		\prod_{j=0}^{L_y-1}\mathcal T_{f,z}(i+1/2,j)=\prod_{k=0}^{L_z}C_k(i+1/2)\ \ ,\ \ C_k(i+1/2)\equiv\prod_{j=0}^{L_y-1}Z_{[i+1/2,j,k]}\,.
	\end{equation*}
	For $i=0$, each $C_k(1/2)$ is exactly $\tilde Z^{\text{red.}}_x(k)$ in Eq.~(\ref{eq_X_cube_redundant_intrinsic_logical_smooth_top}), so this product reduces to Eq.~(\ref{eq_X_cube_f_transport_redundant_relation}).
	For other $i$, two strings at neighboring $x$ on the same layer differ by a line of $B_{v,xy}$ stabilizers,
	\begin{equation*}
		C_k(i-1/2)\,C_k(i+1/2)=\prod_{j=0}^{L_y-1}B_{[i,j,k],xy}\,,
	\end{equation*}
	since each $y$-edge on the line $x=i$ belongs to two neighboring $B_{v,xy}$ and cancels.
	All the $B_{v,xy}$ terms, including those on the bottom and top layers, are stabilizers.
	Therefore, for any $i$,
	\begin{equation}\label{FTH_temp_1}
		\prod_{j=0}^{L_y-1}\mathcal T_{f,z}(i+1/2,j)\sim\prod_{\tau=0}^{L_z}\tilde Z^{\text{red.}}_x(\tau)\,.
	\end{equation}
	Second, the factors with $j\neq0$ are identified by Eq.~(\ref{eq_X_cube_pure_f_transport_x}), and their product equals to
	\begin{equation}\label{FTH_temp_2}
		\prod_{j=1}^{L_y-1}\mathcal T_{f,z}(i+1/2,j)\sim\prod_{j=1}^{L_y-1}\tilde Z_{[i+1/2,j-1/2]}\tilde Z_{[i+1/2,j+1/2]}=\tilde Z_{[i+1/2,1/2]}\tilde Z_{[i+1/2,L_y-1/2]}\,.
	\end{equation}
	Dividing the product in Eq.~(\ref{FTH_temp_1}) by the product in Eq.~(\ref{FTH_temp_2}) (and similarly when $x\leftrightarrow y$) gives the seam operators,
	\begin{gather}
		\mathcal T_{f,z}(i+1/2,0)\sim\tilde Z_{[i+1/2,L_y-1/2]}\tilde Z_{[i+1/2,1/2]}\cdot\prod_{\tau=0}^{L_z}\tilde Z^{\text{red.}}_x(\tau)\ \ ,\ \ i\in\mathbb Z_{L_x}\,,\notag\\
		\mathcal T_{f,z}(0,j+1/2)\sim\tilde Z_{[L_x-1/2,j+1/2]}\tilde Z_{[1/2,j+1/2]}\cdot\prod_{\tau=0}^{L_z}\tilde Z^{\text{red.}}_y(\tau)\ \ ,\ \ j\in\mathbb Z_{L_y}\,.\label{eq_X_cube_f_transport_seam}
	\end{gather}
	That is, $\mathcal T_{f,z}(i+1/2,0)$ is the neighboring pair across the seam $y=0$, times the redundant logical operators $\prod_\tau\tilde Z^{\text{red.}}_x(\tau)$.
	The redundant factor cannot be avoided: in the TFPIM, the neighboring pairs $\tilde Z\tilde Z$ around the periodic $y$-direction multiply to $1$, while on the preimage side, the whole family of $\mathcal T_{f,z}(i+1/2,j)$ multiplies to $\prod_\tau\tilde Z^{\text{red.}}_x(\tau)\neq1$.
	This is also why $\tilde X^{\text{red.}}_x(\tau)$ is redundant even though it anticommutes with $\mathcal T_{f,z}(i+1/2,0)$: the anticommutation comes only from the factor $\tilde Z^{\text{red.}}_x(\tau)$, and $\tilde X^{\text{red.}}_x(\tau)$ commutes with all the TFPIM operators.
	
	\paragraph{Truncated $B_{v,l}$ identification without logical ambiguity. }Multiplying the two generators on either side of a bottom vertex $v=[i,j,0]$ gives $\mathcal T_{f,z}(i,j-1/2)\,\mathcal T_{f,z}(i,j+1/2)=\prod_{k=0}^{L_z}B_{[i,j,k],yz}$, since the $z$-edges cancel in pairs, and likewise for $zx$; this is the bottom truncated $B_{[i,j,0],l}$ times bulk stabilizers times the top truncated $B_{[i,j,L_z],l}$, which is a stabilizer unless $0\in\{i,j\}$, where it is the deleted term $\tilde Z^{\text{twist}}(i,j)$ of Eqs.~(\ref{eq_X_cube_FTH_dual_Z_loop_terminology}) and~(\ref{eq_X_cube_FTH_twist_00}).
	Hence the truncated $B_{v,zx},B_{v,yz}$ terms on the bottom boundary are identified without ambiguity of logical operator as
	\begin{gather}
		B_{[i,j,0],zx}\sim B_{[i,j,0],yz}\sim \prod_{p\revsubset [i,j]}\tilde{Z}_p\ \ ,\ \ i\in\mathbb{Z}_{L_x},\ j\in\mathbb{Z}_{L_y},\ 0\notin\{i,j\},\\
		B_{[i,j,0],zx}\sim B_{[i,j,0],yz}\sim \prod_{p\revsubset [i,j]}\tilde{Z}_p\cdot\tilde{Z}^{\text{twist}}(i,j)\ \ ,\ \ i\in\mathbb{Z}_{L_x},\ j\in\mathbb{Z}_{L_y},\ 0\in\{i,j\}.\label{eq_X_cube_bottom_B_twist_identification}
	\end{gather}

	\paragraph{$l_z$-transport operator and twist toggler. } 
	We illustrate the $l_z$-transport operator at $x=0,y=1$ as an example, which is identified as $\tilde{X}^{\text{twist}}(0,1)$:
	

\begin{equation}\label{eq::X-cube_FTH_smooth_6}
		\begin{tikzpicture}[baseline=4ex]
			\foreach \j in {0,1,2}{\foreach \k in {0,1}{\draw[gray](0,\j,\k)--(3,\j,\k);}}
			\foreach \i in {0,1,2}{\foreach \j in {0,1,2}{\draw[gray](\i,\j,0)--(\i,\j,2);}}
			\foreach \i in {0,1,2}{\foreach \k in {0,1}{\draw[gray](\i,0,\k)--(\i,2,\k);}}
			
			\foreach \i in {0,1,2,3}{\draw[dashed, gray] (\i,0,2)--(\i,2,2);}
			\foreach \j in {0,1,2}{\draw[dashed, gray] (0,\j,2)--(3,\j,2)   (3,\j,0)--(3,\j,2);}
			\foreach \k in {0,1}{\draw[dashed, gray] (3,0,\k)--(3,2,\k);}
			
			\foreach \i/\k in {0/0, 1/0, 2/0, 0/1}{
				\draw[blue, line width=1pt] (\i-0.2,2,\k) arc [start angle=180, end angle=360, radius=0.2];
				\begin{scope}[canvas is yz plane at x=\i]
					\draw[blue, line width=1pt] (2,\k+0.2) arc [start angle=90, end angle=270, radius=0.2];
				\end{scope}
			}
			
			\foreach \i/\k in {1/0}{
				\fill[red, opacity=0.1] (\i-0.2,0,\k)--(\i+0.2,0,\k)--(\i+0.2,2,\k)--(\i-0.2,2,\k);
				\foreach \j in {0.5, 1.5}{\node[red] at (\i,\j,\k) {$X$};}
			}
			\node[red] at (1.3,-0.3,0) {$\tilde{X}^{\text{twist}}(0,1)$};
			
			\draw[->, blue, line width=1pt] (-2-0.7,0,0)--(-2-0.7,0,1);
			\draw[->, blue, line width=1pt] (-2-0.7,0,0)--(-1.2-0.7,0,0);
			\draw[->, blue, line width=1pt] (-2-0.7,0,0)--(-2-0.7,0.8,0);
			\node[blue] at (-2.18-0.7,0,1) {$x$};
			\node[blue] at (-1.2-0.7,-0.22,0) {$y$};
			\node[blue] at (-1.75-0.7,0.8,0) {$z$};
		\end{tikzpicture}\ .
	\end{equation}
	$\mathcal T_{l_z,z}(0,1)$ creates an $l_z$-lineon defect at $[i,j,L_z]$ on the top boundary, and add a twist at $x=0,y=1$ to $\tilde H$ simultaneously.

	\paragraph{Redundant logical operators. } The $2(L_z+1)$ pairs of redundant logical operators of X-cube FTH are as follows:
	\begin{align}\label{eq_X_cube_redundant_intrinsic_logical_smooth_top}
		&\tilde{X}^{\text{red.}}_x(\tau)\equiv\prod_{i=0}^{L_x-1}X_{[i+1/2,0,\tau]}\ \ ,\ \ \tilde{Z}^{\text{red.}}_x(\tau)\equiv \prod_{j=0}^{L_y-1}Z_{[1/2,j,\tau]}\,,\notag\\
		&\tilde{X}^{\text{red.}}_y(\tau)\equiv\prod_{j=0}^{L_y-1}X_{[0,j+1/2,\tau]}\ \ ,\ \ \tilde{Z}^{\text{red.}}_y(\tau)\equiv \prod_{i=0}^{L_x-1}Z_{[i,1/2,\tau]}\,,
	\end{align} 
	where $\tau=0,1,\cdots,L_z$.
	As an example, we illustrate the anti-commuting pair $\left(\tilde{X}^{\text{red.}}_y(1),\tilde{Z}^{\text{red.}}_y(1)\right)$:
	\begin{equation}
		\begin{tikzpicture}[baseline=4ex]
			\foreach \j in {0,1,2}{\foreach \k in {0,1}{\draw[gray](0,\j,\k)--(3,\j,\k);}}
			\foreach \i in {0,1,2}{\foreach \j in {0,1,2}{\draw[gray](\i,\j,0)--(\i,\j,2);}}
			\foreach \i in {0,1,2}{\foreach \k in {0,1}{\draw[gray](\i,0,\k)--(\i,2,\k);}}
			
			\foreach \i in {0,1,2,3}{\draw[dashed, gray] (\i,0,2)--(\i,2,2);}
			\foreach \j in {0,1,2}{\draw[dashed, gray] (0,\j,2)--(3,\j,2)   (3,\j,0)--(3,\j,2);}
			\foreach \k in {0,1}{\draw[dashed, gray] (3,0,\k)--(3,2,\k);}
			
			\foreach \j in {1}{
				\fill[red, opacity=0.1] (0,\j-0.2,0)--(0,\j+0.2,0)--(3,\j+0.2,0)--(3,\j-0.2,0);
				\foreach \i in {1.5,2.5}{
					\node[red] at (\i,\j,0) {$X$};
				}
			}
			\node[red] at (0.6,1,0) {$X$};
			\node[red] at (3.8,1,0) {$\equiv\tilde{X}^{\text{red.}}_y(1)$};
			
			\node[ForestGreen] at (0.4,1,0) {$Z$};
			\node[ForestGreen] at (0.5,1,1) {$Z$};
			\fill[ForestGreen, opacity=0.1] (0.3,1,-0.4)--(0.7,1,-0.4)--(0.7,1,2)--(0.3,1,2);
			\node[ForestGreen] at (0.5,1,2.4) {$\equiv\tilde{Z}^{\text{red.}}_y(1)$};
			
			\draw[->, blue, line width=1pt] (-2-0.7,0,0)--(-2-0.7,0,1);
			\draw[->, blue, line width=1pt] (-2-0.7,0,0)--(-1.2-0.7,0,0);
			\draw[->, blue, line width=1pt] (-2-0.7,0,0)--(-2-0.7,0.8,0);
			\node[blue] at (-2.18-0.7,0,1) {$x$};
			\node[blue] at (-1.2-0.7,-0.22,0) {$y$};
			\node[blue] at (-1.75-0.7,0.8,0) {$z$};
		\end{tikzpicture}\ .
	\end{equation}

	\subsection{Under \texorpdfstring{$(ee)$}{(ee)} top boundary}\label{appendix_X_cube_details_lineon}
	
	As in Appendix~\ref{appendix_low_energy_subspace_dim_under_lineon_condensed_top}, take the $(ee)$ sandwich to have $L_z+1$ layers.
	Under the $(ee)$ top boundary, the bottom-boundary truncated $A_c$ terms are identified without logical ambiguity as follows:
	\begin{gather}
		\prod_{e\subset[i+1/2,j+1/2,0]}X_e \sim \prod_{v\subset[i+1/2,j+1/2]}\tilde Z_v,\quad i\in\mathbb Z_{L_x},\ j\in\mathbb Z_{L_y},\ i\neq0,\ j\neq0,\\
		\prod_{e\subset[i+1/2,j+1/2,0]}X_e\sim\prod_{v\subset[i+1/2,j+1/2]}\tilde{Z}_v\cdot \tilde{Z}^{\text{twist}}(i+1/2,j+1/2)\ \ ,\ \ i\in\mathbb{Z}_{L_x},\ j\in\mathbb{Z}_{L_y},\ 0\in\{i,j\}\,,
	\end{gather}
	where $\tilde Z^{\text{twist}}(i+1/2,j+1/2)\equiv A_{[i+1/2,j+1/2,L_z+1/2]}^{\text{top}}$ are the deleted top cube stabilizers in Fig.~\ref{fig_mm_ee_twists}(a).
	We illustrate $\tilde{Z}^{\text{twist}}(3/2,1/2)=A_{[3/2,1/2,L_z+1/2]}\sim\mathcal{X}_y(1)\mathcal{X}_y(2)$ as an example below ($c=[3/2,1/2,L_z+1/2]$ in the following figure),
	\begin{equation}
		\begin{tikzpicture}[baseline=4ex]
			\foreach \j in {0,1,2}{\foreach \k in {0,1}{\draw[gray](0,\j,\k)--(3,\j,\k);}}
			\foreach \i in {0,1,2}{\foreach \j in {0,1,2}{\draw[gray](\i,\j,0)--(\i,\j,2);}}
			\foreach \i in {0,1,2}{\foreach \k in {0,1}{\draw[gray](\i,0,\k)--(\i,3,\k);}}
			
			\foreach \i in {0,1,2,3}{\draw[dashed, gray] (\i,0,2)--(\i,3,2);}
			\foreach \j in {0,1,2}{\draw[dashed, gray] (0,\j,2)--(3,\j,2)   (3,\j,0)--(3,\j,2);}
			\foreach \k in {0,1}{\draw[dashed, gray] (3,0,\k)--(3,3,\k);}
			
			\fill[red, opacity=0.1] (0,0,0.6)--(0,0,1.4)--(3,0,1.4)--(3,0,0.6);
			\foreach \i in {0.5, 1.5, 2.5}{\node[red] at (\i,0,1) {$X$};}
			\fill[red, opacity=0.1] (0,0,1.6)--(0,0,2.4)--(3,0,2.4)--(3,0,1.6);
			\foreach \i in {0.5, 1.5, 2.5}{\node[red] at (\i,0,2) {$X$};}
			\node[red] at (4.9,0,1.5) {$\sim A_c=\tilde{Z}^{\text{twist}}(3/2,1/2)$};
			
			\fill[red, opacity=0.1] (0,3,1)--(1,3,1)--(1,2,1)--(1,2,2)--(0,2,2)--(0,3,2);
			\node[red] at (0.5,2.5,1.5) {$c$};
			
			\fill[blue] (0,0,0) circle[radius=1.5pt];
			
			\draw[->, blue, line width=1pt] (-2-0.7,0,0)--(-2-0.7,0,1);
			\draw[->, blue, line width=1pt] (-2-0.7,0,0)--(-1.2-0.7,0,0);
			\draw[->, blue, line width=1pt] (-2-0.7,0,0)--(-2-0.7,0.8,0);
			\node[blue] at (-2.18-0.7,0,1) {$x$};
			\node[blue] at (-1.2-0.7,-0.22,0) {$y$};
			\node[blue] at (-1.75-0.7,0.8,0) {$z$};
		\end{tikzpicture}
	\end{equation}
	Then, we identify the $l_z$-transport operators along $z$-direction as single $\tilde Z$ without logical operator ambiguity:
	\begin{equation}
		\mathcal T_{l_z,z}(i,j)=\prod_{k=0}^{L_z}X_{[i,j,k+1/2]}\sim\tilde Z_{[i,j]}\,,
	\end{equation}
	where $i\in\mathbb Z_{L_x},\ j\in\mathbb Z_{L_y}$.
	Finally, the redundant DOFs under the $(ee)$ top boundary are as follows.
	There are $L_x+L_y-2$ independent pairs of $\tilde X^{\text{twist}},\tilde Z^{\text{twist}}$.
	They generate the operator algebra of $L_x+L_y-2$ twist DOFs' Hilbert space.
	The operator algebra of low-energy subspace $\mathcal{H}_{\text{low}}^{(ee)}$ has $2(L_z+1)$ more pairs of independent generators, which are redundant without physical meaning in the identified 2d boundary system, these generators are
	\begin{align}\label{eq_X_cube_redundant_intrinsic_logical_rough_top}
		&\tilde{X}^{\text{red.}}_x(\tau)\equiv\prod_{i=0}^{L_x-1}X_{[i+1/2,0,\tau]}\ \ ,\ \ \tilde{Z}^{\text{red.}}_x(\tau)\equiv \prod_{j=0}^{L_y-1}Z_{[1/2,j,\tau]}\,,\notag\\
		&\tilde{X}^{\text{red.}}_y(\tau)\equiv\prod_{j=0}^{L_y-1}X_{[0,j+1/2,\tau]}\ \ ,\ \ \tilde{Z}^{\text{red.}}_y(\tau)\equiv \prod_{i=0}^{L_x-1}Z_{[i,1/2,\tau]}\,,
	\end{align} 
	where $\tau=0,1,\cdots,L_z$.
	As an example, we illustrate the anti-commuting pair $\left(\tilde{X}^{\text{red.}}_y(1),\tilde{Z}^{\text{red.}}_y(1)\right)$:
	\begin{equation}
		\begin{tikzpicture}[baseline=4ex]
			\foreach \j in {0,1,2}{\foreach \k in {0,1}{\draw[gray](0,\j,\k)--(3,\j,\k);}}
			\foreach \i in {0,1,2}{\foreach \j in {0,1,2}{\draw[gray](\i,\j,0)--(\i,\j,2);}}
			\foreach \i in {0,1,2}{\foreach \k in {0,1}{\draw[gray](\i,0,\k)--(\i,3,\k);}}
			
			\foreach \i in {0,1,2,3}{\draw[dashed, gray] (\i,0,2)--(\i,3,2);}
			\foreach \j in {0,1,2}{\draw[dashed, gray] (0,\j,2)--(3,\j,2)   (3,\j,0)--(3,\j,2);}
			\foreach \k in {0,1}{\draw[dashed, gray] (3,0,\k)--(3,3,\k);}
			
			\foreach \j in {1}{
				\fill[red, opacity=0.1] (0,\j-0.2,0)--(0,\j+0.2,0)--(3,\j+0.2,0)--(3,\j-0.2,0);
				\foreach \i in {1.5,2.5}{
					\node[red] at (\i,\j,0) {$X$};
				}
			}
			\node[red] at (0.6,1,0) {$X$};
			\node[red] at (3.8,1,0) {$\equiv\tilde{X}^{\text{red.}}_y(1)$};
			
			\node[ForestGreen] at (0.4,1,0) {$Z$};
			\node[ForestGreen] at (0.5,1,1) {$Z$};
			\fill[ForestGreen, opacity=0.1] (0.3,1,-0.4)--(0.7,1,-0.4)--(0.7,1,2)--(0.3,1,2);
			\node[ForestGreen] at (0.5,1,2.4) {$\equiv\tilde{Z}^{\text{red.}}_y(1)$};
			
			\draw[->, blue, line width=1pt] (-2-0.7,0,0)--(-2-0.7,0,1);
			\draw[->, blue, line width=1pt] (-2-0.7,0,0)--(-1.2-0.7,0,0);
			\draw[->, blue, line width=1pt] (-2-0.7,0,0)--(-2-0.7,0.8,0);
			\node[blue] at (-2.18-0.7,0,1) {$x$};
			\node[blue] at (-1.2-0.7,-0.22,0) {$y$};
			\node[blue] at (-1.75-0.7,0.8,0) {$z$};
		\end{tikzpicture}\ .
	\end{equation}

	\subsection{\texorpdfstring{Under $(em)$ top boundary}{Under (em) top boundary}}\label{appendix_X_cube_details_em}
	As in Appendix~\ref{app_X_cube_em_me_counting}, the $(em)$ sandwich here has $L_z+1$ layers of vertices: the $z=0,1,\cdots,L_z$ layers are complete, while the top layer $z=L_z+1$ hosts only the $x$-edges.
	The operator identification and the twist DOFs under the $(em)$ top boundary are shown in Figs.~\ref{fig_em_scheme2A}, \ref{fig_em_symmetries_2fig}, and \ref{fig_em_twists_2fig}.
	Explicitly, the QCM spins on the bottom $x$-edges are identified as
	\begin{align}\label{eq_X_cube_em_QCM_identification}
		\mathcal T^{\sqcap}_{l_z,z}(i+\tfrac12,j)&\equiv X_{[i+1/2,j,L_z+1]}\prod_{k=0}^{L_z}X_{[i,j,k+1/2]}X_{[i+1,j,k+1/2]}\sim\tilde X_{[i+1/2,j]}\,,\notag\\
		\mathcal T_{f,z}(i+\tfrac12,j)&\equiv\prod_{k=0}^{L_z+1}Z_{[i+1/2,j,k]}\sim\tilde Z_{[i+1/2,j]}\,,
	\end{align}
	where $i\in\mathbb Z_{L_x}$, $j\in\mathbb Z_{L_y}$.
	In this subsection, we list the redundant DOFs and the anticommuting pair of subsystem symmetry operators of the QCM.

	\paragraph{Redundant logical operators. } Denote the $X$-string along the $y=0$ row of $x$-edges on the $\tau$-th layer by
	\begin{equation}\label{eq_X_cube_row_string_Rx}
		R_x(\tau):=\prod_{i=0}^{L_x-1}X_{[i+1/2,0,\tau]}\,,
	\end{equation}
	so that $\tilde X^{\text{red.}}_x(\tau)=R_x(\tau)$ under the $(mm)$ and $(ee)$ top boundaries [Eq.~(\ref{eq_X_cube_redundant_intrinsic_logical_smooth_top})].
	The $2(L_z+1)$ pairs of redundant logical operators under the $(em)$ top boundary are
	\begin{align}\label{eq_X_cube_redundant_em_top}
		&\tilde{X}^{\text{red.}}_x(\tau)=R_x(\tau)R_x(L_z+1)\ \ ,\ \ \tilde{Z}^{\text{red.}}_x(\tau)= \prod_{j=0}^{L_y-1}Z_{[1/2,j,\tau]}\,,\notag\\
		&\tilde{X}^{\text{red.}}_y(\tau)=\prod_{j=0}^{L_y-1}X_{[0,j+1/2,\tau]}\ \ ,\ \ \tilde{Z}^{\text{red.}}_y(\tau)= \prod_{i=0}^{L_x-1}Z_{[i,1/2,\tau]}\,,
	\end{align}
	where $\tau=0,1,\cdots,L_z$, the same range as under the $(mm)$ and $(ee)$ top boundaries.
	Compared with Eq.~(\ref{eq_X_cube_redundant_intrinsic_logical_smooth_top}), only $\tilde X^{\text{red.}}_x(\tau)$ is modified, by the top-layer factor $R_x(L_z+1)$.
	This factor is necessary: the single row $R_x(\tau)$ alone anticommutes with $\tilde Z_{[1/2,0]}\sim\mathcal T_{f,z}(\frac12,0)$, since they overlap only at the edge $[1/2,0,\tau]$, while $R_x(\tau)R_x(L_z+1)$ overlaps with $\mathcal T_{f,z}(\frac12,0)$ at two edges.
	This is illustrated in Fig.~\ref{fig_X_cube_em_red_vs_mm}.
	\begin{figure}[!htbp]
		\centering
		\begin{tikzpicture}[scale=1.06]
			% standard coordinates: 1st = y, 2nd = z, 3rd = x (matching Figs. 7-11 and Eqs. D10-D18)
			% (a) (mm) top boundary, L_z=2
			\drawsmoothtoplatticegray
			% redundant X-row at y=0, z=1
			\fill[red, opacity=0.14] (-0.16,1,0)--(0.16,1,0)--(0.16,1,2)--(-0.16,1,2)--cycle;
			\foreach \k in {0.53, 1.53}{\node[red] at (0,1,\k) {$X$};}
			\node[red, anchor=east] at (-0.12,1,2) {$R_x(1)\equiv$};
			% pure f-transport at y=1, x=1/2
			\fill[ForestGreen, opacity=0.14] (1,-0.22,0.24)--(1,-0.22,0.76)--(1,2.22,0.76)--(1,2.22,0.24)--cycle;
			\foreach \j in {0,1,2}{\node[ForestGreen] at (1,\j,0.53) {$Z$};}
			\node[ForestGreen] at (1.5,-0.80,1.15) {$\mathcal{T}_{f,z}(\tfrac12,1)\sim\tilde{Z}_{[\frac12,\frac12]}\tilde{Z}_{[\frac12,\frac32]}$};
			\fill[blue] (0,0,0) circle[radius=1.5pt];
			\begin{scope}[xshift=-0.80cm]
				\drawaxes
			\end{scope}
			\node[black] at (1.5,-1.30,1.15) {(a) $(mm)$};
		\end{tikzpicture}
		\qquad\quad
		\begin{tikzpicture}[scale=1.06]
			% (b) (em) top boundary, L_z+1=3 layers of cubes; the top layer z=3 hosts only x-edges
			\foreach \j in {0,1,2}{\foreach \k in {0,1}{\draw[gray](0,\j,\k)--(3,\j,\k);}}
			\foreach \i in {0,1,2}{\foreach \j in {0,1,2,3}{\draw[gray](\i,\j,0)--(\i,\j,2);}}
			\foreach \i in {0,1,2}{\foreach \k in {0,1}{\draw[gray](\i,0,\k)--(\i,3,\k);}}
			\foreach \i in {0,1,2,3}{\draw[dashed, gray] (\i,0,2)--(\i,3,2);}
			\foreach \j in {0,1,2}{\draw[dashed, gray] (0,\j,2)--(3,\j,2);}
			\foreach \j in {0,1,2,3}{\draw[dashed, gray] (3,\j,0)--(3,\j,2);}
			\foreach \k in {0,1}{\draw[dashed, gray] (3,0,\k)--(3,3,\k);}
			% f-transport at y=0, x=1/2, identified as a single QCM Z
			\fill[ForestGreen, opacity=0.14] (0,-0.22,0.24)--(0,-0.22,0.76)--(0,3.22,0.76)--(0,3.22,0.24)--cycle;
			\foreach \j in {0,2}{\node[ForestGreen] at (0,\j,0.53) {$Z$};}
			% X-row at z=1 and the top-layer X-row at z=3, sharing edges [1/2,0,1] and [1/2,0,3] with T_{f,z}(1/2,0)
			\foreach \j in {1,3}{
				\fill[red, opacity=0.14] (-0.16,\j,0)--(0.16,\j,0)--(0.16,\j,2)--(-0.16,\j,2)--cycle;
				\node[ForestGreen] at ([xshift=-3.6pt] 0,\j,0.56) {$Z$};
				\node[red] at ([xshift=3.6pt] 0,\j,0.56) {$X$};
				\node[red] at (0,\j,1.55) {$X$};
			}
			\node[red, anchor=east] at (-0.12,1,2) {$R_x(1)\equiv$};
			\node[red, anchor=east] at (-0.12,3,2) {$S_x^{(X)}(0)\sim R_x(3)\equiv$};
			\node[ForestGreen] at (1.5,-0.80,1.15) {$\mathcal{T}_{f,z}(\tfrac12,0)\sim\tilde{Z}_{[\frac12,0]}$};
			\fill[blue] (0,0,0) circle[radius=1.5pt];
			\node[black] at (1.5,-1.30,1.15) {(b) $(em)$};
		\end{tikzpicture}
		\caption{The row string $R_x(\tau)$ (red) and redundant logical operator $\tilde X_x^{\text{red.}}$ under the $(mm)$ and $(em)$ top boundaries, for $\tau=1$ and $L_z=2$.
		(a) Under the $(mm)$ top boundary, $R_x(1)$ on the $y=0$ row commutes with the pure $f$-transport operators $\mathcal T_{f,z}(i+\frac12,j)$ with $j\neq0$, e.g., $\mathcal T_{f,z}(\frac12,1)$ (green), which are identified as neighboring pairs $\tilde Z_{[i+1/2,j-1/2]}\tilde Z_{[i+1/2,j+1/2]}$ [see Eq.~(\ref{eq_X_cube_pure_f_transport_x})].
		It anticommutes with $\mathcal T_{f,z}(i+\frac12,0)$ (not shown), but by Eq.~(\ref{eq_X_cube_f_transport_seam}), $\mathcal T_{f,z}(i+\frac12,0)$ is the neighboring pair across $y=0$ times $\prod_\tau\tilde Z^{\text{red.}}_x(\tau)$, and the anticommutation only comes from the factor $\tilde Z^{\text{red.}}_x(1)$, the anticommuting partner of $R_x(1)$.
		So $R_x(1)$ commutes with all the TFPIM operators and is redundant, $\tilde X^{\text{red.}}_x(1)=R_x(1)$.
		(b) Under the $(em)$ top boundary, with $L_z+1=3$ layers of cubes and the top layer $z=3$ hosting only the $x$-edges, the corresponding product $\prod_j\tilde Z_{[1/2,j]}$ is a nontrivial QCM operator, and every $\mathcal T_{f,z}(i+\frac12,j)$, including $j=0$, is identified as a single QCM operator $\tilde Z_{[i+1/2,j]}$, without redundant factor.
		Then $R_x(1)$ shares exactly one edge with each $\mathcal T_{f,z}(i+\frac12,0)$, e.g., the edge $[1/2,0,1]$ with $\mathcal T_{f,z}(\frac12,0)$, i.e., it acts on the QCM in the same way as $S_x^{(X)}(0)=\prod_i\tilde X_{[i+1/2,0]}$.
		Multiplying by the top-layer row $R_x(3)\sim S_x^{(X)}(0)$, which shares the edge $[1/2,0,3]$ with $\mathcal T_{f,z}(\frac12,0)$, makes the operator commuting with $\mathcal T_{f,z}(\frac12,0)$, so that $\tilde X^{\text{red.}}_x(1)=R_x(1)R_x(3)$ is redundant.}
		\label{fig_X_cube_em_red_vs_mm}
	\end{figure}

	The operators in Eq.~(\ref{eq_X_cube_redundant_em_top}) commute with all the QCM operators in Eq.~(\ref{eq_X_cube_em_QCM_identification}), all the twist indicators, and the twist togglers $\mathcal T_{l_z,z}(0,j)$ ($j\in\mathbb Z_{L_y}$) and $\mathcal T_{f,z}(i,\frac12)$ ($i=1,\cdots,L_x-1$) in Fig.~\ref{fig_em_twists_2fig}.
	As shown in Eq.~(\ref{eq_X_cube_f_transport_redundant_relation}), $\mathcal T_{f,z}(0,\frac12)$ is not included because
	\begin{equation}
		\prod_{i=0}^{L_x-1}\mathcal T_{f,z}(i,\tfrac12)=\prod_{\tau=0}^{L_z}\tilde Z^{\text{red.}}_y(\tau)\,.
	\end{equation}

	\paragraph{Anticommuting pair of subsystem symmetry operators. } The top layer $z=L_z+1$, which hosts only the $x$-edges, carries one more pair of logical operators.
	They are not redundant, but a pair of genuine subsystem (line-like) symmetry operators of the QCM, the single-row symmetry indicator $S_x^{(X)}(0)$ and the single-column string $\prod_j\tilde Z_{[1/2,j]}$,
	\begin{equation}\label{eq_X_cube_em_global_pair}
		S_x^{(X)}(0)=\prod_{i=0}^{L_x-1}\tilde X_{[i+1/2,0]}\sim R_x(L_z+1)\ \ ,\ \ \prod_{j=0}^{L_y-1}\tilde Z_{[1/2,j]}\sim\prod_{k=0}^{L_z+1}\prod_{j=0}^{L_y-1}Z_{[1/2,j,k]}\,.
	\end{equation}
	Here the first identification follows from $\prod_{i}\mathcal T^{\sqcap}_{l_z,z}(i+\frac12,0)=R_x(L_z+1)$, since the $z$-edges cancel in pairs, while the second one is $\prod_j\mathcal T_{f,z}(\frac12,j)$.
	There is no counterpart under $x\leftrightarrow y$, since the top $y$-edges are truncated.
	The two operators in Eq.~(\ref{eq_X_cube_em_global_pair}) anticommute with each other, but commute with all the QCM bonds, the twist indicators, and the twist togglers.
	Hence this pair carries a qubit on which the bottom Hamiltonian acts trivially, so every energy level of the QCM is (at least) two-fold degenerate.

	In all, the $L_xL_y$ QCM spins, the $L_x+L_y-1$ twist DOFs, and the $2(L_z+1)$ redundant qubits generated by Eq.~(\ref{eq_X_cube_redundant_em_top}) account for all the $L_xL_y+L_x+L_y+2L_z+1$ qubits in Eq.~(\ref{eq_X_cube_FTH_em_GSD_Lz1}).
	Under the $(me)$ top boundary, the same holds after exchanging $x\leftrightarrow y$.

	\section{\texorpdfstring{Step-by-step explanation of the SQC switching the $(ee)$ top boundary to the $(mm)$ top boundary}{Step-by-step explanation of the SQC switching the (ee) top boundary to the (mm) top boundary}}\label{app_X_cube_SQC}

	The blue colored cubes in Eq.~(\ref{X_cube_rough_to_smooth_top_SQC}) host the deleted $A_c$ terms, as shown in Fig.~\ref{fig_mm_ee_twists}(a).
	Each arrow represents a CNOT gate with the tail being control qubit and head being target qubit.
	The arrows/CNOT gates with the same color are commuting, thus form a single layer of the SQC\footnote{The usual convention of a layer in the quantum circuit context is that local unitary gates do not have intersecting DOFs.
	Here we adopt another convention where local unitary gates within the same layer could intersect, but must be commuting.
	The depth of the circuit in these two conventions differs at most by a constant factor, as long as the circuit is an SQC.}.
	The SQC is LU and linear depth, with $L_x+L_y-2$ layers in total.
	The CNOT gate $\text{CNOT}_{c,t}$ maps the Pauli $X,Z$ of control $c$ and target $t$ as follows:
	\begin{align}
		& \text{CNOT}_{c,t}X_c\text{CNOT}_{c,t}^\dagger = X_cX_t\notag \\
		& \text{CNOT}_{c,t}Z_c\text{CNOT}_{c,t}^\dagger = Z_c\notag    \\
		& \text{CNOT}_{c,t}X_t\text{CNOT}_{c,t}^\dagger = X_t\notag    \\
		& \text{CNOT}_{c,t}Z_t\text{CNOT}_{c,t}^\dagger = Z_cZ_t\,.
	\end{align}
	The first $L_x+L_y-3$ layers map the undeleted $A_c$ terms on the top boundary to a single Pauli $X$ as follows:
	\begin{equation}
		\begin{tikzpicture}[>={Triangle[round,length=1.5mm,width=1mm]}]
			\draw[line width=1pt] (0,0,0)--(1,0,0)--(1,0,1)--(0,0,1)--cycle;
			\foreach \i in {0,1}{\foreach \k in {0,1}{
					\draw[line width=1pt] (\i,0,\k)--(\i,1,\k);
			}}
			
			\foreach \i/\j/\k in {0.5/0/0, 1/0/0.5, 0.5/0/1, 0/0/0.5, 0/0.5/0, 1/0.5/0, 1/0.5/1, 0/0.5/1}{
				\node[red] at (\i,\j,\k) {$X$};
			}

			\draw[line width=1pt, ->] (2,0.5,0.5)--(3,0.5,0.5);
			
			\draw (0+2.15,0+0.75,0)--(0.5+2.15,0+0.75,0)--(0.5+2.15,0+0.75,0.5)--(0+2.15,0+0.75,0.5)--cycle;
			\foreach \i in {0+2.15,0.5+2.15}{\foreach \k in {0,0.5}{
					\draw (\i,0+0.75,\k)--(\i,0.5+0.75,\k);
			}}
			\draw[blue, line width=0.5pt, ->] (0+2.15,0.25+0.75,0)--(0.5+2.15,0.25+0.75,0);
			\draw[blue, line width=0.5pt, ->] (0+2.15,0.25+0.75,0)--(0.5+2.15,0.25+0.75,0.5);
			\draw[blue, line width=0.5pt, ->] (0+2.15,0.25+0.75,0)--(0+2.15,0.25+0.75,0.5);
			\draw[blue, line width=0.5pt, ->] (0+2.15,0.25+0.75,0)--(0.25+2.15,0+0.75,0);
			\draw[blue, line width=0.5pt, ->] (0+2.15,0.25+0.75,0)--(0.5+2.15,0+0.75,0.25);
			\draw[blue, line width=0.5pt, ->] (0+2.15,0.25+0.75,0)--(0.25+2.15,0+0.75,0.5);
			\draw[blue, line width=0.5pt, ->] (0+2.15,0.25+0.75,0)--(0+2.15,0+0.75,0.25);

			\draw[line width=1pt] (0+4,0,0)--(1+4,0,0)--(1+4,0,1)--(0+4,0,1)--cycle;
			\foreach \i in {0+4,1+4}{\foreach \k in {0,1}{
					\draw[line width=1pt] (\i,0,\k)--(\i,1,\k);
			}}
			
			\node[red] at (0+4,0.5,0) {$X$};
		\end{tikzpicture}
	\end{equation}
	At the same time, each layer of the first $L_x+L_y-3$ layers maps the two $B_{v,yz},B_{v,zx}$ terms as follows (while keeping other $B_{v,l}$ terms invariant):
	\begin{equation}
		\begin{tikzpicture}[>={Triangle[round,length=1.5mm,width=1mm]}]
			%左一图
			\draw[line width=1pt] (0,0,0)--(1,0,0)--(1,0,1)--(0,0,1)--cycle;
			\foreach \i in {0,1}{\foreach \k in {0,1}{
					\draw[line width=1pt] (\i,0,\k)--(\i,1,\k);
			}}
			\draw[line width=1pt]
			(-1,0,0)--(0,0,0)
			(0,-1,0)--(0,0,0)
			(0,0,-1)--(0,0,0);
			
			\foreach \i/\j/\k in {0.5/0/0, 0/0.5/0, -0.5/0/0, 0/-0.5/0}{
				\node[ForestGreen] at (\i,\j,\k) {$Z$};
			}
			\fill[ForestGreen, opacity=0.2] (-0.875,0.2,0)--(0.875,0.2,0)--(0.875,-0.2,0)--(-0.875,-0.2,0);
			\fill[ForestGreen, opacity=0.2] (0.2,0.875,0)--(0.2,-0.875,0)--(-0.2,-0.875,0)--(-0.2,0.875,0);

			%箭头
			\draw[line width=1pt, ->] (2,0.2,0.5)--(3.25,0.2,0.5);
			
			\draw (0+2.4,0+0.7,0)--(0.5+2.4,0+0.7,0)--(0.5+2.4,0+0.7,0.5)--(0+2.4,0+0.7,0.5)--cycle;
			\foreach \i in {0+2.4,0.5+2.4}{\foreach \k in {0,0.5}{
					\draw (\i,0+0.7,\k)--(\i,0.5+0.7,\k);
			}}
			\draw[blue, line width=0.5pt, ->] (0+2.4,0.25+0.7,0)--(0.5+2.4,0.25+0.7,0);
			\draw[blue, line width=0.5pt, ->] (0+2.4,0.25+0.7,0)--(0.5+2.4,0.25+0.7,0.5);
			\draw[blue, line width=0.5pt, ->] (0+2.4,0.25+0.7,0)--(0+2.4,0.25+0.7,0.5);
			\draw[blue, line width=0.5pt, ->] (0+2.4,0.25+0.7,0)--(0.25+2.4,0+0.7,0);
			\draw[blue, line width=0.5pt, ->] (0+2.4,0.25+0.7,0)--(0.5+2.4,0+0.7,0.25);
			\draw[blue, line width=0.5pt, ->] (0+2.4,0.25+0.7,0)--(0.25+2.4,0+0.7,0.5);
			\draw[blue, line width=0.5pt, ->] (0+2.4,0.25+0.7,0)--(0+2.4,0+0.7,0.25);
			\draw
			(-0.5+2.4,0+0.7,0)--(0+2.4,0+0.7,0)
			(0+2.4,-0.5+0.7,0)--(0+2.4,0+0.7,0)
			(0+2.4,0+0.7,-0.5)--(0+2.4,0+0.7,0);

			%左二图
			\draw[line width=1pt] (0+4.75,0,0)--(1+4.75,0,0)--(1+4.75,0,1)--(0+4.75,0,1)--cycle;
			\foreach \i in {0+4.75,1+4.75}{\foreach \k in {0,1}{
					\draw[line width=1pt] (\i,0,\k)--(\i,1,\k);
			}}
			\draw[line width=1pt]
			(-1+4.75,0,0)--(0+4.75,0,0)
			(0+4.75,-1,0)--(0+4.75,0,0)
			(0+4.75,0,-1)--(0+4.75,0,0);
			
			\foreach \i/\j/\k in {0.5+4.75/0/0, -0.5+4.75/0/0, 0+4.75/-0.5/0}{
				\node[ForestGreen] at (\i,\j,\k) {$Z$};
			}
			\fill[ForestGreen, opacity=0.2] (-0.875+4.75,0.2,0)--(0.875+4.75,0.2,0)--(0.875+4.75,-0.2,0)--(-0.875+4.75,-0.2,0);
			\fill[ForestGreen, opacity=0.2] (0.2+4.75,0.2,0)--(0.2+4.75,-0.875,0)--(-0.2+4.75,-0.875,0)--(-0.2+4.75,0.2,0);

			%逗号
			\node at (6.5,-0.5,0) {$,$};

			\begin{scope}[xshift=-8.25cm, yshift=-2.8cm]
				%右一图
				\draw[line width=1pt] (0+8.25,0,0)--(1+8.25,0,0)--(1+8.25,0,1)--(0+8.25,0,1)--cycle;
				\foreach \i in {0+8.25,1+8.25}{\foreach \k in {0,1}{
						\draw[line width=1pt] (\i,0,\k)--(\i,1,\k);
				}}
				\draw[line width=1pt]
				(-1+8.25,0,0)--(0+8.25,0,0)
				(0+8.25,-1,0)--(0+8.25,0,0)
				(0+8.25,0,-1)--(0+8.25,0,0);
				
				\foreach \i/\j/\k in {0+8.25/0/0.5, 0+8.25/0.5/0, 0+8.25/0/-0.5, 0+8.25/-0.5/0}{
					\node[ForestGreen] at (\i,\j,\k) {$Z$};
				}
				\fill[ForestGreen, opacity=0.2] (0.2+8.25,0.875,0)--(0.2+8.25,-0.875,0)--(-0.2+8.25,-0.875,0)--(-0.2+8.25,0.875,0);
				\fill[ForestGreen, opacity=0.2] (-0.2+8.25,0,1)--(0.2+8.25,0,1)--(0.2+8.25,0,-1)--(-0.2+8.25,0,-1);

				%箭头
				\draw[line width=1pt, ->] (2+8.25,0.2,0.5)--(3.25+8.25,0.2,0.5);
				
				\draw (0+10.65,0+0.7,0)--(0.5+10.65,0+0.7,0)--(0.5+10.65,0+0.7,0.5)--(0+10.65,0+0.7,0.5)--cycle;
				\foreach \i in {0+10.65,0.5+10.65}{\foreach \k in {0,0.5}{
						\draw (\i,0+0.7,\k)--(\i,0.5+0.7,\k);
				}}
				\draw[blue, line width=0.5pt, ->] (0+10.65,0.25+0.7,0)--(0.5+10.65,0.25+0.7,0);
				\draw[blue, line width=0.5pt, ->] (0+10.65,0.25+0.7,0)--(0.5+10.65,0.25+0.7,0.5);
				\draw[blue, line width=0.5pt, ->] (0+10.65,0.25+0.7,0)--(0+10.65,0.25+0.7,0.5);
				\draw[blue, line width=0.5pt, ->] (0+10.65,0.25+0.7,0)--(0.25+10.65,0+0.7,0);
				\draw[blue, line width=0.5pt, ->] (0+10.65,0.25+0.7,0)--(0.5+10.65,0+0.7,0.25);
				\draw[blue, line width=0.5pt, ->] (0+10.65,0.25+0.7,0)--(0.25+10.65,0+0.7,0.5);
				\draw[blue, line width=0.5pt, ->] (0+10.65,0.25+0.7,0)--(0+10.65,0+0.7,0.25);
				\draw
				(-0.5+10.65,0+0.7,0)--(0+10.65,0+0.7,0)
				(0+10.65,-0.5+0.7,0)--(0+10.65,0+0.7,0)
				(0+10.65,0+0.7,-0.5)--(0+10.65,0+0.7,0);

				%右一图
				\draw[line width=1pt] (0+13,0,0)--(1+13,0,0)--(1+13,0,1)--(0+13,0,1)--cycle;
				\foreach \i in {0+13,1+13}{\foreach \k in {0,1}{
						\draw[line width=1pt] (\i,0,\k)--(\i,1,\k);
				}}
				\draw[line width=1pt]
				(-1+13,0,0)--(0+13,0,0)
				(0+13,-1,0)--(0+13,0,0)
				(0+13,0,-1)--(0+13,0,0);
				
				\foreach \i/\j/\k in {0+13/0/0.5, 0+13/0/-0.5, 0+13/-0.5/0}{
					\node[ForestGreen] at (\i,\j,\k) {$Z$};
				}
				\fill[ForestGreen, opacity=0.2] (0.2+13,0,0)--(0.2+13,-0.875,0)--(-0.2+13,-0.875,0)--(-0.2+13,0,0);
				\fill[ForestGreen, opacity=0.2] (-0.2+13,0,1)--(0.2+13,0,1)--(0.2+13,0,-1)--(-0.2+13,0,-1);
			\end{scope}
			
			\node at (6.5,-3.5,0) {$.$};
		\end{tikzpicture}
	\end{equation}
	The final layer of the SQC consists of two parts, part 1 lies along the $x=0,z=L_z$ line, part 2 lies along the $y=0,z=L_z$ line.
	The two lines meet at $v=[0,0,L_z]$, and this vertex is assigned to part 1 only: part 1 acts on the top boundary vertices with $x(v)=0$, while part 2 acts on those with $y(v)=0$ \textit{and} $x(v)\neq0$, as drawn in Eq.~(\ref{X_cube_rough_to_smooth_top_SQC}).
	The two parts therefore cover $L_x+L_y-1$ vertices in total, matching the $L_x+L_y-1$ top stabilizers deleted under the $(mm)$ top boundary.
	Part 1 maps the $B_{v,yz},B_{v,zx}$ terms attached to vertices $v$ on the top boundary with $x=0$ as follows (while keeping $B_{v,xy}$ and $A_c$ invariant):
	\begin{equation}\label{X_cube_rough_to_smooth_top_SQC_2}
		\begin{tikzpicture}[>={Triangle[round,length=1.5mm,width=1mm]}, baseline=-3ex]
			%左一图
			\draw[line width=1pt]
			(-1,0,0)--(1,0,0)
			(0,-1,0)--(0,1,0)
			(0,0,-1)--(0,0,1);
			
			\foreach \i/\j/\k in {0.5/0/0, 0/0.5/0, -0.5/0/0, 0/-0.5/0}{
				\node[ForestGreen] at (\i,\j,\k) {$Z$};
			}
			\fill[ForestGreen, opacity=0.2] (-0.875,0.2,0)--(0.875,0.2,0)--(0.875,-0.2,0)--(-0.875,-0.2,0);
			\fill[ForestGreen, opacity=0.2] (0.2,0.875,0)--(0.2,-0.875,0)--(-0.2,-0.875,0)--(-0.2,0.875,0);

			%箭头
			\draw[line width=1pt, ->] (2,0.2,0.5)--(3.25,0.2,0.5);
			
			\draw
			(-0.5+2.4,0+0.7,0)--(0.5+2.4,0+0.7,0)  (0+2.4,-0.5+0.7,0)--(0+2.4,0.5+0.7,0)
			(0+2.4,0+0.7,-0.5)--(0+2.4,0+0.7,0.5);
			
			\draw[blue, line width=1pt, ->] (0+2.4,-0.225+0.7,0)--(0+2.4,0.225+0.7,0);
			\draw[blue, line width=1pt, ->] (-0.25+2.4,0+0.7,0)--(0-0.05+2.4,0.225+0.7,0);
			\draw[blue, line width=1pt, ->] (+0.25+2.4,0+0.7,0)--(0+0.05+2.4,0.225+0.7,0);

			%左图二
			\draw[line width=1pt]
			(-1+4.75,0,0)--(1+4.75,0,0)
			(0+4.75,-1,0)--(0+4.75,1,0)
			(0+4.75,0,-1)--(0+4.75,0,1);
			
			\foreach \i/\j/\k in {0+4.75/0.5/0}{
				\node[ForestGreen] at (\i,\j,\k) {$Z$};
			}
			\fill[ForestGreen, opacity=0.2] (0.2+4.75,0.875,0)--(0.2+4.75,0,0)--(-0.2+4.75,0,0)--(-0.2+4.75,0.875,0);

			%逗号
			\node at (6.5,-0.5,0) {$,$};

			\begin{scope}[xshift=-8.25cm, yshift=-2.8cm]
				%右图一
				\draw[line width=1pt]
				(-1+8.25,0,0)--(1+8.25,0,0)
				(0+8.25,-1,0)--(0+8.25,1,0)
				(0+8.25,0,-1)--(0+8.25,0,1);
				
				\foreach \i/\j/\k in {0+8.25/0.5/0, 0+8.25/-0.5/0, 0+8.25/0/0.5, 0+8.25/0/-0.5}{
					\node[ForestGreen] at (\i,\j,\k) {$Z$};
				}
				\fill[ForestGreen, opacity=0.2] (0.2+8.25,0.875,0)--(0.2+8.25,-0.875,0)--(-0.2+8.25,-0.875,0)--(-0.2+8.25,0.875,0);
				\fill[ForestGreen, opacity=0.2] (-0.2+8.25,0,-1)--(0.2+8.25,0,-1)--(0.2+8.25,0,1)--(-0.2+8.25,0,1);

				%箭头
				\draw[line width=1pt, ->] (2+8.25,0.2,0.5)--(3.25+8.25,0.2,0.5);
				
				\draw
				(-0.5+10.65,0+0.7,0)--(0.5+10.65,0+0.7,0)  (0+10.65,-0.5+0.7,0)--(0+10.65,0.5+0.7,0)
				(0+10.65,0+0.7,-0.5)--(0+10.65,0+0.7,0.5);
				
				\draw[blue, line width=1pt, ->] (0+10.65,-0.225+0.7,0)--(0+10.65,0.225+0.7,0);
				\draw[blue, line width=1pt, ->] (-0.25+10.65,0+0.7,0)--(0-0.05+10.65,0.225+0.7,0);
				\draw[blue, line width=1pt, ->] (+0.25+10.65,0+0.7,0)--(0+0.05+10.65,0.225+0.7,0);

				%右图二
				\draw[line width=1pt]
				(-1+13,0,0)--(1+13,0,0)
				(0+13,-1,0)--(0+13,1,0)
				(0+13,0,-1)--(0+13,0,1);
				
				\foreach \i/\j/\k in {0+13/0.5/0, 0.5+13/0/0, -0.5+13/0/0, 0+13/0/-0.5, 0+13/0/0.5}{
					\node[ForestGreen] at (\i,\j,\k) {$Z$};
				}
				\fill[ForestGreen, opacity=0.2] (0.2+13,0.875,0)--(0.2+13,0,0)--(-0.2+13,0,0)--(-0.2+13,0.875,0);
				\fill[ForestGreen, opacity=0.2] (-0.2+13,0,-1)--(0.2+13,0,-1)--(0.2+13,0,1)--(-0.2+13,0,1);
				\fill[ForestGreen, opacity=0.2] (-0.875+13,0.2,0)--(-0.875+13,-0.2,0)--(0.875+13,-0.2,0)--(0.875+13,0.2,0);
			\end{scope}
			
			\node at (6.5,-3.5,0) {$.$};
		\end{tikzpicture}
	\end{equation}
	The two operators after the last layer in Eq.~(\ref{X_cube_rough_to_smooth_top_SQC_2}) are equivalent, in the sense that they differ only by multiplication by a stabilizer $B_{v,xy}$.
	The multiplication relation $B_{v,xy}B_{v,zx}B_{v,yz}=1$ is preserved by the last layer, which is consistent with the fact that CNOT is invertible, and any invertible adjoint action preserves the relation $B_{v,xy}B_{v,zx}B_{v,yz}=1$.
	Similarly, part 2 of the last layer maps the $B_{v,zx},B_{v,yz}$ terms attached to vertices $v$ on the top boundary with $y=0$ and $x\neq0$ as follows (while keeping $B_{v,xy}$ and $A_c$ invariant):
	\begin{equation}\label{X_cube_rough_to_smooth_top_SQC_3}
		\begin{tikzpicture}[>={Triangle[round,length=1.5mm,width=1mm]}, baseline=-3ex]
			%左一图
			\draw[line width=1pt]
			(-1,0,0)--(1,0,0)
			(0,-1,0)--(0,1,0)
			(0,0,-1)--(0,0,1);
			
			\foreach \i/\j/\k in {0.5/0/0, 0/0.5/0, -0.5/0/0, 0/-0.5/0}{
				\node[ForestGreen] at (\i,\j,\k) {$Z$};
			}
			\fill[ForestGreen, opacity=0.2] (-0.875,0.2,0)--(0.875,0.2,0)--(0.875,-0.2,0)--(-0.875,-0.2,0);
			\fill[ForestGreen, opacity=0.2] (0.2,0.875,0)--(0.2,-0.875,0)--(-0.2,-0.875,0)--(-0.2,0.875,0);

			%箭头
			\draw[line width=1pt, ->] (2,0.2,0.5)--(3.25,0.2,0.5);
			
			\draw
			(-0.5+2.4,0+0.7,0)--(0.5+2.4,0+0.7,0)  (0+2.4,-0.5+0.7,0)--(0+2.4,0.5+0.7,0)
			(0+2.4,0+0.7,-0.5)--(0+2.4,0+0.7,0.5);
			
			\draw[blue, line width=1pt, ->] (0+2.4,-0.225+0.7,0)--(0+2.4,0.225+0.7,0);
			\draw[blue, line width=1pt, ->] (0+2.4,0+0.7,-0.25)--(0+2.4,0.225+0.7,-0.1);
			\draw[blue, line width=1pt, ->] (0+2.4,0+0.7,+0.25)--(0+2.4,0.225+0.7,+0.1);

			%左图二
			\draw[line width=1pt]
			(-1+4.75,0,0)--(1+4.75,0,0)
			(0+4.75,-1,0)--(0+4.75,1,0)
			(0+4.75,0,-1)--(0+4.75,0,1);
			
			\foreach \i/\j/\k in {0+4.75/0.5/0, 0.5+4.75/0/0, -0.5+4.75/0/0, 0+4.75/0/-0.5, 0+4.75/0/0.5}{
				\node[ForestGreen] at (\i,\j,\k) {$Z$};
			}
			\fill[ForestGreen, opacity=0.2] (0.2+4.75,0.875,0)--(0.2+4.75,0,0)--(-0.2+4.75,0,0)--(-0.2+4.75,0.875,0);
			\fill[ForestGreen, opacity=0.2] (-0.2+4.75,0,-1)--(0.2+4.75,0,-1)--(0.2+4.75,0,1)--(-0.2+4.75,0,1);
			\fill[ForestGreen, opacity=0.2] (-0.875+4.75,0.2,0)--(-0.875+4.75,-0.2,0)--(0.875+4.75,-0.2,0)--(0.875+4.75,0.2,0);

			%逗号
			\node at (6.5,-0.5,0) {$,$};

			\begin{scope}[xshift=-8.25cm, yshift=-2.8cm]
				%右图一
				\draw[line width=1pt]
				(-1+8.25,0,0)--(1+8.25,0,0)
				(0+8.25,-1,0)--(0+8.25,1,0)
				(0+8.25,0,-1)--(0+8.25,0,1);
				
				\foreach \i/\j/\k in {0+8.25/0.5/0, 0+8.25/-0.5/0, 0+8.25/0/0.5, 0+8.25/0/-0.5}{
					\node[ForestGreen] at (\i,\j,\k) {$Z$};
				}
				\fill[ForestGreen, opacity=0.2] (0.2+8.25,0.875,0)--(0.2+8.25,-0.875,0)--(-0.2+8.25,-0.875,0)--(-0.2+8.25,0.875,0);
				\fill[ForestGreen, opacity=0.2] (-0.2+8.25,0,-1)--(0.2+8.25,0,-1)--(0.2+8.25,0,1)--(-0.2+8.25,0,1);

				%箭头
				\draw[line width=1pt, ->] (2+8.25,0.2,0.5)--(3.25+8.25,0.2,0.5);
				
				\draw
				(-0.5+2.4+8.25,0+0.7,0)--(0.5+2.4+8.25,0+0.7,0)  (0+2.4+8.25,-0.5+0.7,0)--(0+2.4+8.25,0.5+0.7,0)
				(0+2.4+8.25,0+0.7,-0.5)--(0+2.4+8.25,0+0.7,0.5);
				
				\draw[blue, line width=1pt, ->] (0+2.4+8.25,-0.225+0.7,0)--(0+2.4+8.25,0.225+0.7,0);
				\draw[blue, line width=1pt, ->] (0+2.4+8.25,0+0.7,-0.25)--(0+2.4+8.25,0.225+0.7,-0.1);
				\draw[blue, line width=1pt, ->] (0+2.4+8.25,0+0.7,+0.25)--(0+2.4+8.25,0.225+0.7,+0.1);

				%右图二
				\draw[line width=1pt]
				(-1+13,0,0)--(1+13,0,0)
				(0+13,-1,0)--(0+13,1,0)
				(0+13,0,-1)--(0+13,0,1);
				
				\foreach \i/\j/\k in {0+13/0.5/0}{
					\node[ForestGreen] at (\i,\j,\k) {$Z$};
				}
				\fill[ForestGreen, opacity=0.2] (0.2+13,0.875,0)--(0.2+13,0,0)--(-0.2+13,0,0)--(-0.2+13,0.875,0);
			\end{scope}
			
			\node at (6.5,-3.5,0) {$.$};
		\end{tikzpicture}
	\end{equation}
	Thus the linear-depth LU SQC changes the $(ee)$ top boundary to the $(mm)$ top boundary and realizes the TFPIM KW duality summarized in Fig.~\ref{fig_X_cube_4_boundaries_dualities}.

	\section{Pedagogical review of the translation-invariant stabilizer code formalism}\label{appendix_pedagogical_review}\label{app_poly_review}
	
	We provide a pedagogical review of translation-invariant $\mathbb Z_p$ stabilizer code formalism and the language for describing the boundaries of stabilizer codes in this formalism.
	We focus on the case where $p$ is prime in this paper, which was introduced by Haah in Ref.~\cite{haah_2013_modules}, and generalized to the case where the qudit dimension could be non-prime by Ruba-Yang~\cite{ruba_yang_2024_homological}. 
	
	\subsection{Laurent polynomial ring and translation-invariant stabilizer codes}
	
	For a $p$-dimensional qudit, the Pauli group can be generated by Pauli $X$ and Pauli $Z$ purely by multiplication, where the Pauli $X,Z$ are defined as
	\begin{equation}
		X|j\rangle = |j+1\rangle,\qquad Z|j\rangle=\omega^j|j\rangle,\qquad \omega=e^{2\pi i/p}\,,\qquad ZX=\omega XZ.
	\end{equation}
	Consequently, the Pauli group modulo phase factors can be represented as a vector space $\mathbb F_p^2$, spanned by the Pauli $X$ and Pauli $Z$ of the qudit, where the multiplication of Pauli operators now becomes addition of vectors.
	The commutation phase information can be encoded into a symplectic bilinear form $\Omega:\mathbb F_p^2\times\mathbb F_p^2\to\mathbb F_p$, represented as a $2\times 2$ matrix 
	\[
	\lambda = \begin{pmatrix} 0 & 1 \\ -1 & 0 \end{pmatrix}
	\]
	in the $X,Z$ basis, i.e. $X:=(1,0)^T,\ Z:=(0,1)^T$, where $\Omega(u,v) := u^T\lambda v$.
	Specifically, denoting the corresponding Pauli operators of $u,v\in\mathbb{F}_p^2$ as $\mathsf p(u),\mathsf p(v)$, we have
	\[
	\mathsf p(u)\mathsf p(v) = \omega^{-u^T\lambda v} \mathsf p(v)\mathsf p(u).
	\]
	
	In Ref.~\cite{haah_2013_modules}, Haah generalized this formalism to describe Pauli operators on an abelian group lattice $\Lambda$ (e.g. $\mathbb Z^d$).
	For concreteness, we use the lattice $\Lambda=\mathbb Z^2$ for illustration.
	Consider a system with each site in $\Lambda$ hosting $q$ $p$-dimensional qudits.
	Then its finite-support Pauli group modulo phase factors is represented as a free module $R^{2q}$ over the base ring
	\begin{equation}
		R = \mathbb F_p[\mathbb Z^2] = \mathbb F_p[x^{\pm1},y^{\pm1}].
	\end{equation}
	The elements in $R$ are finite Laurent polynomials like
	\[
	1,\qquad 1+x,\qquad 1+x+x^2y+\bx\by^3.
	\]
	where $\overline{(\cdot)}=(\cdot)^{-1}$.
	The finite support Pauli operators (up to phase factor) are represented as elements of $R^{2q}$ as follows:
	\begin{enumerate}
		\item Choose a basis of the free module $R^{2q}$ that corresponds to single Pauli $X,Z$ of the system:
		\begin{gather}
			X_1\equiv X_{(0,0),1}:=(1,0,\cdots,0;0,0,\cdots,0)^T,\qquad\cdots\cdots,\qquad X_q\equiv X_{(0,0),q}:=(0,0,\cdots,1;0,0,\cdots,0)^T,\notag\\
			Z_1\equiv Z_{(0,0),1}:=(0,0,\cdots,0;1,0,\cdots,0)^T,\qquad\cdots\cdots,\qquad Z_q\equiv Z_{(0,0),q}:=(0,0,\cdots,0;0,0,\cdots,1)^T,\label{eq_TI_stabilizer_code_basis}
		\end{gather}
		where $X_{(0,0),\alpha}$ is the $\alpha$-th qudit in the cell at $(0,0)$.
		$Z$ basis are read similarly. 
		\item Let
		\begin{equation}
			x^iy^j X_{(0,0),\alpha}\equiv X_{(i,j),\alpha},
		\end{equation}
		where $X_{(i,j),\alpha}$ is the $\alpha$-th qudit in the cell at $(i,j)$.
		\item According to $R$-linearity, other elements in $R^{2q}$ can be read.
		For example,
		\begin{equation*}
			(1+x^2,0,\cdots,\by;2x,0,\cdots,0)^T
		\end{equation*}
		corresponds to
		\begin{equation*}
			X_{(0,0),1}X_{(2,0),1}X_{(0,-1),q}Z_{(1,0),1}^2.
		\end{equation*}
		Here $\by\equiv y^{-1}$.
		Similarly, $\bx\equiv x^{-1}$.
		For any $u=\sum_{ij}u_{ij}x^iy^j\in R$, $\bar u\equiv\sum_{ij}u_{ij}\bx^i\by^j$.
	\end{enumerate}
	Like before, the commutation phase information can be encoded into a symplectic form $\Omega:R^{2q}\times R^{2q}\to R$, represented by a $2q\times 2q$ matrix
	\begin{equation}\label{eq_def_lambda_q}
		\lambda_q=\left(\begin{array}{cc}
			0 & 1_q \\
			-1_q & 0
		\end{array}\right),
	\end{equation}
	in the basis in Eq.~(\ref{eq_TI_stabilizer_code_basis}), where $\Omega(U,V)=U^\dagger\lambda_q V$ for any $U,V\in R^{2q}$, $1_q$ is the $q\times q$ identity matrix, $U^\dagger\equiv \bar U^T$.
	Specifically, denoting the corresponding Pauli operator of $U\in R^{2q}$ as $\mathsf p(U)$, then
	\begin{equation}\label{eq_commutation_phase_TI_stabilizer_code}
		\mathsf p(U)\mathsf p(V) = \omega^{-\tr(U^\dagger\lambda_q V)} \mathsf p(V)\mathsf p(U),
	\end{equation}
	where $\forall u = \sum_{ij} u_{ij}x^iy^j \in R$,
	\[
	\tr(u) := u_{00}.
	\]
	Moreover, $\forall g\in\Lambda$ (or equivalently, for any monomial/unit $g\in R$), 
	\[
	\mathsf p(U)\mathsf p(gV) = \omega^{-\tr(U^\dagger\lambda_q gV)} \mathsf p(gV)\mathsf p(U).
	\]
	Consequently, $\mathsf p(U)$ commutes with all the translations of $\mathsf p(V)$, i.e.
	\[
	\forall g\in\Lambda,\quad \mathsf p(U)\mathsf p(gV)=\mathsf p(gV)\mathsf p(U),
	\]
	iff
	\[
	U^\dagger\lambda_q V = 0 \in R.
	\]
	Denote the Pauli module as
	\begin{equation}
		P:=\text{span}_R\{X_\alpha,Z_\alpha\mid \alpha=1,\cdots,q\}\cong R^{2q}.
	\end{equation}
	Suppose there are $t$ stabilizer generators w.r.t.
	$R$.
	Let $G\cong R^t$ be their generator-label module.
	The physical Pauli operator represented by each generator is specified by an $R$-linear stabilizer map
	\begin{equation}
		\sigma:G\to P.
	\end{equation}
	Thus, $\sigma$ can be reprensented by a $2q\times t$ matrix over $R$, and $\imop\sigma$ is the module of Pauli operators generated by all lattice translations of the generators.
	Relations among the chosen generators form $\ker\sigma$.
	
	Let $E\cong R^t$ be the excitation module, whose basis labels the same $t$ stabilizer types.
	The syndrome created by a finite-support Pauli operator is determined by the excitation map
	\begin{equation}
		\epsilon:P\longrightarrow E,\qquad \epsilon:=\sigma^\dagger\lambda_q.
	\end{equation}
	The entries of $\epsilon(p)$ give the configurations of violated stabilizers produced by $p\in P$.
	The stabilizer generators commute with all their translations precisely when
	\begin{equation}
		\epsilon\circ\sigma=0,
	\end{equation}
	A translation-invariant stabilizer code therefore defines a chain complex
	\begin{equation}\label{eq_bulk_stabilizer_complex_review}
		G\xrightarrow{\ \sigma\ }P\xrightarrow{\ \epsilon\ }E.
	\end{equation}
	Here $\ker\epsilon$ consists of finite-support Pauli operators commuting with every stabilizer, while $\imop\sigma$ consists of finite-support stabilizers.
	Hence
	\begin{equation}
		\ker\epsilon/\imop\sigma
	\end{equation}
	is the module of finite-support logical operators on the infinite lattice.
	The bulk topological-order condition, i.e., the absense of nontrivial finite-support logical operator, is therefore the exactness of the chain complex in Eq.~(\ref{eq_bulk_stabilizer_complex_review}),
	\begin{equation}
		\imop\sigma=\ker\epsilon,
	\end{equation}
	which does not exclude extended logical operators after imposing periodic boundary conditions.
	
	The submodule $\imop\epsilon\subset E$ consists of syndromes created by finite-support Pauli operators.
	Consequently,
	\begin{equation}
		\coker\epsilon:=E/\imop\epsilon
	\end{equation}
	classifies syndrome configurations modulo finite-support Pauli operations.
	Point-like bulk topological excitation classes form the torsion submodule of $\coker\epsilon$.
	For $e\in\coker\epsilon$, its annihilator ideal is
	\begin{equation}
		\ann(e):=\{r\in R\mid re=0\text{ in }\coker\epsilon\},
	\end{equation}
	which records the finite excitation configurations of $e$ that can be created or annihilated by finite-support Pauli operators.

	\subsection{Boundary description of translation-invariant stabilizer codes}\label{appendix_general_boundary_gauge_syndromes}
	
	Boundary gauge operators are the operators commuting with the bulk stabilizers but not generated by the bulk stabilizers near the boundary (see Sec.~\ref{sec_TH_review}). a precise definition of their syndromes that is used in the FTH construction.
	Fix a boundary $\square$, let $R_\partial$ be the Laurent polynomial ring for translations parallel to it, and let $G^\square$, $P^\square$, and $E^\square$ be the boundary gauge-operator generator-label, Pauli, and syndrome modules, respectively.
	The map
	\begin{equation}
		\sigma^\square:G^\square\longrightarrow P^\square
	\end{equation}
	assigns Pauli representatives to the boundary gauge-operator labels.
	Writing the symplectic matrix on $P^\square$ as $\lambda$, define
	\begin{equation}\label{eq_general_boundary_syndrome_map}
		\epsilon^\square
		:=
		(\sigma^\square)^\dagger\lambda:
		P^\square\longrightarrow E^\square.
	\end{equation}
	Analogous to Eq.~(\ref{eq_commutation_phase_TI_stabilizer_code}), for any $O\in P^\square$ and $g\in\sigma^\square(G^\square)$, their corresponding Pauli operators $\mathsf p(O)$ and $\mathsf p(g)$ have the following commutation relation:
	\begin{equation}
		\mathsf p(g)\mathsf p(O)
		=
		\omega^{-\tr\!\left(g^\dagger\lambda O\right)}
		\mathsf p(O)\mathsf p(g).
	\end{equation}
	Writing $g=\sigma^\square(U)$ with $U\in G^\square$, we have
	\begin{equation}
		g^\dagger\lambda O
		=U^\dagger(\sigma^\square)^\dagger\lambda O=
		U^\dagger\epsilon^\square(O).
	\end{equation}
	Thus, $\epsilon^\square(O)$ records the commutation phases of $O$ with all boundary gauge operators and is called the $\square$ boundary gauge syndrome of $O$.
	
	Unlike stabilizer generators, boundary gauge operators need not commute with one another.
	Their mutual commutation is encoded by
	\begin{equation}
		\eta^\square
		:=
		\epsilon^\square\circ\sigma^\square
		=
		(\sigma^\square)^\dagger\lambda\sigma^\square,
	\end{equation}
	which need not vanish.
	Hence
	\begin{equation}
		G^\square\xrightarrow{\sigma^\square}
		P^\square\xrightarrow{\epsilon^\square}
		E^\square
	\end{equation}
	is not generally a chain complex.
	
	Multiplying $O$ by the boundary gauge operator $g\in\sigma^\square(G^\square)$ changes its syndrome by an element of $\imop\eta^\square$:
	\begin{equation}
		\epsilon^\square(O+g)
		=
		\epsilon^\square(O)+\epsilon^\square(g),
		\qquad
		\epsilon^\square(g)\in\imop\eta^\square.
	\end{equation}
	Therefore, a boundary topological excitation is represented by
	\begin{equation}
		[\mathsf b]_\square
		:=
		\mathsf b+\imop\eta^\square
		\in
		\coker\eta^\square.
	\end{equation}
	The point-like boundary topological excitation classes form the torsion submodule of $\coker\eta^\square$.
	A concrete example is given for Haah's cubic code in Appendix~\ref{appendix_boundary_gauge_syndrome_matrices}.

	\section{Boundary data for Haah's code FTH}
	\label{appendix_Haah_boundary_data}

	\subsection{Notation and setup}\label{appendix_Haah_setup}
	
	In this and the following appendices, we use the module language reviewed in Appendix~\ref{appendix_pedagogical_review}.
	Denote the finite-support stabilizer module, Pauli module and excitation module of Haah's code on $\mathbb Z^3$ as $G,P,E$, respectively, over $R=\mathbb F_2[x^{\pm1},y^{\pm1},z^{\pm1}]$.
	Choose the basis $S_X:=(1,0)^T$, $S_Z:=(0,1)^T$ of $G$ and $E$, and the basis $X_1:=(1,0;0,0)^T$, $X_2:=(0,1;0,0)^T$, $Z_1:=(0,0;1,0)^T$, $Z_2:=(0,0;0,1)^T$ of $P$, where $X_{1,2},Z_{1,2}$ are associated with the qubits at $[1,1,1]$ and $S_X,S_Z$ are associated with the cube $[\frac12,\frac12,\frac12]$.
	Then the stabilizer generators in Eq.~(\ref{eq_Haah_SX_SZ}) are encoded in the stabilizer map $\sigma:G\to P$,
	\begin{equation}\label{eq_Haah_code_sigma}
		\sigma=
		\begin{pmatrix}
			\bx\by\bz(1+xy+xz+yz) & 0\\
			\bx\by\bz(1+x+y+z) & 0\\
			0 & \bx\by\bz(xy+xz+yz+xyz)\\
			0 & \bx\by\bz(x+y+z+xyz)
		\end{pmatrix},
	\end{equation}
	and the excitation map $\epsilon:P\to E$ is $\epsilon=\sigma^\dagger\lambda_2$, i.e.,
	\begin{equation}\label{eq_Haah_code_epsilon}
		\epsilon=
		\begin{pmatrix}
			0&0&d&c\\
			b&a&0&0
		\end{pmatrix},
	\end{equation}
	where
	\begin{align}
		a:=1+xy+xz+yz&,\quad b:=1+x+y+z,\notag\\
		c:=xy+xz+yz+xyz&,\quad d:=x+y+z+xyz,\label{eq_Haah_code_abcd_def}
	\end{align}
	in agreement with Eq.~(\ref{eq_Haah_single_Pauli_syndromes}).

	Consider the truncation $\pi$ of Haah's code onto the slab with cells from $z=0$ to $z=L_z-1$.
	$\pi P$ is the Pauli group modulo phase factors supported on the slab.
	Denote the bulk stabilizer module as $\mathcal S_B$, and the truncated stabilizer module, which further includes the truncated stabilizers on the two boundaries, as $\mathcal S_T\supset\mathcal S_B$.
	Over the boundary base ring $R_\partial=\mathbb F_2[x^{\pm1},y^{\pm1}]$,
	\begin{equation}
		\pi P\cong R_\partial^{4L_z},\qquad \mathcal S_B\cong R_\partial^{2(L_z-1)},\qquad \mathcal S_T\cong R_\partial^{2(L_z+1)}.
	\end{equation}
	The boundary gauge-operator class module is $\mathcal G\equiv\mathcal S_T/\mathcal S_B=\mathcal G^{\text{top}}\oplus\mathcal G^{\text{bot}}$, with $\mathcal G^{\text{top}}\cong\mathcal G^{\text{bot}}\cong R_\partial^2$. There is no secondary boundary gauge operator.
	For each $\square\in\{\text{top},\text{bot}\}$, let $G^\square\cong R_\partial^2$ be the boundary gauge-operator generator-label module with basis $\mathcal G_1^\square,\mathcal G_2^\square$, and let $\sigma^\square:G^\square\to\pi P$ map each label to its corresponding Pauli operator,
	\begin{align}
		\sigma^{\text{bot}}(\mathcal G_1^{\text{bot}})&:=S_X^{\text{bot}}\equiv\pi\bz S_X,\notag\\
		\sigma^{\text{bot}}(\mathcal G_2^{\text{bot}})&:=S_Z^{\text{bot}}\equiv\pi\bz S_Z,\label{def_boundary_gauge_bot}\\
		\sigma^{\text{top}}(\mathcal G_1^{\text{top}})&:=S_X^{\text{top}}\equiv \pi z^{L_z-1}S_X,\notag\\
		\sigma^{\text{top}}(\mathcal G_2^{\text{top}})&:=S_Z^{\text{top}}\equiv \pi z^{L_z-1}S_Z,\label{def_boundary_gauge_top}
	\end{align}
	which are the truncated stabilizers $S_{X,Z}^{\text{bot}}$ and $S_{X,Z}^{\text{top}}$ in Sec.~\ref{sec_Haah_code_FTH}.
	The boundary excitation module $E^\square:=\bigoplus_iR_\partial\mathcal G_i^\square$ reuses the symbols $\mathcal G_i^\square$ as syndrome-basis labels; the boundary syndrome map $\epsilon^\square$, the map $\eta^\square\equiv\epsilon^\square\circ\sigma^\square$, and the boundary topological excitation classes $[\mathsf b]_\square:=\mathsf b+\imop\eta^\square$ are defined as in Appendix~\ref{appendix_general_boundary_gauge_syndromes}.
	By an abuse of notation, we use $\mathcal G_i^\square$ to represent the generator label, its representative $\sigma^\square(\mathcal G_i^\square)$, and the boundary gauge syndrome $\mathcal G_i^\square$ when the intended meaning is clear; in particular, $\mathcal G_{1,2}^{\text{bot}}$ stand for $S_{X,Z}^{\text{bot}}$ in Sec.~\ref{sec_Haah_code_FTH}.
	In the calculation, we replace $\pi P$ by the boundary Pauli submodule $P^\square$, the minimal set of layers that $\sigma^\square(G^\square)$ involves; for an $\ell$-layer $P^\square$, $\epsilon^\square=(\sigma^\square)^\dagger\lambda_{q\ell}$ with
	\begin{equation}
		\lambda_{q\ell}=\left(\begin{array}{cc}
			0&I_{q\ell}\\-I_{q\ell}&0
		\end{array}\right),
	\end{equation}
	where $q$ is the number of qudits per cell.
	For Haah's code, $q=2$ and $\ell=1$.
	In the bulk, the analogues of $\sigma,\epsilon$ are $\sigma_B:\mathcal S_B\to\pi P$ and $\epsilon_B:\pi P\to E_B$, where $E_B\cong\mathcal S_B$ is the bulk excitation module.
	We use the shorthand polynomials
	\begin{gather}\label{eq_Haah_code_FTH_15}
		A :=1+x+y,\ \  B:=1+xy,\ \  C:=1,\ \ D :=x+y,\notag\\
		F :=AD+B=1+x+x^2+y+y^2+xy\,.
	\end{gather}

	Under the $(Z)$ top boundary, the stabilizer module is
	\begin{equation}
		\mathcal S_{(Z)} = \text{span}_{R_\partial}\big\{z^kS_X,z^kS_Z,\mathcal G_1^{\text{top}}\mid k=0,1,\cdots,L_z-2\big\},
	\end{equation}
	and $\mathcal S_{(X)}$ is defined similarly with $\mathcal G_1^{\text{top}}$ replaced by $\mathcal G_2^{\text{top}}$.
	For $O\in\mathcal S_{(Z)}^\Omega$, where $W^\Omega$ denotes the symplectic complement of $W\subset\pi P$, we write $[O]_{(Z)}:=O+\mathcal S_{(Z)}$.
	The low-energy-preserving operator module is $\mathcal O_{(Z)}:=\mathcal S_{(Z)}^\Omega/\mathcal S_{(Z)}$, and the logical operator module, consisting of the operators commuting with all stabilizers and bottom boundary gauge operators, is $\mathcal L_{(Z)}=\big(\mathcal S_{(Z)}+\mathcal G^{\text{bot}}\big)^\Omega/\mathcal S_{(Z)}$; similarly for $(X)$.
	The operator identification under the $(Z)$ [$(X)$] top boundary is an $R_\partial$-linear symplectic homomorphism $\Phi_Z$ [$\Phi_X$] from $\mathcal O_{(Z)}$ [$\mathcal O_{(X)}$] to the Pauli operators of the identified 2d system, whose commutation phases are encoded by the symplectic form $\tilde\Omega$.
	In this and the following appendices, we drop the primes of the operators identified under the $(X)$ top boundary.
	The operator identification shown in Sec.~\ref{sec_Haah_code_FTH} are
	\begin{gather}
		\Phi_Z\big([\mathcal G_1^{\text{bot}}]_{(Z)}\big)=\bx\by F\tilde Z,\qquad \Phi_Z\big([\mathcal T_{e_Z,z}]_{(Z)}\big)=A^{L_z}\tilde Z,\notag\\
		\Phi_X\big([\mathcal G_2^{\text{bot}}]_{(X)}\big)=xy\bF\tilde Z,\qquad \Phi_X\big([\mathcal T_{e_X,z}]_{(X)}\big)=\tilde Z,
	\end{gather}
	where the transport operators $\mathcal T_{e_Z,z}$ and $\mathcal T_{e_X,z}$ are given in Appendix~\ref{appendix_Haah_low_energy_preserving}.
	The symplectic homomorphism conditions of $\Phi_Z$ and $\Phi_X$ only fix
	\begin{equation}\label{eq_Haah_code_FTH_gauge_2_Z_identification}
		\Phi_Z\big([\mathcal G_2^{\text{bot}}]_{(Z)}\big) = \tilde X + v\tilde Z,\qquad \Phi_X\big([\mathcal G_1^{\text{bot}}]_{(X)}\big) = \tilde X + v'\tilde Z,\qquad v=\bar v,\quad v'=\bar v',
	\end{equation}
	(see Appendix~\ref{appendix_identification_of_other_boundary_gauge_operators}), where $v$ and $v'$ are independent.
	The choice $v=v'=0$ in Sec.~\ref{sec_Haah_code_FTH} is connected to the general cases by finite-depth LU circuits (see Appendix~\ref{appendix_LU_circuit_v_neq_0}).

	\subsection{Layer-by-layer bulk syndrome and stabilizer representation}\label{appendix_layer_by_layer_calculations}
	
	To facilitate analysis near the boundaries, we separate the stabilizer module $\mathcal S_B$, the truncated Pauli module $\pi P$, and the bulk excitation module $E_B$ into layer-by-layer submodules:
	\[
	\mathcal S_B=\bigoplus_{k=0}^{L_z-2}\mathcal S_{k+1/2},\qquad\pi P=\bigoplus_{k=0}^{L_z-1}P_k,
	\qquad
	E_B=\bigoplus_{k=0}^{L_z-2}E_{k+1/2},
	\]
	where $P_k\cong R_\partial^4$ is the Pauli submodule of qubits in the $z=k$ layer, generated by $z^{k-1}\{X_1,X_2,Z_1,Z_2\}$; $\mathcal S_{k+1/2}\cong E_{k+1/2}\cong R_\partial^2$ are the bulk stabilizer submodule and excitation submodule, both generated by $z^kS_X,\ z^kS_Z$.
	When the commutation relation we are talking about only involves qubits in some specific layers, we can use the symplectic form in the corresponding Pauli submodule.
	Specifically, from Eq.~(\ref{eq_Haah_code_sigma}) we can write
	\begin{gather}
		\sigma_B(z^kS_X) = z^{k-1}\bx\by(BX_1+AX_2) + z^k\bx\by(DX_1+X_2)\notag\\
		\sigma_B(z^kS_Z) = z^{k-1}\bx\by(xyZ_1+DZ_2) + z^k\bx\by(xy\bar A Z_1+BZ_2),
	\end{gather}
	for any $k=0,1,\cdots,L_z-2$, where
	\begin{equation*}
		A:=1+x+y,
		\qquad
		B:=1+xy,
		\qquad
		C:=1,
		\qquad
		D:=x+y.
	\end{equation*}
	Similarly, from Eqs.~(\ref{eq_Haah_code_abcd_def},\ref{eq_Haah_code_epsilon}) we can write
	\begin{gather}
		\epsilon_B(z^kX_1)=z^k(A+z)S_Z,\qquad \epsilon_B(z^kX_2)=z^k(B+Dz)S_Z,\notag\\
		\epsilon_B(z^kZ_1)=z^k(D+Bz)S_X,\qquad \epsilon_B(z^kZ_2)=z^k(xy+xy\bA z)S_X,\label{eq_epsilon_B_value_bulk}
	\end{gather}
	for any $k=0,1,\cdots,L_z-3$.
	While on the bottom boundary, since $E_B$ is truncated,
	\begin{gather}
		\epsilon_B(\bz X_1) = S_Z,\quad \epsilon_B(\bz X_2) = DS_Z,\quad \epsilon_B(\bz Z_1) = BS_X,\quad \epsilon_B(\bz Z_2) = xy\bA S_X.\label{eq_epsilon_B_value_bottom_boundary}
	\end{gather}
	Similarly, on the top boundary, since $E_B$ is truncated,
	\begin{gather}
		\epsilon_B(z^{L_z-2}X_1) = z^{L_z-2}AS_Z,\qquad \epsilon_B(z^{L_z-2}X_2) = z^{L_z-2}BS_Z,\notag\\ \epsilon_B(z^{L_z-2}Z_1) = z^{L_z-2}DS_X,\qquad \epsilon_B(z^{L_z-2}Z_2) = z^{L_z-2}xyS_X.\label{eq_epsilon_B_value_top_boundary}
	\end{gather}
	
	\subsection{Boundary gauge operators, syndromes and boundary topological excitations}\label{appendix_boundary_gauge_syndrome_matrices}
	
	\textit{Bottom boundary.}
	By convention, use the following ordered basis labels for both the bottom boundary generator-label module $G^{\text{bot}}$ and the bottom boundary syndrome module $E^{\text{bot}}$:
	\begin{equation}
		\mathcal G_1^{\text{bot}}:=\left(\begin{array}{c}
			1\\0
		\end{array}\right),\qquad\mathcal G_2^{\text{bot}}:=\left(\begin{array}{c}
			0\\1
		\end{array}\right)\,,
	\end{equation}
	The two copies use the same coordinate labels but are distinct modules.
	Under $\sigma^{\text{bot}}$, the labels in $G^{\text{bot}}$ map to the bottom truncated stabilizers defined in Eq.~(\ref{def_boundary_gauge_bot}).
	These representatives only involve Pauli operators in $P_0$, so $P^{\text{bot}}=P_0=\text{span}_{R_\partial}\{\bz X_1,\bz X_2,\bz Z_1,\bz Z_2\}$, and we can write $\sigma^{\text{bot}}$ as a $4\times 2$ matrix:
	\begin{equation}
		\sigma^{\text{bot}} = \left(\begin{array}{cc}
			\bx\by(x+y) & 0\\
			\bx\by & 0\\
			0 & \bx\by(x+y+xy)\\
			0 & \bx\by(1+xy) 
		\end{array}\right) = \bx\by\left(\begin{array}{cc}
			D&0\\C&0\\0&xy\bA\\0&B
		\end{array}\right)\,,
	\end{equation}
	where the basis of $P^{\text{bot}}$ is ordered as $(\bz X_1,\bz X_2,\bz Z_1,\bz Z_2)^T$.
	Then, analogous to the bulk theory, $\epsilon^{\text{bot}}$ can be calculated as
	\begin{equation}\label{eq_epsilon_bot_values}
		\epsilon^{\text{bot}}=(\sigma^{\text{bot}})^\dagger\lambda_2 = xy\left(\begin{array}{cccc}
			0&0&\bar D&\bar C\\\bx\by A&\bar B&0&0
		\end{array}\right)\,.
	\end{equation}
	Then, we can write
	\begin{equation}\label{eq_Haah_code_FTH_13}
		\eta^{\text{bot}} = \epsilon^{\text{bot}} \sigma^{\text{bot}} = \left(\begin{array}{cc}
			0&xy\bA\bar D+B\\\bx\by AD+\bar B&0
		\end{array}\right) := \left(\begin{array}{cc}
			0& xy\bF\\ \bx\by F&0
		\end{array}\right)\,,
	\end{equation}
	where
	\begin{equation}\label{eq_Haah_code_FTH_16}
		F:=xy(\bx\by AD+\bar B) = AD+B = 1+x+x^2+y+y^2+xy\,. 
	\end{equation}
	Eq.~(\ref{eq_Haah_code_FTH_13}) means that 
	\begin{equation}\label{eq_Haah_code_FTH_20}
		\eta^{\text{bot}}\mathcal G_1^{\text{bot}}=\bx\by F\mathcal G_2^{\text{bot}}\,,
	\end{equation}
	i.e. $\mathcal G_1^{\text{bot}}$ violates $\mathcal G_2^{\text{bot}}$ with configuration $\bar x\bar yF$, and 
	\begin{equation}\label{eq_Haah_code_FTH_22}
		\eta^{\text{bot}}\mathcal G_2^{\text{bot}}=xy\bF\mathcal G_1^{\text{bot}}\,,
	\end{equation}
	i.e. $\mathcal G_2^{\text{bot}}$ violates $\mathcal G_1^{\text{bot}}$ with configuration $xy\bar F$.
	In all, $\mathcal G_i^{\text{bot}}$ violates $\mathcal G_j^{\text{bot}}$ with the configuration $(\eta^{\text{bot}})_{ji}$.
	More precisely, $(\eta^{\text{bot}})_{ji}$ is the symplectic product of $\mathcal G_j^{\text{bot}}$ and $\mathcal G_i^{\text{bot}}$, since
	\begin{gather}
		\eta^{\text{bot}}=\epsilon^{\text{bot}}\circ\sigma^{\text{bot}} = (\sigma^{\text{bot}})^\dagger \lambda_{q\ell} \sigma^{\text{bot}},\label{eq_Haah_FTH_eta_bot_1}\\
		(\eta^{\text{bot}})_{ji} = (\mathcal G_j^{\text{bot}})^\dagger(\sigma^{\text{bot}})^\dagger \lambda_{q\ell} \sigma^{\text{bot}} \mathcal G_i^{\text{bot}} = \left(\sigma^{\text{bot}}(\mathcal G_j^{\text{bot}})\right)^\dagger\lambda_{q\ell}\sigma^{\text{bot}}(\mathcal G_i^{\text{bot}}) = \Omega\left(\sigma^{\text{bot}}(\mathcal G_j^{\text{bot}}),\sigma^{\text{bot}}(\mathcal G_i^{\text{bot}})\right).\label{eq_Haah_FTH_eta_bot_2}
	\end{gather}
	
	From $\eta^{\text{bot}}$ we see
	\[
	\imop(\eta^{\rm bot})=(\bF)\mathcal G_1^{\rm bot}\oplus(F)\mathcal G_2^{\rm bot},\qquad\coker(\eta^{\text{bot}})
	=
	R_\partial/(\bar F)\mathcal G_1^{\rm bot}\oplus R_\partial/(F)\mathcal G_2^{\rm bot}.
	\]
	The elements of $\coker(\eta^{\rm bot})$ are the boundary boundary gauge syndrome equivalence classes up to the syndromes local bottom boundary gauge operators.
	We use the following notation for such equivalence classes: $\forall \mathsf b\in E^{\rm bot}$,
	\begin{equation}
		[\mathsf b]_{\rm bot}:=\mathsf b+\imop(\eta^{\rm bot})\in\coker(\eta^{\rm bot})\,,
	\end{equation}
	which has the meaning of bottom boundary topological excitation in our context.
	Precisely speaking, the point-like $\square$ boundary topological excitations are those that can be individually created by possibly infinite support Pauli operators, they form the torsion submodule of $\coker\eta^\square$.
	The torsion submodule of $\coker\eta^\square$ is defined as 
	\begin{equation}
		\torsion\coker\eta^\square:=\{\mathsf b\in\coker\eta^\square:\exists(r\neq 0)\in R_\partial,\ \text{s.t. }r\mathsf b=0\},
	\end{equation}
	which physically means the point-like boundary topological excitation $\mathsf b$ becomes annihilatable by $\square$ boundary gauge operators after multiplying a proper non-zero finite polynomial $r\in R_\partial$.
	Indeed, the boundary topological excitations $e_X^{\text{bot}}$ and $e_Z^{\text{bot}}$ are individually creatable by infinite support Pauli operators [which we illustrate soon, see words around Eqs.~(\ref{eq_def_infinite_support_R},\ref{eq_infinite_to_finite_syndrome})].
	Moreover, we can straightforwardly check that the torsion submodule of $\coker\eta^{\text{bot}}$ equals $\coker\eta^{\text{bot}}$ itself: For any $\mathsf b=u\mathcal G_1^{\text{bot}}+v\mathcal G_2^{\text{bot}}\in E^{\text{bot}}$,
	\begin{align}
		F\bF[\mathsf b]_{\rm bot}&=F\bF\left[u\mathcal G_1^{\rm bot}+v\mathcal G_2^{\rm bot}\right]_{\rm bot}=F\bF u\mathcal G_1^{\rm bot} + (\bF)\mathcal G_1^{\rm bot} + F\bF v\mathcal G_2^{\rm bot} + (F)\mathcal G_2^{\rm bot}\notag\\
		&=(\bF)\mathcal G_1^{\rm bot} + (F)\mathcal G_2^{\rm bot} = \imop(\eta^{\rm bot}) = [0]_{\rm bot}\,,
	\end{align}
	so
	\begin{equation}
		\torsion\ \coker(\eta^{\text{bot}}) = \coker(\eta^{\text{bot}})\,.
	\end{equation}
	The two natural generators of $\torsion\coker(\eta^{\text{bot}})$ are
	\begin{equation}\label{eq_Haah_code_FTH_21_app}
		e_X^{\text{bot}}:=[\mathcal G_1^{\text{bot}}]_{\text{bot}},
		\qquad
		e_Z^{\text{bot}}:=[\mathcal G_2^{\text{bot}}]_{\text{bot}}.
	\end{equation}
	Their annihilators are
	\[
	\operatorname{ann}(e_X^{\text{bot}})=(\bar F),
	\qquad
	\operatorname{ann}(e_Z^{\text{bot}})=(F).
	\]

	\textit{Top boundary.}
	Similarly, use the following ordered basis labels for both $G^{\text{top}}$ and $E^{\text{top}}$:
	\begin{equation}
		\mathcal G_1^{\text{top}}:=\left(\begin{array}{c}
			1\\0
		\end{array}\right),\qquad
		\mathcal G_2^{\text{top}}:=\left(\begin{array}{c}
			0\\1
		\end{array}\right)\,,
	\end{equation}
	Again, the two copies use the same coordinate labels but are distinct modules.
	Under $\sigma^{\text{top}}$, the labels in $G^{\text{top}}$ map to the top truncated stabilizers defined in Eq.~(\ref{def_boundary_gauge_top}).
	These representatives only involve Pauli operators in $P_{L_z-1}$, so $P^{\text{top}}=P_{L_z-1}=\text{span}_{R_\partial}\{z^{L_z-2}X_1,z^{L_z-2}X_2,z^{L_z-2}Z_1,z^{L_z-2}Z_2\}$, and we can write $\sigma^{\text{top}}$ as a $4\times 2$ matrix:
	\begin{equation}
		\sigma^{\text{top}} = \bx\by\left(\begin{array}{cc}
			B&0\\A&0\\0&xy\\0&D
		\end{array}\right)\,,
	\end{equation}
	where the basis of $P^{\text{top}}$ is ordered as $(z^{L_z-2}X_1,z^{L_z-2}X_2,z^{L_z-2}Z_1,z^{L_z-2}Z_2)^T$.
	
	Then, analogous to the bulk theory, $\epsilon^{\text{top}}$ can be calculated as
	\begin{equation}\label{eq_Haah_code_FTH_17}
		\epsilon^{\text{top}}=(\sigma^{\text{top}})^\dagger\lambda_2
		= xy\left(\begin{array}{cccc}
			0&0&\bar B&\bar A\\
			\bx\by&\bar D&0&0
		\end{array}\right)=\left(\begin{array}{cccc}
			0&0&B&xy\bar A\\1&D&0&0
		\end{array}\right)\,,
	\end{equation}
	where we have used $xy\bar D=D$, $xy\bar B=B$.
	Eq.~(\ref{eq_Haah_code_FTH_17}) means
	\begin{gather}
		\epsilon^{\text{top}}(z^{L_z-2}X_1)=\mathcal G_2^{\text{top}},
		\qquad
		\epsilon^{\text{top}}(z^{L_z-2}X_2)=D\,\mathcal G_2^{\text{top}},\label{eq_Haah_code_FTH_18}
		\\
		\epsilon^{\text{top}}(z^{L_z-2}Z_1)=B\,\mathcal G_1^{\text{top}},
		\qquad
		\epsilon^{\text{top}}(z^{L_z-2}Z_2)=xy\bar A\,\mathcal G_1^{\text{top}}\,.\label{eq_Haah_code_FTH_19}
	\end{gather}
	Next, we can write
	\begin{equation}
		\eta^{\text{top}} = \epsilon^{\text{top}} \sigma^{\text{top}}
		= \left(\begin{array}{cc}
			0&B+xy\bar A\bar D\\
			\bar B+\bx\by AD&0
		\end{array}\right)\,.
	\end{equation}
	Notice that
	\begin{equation}
		B+xy\bar A\bar D = xy\bar F\,,
	\end{equation}
	where $F=1+x+x^2+y+y^2+xy$ is defined in Eq.~(\ref{eq_Haah_code_FTH_16}).
	Then we can write
	\begin{equation}
		\eta^{\text{top}}
		= \left(\begin{array}{cc}
			0&xy\bar F\\
			\bx\by F&0
		\end{array}\right)\,,
	\end{equation}
	which means
	\begin{equation*}
		\eta^{\text{top}}\mathcal G_1^{\text{top}} = \bx\by F\mathcal G_2^{\text{top}},\qquad \eta^{\text{top}}\mathcal G_2^{\text{top}} = xy\bF\mathcal G_1^{\text{top}}.
	\end{equation*}
	The same argument as for the bottom boundary [see Eqs.~(\ref{eq_Haah_FTH_eta_bot_1},\ref{eq_Haah_FTH_eta_bot_2})] gives
	\begin{equation}
		(\eta^{\text{top}})_{ji} = \Omega\left(\sigma^{\text{top}}(\mathcal G_j^{\text{top}}),\sigma^{\text{top}}(\mathcal G_i^{\text{top}})\right).
	\end{equation}
	From $\eta^{\text{top}}$ we see
	\begin{equation}\label{eq_Haah_FTH_coker_eta_top}
		\imop(\eta^{\rm top})=(\bF)\mathcal G_1^{\rm top}\oplus(F)\mathcal G_2^{\rm top},\qquad\coker(\eta^{\text{top}})=
		R_\partial/(\bar F)\mathcal G_1^{\rm top}\oplus R_\partial/(F)\mathcal G_2^{\rm top}\,.
	\end{equation}
	The elements of $\coker(\eta^{\rm top})$ are equivalence classes up to syndromes of local top boundary gauge operators, denoted as
	\begin{equation}
		[\mathsf b]_{\rm top} := \mathsf b + \imop(\eta^{\rm top}) \in \coker(\eta^{\rm top})\,,\quad \forall \mathsf b\in E^{\rm top}\,.
	\end{equation}
	For any $\mathsf b=u\mathcal G_1^{\text{top}}+v\mathcal G_2^{\text{top}}\in E^{\text{top}}$,
	\begin{equation}
		\bar FF[\mathsf b]_{\text{top}}=\imop(\eta^{\rm top})=[0]_{\text{top}}\,,
	\end{equation}
	so
	\begin{equation}
		\torsion\ \coker(\eta^{\text{top}}) = \coker(\eta^{\text{top}})\,.
	\end{equation}
	The two natural generators of $\torsion\,\coker(\eta^{\text{top}})$ are
	\begin{equation}\label{eq_Haah_FTH_top_excitation}
		e_X^{\text{top}}:=[\mathcal G_1^{\text{top}}]_{\text{top}},
		\qquad
		e_Z^{\text{top}}:=[\mathcal G_2^{\text{top}}]_{\text{top}}.
	\end{equation}
	Their annihilators are
	\begin{equation}\label{eq_Haah_FTH_top_excitation_ann}
		\operatorname{ann}(e_X^{\text{top}})=(\bar F),
		\qquad
		\operatorname{ann}(e_Z^{\text{top}})=(F)\,.
	\end{equation}
	As the discussion of $e_X^{\text{bot}}$ and $e_Z^{\text{bot}}$, we can use Freshman's Dream to construct single $e_X^{\text{top}}$ and $e_Z^{\text{top}}$ creators, so $e_X^{\text{top}}$ and $e_Z^{\text{top}}$ are point-like boundary topological excitations.

	\subsection{Point-like excitations via infinite support operators}\label{appendix_point_like_excitations}\label{app_Haah_infinite_support}
	
	Now we show single $e_X^{\text{bot}}$ and single $e_Z^{\text{bot}}$ can be created by infinite support Pauli operators\footnote{By definition, the support of an operator $O$ is the set of all $(i,j)\in\mathbb Z^2$ where $O$ acts nontrivially on.
	An operator $O$ is called finite support if $\exists r\in\mathbb N$, s.t.
	$\supp(O)\subset\big\{(i,j)\mid -r\leq i,j\leq r\big\}.$ Conversely, an operator $O$ is called infinite support if it is not finite support.}, so they are indeed point-like bottom boundary topological excitations.
	Denote the $R_\partial$-module of possibly infinite support series
	\begin{equation}\label{eq_def_infinite_support_R}
		\widehat R_\partial:=\mathbb F_2[[x^{\pm1},y^{\pm1}]],
	\end{equation}
	so that the possibly infinite support Pauli operators form the $R_\partial$-module $\hat P\cong \widehat R_\partial^{4L_z}$.
	According to Eq.~(\ref{eq_Haah_code_FTH_20}), the syndrome caused by applying $xy\mathcal G_1^{\text{bot}}$ is $F\mathcal G_2^{\text{bot}}$.
	Using Freshman's Dream (i.e. $f(x,y)^2=f(x^2,y^2)$, which is true over $\mathbb F_2$), the syndrome caused by applying $xyF\mathcal G_1^{\text{bot}}$ is
	$$
	F^2\mathcal G_2^{\text{bot}} = \left(1+x^2+x^4+y^2+y^4+x^2y^2\right)\mathcal G_2^{\text{bot}}.
	$$
	The syndrome caused by applying $xyF\mathcal G_1^{\text{bot}}$ is the two times scaling of the syndrome caused by applying $xy\mathcal G_1^{\text{bot}}$.
	By repeatedly using Freshman's Dream, this scaling can be iteratively done, implying the syndrome caused by applying $xyF^{2^n-1}\mathcal G_1^{\text{bot}}$ is
	\begin{align}
		F^{2^n}\mathcal G_2^{\text{bot}} &= \left(1+x^2+x^4+y^2+y^4+x^2y^2\right)^{2^{n-1}}\mathcal G_2^{\text{bot}}\notag \\&= \cdots = \left(1+x^{2^n}+x^{2^{n+1}}+y^{2^n}+y^{2^{n+1}}+x^{2^n}y^{2^n}\right)\mathcal G_2^{\text{bot}}.\label{eq_infinite_to_finite_syndrome}
	\end{align}
	Under the limit $n\to\infty$ in the coefficientwise sense of Appendix~\ref{appendix_coefficientwise_limit}, this becomes an infinite support operator $\lim_{n\to\infty}xyF^{2^n-1}\mathcal G_1^{\text{bot}}$, with the finite support syndrome $\mathcal G_2^{\text{bot}}$, since $x^{2^n}\to 0$ when $n\to\infty$\footnote{$(x^{2^n},0)$ leaves any fixed finite subset of $\mathbb Z^2$ when $n\to\infty$.}, and so are other non-identity coefficients in Eq.~(\ref{eq_infinite_to_finite_syndrome}).
	$\mathcal G_2^{\text{bot}}$ is a representative of $e_Z^{\text{bot}}$.
	So indeed, a single $e_Z^{\text{bot}}$ can be created by an infinite-support bottom boundary gauge operator $\lim_{n\to\infty}xyF^{2^n-1}\mathcal G_1^{\text{bot}}$, $e_Z^{\text{bot}}$ is a point-like bottom boundary topological excitation.
	Similarly, a single $e_X^{\text{bot}}$ can be created by an infinite-support bottom boundary gauge operator $\lim_{n\to\infty}\bx\by\bF^{2^n-1}\mathcal G_2^{\text{bot}}$, so $e_X^{\text{bot}}$ is a point-like bottom boundary topological excitation.
	The top boundary topological excitations $e_Z^{\text{top}},e_X^{\text{top}}$ can be similarly created by infinite-support top boundary gauge operators.

	\subsection{Coefficientwise limit of infinite-support operators}
	\label{appendix_coefficientwise_limit}
	
	An element of \(\widehat R_\partial=\mathbb F_2[[x^{\pm1},y^{\pm1}]]\) is a formal sum
	\[
	f=\sum_{(i,j)\in\mathbb Z^2} f_{ij}x^iy^j,\qquad f_{ij}\in\mathbb F_2,
	\]
	where the number of non-zero $f_{ij}$ could be infinite.
	We write \([x^iy^j]f=f_{ij}\) for any $f\in\widehat R_\partial$.
	A sequence \(f^{(n)}\in\widehat R_\partial\) converges coefficientwise to \(f\in\widehat R_\partial\) if, for every fixed \((i,j)\),
	\[
	[x^iy^j]f^{(n)}=f_{ij}
	\]
	for all sufficiently large \(n\).
	Physically, this means that on every fixed finite region, the operator pattern eventually converges.
	Such a limit is unique: if \(f^{(n)}\) converges coefficientwise to both \(f\) and \(f'\), then for every fixed \((i,j)\), both \([x^iy^j]f\) and \([x^iy^j]f'\) are equal to the same eventual value of \([x^iy^j]f^{(n)}\).
	
	We now justify the particular limit used in the construction of point-like boundary excitations.
	Let
	\[
	g_n:=F^{2^n-1}.
	\]
	Since \(2^n-1=1+2+\cdots+2^{n-1}\), we have
	\[
	g_n=\prod_{k=0}^{n-1}F^{2^k}.
	\]
	Using Freshman's Dream over \(\mathbb F_2\),
	\[
	F^{2^k}
	=
	1+x^{2^k}+x^{2^{k+1}}+y^{2^k}+y^{2^{k+1}}+x^{2^k}y^{2^k}.
	\]
	Fix a monomial \(x^ay^b\).
	If \(a<0\) or \(b<0\), then \([x^ay^b]g_n=0\) for all \(n\), since \(g_n\) contains only nonnegative powers of \(x\) and \(y\).
	For $a,b\geq 0$, compare $g_{n+1}=g_nF^{2^n}$ with $g_n$: every nonidentity term of $F^{2^n}$ raises the power of $x$ or of $y$ by at least $2^n$, so $g_{n+1}$ and $g_n$ differ only at monomials with $a\geq 2^n$ or $b\geq 2^n$.
	Let
	\[
	Q_n:=\{x^ay^b\,:\,0\leq a,b<2^n\}
	\]
	be the $2^n\times 2^n$ square window of monomials with a corner at the origin.
	Then $g_m$ and $g_n$ have the same coefficients on $Q_n$ for every $m\geq n$: on a fixed window, the pattern of $g_n$ stops changing as soon as $2^n$ exceeds the size of the window.
	Hence \([x^ay^b]g_n\) is eventually independent of \(n\), which defines a unique coefficientwise limit
	\[
	g_Z:=\lim_{n\to\infty}F^{2^n-1}\in\widehat R_\partial.
	\]
	
	We can now compute the syndrome of $xyg_Z\mathcal G_1^{\text{bot}}$.
	According to Eq.~(\ref{eq_Haah_code_FTH_20}) it is
	\[
	\bx\by F\, xyg_Z\,\mathcal G_2^{\text{bot}} = Fg_Z\,\mathcal G_2^{\text{bot}},
	\]
	where we used $\bx\by\cdot xy=1$.
	Every entry of $\eta^{\text{bot}}$ is a finite Laurent polynomial, so this expression is well defined even though $g_Z$ has infinite support.
	The same window works for $Fg_n$.
	Since $F=1+x+x^2+y+y^2+xy$, for any $h\in\widehat R_\partial$
	\[
	[x^ay^b](Fh)=\sum_{(p,q)\in S}[x^{a-p}y^{b-q}]h\,,
	\]
	where $S=\{(0,0),(1,0),(2,0),(0,1),(0,2),(1,1)\}$ is the support of $F$: one coefficient of $Fh$ is a sum of six coefficients of $h$, taken at
	\[
	(a,b),\ (a-1,b),\ (a-2,b),\ (a,b-1),\ (a,b-2),\ (a-1,b-1)\,,
	\]
	none of which lies to the right of or above $(a,b)$.
	Now take $h=g_n$ and let $x^ay^b$ lie in $Q_n$.
	Those of the six points with both exponents nonnegative then lie in $Q_n$ as well, where $g_m$ and $g_n$ already have the same coefficients for every $m\geq n$; those with a negative exponent have coefficient $0$ in every $g_n$, since $g_n$ has only nonnegative powers.
	Either way the six terms of the sum are the same for $g_m$ and for $g_n$, and so is the sum: $[x^ay^b](Fg_m)=[x^ay^b](Fg_n)$ for every $m\geq n$.
	Multiplication by $F$ therefore commutes with the coefficientwise limit:
	\[
	Fg_Z=F\lim_{n\to\infty}g_n=\lim_{n\to\infty}Fg_n=\lim_{n\to\infty}F\cdot F^{2^n-1}=\lim_{n\to\infty}F^{2^n}\,.
	\]
	Again by Freshman's Dream,
	\[
	F^{2^n}=1+x^{2^n}+x^{2^{n+1}}+y^{2^n}+y^{2^{n+1}}+x^{2^n}y^{2^n}
	\to 1
	\]
	coefficientwise when $n\to\infty$, since each nonidentity monomial leaves any fixed finite region.
	Both $Fg_Z$ and $1$ are therefore coefficientwise limits of the same sequence $F^{2^n}$, so by the uniqueness of coefficientwise limits established above, $Fg_Z=1$.
	Thus the syndrome of $xyg_Z\mathcal G_1^{\text{bot}}$ is $\mathcal G_2^{\text{bot}}$.
	Consequently \(xyg_Z\mathcal G_1^{\text{bot}}\) is a well-defined infinite-support boundary gauge operator.
	The same argument after applying \((x,y)\mapsto(\bar x,\bar y)\) gives the coefficientwise limit of \(\bF^{2^n-1}\), and shows that \(\bx\by\lim_{n\to\infty}\bF^{2^n-1}\mathcal G_2^{\text{bot}}\) has boundary gauge syndrome \(\mathcal G_1^{\text{bot}}\).
	For the top boundary topological excitations $e_Z^{\text{top}},e_X^{\text{top}}$, it can be similarly shown that they can be individually created by infinite-support top boundary gauge operators.

	\section{Low-energy preserving operators in Haah's code FTH}
	\label{appendix_Haah_low_energy_preserving}

	\subsection{Generators of low-energy-preserving operator module under \texorpdfstring{$(Z)$}{(Z)} top boundary}\label{appendix_generators_of_low_energy_preserving_module_Z_top}
	
	In this Appendix, we derive the generators of low-energy-preserving operator module under $(Z)$ top boundary, namely, $\mathcal O_{(Z)}$.
	We first derive the general form of low-energy-preserving Pauli operators, i.e. elements in $\mathcal S_{(Z)}^\Omega$; then we prove that two low-energy-preserving Pauli operators with the same bottom boundary gauge syndrome differ by a stabilizer; finally, we conclude the generators of $\mathcal O_{(Z)}$.
	
	\textit{Derive all finite-support low-energy-preserving operators.}
	Since Haah's code is CSS, we can separate the generators of $\mathcal S_{(Z)}^\Omega$ into $X$-part and $Z$-part.
	We start from the $X$-part.
	Consider the general finite-support $X$-type operator, and apply low-energy-preserving condition on it.
	
	For a general finite-support $X$-type operator $O_X$ under $(Z)$ top boundary, we separate it into layers:
	\begin{equation}\label{eq_O_X_def_1}
		O_X = \sum_{k=0}^{L_z-1}O_{X,k},
	\end{equation}
	where each layer operator is represented as:
	\begin{equation}\label{eq_O_X_def_2}
		O_{X,k} = z^{k-1}\left(u_kX_1+v_kX_2+w_kxyz\mathcal G_1^{\text{bot}}\right)\in P_k\,,
	\end{equation}
	with $u_k,v_k,w_k \in R_\partial$.
	Eq.~(\ref{eq_O_X_def_2}) is a redundant (thus faithful) representation of $O_X$.
	Under the $(Z)$ top boundary, requiring $O_X$ to violate no stabilizers is equivalent to having no bulk syndrome.
	By analyzing the bulk syndrome layer-by-layer, the cancellation of $E_{k+1/2}$ syndromes between $O_{X,k}$ and $O_{X,k+1}$ gives the no-bulk-syndrome condition:
	\begin{equation}\label{eq_O_X_condition}
		u_k+v_kD = u_{k+1}A+v_{k+1}B+w_{k+1}F,
	\end{equation}
	for all $k=0,1,\cdots,L_z-2$. 
	
	Solving these constraints layer-by-layer gives the general form of finite-support $X$-type Pauli operators with no stabilizer violation.
	The only excitation created by $O_X$ is on the bottom boundary:
	\begin{equation}
		(u_0A+v_0B+w_0F)S_Z^{\text{bot}}\in E^{\text{bot}}.
	\end{equation}
	We now figure out what constraints do the conditions in Eq.~(\ref{eq_O_X_condition}) give to the bottom boundary excitation/syndrome of $O_X$.
	
	Denote the common $E_{k+1/2}$ syndrome of $O_{X,k}$ and $O_{X,k+1}$ as
	\begin{equation}\label{eq_rk_def}
		r_{k+1/2}:=u_k+v_kD=u_{k+1}A+v_{k+1}B+w_{k+1}F,\quad k=0,1,\cdots,L_z-2.
	\end{equation}
	Take $k=L_z-2$, the form of $r_{L_z-3/2}$ directly gives
	\begin{equation}\label{eq_rLz32_ideal}
		r_{L_z-3/2}\in(A,B,F)=(A,F),
	\end{equation}
	where we have used $B=AD+F\Rightarrow(A,B,F)=(A,F)$.
	From Eq.~(\ref{eq_rk_def}) we can write $u_k=r_{k+1/2}+Dv_k$.
	Taking it into $r_{k-1/2}=u_kA+v_kB+w_kF$, we get
	\begin{align}
		r_{k-1/2} &= (r_{k+1/2}+Dv_k)A + v_kB + w_kF = Ar_{k+1/2} + (AD+B)v_k + Fw_k = Ar_{k+1/2} + Ft_k,\label{eq_Haah_code_FTH_35}
	\end{align}
	where
	\begin{equation}
		t_k := v_k+w_k.
	\end{equation}
	Using Eq.~(\ref{eq_Haah_code_FTH_35}) iteratively, we get
	\begin{align}
		r_{1/2} &= Ar_{3/2} + Ft_1 = A(Ar_{5/2}+Ft_2) + Ft_1 = \cdots\notag\\
		&= A^{L_z-2}r_{L_z-3/2} + F(t_1+At_2+\cdots+A^{L_z-3}t_{L_z-2}).\label{eq_r12_iteration}
	\end{align}
	Combining Eqs.~(\ref{eq_rLz32_ideal},\ref{eq_r12_iteration}), we get 
	\begin{equation}
		r_{1/2} \in A^{L_z-2}(A,F) + (F) = (A^{L_z-1},F).
	\end{equation}
	Finally, from Eq.~(\ref{eq_rk_def}) we can write $u_0=r_{1/2}+v_0D$.
	Since $r_{1/2}\in(A^{L_z-1},F)$, the coefficient of the bottom boundary gauge syndrome of $O_X$ is
	\begin{align}
		r_{-1/2} := u_0A+v_0B+w_0F = (r_{1/2}+v_0D)A + v_0B + w_0F = Ar_{1/2} + F(v_0+w_0).
	\end{align}
	Taking $r_{1/2}\in(A^{L_z-1},F)$ into above, we get the constraint of bottom boundary gauge syndrome of $O_X$,
	\begin{equation}\label{eq_Haah_code_FTH_39}
		\epsilon^{\text{bot}}(O_X) = r_{-1/2}S_Z^{\text{bot}} \in (A^{L_z},F)S_Z^{\text{bot}}.
	\end{equation}

	Now we turn to the general $Z$-type low-energy-preserving operator $O_Z$ under $(Z)$ top boundary.
	Again, we separate it into layers:
	\begin{equation}\label{eq_Z_top_OZ_def_1}
		O_Z = \sum_{k=0}^{L_z-1}O_{Z,k},
	\end{equation}
	where each layer is represented as
	\begin{equation}\label{eq_Z_top_OZ_def_2}
		O_{Z,k} = z^{k-1}(u_kZ_1+v_kZ_2)\in P_k,
	\end{equation}
	with $u_k,v_k\in R_\partial$. 
	Denote the $E_{k+1/2}$ syndrome coefficient\footnote{Since $O_{Z,k}$ is $Z$-type, it can only have $S_X$ syndrome, here we mean the $S_X$ syndrome by default.}
	($k=0,1,2,\cdots,L_z-2$) and $E_{k-1/2}$ syndrome coefficient ($k=1,2,\cdots,L_z-2,L_z-1$) of $O_{Z,k}$ as $t_k$ and $s_k$, respectively,
	\begin{equation}
		t_k = u_kB + v_kxy\bA\quad,\quad s_k = u_kD + v_kxy.
	\end{equation}
	Moreover, denote the $E^{\text{top}}$ syndrome coefficient of $O_{Z,L_z-1}$ as $t_{L_z-1}$ and the $E^{\text{bot}}$ syndrome coefficient of $O_{Z,0}$ as $s_0$,
	\begin{equation}
		t_{L_z-1} = u_{L_z-1}B + v_{L_z-1}xy\bA\quad,\quad s_0 = u_0D + v_0xy.
	\end{equation} 
	The no-bulk syndrome condition requires the $E_{k+1/2}$ syndromes of $O_{Z,k}$ and $O_{Z,k+1}$ to cancel, which gives
	\begin{equation}
		t_k = s_{k+1},
	\end{equation}
	for all $k=0,1,\cdots,L_z-2$.
	Furthermore, since $\mathcal G_1^{\text{top}}$ is a generator of $\mathcal S_{(Z)}$, the low-energy-preserving condition then requires $O_Z$ to have no $\mathcal G_1^{\text{top}}$ syndrome, which gives
	\begin{equation}
		t_{L_z-1} = 0.
	\end{equation}
	Note that if $t_k\in(\bF)$, then $s_k\in(\bF)$:
	$xy\bF = \bA D + B$, so $$t_k=u_kB+v_kxy\bA=u_k(\bA D+B)+\bA(u_kD+v_kxy)=u_kxy\bF+\bA s_k,$$ therefore if $t_k\in(xy\bF)=(\bF)$, then $\bA s_k\in(\bF)$, which implies $s_k\in(\bF)$ since $\text{gcd}(\bA,\bF)=1$.
	
	Using this property iteratively, we can get
	\begin{align}
		t_{L_z-1}=0\in(\bF)\quad&\Longrightarrow\quad t_{L_z-2}=s_{L_z-1}\in(\bF)\quad\Longrightarrow\quad t_{L_z-3}=s_{L_z-2}\in(\bF)\notag\\
		\cdots\cdots\quad&\Longrightarrow\quad t_0=s_1\in(\bF)\quad\Longrightarrow\quad s_0\in(\bF).
	\end{align}
	Since $O_{Z,0}$ is the only layer of $O_Z$ which can have bottom boundary gauge syndrome, $s_0$ is the bottom syndrome coefficient of $O_Z$, in other words,
	\begin{equation}
		\epsilon^{\text{bot}}(O_Z) = s_0 S_X^{\text{bot}} \in (\bF)S_X^{\text{bot}}.
	\end{equation}

	\textit{Same bottom gauge syndrome decides low-energy-preserving operators up to stabilizers.}
	\begin{lemma}\label{theorem_Z_top_OX_differ_by_stabilizer}
		Under $(Z)$ top boundary, for two $X$-type low-energy-preserving operators $O_X,O'_X\in\mathcal S_{(Z)}^\Omega$ with the same bottom syndrome, $O_X-O'_X\in\mathcal S_{(Z)}$.
	\end{lemma}
	\textbf{Proof. }$O_X,O'_X$ both have the form of Eqs.~(\ref{eq_O_X_def_1},\ref{eq_O_X_def_2}), so $\delta O_X=O_X-O'_X$ also has the form of Eqs.~(\ref{eq_O_X_def_1},\ref{eq_O_X_def_2}).
	Suppose
	\begin{equation}\label{eq_pos_3}
		\delta O_X=\sum_{k=0}^{L_z-1}\delta O_{X,k} = \sum_{k=0}^{L_z-1}z^{k-1}\left(u_kX_1+v_kX_2+w_kxyz\mathcal G_1^{\text{bot}}\right).
	\end{equation}
	Since $O_X,O'_X$ have the same bottom boundary gauge syndrome, $\delta O_X$ has zero bottom boundary gauge syndrome.
	The bottom syndromes of $\bz X_1,\bz X_2,xy\mathcal G_1^{\text{bot}}$ are $AS_Z^{\text{bot}},BS_Z^{\text{bot}},FS_Z^{\text{bot}}$, respectively, so the zero bottom syndrome condition means
	\begin{align}
		u_0A+v_0B+w_0F=0\quad\Longleftrightarrow\quad u_0A+v_0B+w_0(AD+B)=0\quad\Longleftrightarrow\quad A(Dw_0+u_0)+B(w_0+v_0)=0.
	\end{align}
	Denote
	\begin{equation}
		g_1^{(0)}:=Dw_0+u_0,\quad g_2^{(0)}:=w_0+v_0,
	\end{equation}
	we have
	\begin{equation}
		Ag_1^{(0)}+Bg_2^{(0)}=0\quad\Longleftrightarrow\quad Ag_1^{(0)}=Bg_2^{(0)}\quad\Longrightarrow\quad A|Bg_2^{(0)}\ \ \&\ \ B|Ag_1^{(0)},
	\end{equation}
	where $a|b$ means $a$ divides $b$.
	Since $R_\partial$ is a UFD, and $\text{gcd}(A,B)=1$, we can use Theorem~\ref{theorem_UFD_divide_property} to derive from $A|Bg_2^{(0)}$ and $B|Ag_1^{(0)}$ that
	\begin{equation}
		A|g_2^{(0)}\ \ \&\ \ B|g_1^{(0)}.
	\end{equation}
	Suppose $Af_1=g_2^{(0)}$, $Bf_2=g_1^{(0)}$, taking them back to $Ag_1^{(0)}=Bg_2^{(0)}$, we get $ABf_2=BAf_1$.
	Since $R_\partial$ is a UFD and $AB$ is non-zero, $f_1=f_2$.
	Denote $U_{1/2}:=f_1=f_2$, we have $g_1^{(0)}=BU_{1/2},\ g_2^{(0)}=AU_{1/2}$.
	On the other hand, $g_1^{(0)},g_2^{(0)}$ have a more straightforward physical meaning, i.e. the coefficients of $X_1,X_2$ of $O_X$ in layer $k=0$:
	\begin{equation*}
		u_0X_1+v_0X_2+w_0xyz\mathcal G_1^{\text{bot}} = u_0X_1+v_0X_2+w_0(DX_1+X_2) = g_1^{(0)}X_1+g_2^{(0)}X_2.
	\end{equation*}
	Therefore,
	\begin{equation*}
		\delta O_X:=O_X-O'_X=\bz\left(g_1^{(0)}X_1+g_2^{(0)}X_2\right) + \sum_{k=1}^{L_z-1}\delta O_{X,k} = \bz\left(BX_1+AX_2\right)U_{1/2} + \sum_{k=1}^{L_z-1}\delta O_{X,k}.
	\end{equation*}
	Recall that
	\begin{equation*}
		\mathcal G_1^{\text{bot}} = \bx\by\bz\left(DX_1+X_2\right),\quad S_X=\bx\by\bz\big(z(DX_1+X_2)+(BX_1+AX_2)\big),
	\end{equation*}
	so we have
	\begin{equation*}
		S_X=\bx\by\big(xyz\mathcal G_1^{\text{bot}} + \bz(BX_1+AX_2)\big).
	\end{equation*}
	So,
	\begin{equation*}
		\delta O_X^{(1)} := \delta O_X + xyS_XU_{1/2} = xyz\mathcal G_1^{\text{bot}}U_{1/2} + \sum_{k=1}^{L_z-1}\delta O_{X,k},
	\end{equation*}
	which has no support in the $z=0$ layer.
	
	We then do the same thing to $\delta O_X^{(1)}$ to clear the $z=1$ layer Pauli operators by adding multiples of $S_X$ in the $z=3/2$ layer.
	The process is shown as follows:
	\begin{enumerate}
		\item $\delta O_X^{(1)}$ has no bulk syndrome, so it has zero syndrome in the $k=1/2$ bulk layer.
		Only the $k=1$ layer part of $\delta O_X^{(1)}$ can have $k=1/2$ syndrome, so the $k=1$ layer part of $\delta O_X^{(1)}$ has zero $k=1/2$ syndrome, 
		\begin{align*}
			U_{1/2}F + u_1A + v_1B + w_1F = 0\quad&\Longleftrightarrow\quad u_1A + v_1B + (U_{1/2}+w_1)(AD+B) = 0\\
			&\Longleftrightarrow\quad A(u_1+DU_{1/2}+Dw_1) + B(v_1+U_{1/2}+w_1) = 0.
		\end{align*}
		\item Denote
		\begin{equation*}
			g_1^{(1)} := u_1 + DU_{1/2} + Dw_1,\quad g_2^{(1)} := v_1 + U_{1/2} + w_1,
		\end{equation*}
		we get
		\begin{equation*}
			Ag_1^{(1)} + Bg_2^{(1)} = 0\quad\Longleftrightarrow\quad Ag_1^{(1)}=Bg_2^{(1)}\quad\Longrightarrow\quad A|Bg_2^{(1)}\ \ \&\ \ B|Ag_1^{(1)}.
		\end{equation*}
		Using $\text{gcd}(A,B)=1$ and Theorem~\ref{theorem_UFD_divide_property}, we get $A|g_2^{(1)}$ and $B|g_1^{(1)}$.
		Combining with $Ag_1^{(1)}=Bg_2^{(1)}$, we get
		\begin{equation*}
			g_1^{(1)} = BU_{3/2},\quad g_2^{(1)} = AU_{3/2},
		\end{equation*}
		for some $U_{3/2}\in R_\partial$.
		\item $g_1^{(1)},g_2^{(1)}$ are the $X_1,X_2$ coefficients of $\delta O_X^{(1)}$ in the $z=1$ layer:
		\begin{align*}
			xyz\mathcal G_1^{\text{bot}}U_{1/2} + \delta O_{X,1} =& (DX_1+X_2)U_{1/2} + u_1X_1 + v_1X_2 + w_1(DX_1+X_2)\\
			=&(DU_{1/2}+u_1+w_1D)X_1 + (U_{1/2}+v_1+w_1)X_2 = g_1^{(1)}X_1 + g_2^{(1)}X_2.
		\end{align*}
		Therefore,
		\begin{align*}
			\delta O_X^{(1)} = g_1^{(1)}X_1 + g_2^{(1)}X_2 + \sum_{k=2}^{L_z-1}\delta O_{X,k} = (BX_1+AX_2)U_{3/2} + \sum_{k=2}^{L_z-1}\delta O_{X,k}.
		\end{align*}
		\item Recall that
		\begin{equation*}
			S_X = \bx\by\left(xyz\mathcal G_1^{\text{bot}} + \bz(BX_1+AX_2)\right),
		\end{equation*}
		we have
		\begin{equation*}
			\delta O_X^{(2)} := \delta O_X^{(1)} + xyzS_XU_{3/2} = xyz^2\mathcal G_1^{\text{bot}}U_{3/2} + \sum_{k=2}^{L_z-1}\delta O_{X,k},
		\end{equation*}
		which has no support in the $z=0,1$ layers.
	\end{enumerate}
	This can be iteratively done, until we get
	\begin{equation*}
		\delta O_X^{(L_z-1)} := \delta O_X^{(L_z-2)} + xyz^{L_z-2}S_XU_{L_z-3/2} = xyz^{L_z-1}\mathcal G_1^{\text{bot}}U_{L_z-3/2} + \delta O_{X,L_z-1},
	\end{equation*}
	which has support only in the $z=L_z-1$ top layer. 
	
	The final step is similar, except that we will add multiples of top truncated $S_X$, instead of complete bulk $S_X$.
	Final step:
	\begin{enumerate}
		\item $\delta O_X^{(L_z-1)}$ has no bulk syndrome,
		\begin{align*}
			&U_{L_z-3/2}F + u_{L_z-1}A + v_{L_z-1}B + w_{L_z-1}F = 0\\
			\Longleftrightarrow\quad&A\left(u_{L_z-1}+DU_{L_z-3/2}+Dw_{L_z-1}\right) + B\left(v_{L_z-1}+U_{L_z-3/2}+w_{L_z-1}\right) = 0.
		\end{align*}
		\item Denote
		\begin{equation*}
			g_1^{(L_z-1)} := u_{L_z-1}+DU_{L_z-3/2}+Dw_{L_z-1},\quad g_2^{(L_z-1)} := v_{L_z-1}+U_{L_z-3/2}+w_{L_z-1},
		\end{equation*}
		so that the no bulk syndrome condition of $\delta O_X^{(L_z-1)}$ can be written as
		\begin{equation*}
			Ag_1^{(L_z-1)} + Bg_2^{(L_z-1)} = 0\quad\Longleftrightarrow\quad Ag_1^{(L_z-1)} = Bg_2^{(L_z-1)}\quad\Longrightarrow\quad A|Bg_2^{(L_z-1)}\ \ \&\ \ B|Ag_1^{(L_z-1)}.
		\end{equation*}
		Using $\text{gcd}(A,B)=1$ and Theorem~\ref{theorem_UFD_divide_property}, we get $A|g_2^{(L_z-1)}$ and $B|g_1^{(L_z-1)}$.
		Combining with $Ag_1^{(L_z-1)}=Bg_2^{(L_z-1)}$, we get
		\begin{equation*}
			g_1^{(L_z-1)}=BU_{L_z-1/2},\quad g_2^{(L_z-1)}=AU_{L_z-1/2},
		\end{equation*}
		for some $U_{L_z-1/2}\in R_\partial$.
		\item $g_1^{(L_z-1)},g_2^{(L_z-1)}$ are the $X_1,X_2$ coefficients of $\delta O_X^{(L_z-1)}$ in the $k=L_z-1$ top layer.
		Direct calculation gives
		\begin{align*}
			\delta O_X^{(L_z-1)} &= xyz^{L_z-1}\mathcal G_1^{\text{bot}}U_{L_z-3/2} + \delta O_{X,L_z-1} \\&= z^{L_z-2}(DX_1+X_2)U_{L_z-3/2} + u_{L_z-1}z^{L_z-2}X_1 + v_{L_z-1}z^{L_z-2}X_2 + w_{L_z-1}z^{L_z-2}(DX_1+X_2)\\
			&= g_1^{(L_z-1)}z^{L_z-2}X_1 + g_2^{(L_z-1)}z^{L_z-2}X_2 = z^{L_z-2}(BX_1+AX_2)U_{L_z-1/2}.
		\end{align*}
		\item Recall that
		\begin{equation*}
			\mathcal G_1^{\text{top}} = \bx\by z^{L_z-2}(BX_1+AX_2).
		\end{equation*}
		$\mathcal G_1^{\text{top}}$ is a stabilizer on the $(Z)$ top boundary.
		\begin{equation*}
			\delta O_X^{(L_z)} := \delta O_X^{(L_z-1)} + xy\mathcal G_1^{\text{top}}U_{L_z-1/2} = 0.
		\end{equation*}
	\end{enumerate}
	In all,
	\begin{equation}\label{eq_pos_4}
		\delta O_X^{(L_z)} = \delta O_X + \sum_{k=0}^{L_z-2}xyz^k S_X U_{k+1/2} + xy\mathcal G_1^{\text{top}}U_{L_z-1/2} = 0,
	\end{equation}
	$\delta O_X$ is a stabilizer, $O_X$ and $O'_X$ differ by a stabilizer.
	\textbf{QED. }
	
	The symbols $u_k$, $v_k$, $r_{k+1/2}$, $t_k$, $g_1^{(k)}$, $g_2^{(k)}$, and $U_{k+1/2}$ are local to the above proof and will be redefined below.

	\begin{lemma}\label{theorem_Z_top_OZ_differ_by_stabilizer}
		Under $(Z)$ top boundary, for two $Z$-type low-energy-preserving operators $O_Z,O'_Z\in\mathcal S_{(Z)}^\Omega$ with the same bottom syndrome, $O_Z-O'_Z\in\mathcal S_B\subset\mathcal S_{(Z)}$.
	\end{lemma}
	\textbf{Proof:} $O_Z,O'_Z$ both have the form of Eqs.~(\ref{eq_Z_top_OZ_def_1},\ref{eq_Z_top_OZ_def_2}), so $\delta O_Z=O_Z-O'_Z$ also has the form of Eqs.~(\ref{eq_Z_top_OZ_def_1},\ref{eq_Z_top_OZ_def_2}).
	Suppose
	\begin{equation}
		\delta O_Z = \sum_{k=0}^{L_z-1}\delta O_{Z,k} = \sum_{k=0}^{L_z-1}z^{k-1}\left(u_kZ_1+v_kZ_2\right).
	\end{equation}
	Since $O_Z,O'_Z$ have the same bottom syndrome, $\delta O_Z$ has zero bottom syndrome.
	The bottom syndrome of $\bz Z_1,\bz Z_2$ are $DS_X^{\text{bot}}$, $xyS_X^{\text{bot}}$, respectively, so the zero bottom syndrome condition means
	\begin{equation}
		u_0D+v_0xy=0.
	\end{equation}
	Denote
	\begin{equation}\label{eq_position_1}
		g_1^{(0)}:=u_0,\quad g_2^{(0)}:=v_0,
	\end{equation}
	then
	\begin{equation}
		Dg_1^{(0)}+xyg_2^{(0)}=0\quad\Longleftrightarrow\quad Dg_1^{(0)}=xyg_2^{(0)}\quad\Longrightarrow\quad D|g_2^{(0)}.
	\end{equation}  
	Suppose $g_2^{(0)}=DU_{1/2}$, taking it back to $Dg_1^{(0)}=xyg_2^{(0)}$, we get $Dg_1^{(0)}=xyDU_{1/2}$, which implies $g_1^{(0)}=xyU_{1/2}$ since $R_\partial$ is a UFD.
	Hence the bottom layer ($z=0$) part of $\delta O_Z$ can be written as
	\begin{equation}
		\delta O_{Z,0} = \bar z(xyZ_1+DZ_2)U_{1/2}.
	\end{equation}
	Notice that
	\begin{equation*}
		xyS_Z = \bar z(xyZ_1+DZ_2)+xy\bar A Z_1+BZ_2,
	\end{equation*}
	we can write
	\begin{equation}
		\delta O_Z^{(1)} := \delta O_Z + xyS_ZU_{1/2} = (xy\bA Z_1+BZ_2)U_{1/2} + \sum_{k=1}^{L_z-1}\delta O_{Z,k},
	\end{equation}
	which does not have support in the $k=0$ layer.
	Next, we add multiples of $S_Z$ in the $k=3/2$ layer to cancel the $k=1$ part of $\delta O_Z^{(1)}$ as follows:
	\begin{enumerate}
		\item $\delta O_Z^{(1)}$ has no bulk syndrome, so it has zero syndrome in the $k=1/2$ bulk layer.
		Only the $k=1$ layer part of $\delta O_Z^{(1)}$ can have $k=1/2$ syndrome, so the $k=1$ part of $\delta O_Z^{(1)}$ has zero $k=1/2$ syndrome,
		\begin{align*}
			D(xy\bA U_{1/2} + u_1) + xy(BU_{1/2}+v_1) = 0.
		\end{align*}
		\item Denote
		\begin{equation*}
			g_1^{(1)} := xy\bA U_{1/2} + u_1,\quad g_2^{(1)} := BU_{1/2}+v_1,
		\end{equation*}
		we get
		\begin{equation*}
			Dg_1^{(1)}+xyg_2^{(1)} = 0\quad\Longleftrightarrow\quad Dg_1^{(1)}=xyg_2^{(1)}\quad\Longrightarrow\quad D|g_2^{(1)}.
		\end{equation*}
		Suppose $g_2^{(1)}=DU_{3/2}$, taking it back to $Dg_1^{(1)}=xyg_2^{(1)}$, we get $Dg_1^{(1)}=xyDU_{3/2}$, which implies $g_1^{(1)}=xyU_{3/2}$ since $R_\partial$ is a UFD.
		\item $g_1^{(1)},g_2^{(1)}$ are the $Z_1,Z_2$ coefficients of $\delta O_Z^{(1)}$ in the $z=1$ layer.
		Therefore,
		\begin{equation*}
			\delta O_Z^{(1)} = g_1^{(1)}Z_1 + g_2^{(1)}Z_2 + \sum_{k=2}^{L_z-1}\delta O_{Z,k} = (xyZ_1+DZ_2)U_{3/2} + \sum_{k=2}^{L_z-1}\delta O_{Z,k}.
		\end{equation*}
		\item Using
		\begin{equation*}
			xyS_Z = \bar z(xyZ_1+DZ_2)+xy\bA Z_1+BZ_2,
		\end{equation*}
		we can write
		\begin{equation*}
			\delta O_Z^{(2)} := \delta O_Z^{(1)} + xyzS_ZU_{3/2} = z(xy\bA Z_1 + BZ_2)U_{3/2} + \sum_{k=2}^{L_z-1}\delta O_{Z,k},
		\end{equation*}
		which has no support in the $z=0,1$ layers.
	\end{enumerate}
	This can be iteratively done, until we get
	\begin{equation*}
		\delta O_Z^{(L_z-1)} := \delta O_Z^{(L_z-2)} + xyz^{L_z-2}S_ZU_{L_z-3/2} = z^{L_z-2}(xy\bA Z_1+BZ_2)U_{L_z-3/2} + \delta O_{Z,L_z-1},
	\end{equation*}
	which has support only in the $z=L_z-1$ top layer.
	
	The final step is similar, except that we will add multiples of top truncated $S_Z$, instead of complete bulk $S_Z$.
	Final step:
	\begin{enumerate}
		\item $\delta O_Z^{(L_z-1)}$ has no bulk syndrome in the $z=L_z-3/2$ layer, so
		\begin{align*}
			D(xy\bA U_{L_z-3/2} + u_{L_z-1}) + xy(BU_{L_z-3/2}+v_{L_z-1}) = 0.
		\end{align*}
		\item Denote
		\begin{equation*}
			g_1^{(L_z-1)} := xy\bA U_{L_z-3/2} + u_{L_z-1},\quad g_2^{(L_z-1)} := BU_{L_z-3/2}+v_{L_z-1},
		\end{equation*}
		we get
		\begin{equation*}
			Dg_1^{(L_z-1)}+xyg_2^{(L_z-1)} = 0\quad\Longleftrightarrow\quad Dg_1^{(L_z-1)}=xyg_2^{(L_z-1)}\quad\Longrightarrow\quad D|g_2^{(L_z-1)}.
		\end{equation*}
		Suppose $g_2^{(L_z-1)}=DU_{L_z-1/2}$, taking it back to $Dg_1^{(L_z-1)}=xyg_2^{(L_z-1)}$, we get $Dg_1^{(L_z-1)}=xyDU_{L_z-1/2}$, which implies $g_1^{(L_z-1)}=xyU_{L_z-1/2}$ since $R_\partial$ is a UFD.
		\item $g_1^{(L_z-1)},g_2^{(L_z-1)}$ are the $Z_1,Z_2$ coefficients of $\delta O_Z^{(L_z-1)}$ in the $z=L_z-1$ layer.
		Therefore,
		\begin{equation*}
			\delta O_Z^{(L_z-1)} = z^{L_z-2}(g_1^{(L_z-1)}Z_1 + g_2^{(L_z-1)}Z_2) = z^{L_z-2}(xyZ_1+DZ_2)U_{L_z-1/2}.
		\end{equation*}
		\item Recall that
		\begin{equation*}
			xy\mathcal G_2^{\text{top}} = z^{L_z-2}(xyZ_1+DZ_2),
		\end{equation*}
		so we can write
		\begin{equation*}
			\delta O_Z^{(L_z)} := \delta O_Z^{(L_z-1)} + xy\mathcal G_2^{\text{top}}U_{L_z-1/2} = 0.
		\end{equation*}
	\end{enumerate}
	In all,
	\begin{equation}\label{eq_position_2}
		\delta O_Z^{(L_z)} = \delta O_Z + \sum_{k=0}^{L_z-2}xyz^k S_Z U_{k+1/2} + xy\mathcal G_2^{\text{top}}U_{L_z-1/2} = 0.
	\end{equation}
	In other words
	\begin{equation}
		\delta O_Z = \sum_{k=0}^{L_z-2}xyz^k S_Z U_{k+1/2} + xy\mathcal G_2^{\text{top}}U_{L_z-1/2}.
	\end{equation}
	Since $\mathcal G_1^{\text{top}}$ is in the generator set of $\mathcal S_{(Z)}$, the low-energy-preserving operator $\delta O_Z$ must have no $\mathcal G_1^{\text{top}}$ syndrome.
	Thus
	\begin{equation}
		0=\epsilon^{\text{top}}\left(xy\mathcal G_2^{\text{top}}U_{L_z-1/2}\right)
		=x^2y^2\bF\,U_{L_z-1/2}\mathcal G_1^{\text{top}}.
	\end{equation}
	Since $R_\partial$ is a UFD and $x^2y^2\bF\neq0$, this implies $U_{L_z-1/2}=0$, thus $\delta O_Z$ is a bulk stabilizer.
	\text{QED.}

	\begin{theorem}\label{theorem_Z_top_same_bottom_syndrome_decide_up_to_stabilizer}
		Under $(Z)$ top boundary, for two low-energy-preserving Pauli operators $O,O'\in\mathcal S_{(Z)}^\Omega$ with the same bottom syndrome, $O-O'\in\mathcal S_{(Z)}$.
	\end{theorem}
	\textbf{Proof.}
	Any low-energy-preserving Pauli operator $O$ can be decomposed to the product of $X$-part and $Z$-part.
	Suppose $O=O_X+O_Z$, $O'=O'_X+O'_Z$, where $O_X,O'_X$ are $X$-type and $O_Z,O'_Z$ are $Z$-type.
	The $\mathcal G_1^{\text{bot}}$ syndromes of $O,O'$, which are decided by $O_Z,O'_Z$, are independent of the $\mathcal G_2^{\text{bot}}$ syndromes of $O,O'$, which are decided by $O_X,O'_X$.
	So $O,O'$ have the same bottom syndrome implies $O_Z,O'_Z$ have the same bottom syndrome, and $O_X,O'_X$ have the same bottom syndrome.
	Then, according to Lemmas~\ref{theorem_Z_top_OX_differ_by_stabilizer},\ref{theorem_Z_top_OZ_differ_by_stabilizer}, $O_X-O'_X\in\mathcal S_{(Z)}$, $O_Z-O'_Z\in\mathcal S_{(Z)}$.
	Therefore, $O-O'=O_Z-O'_Z+O_X-O'_X\in\mathcal S_{(Z)}$.

	\begin{corollary}\label{corollary_no_logical_Z_top}
		Theorem~\ref{theorem_Z_top_same_bottom_syndrome_decide_up_to_stabilizer} implies that there is no nontrivial finite-support Pauli logical operator under $(Z)$ top boundary, $\mathcal L_{(Z)}=\{[0]_{(Z)}\}$.
	\end{corollary}

	\textit{Conclusion.}
	We have derived before that under $(Z)$ top boundary, the bottom syndrome of $X$-type low-energy-preserving operator $O_X$ lies in $(A^{L_z},F)\mathcal G_2^{\text{bot}}$, and the bottom syndrome of $Z$-type low-energy-preserving operator $O_Z$ lies in $(\bF)\mathcal G_1^{\text{bot}}$.
	Now, we conclude with three canonical generators of low-energy-preserving operator equivalence classes under $(Z)$ top boundary.
	
	Note that the $e_Z$-transport operator
	\begin{equation*}
		\mathcal T_{e_Z,z} = \sum_{k=0}^{L_z-1}A^{L_z-1-k}z^{k-1}X_1
	\end{equation*}
	is finite support, $X$-type, low-energy-preserving, and has the bottom syndrome $A^{L_z}\mathcal G_2^{\text{bot}}$.
	The bottom boundary gauge operator $xy\mathcal G_1^{\text{bot}}$ is finite support, $X$-type, low-energy-preserving, and has the bottom syndrome $F\mathcal G_2^{\text{bot}}$.
	Therefore, by Lemma~\ref{theorem_Z_top_OX_differ_by_stabilizer}, any finite support $X$-type low-energy-preserving operator $O_X$ with the bottom boundary gauge syndrome $(uA^{L_z}+vF)\mathcal G_2^{\text{bot}}$ differs from $u\mathcal T_{e_Z,z}+vxy\mathcal G_1^{\text{bot}}$ by a stabilizer only.
	The $X$-type low-energy-preserving operator equivalence classes are generated by $[\mathcal T_{e_Z,z}]_{(Z)}$ and $[\mathcal G_1^{\text{bot}}]_{(Z)}$, with the equivalence relation being differing by a stabilizer in $\mathcal S_{(Z)}$.
	
	On the other hand, the bottom boundary gauge operator $\bx\by\mathcal G_2^{\text{bot}}$ is finite-support, $Z$-type, low-energy-preserving, and has the bottom syndrome $\bF\mathcal G_1^{\text{bot}}$, so by Lemma~\ref{theorem_Z_top_OZ_differ_by_stabilizer}, any finite support $Z$-type low-energy-preserving operator $O_Z$ with the bottom syndrome $u\bF\mathcal G_1^{\text{bot}}$ differs from $u\bx\by\mathcal G_2^{\text{bot}}$ by a stabilizer only.
	The $Z$-type low-energy-preserving operator equivalence classes are generated by $[\mathcal G_2^{\text{bot}}]_{(Z)}$.
	
	In all, we conclude that the low-energy-preserving Pauli operator equivalence class module $\mathcal O_{(Z)}=\mathcal S_{(Z)}^\Omega/\mathcal S_{(Z)}$ is generated by $[\mathcal G_1^{\text{bot}}]_{(Z)},[\mathcal T_{e_Z,z}]_{(Z)},[\mathcal G_2^{\text{bot}}]_{(Z)}$.
	The three generators all have nontrivial bottom boundary gauge syndrome, so they are not logical operators, which is consistent with Corollary~\ref{corollary_no_logical_Z_top}.

	\subsection{Generators of low-energy-preserving operator module under \texorpdfstring{$(X)$}{(X)} top boundary}\label{appendix_generators_of_low_energy_preserving_operators_X_top}
	
	In this Appendix, we derive the generators of low-energy-preserving operator module under $(X)$ top boundary, namely, $\mathcal O_{(X)}$.
	We first derive the general form of low-energy-preserving Pauli operators, i.e. elements in $\mathcal S_{(X)}^\Omega$; then we prove that two low-energy-preserving operators with the same bottom boundary gauge syndrome differ by a stabilizer only; finally, we conclude the generators of $\mathcal O_{(X)}$.
	
	\textit{Derive all finite-support low-energy-preserving operators.}
	Like for the $(Z)$ top boundary case, we separate the generators of $\mathcal S_{(X)}^\Omega$ into $X$-part and $Z$-part.
	We first discuss the $Z$-part by considering the general finite-support $Z$-type operator and applying low-energy-preserving condition on it.
	
	For a general finite-support $Z$-type operator $O_Z$, we separate it into layers:
	\begin{equation}\label{eq_O_Z_def_1}
		O_Z = \sum_{k=0}^{L_z-1}O_{Z,k},
	\end{equation}
	where each layer operator is:
	\begin{equation}\label{eq_O_Z_def_2}
		O_{Z,k} = z^{k-1}(u_kZ_1+v_kZ_2)\in P_k,
	\end{equation}
	with $u_k,v_k \in R_\partial$.
	The no-bulk-syndrome condition requires the $E_{k+1/2}$ syndromes of $O_{Z,k}$ and $O_{Z,k+1}$ to cancel, which gives:
	\begin{equation}\label{eq_Haah_code_O_Z_no_bulk_syndrome_condition}
		u_kB+v_kxy\bar A = u_{k+1}D+v_{k+1}xy,
	\end{equation}
	for all $k=0,1,\cdots,L_z-2$. 
	
	Solving these constraints layer-by-layer gives the general form of finite-support $Z$-type Pauli operators with no stabilizer violation.
	The only excitation created by $O_Z$ is on the bottom boundary:
	\begin{equation}
		(u_0D+v_0xy)S_X^{\text{bot}}\in E^{\text{bot}}.
	\end{equation}
	We now figure out what constraints do the conditions in Eq.~(\ref{eq_Haah_code_O_Z_no_bulk_syndrome_condition}) give to the bottom boundary excitation/syndrome of $O_Z$.
	In fact, for any $k=0,1,\cdots,L_z-2$, and any choice of $u_k,v_k,u_{k+1}$, Eq.~(\ref{eq_Haah_code_O_Z_no_bulk_syndrome_condition}) always has a solution:
	\begin{equation}
		v_{k+1} = \bx\by\left(u_kB+ v_kxy\bA+u_{k+1}D\right).
	\end{equation}
	This is because $\epsilon_B(z^{k-1}Z_2)$ has a monomial syndrome configuration $xy$ in the $z=k-1/2$ layer $E_{k-1/2}$, so we can always apply a proper multiple of $z^{k-1}Z_2$ to move all the $S_X$ syndromes in $E_{k-1/2}$ layer above to the $E_{k+1/2}$ layer.
	Consequently, the set of bottom syndromes that operators in a general low-energy-preserving $Z$-type operator $O_Z$ could leave is $R_\partial\mathcal G_1^{\text{bot}}$. 
	
	Then we discuss the $X$-part similarly: first considering the general finite-support $X$-type operator and applying low-energy-preserving condition on it.

	For a general finite-support $X$-type operator $O_X$ under $(X)$ top boundary, we separate it into layers:
	\begin{equation}\label{eq_X_top_OX_def_1}
		O_X=\sum_{k=0}^{L_z-1}O_{X,k},
	\end{equation}
	where each layer is represented as
	\begin{equation}\label{eq_X_top_OX_def_2}
		O_{X,k}=z^{k-1}(u_kX_1+v_kX_2)\in P_k,
	\end{equation}
	with $u_k,v_k\in R_\partial$.
	Denote the $E_{k+1/2}$ syndrome coefficient\footnote{Since $O_X$ is $X$-type, it can only have $S_Z$ syndrome, here we mean the $S_Z$ syndrome coefficient by default.} of $O_{X,k}$ and the $E_{k-1/2}$ syndrome coefficient of $O_{X,k}$ as $t_k$ and $s_k$, respectively,
	\begin{equation}
		t_k=u_k+v_kD\quad,\quad s_k=u_kA+v_kB.
	\end{equation}
	Here $t_k$ is defined for $k=0,1,\cdots,L_z-2$ as a bulk $E_{k+1/2}$ syndrome coefficient, while $s_k$ is defined for $k=1,2,\cdots,L_z-1$ as a bulk $E_{k-1/2}$ syndrome coefficient.
	Moreover, the $E^{\text{top}}$ syndrome coefficient of $O_{X,L_z-1}$ and the $E^{\text{bot}}$ syndrome coefficient of $O_{X,0}$ are
	\begin{equation}
		t_{L_z-1}=u_{L_z-1}+v_{L_z-1}D\quad,\quad s_0=u_0A+v_0B.
	\end{equation}
	The no-bulk-syndrome condition requires the $E_{k+1/2}$ syndromes of $O_{X,k}$ and $O_{X,k+1}$ to cancel, which gives
	\begin{equation}\label{eq_X_top_OX_no_bulk_syndrome_condition}
		t_k=s_{k+1},
	\end{equation}
	for all $k=0,1,\cdots,L_z-2$.
	Furthermore, since $\mathcal G_2^{\text{top}}$ is a generator of $\mathcal S_{(X)}$, the low-energy-preserving condition requires $O_X$ to have no $\mathcal G_2^{\text{top}}$ syndrome, which gives
	\begin{equation}\label{eq_X_top_OX_no_top_syndrome_condition}
		t_{L_z-1}=0.
	\end{equation}
	
	Note that if $t_k\in(F)$, then $s_k\in(F)$:
	\begin{equation}
		s_k=u_kA+v_kB=A(u_k+v_kD)+v_k(B+AD)=At_k+v_kF,
	\end{equation}
	where we have used $F=AD+B$.
	Therefore $t_k\in(F)$ implies $s_k\in(F)$.
	
	Using this property iteratively, we can get
	\begin{align}
		t_{L_z-1}=0\in(F)\quad&\Longrightarrow\quad s_{L_z-1}\in(F)\quad\Longrightarrow\quad t_{L_z-2}=s_{L_z-1}\in(F)\notag\\
		&\Longrightarrow\quad s_{L_z-2}\in(F)\quad\Longrightarrow\quad t_{L_z-3}=s_{L_z-2}\in(F)\notag\\
		\cdots\cdots\quad&\Longrightarrow\quad t_0=s_1\in(F)\quad\Longrightarrow\quad s_0\in(F).
	\end{align}
	Since $O_{X,0}$ is the only layer of $O_X$ which can have bottom boundary gauge syndrome, $s_0$ is the bottom syndrome coefficient of $O_X$, in other words,
	\begin{equation}\label{eq_X_top_OX_bottom_syndrome_constraint}
		\epsilon^{\text{bot}}(O_X)=s_0S_Z^{\text{bot}}\in(F)S_Z^{\text{bot}}.
	\end{equation}

	\textit{Same bottom gauge syndrome decides low-energy operators up to stabilizers.}
	\begin{lemma}\label{theorem_X_top_OZ_differ_by_stabilizer}
		Under $(X)$ top boundary, for two $Z$-type Pauli operators $O_Z,O'_Z\in\mathcal S_{(X)}^\Omega$ with the same bottom syndrome, $O_Z-O'_Z\in\mathcal S_{(X)}$.
	\end{lemma}
	\textbf{Proof. }$O_Z,O'_Z$ both have the form of Eqs.~(\ref{eq_O_Z_def_1},\ref{eq_O_Z_def_2}), so $\delta O_Z=O_Z-O'_Z$ also has the form of Eqs.~(\ref{eq_O_Z_def_1},\ref{eq_O_Z_def_2}).
	Suppose
	\begin{equation}
		\delta O_Z = \sum_{k=0}^{L_z-1}\delta O_{Z,k} = \sum_{k=0}^{L_z-1}z^{k-1}\left(u_kZ_1+v_kZ_2\right).
	\end{equation}
	Since $O_Z,O'_Z$ have the same bottom syndrome, $\delta O_Z$ has zero bottom syndrome.
	The bottom syndrome of $\bz Z_1,\bz Z_2$ are $DS_X^{\text{bot}}$, $xyS_X^{\text{bot}}$, respectively, so the zero bottom syndrome condition means
	\begin{equation}
		u_0D+v_0xy=0.
	\end{equation}
	Following exactly the same procedure as from Eq.~(\ref{eq_position_1}) to Eq.~(\ref{eq_position_2}), we get
	\begin{equation*}
		\delta O_Z^{(L_z)} = \delta O_Z + \sum_{k=0}^{L_z-2}xyz^k S_Z U_{k+1/2} + xy\mathcal G_2^{\text{top}}U_{L_z-1/2} = 0.
	\end{equation*}
	Since $\mathcal G_2^{\text{top}}\in\mathcal S_{(X)}$, $\delta O_Z$ differs from 0 by a stabilizer.
	$O_Z,O'_Z$ differ by a stabilizer.
	\textbf{QED.}

	\begin{lemma}\label{theorem_X_top_OX_differ_by_stabilizer}
		Under $(X)$ top boundary, for two $X$-type Pauli operators $O_X,O'_X\in\mathcal S_{(X)}^\Omega$ with the same bottom syndrome, $O_X-O'_X\in\mathcal S_B\subset\mathcal S_{(X)}$.
	\end{lemma}
	\textbf{Proof. }Although Eqs.~(\ref{eq_O_X_def_1},\ref{eq_O_X_def_2}) were introduced in the discussion of the $(Z)$ top boundary, the redundant representation itself is only a layer-by-layer representation of $X$-type Pauli operators and does not depend on the choice of top boundary.
	Therefore we use the same representation here.
	Suppose
	\begin{equation}
		\delta O_X:=O_X-O'_X
		=\sum_{k=0}^{L_z-1}\delta O_{X,k}
		=\sum_{k=0}^{L_z-1}z^{k-1}\left(u_kX_1+v_kX_2+w_kxyz\mathcal G_1^{\text{bot}}\right).
	\end{equation}
	Since $O_X,O'_X$ have the same bottom boundary gauge syndrome, $\delta O_X$ has zero bottom boundary gauge syndrome:
	\begin{equation}
		u_0A+v_0B+w_0F=0.
	\end{equation}

	Moreover, since $O_X,O'_X\in\mathcal S_{(X)}^\Omega$, $\delta O_X$ has no bulk syndrome.
	So, following exactly the same procedure as from Eq.~(\ref{eq_pos_3}) to Eq.~(\ref{eq_pos_4}), we get
	\begin{equation*}
		\delta O_X^{(L_z)} := \delta O_X + \sum_{k=0}^{L_z-2}xyz^k S_X U_{k+1/2} + xy\mathcal G_1^{\text{top}}U_{L_z-1/2} = 0,
	\end{equation*}
	or equivalently,
	\begin{equation}\label{eq_delta_OX}
		\delta O_X = \sum_{k=0}^{L_z-2}xyz^k S_X U_{k+1/2} + xy\mathcal G_1^{\text{top}}U_{L_z-1/2}.
	\end{equation}
	Recall that $\delta O_X=O_X-O'_X\in\mathcal S_{(X)}^\Omega$ is low-energy-preserving, and $\mathcal G_2^{\text{top}}$ is a generator of $\mathcal S_{(X)}$, so $\delta O_X$ must have no $\mathcal G_2^{\text{top}}$ syndrome, which implies
	\begin{equation*}
		0=\epsilon^{\text{top}}(xy\mathcal G_1^{\text{top}}U_{L_z-1/2}) = FU_{L_z-1/2}\mathcal G_2^{\text{top}}.
	\end{equation*}
	Since $R_\partial$ is a UFD and $F\neq0$, this implies
	\begin{equation*}
		U_{L_z-1/2}=0,
	\end{equation*}
	so according to Eq.~(\ref{eq_delta_OX}), $\delta O_X$ is a bulk stabilizer.
	\textbf{QED.}

	\begin{theorem}\label{theorem_X_top_same_bottom_syndrome_decides_low_energy_operator_up_to_stabilizer}
		Under $(X)$ top boundary, for two low-energy-preserving Pauli operators $O,O'\in\mathcal S_{(X)}^\Omega$ with the same bottom syndrome, $O-O'\in\mathcal S_{(X)}$.
	\end{theorem}
	\textbf{Proof. }Any low-energy-preserving Pauli operator $O$ can be decomposed to the product of $X$-part and $Z$-part.
	Suppose $O=O_X+O_Z$, $O'=O'_X+O'_Z$, where $O_X,O'_X$ are $X$-type and $O_Z,O'_Z$ are $Z$-type.
	The $\mathcal G_1^{\text{bot}}$ syndromes of $O,O'$, which are decided by $O_Z,O'_Z$, are independent of the $\mathcal G_2^{\text{bot}}$ syndromes of $O,O'$, which are decided by $O_X,O'_X$.
	So $O,O'$ have the same bottom syndrome implies $O_Z,O'_Z$ have the same bottom syndrome, and $O_X,O'_X$ have the same bottom syndrome.
	Then, according to Lemmas~\ref{theorem_X_top_OZ_differ_by_stabilizer},\ref{theorem_X_top_OX_differ_by_stabilizer}, $O_X-O'_X\in\mathcal S_{(X)}$, $O_Z-O'_Z\in\mathcal S_{(X)}$.
	Therefore, $O-O'=O_Z-O'_Z+O_X-O'_X\in\mathcal S_{(X)}$.

	\begin{corollary}\label{corollary_X_top_no_logical}
		Theorem~\ref{theorem_X_top_same_bottom_syndrome_decides_low_energy_operator_up_to_stabilizer} implies that there is no nontrivial finite-support Pauli logical operator under $(X)$ top boundary, $\mathcal L_{(X)}=\{[0]_{(X)}\}$.
	\end{corollary}

	\textit{Conclusion.}
	We have derived before that under $(X)$ top boundary, the bottom syndrome of $Z$-type low-energy-preserving operator $O_Z$ lies in $R_\partial\mathcal G_1^{\text{bot}}$, and the bottom syndrome of $X$-type low-energy-preserving operator $O_X$ lies in $(F)\mathcal G_2^{\text{bot}}$.
	Now, we conclude with two canonical generators of low-energy-preserving operator equivalence classes under $(X)$ top boundary.
	
	Note that the $e_X$-transport operator
	\begin{equation*}
		\mathcal T_{e_X,z} := \bx\by
		\sum_{k=0}^{L_z-1}z^{k-1}\bA^k Z_2
	\end{equation*}
	is finite-support, $Z$-type, low-energy-preserving, and has the bottom syndrome $\mathcal G_1^{\text{bot}}$ [see Eq.~(\ref{eq_epsilon_bot_values})].
	Therefore, by Lemma~\ref{theorem_X_top_OZ_differ_by_stabilizer}, any finite-support $Z$-type low-energy-preserving operator $O_Z$ with the bottom syndrome $u\mathcal G_1^{\text{bot}}$ differs from $u\mathcal T_{e_X,z}$ by a stabilizer only.
	The $Z$-type low-energy-preserving operator equivalence classes are generated by $[\mathcal T_{e_X,z}]_{(X)}$.
	
	On the other hand, the bottom boundary gauge operator $xy\mathcal G_1^{\text{bot}}$ is finite-support, $X$-type, low-energy-preserving, and has the bottom syndrome $F\mathcal G_2^{\text{bot}}$, so by Lemma~\ref{theorem_X_top_OX_differ_by_stabilizer}, any finite-support $X$-type low-energy-preserving operator $O_X$ with the bottom syndrome $uF\mathcal G_2^{\text{bot}}$ differs from $uxy\mathcal G_1^{\text{bot}}$ by a stabilizer only.
	The $X$-type low-energy-preserving operator equivalence classes are generated by $[\mathcal G_1^{\text{bot}}]_{(X)}$.
	
	In all, we conclude that the low-energy-preserving Pauli operator equivalence class module $\mathcal O_{(X)}=\mathcal S_{(X)}^\Omega/\mathcal S_{(X)}$ is generated by $[\mathcal T_{e_X,z}]_{(X)},[\mathcal G_1^{\text{bot}}]_{(X)}$.
	The two generators all have nontrivial bottom syndrome, so they are not logical operators, which is consistent with Corollary~\ref{corollary_X_top_no_logical}.

	\subsection{A theorem in unique factorization domain (UFD) and UFD explanation}
	
	The following is a basic theorem that will be repeatedly used in the proof of Theorem~\ref{theorem_Z_top_OX_differ_by_stabilizer}.
	\begin{theorem}\label{theorem_UFD_divide_property}
		Let $R$ be a Unique Factorization Domain (UFD) and let $A,F,s\in R$.
		If $\gcd(A,F)=1$ and $F\mid As$, then $F\mid s$.
	\end{theorem}
	\textbf{Explanation. }
	\begin{itemize}
		\item A UFD is an integral domain $R$, s.t. every nonzero nonunit element $a\in R$ can be written as a finite product of irreducible elements
		\begin{equation*}
			a=p_1p_2\cdots p_n\,,
		\end{equation*}
		and this factorization is unique up to order and multiplication by units.
		More precisely, if
		\begin{equation*}
			a=p_1p_2\cdots p_n=q_1q_2\cdots q_m\,,
		\end{equation*}
		where $p_i$ and $q_j$ are irreducible nonzero nonunit factors, then $n=m$, and after reordering,
		\begin{equation*}
			p_i=u_iq_i
		\end{equation*}
		for all $i$, with some unit $u_i$.
		\item $A$ and $F$ are coprime means $\text{gcd}(A,F)=1$, up to the multiplication of unit.
	\end{itemize}
	\textbf{Proof. }The cases $s=0$ or $F$ a unit are immediate.
	Otherwise, write $F\sim p_1^{\alpha_1}\cdots p_m^{\alpha_m}$, where $\sim$ means differ by multiplying a unit here.
	Since $F|As$, all $p_i^{\alpha_i}$ ($i=1,\cdots,m$) must appear in the factorization of $As$.
	On the other hand, since $\text{gcd}(A,F)=1$, no $p_i$ ($i=1,\cdots,m$) can appear in the factorization of $A$.
	Therefore, all prime factors $p_i^{\alpha_i}$ must appear in the factorization of $s$, which implies $F|s$.
	\textbf{QED.}
	
	In particular, $R_\partial=\mathbb F_2[x^{\pm1},y^{\pm1}]$ is a UFD.

	\section{Operator identification and LU connectivity in Haah's code FTH}
	\label{appendix_Haah_operator_identification}
	
	\subsection{Identification of \texorpdfstring{\textit{other}}{other} bottom boundary gauge operators}\label{appendix_identification_of_other_boundary_gauge_operators}
	
	In this Appendix, we illustrate the detailed calculation for settling down the identification of \textit{other} bottom boundary gauge operator generators, where \textit{other} means the operator's boundary gauge syndrome is not a sum of top condensed topological excitations, or equivalently, those bottom boundary gauge generators whose identification is settled from the symplectic homomorphism condition of operator identification map ($\Phi_Z$ or $\Phi_X$).
	
	\textit{Under (Z) top boundary.}
	We have calculated before [see Eq.~(\ref{eq_Haah_FTH_eta_bot_2})] that
	\begin{equation}
		\Omega\left(\mathcal G_2^{\text{bot}},\mathcal G_1^{\text{bot}}\right) = \eta_{21}^{\text{bot}} = \bx\by F.
	\end{equation}
	On the other hand, denoting the symplectic form of $\tilde P_{(Z)}$ as $\tilde\Omega$, we have
	\begin{equation}
		\tilde\Omega\left(\Phi_Z\big([\mathcal G_2^{\text{bot}}]_{(Z)}\big),\Phi_Z\big([\mathcal G_1^{\text{bot}}]_{(Z)}\big)\right) = \tilde\Omega\left(u\tilde X+v\tilde Z,\bx\by F\tilde Z\right) = (\bar u,\bar v)\left(\begin{array}{cc}
			0&1\\1&0
		\end{array}\right)\left(\begin{array}{c}
			0\\\bx\by F
		\end{array}\right) = \bar u\bx\by F.
	\end{equation}
	The symplectic homomorphism condition of $\Phi_Z$ requires 
	\begin{equation}
		\Omega\left(\mathcal G_2^{\text{bot}},\mathcal G_1^{\text{bot}}\right) = \tilde\Omega\left(\Phi_Z\big([\mathcal G_2^{\text{bot}}]_{(Z)}\big),\Phi_Z\big([\mathcal G_1^{\text{bot}}]_{(Z)}\big)\right)\quad\Longleftrightarrow\quad \bx\by F = \bar u\bx\by F.
	\end{equation}
	Since $R_\partial$ is a UFD and $F$ is non-zero, this implies 
	\begin{equation}
		\bx\by=\bar u\bx\by\quad\Longleftrightarrow\quad u=1\quad\Longrightarrow\quad \Phi_Z\big([\mathcal G_2^{\text{bot}}]_{(Z)}\big) = \tilde X + v\tilde Z.
	\end{equation}
	Another constraint comes from 
	\begin{equation}
		0 = \eta_{22}^{\text{bot}} = \Omega(\mathcal G_2^{\text{bot}},\mathcal G_2^{\text{bot}}) = \tilde\Omega\left(\Phi_Z\big([\mathcal G_2^{\text{bot}}]_{(Z)}\big),\Phi_Z\big([\mathcal G_2^{\text{bot}}]_{(Z)}\big)\right) = (1,\bar v)\left(\begin{array}{cc}
			0&1\\1&0
		\end{array}\right)\left(\begin{array}{c}
			1\\v
		\end{array}\right) = v+\bar v.
	\end{equation}
	So, we get the identification of $\mathcal G_2^{\text{bot}}$ in Eq.~(\ref{eq_Haah_code_FTH_gauge_2_Z_identification}).
	Furthermore, we check the consistency of the identification $\Phi_Z$ by comparing the commutation phases between $\mathcal T_{e_Z,z}$ and the translations of $\mathcal G_2^{\text{bot}}$ before and after $\Phi_Z$.
	$\mathcal T_{e_Z,z}$ and $\mathcal G_2^{\text{bot}}$ only have overlap in the $z=0$ layer $P_0$, so
	\begin{equation}
		\Omega(\mathcal T_{e_Z,z},\mathcal G_2^{\text{bot}}) = \Omega\left(A^{L_z-1}\bz X_1,\bx\by\bz(xy\bA Z_1 + BZ_2)\right) = (\bA^{L_z-1},0,0,0)\left(\begin{array}{cc}
			0&1_2\\1_2&0
		\end{array}\right)\left(\begin{array}{c}
			0\\0\\\bA\\\bx\by B
		\end{array}\right)= \bA^{L_z},
	\end{equation}
	where $1_2$ is the $2\times 2$ identity matrix.
	On the other hand, 
	\begin{equation}
		\tilde\Omega\left(\Phi_Z\big([\mathcal T_{e_Z,z}]_{(Z)}\big),\Phi_Z\big([\mathcal G_2^{\text{bot}}]_{(Z)}\big)\right) = \tilde\Omega\left(A^{L_z}\tilde Z,\tilde X+v\tilde Z\right) = (0,\bA^{L_z})\left(\begin{array}{cc}
			0&1\\1&0
		\end{array}\right)\left(\begin{array}{c}
			1\\v
		\end{array}\right) = \bA^{L_z}.
	\end{equation}
	The commutation phases between $\mathcal T_{e_Z,z}$ and the translations of $\mathcal G_2^{\text{bot}}$ is indeed the same before and after $\Phi_Z$, the identification $\Phi_Z$ is consistent.

	\textit{Under (X) top boundary.}
	Recall that $\mathcal G_1^{\text{bot}}=\bx\by\bz(DX_1+X_2)$, and the $z=0$ layer part of $\mathcal T_{e_X,z}$ is $\bx\by\bz Z_2$ [see Appendix~\ref{appendix_generators_of_low_energy_preserving_operators_X_top}].
	$\mathcal T_{e_X,z}$ and $\mathcal G_1^{\text{bot}}$ only have overlap in the $z=0$ layer $P_0$, so
	\begin{equation}
		\Omega(\mathcal T_{e_X,z},\mathcal G_1^{\text{bot}}) = \Omega\left(\bx\by\bz Z_2,\bx\by\bz(DX_1+X_2)\right) = (0,0,0,xy)\left(\begin{array}{cc}
			0&1_2\\1_2&0
		\end{array}\right)\left(\begin{array}{c}
			\bx\by D\\\bx\by\\0\\0
		\end{array}\right) = 1.
	\end{equation}
	On the other hand,
	\begin{equation}
		\tilde\Omega\left(\Phi_X\big([\mathcal T_{e_X,z}]_{(X)}\big),\Phi_X\big([\mathcal G_1^{\text{bot}}]_{(X)}\big)\right) = \tilde\Omega\left(\tilde Z, u\tilde X+v\tilde Z\right) = (0,1)\left(\begin{array}{cc}
			0&1\\1&0
		\end{array}\right)\left(\begin{array}{c}
			u\\v
		\end{array}\right) = u.
	\end{equation}
	The symplectic homomorphism property of $\Phi_X$ requires
	\begin{equation}
		\Omega(\mathcal T_{e_X,z},\mathcal G_1^{\text{bot}}) = \tilde\Omega\left(\Phi_X\big([\mathcal T_{e_X,z}]_{(X)}\big),\Phi_X\big([\mathcal G_1^{\text{bot}}]_{(X)}\big)\right) = \tilde\Omega\left(\tilde Z, u\tilde X+v\tilde Z\right),
	\end{equation}
	so
	\begin{equation}
		u=1\quad\Longrightarrow\quad \Phi_X\big([\mathcal G_1^{\text{bot}}]_{(X)}\big) = \tilde X + v\tilde Z.
	\end{equation}
	Another constraint comes from
	\begin{equation}
		0=\eta_{11}^{\text{bot}} = \Omega(\mathcal G_1^{\text{bot}},\mathcal G_1^{\text{bot}}) = \tilde\Omega\left(\Phi_X\big([\mathcal G_1^{\text{bot}}]_{(X)}\big),\Phi_X\big([\mathcal G_1^{\text{bot}}]_{(X)}\big)\right) = (1,\bar v)\left(\begin{array}{cc}
			0&1\\1&0
		\end{array}\right)\left(\begin{array}{c}
			1\\v
		\end{array}\right) = \bar v+v.
	\end{equation}

	\subsection{LU circuit connectivity for \texorpdfstring{$v \neq 0$}{v != 0}}\label{appendix_LU_circuit_v_neq_0}
	
	Under the $(Z)$ top boundary, for a general $v=\bar v$, the identified Hamiltonian is
	\begin{equation}
		\tilde H_{(Z),v} := -\mathop{\sum\nolimits_{\mathbb{C}}}\limits_{\text{unit }m\in R_\partial}m\left(F\tilde Z +_{\mathbb C} g(\tilde X+v\tilde Z)\right).
	\end{equation}
	$v=\bar v$ means $v$ is symmetric under $(x,y)\to(\bx,\by)$, so the general form of $v$ is
	\begin{equation}
		v=a_1+\sum_{m\in\mathscr H}a_m(m+\bar m),
	\end{equation}
	where $a_1,a_m\in\mathbb F_2$, and $\mathscr H$ is the half space of $\mathbb Z^2$, s.t.
	$\mathscr H\cup\bar{\mathscr H}=\mathbb Z^2\backslash\{1\}$.
	For example, we can take
	\begin{equation}
		\mathscr H = \big\{x^iy^j\in R_\partial:i>0\big\}\cup\big\{y^j:j>0\big\}.
	\end{equation}
	Since $v\in R_\partial$, only finite number of $a_m$ is non-zero, denote
	\begin{equation}
		M(v):=\big\{m\in\mathscr H:a_m\neq0\big\}\subset R_\partial,
	\end{equation}
	so we can write
	\begin{equation}
		v=a_1+\sum_{m\in M(v)}(m+\bar m).
	\end{equation}
	Note that for $m,m'\in\mathbb Z^2$, the application of controlled-$\tilde Z$ gates map $\tilde X$ and $\tilde Z$ as follows: for any $m,m',m''\in \mathbb Z^2$,
	\begin{gather}
		C\tilde Z_{m,m'}\left(m''\tilde X\right)C\tilde Z_{m,m'} = m''\tilde X+\delta_{m,m''}m'\tilde Z + \delta_{m',m''}m\tilde Z,\\
		C\tilde Z_{m,m'}\left(m''\tilde Z\right)C\tilde Z_{m,m'} = m''\tilde Z.
	\end{gather}
	Therefore, denoting the adjoint action of $C\tilde Z_{m,m'}$ as
	\begin{equation}
		Ad\left(C\tilde Z_{m,m'}\right)(\bullet) := C\tilde Z_{m,m'}\bullet C\tilde Z_{m,m'},
	\end{equation}
	we can write
	\begin{align}
		&\Big(\prod_{\text{unit }m\in R_\partial}\,\prod_{m'\in M(v)}Ad\left(C\tilde Z_{m,mm'}\right)\Big)\left(m''\tilde X\right)\notag\\
		=&m''\tilde X + \sum_{\text{unit }m\in R_\partial}\,\sum_{m'\in M(v)}\left(\delta_{m,m''}mm'\tilde Z + \delta_{mm',m''}m\tilde Z\right)\notag\\
		=&m''\tilde X + \sum_{m'\in M(v)}\left(m'm''\tilde Z + \bar m'm''\tilde Z\right) = m''\left(\tilde X + (v-a_1)\tilde Z\right).
	\end{align}
	Since $Ad\left(C\tilde Z_{m,m'}\right)$ acts on $\tilde Z$ trivially, we have
	\begin{align}
		\Big(\prod_{\text{unit }m\in R_\partial}\,\prod_{m'\in M(v)}Ad\left(C\tilde Z_{m,mm'}\right)\Big)\left(\tilde H_{(Z)}\right) = -\mathop{\sum\nolimits_{\mathbb{C}}}\limits_{\text{unit }m\in R_\partial}m\left(F\tilde Z +_{\mathbb C} g(\tilde X+(v-a_1)\tilde Z)\right).
	\end{align}
	When $a_1=0$, the above Hamiltonian is $\tilde H_{(Z),v}$.
	When $a_1=1$, we can further apply the transversal $z$-rotation gate
	\begin{equation}
		\prod_{\text{unit }m\in R_\partial}R_z^{(m)}\left(\frac\pi 2\right) = \prod_{\text{unit }m\in R_\partial}m\cdot\text{exp}\left(-i\frac\pi 4\tilde Z\right),
	\end{equation}
	which maps $m\tilde X\to m\tilde Y$, $m\tilde Z\to m\tilde Z$, resulting in $\tilde H_{(Z),v}$.
	All the $C\tilde Z$ gates are mutually commuting, and $M(v)$ is finite, each qubit is involved in a finite number ($2|M(v)|$) of $C\tilde Z$ gates.
	So $\tilde H_{(Z)}$ and $\tilde H_{(Z),v}$ are connected by a finite-depth local\footnote{Here ``local'' means finite support rather than bounded linear size%
		; if $\tilde X + v\tilde Z$ is required to have bounded linear size, the circuit is local in that stronger sense as well.} unitary (LU) circuit.
		The $(X)$ top boundary case is completely analogous.

\section{\texorpdfstring{Fractal symmetries, relations, and global constraints in Haah FTH}{Fractal symmetries, relations, and global constraints in Haah FTH}}\label{appendix_Haah_infinite_support_quotients}

\subsection{\texorpdfstring{Fractal symmetry indicators and relations}{Fractal symmetry indicators and relations}}\label{app_Haah_constraints}
For $u=\sum_{\text{unit }m}u_m m\in\widehat{R}_\partial$ and $s\in R_\partial$, define the pairing
\begin{equation}\label{pairing_def}
	\langle u,s\rangle\equiv\sum_{\text{unit }m}u_m s_m=(u\bar{s})_0\in\mathbb{F}_2,
\end{equation}
where $w_0$ denotes the coefficient of $1$ in $w\in\widehat{R}_\partial$.
It is well defined since $s$ has finite support. 
It follows immediately that $\langle u,ms\rangle=(u\bar{s})_m$.
Then
\begin{equation}\label{eq_Haah_XZ_commutation}
	\tilde{X}(u)\tilde{Z}(s)=(-1)^{\langle u,s\rangle}\tilde{Z}(s)\tilde{X}(u),
\end{equation}
and similarly with the roles of $\tilde{X}$ and $\tilde{Z}$ exchanged and the primed operators.

\emph{Symmetry indicators of $\tilde{H}_{(Z)}$.}
Suppose a symmetry indicator $\tilde{X}(u)\tilde{Z}(v)$ with $u,v\in\widehat{R}_\partial$.
Commutation with all the transverse-field terms $m\tilde{X}$ requires $\langle v,m\rangle=v_m=0$ for every $m$, i.e., $v=0$.
Commutation with all the Ising terms $mf\tilde{Z}$ then requires, by Eq.~(\ref{eq_Haah_XZ_commutation}), $\langle u,mf\rangle=(u\bar{f})_m=0$ for every $m$, i.e., $u\bar{f}=0$.
Since $\bar{f}=xy\bar{F}$ differs from $\bar{F}$ by a unit, the symmetry indicators of $\tilde{H}_{(Z)}$ are exactly
\begin{equation}
	u\tilde{X},\qquad u\in\ann_{\widehat{R}_\partial}(\bar{F})\equiv\{u\in\widehat{R}_\partial:\ u\bar{F}=0\}.
\end{equation}
In the sandwich, $u\tilde{X}$ is identified from $uS_Z^{\text{bot}}$, which commutes with all $mS_X^{\text{bot}}$ for the same reason.
Under the $(X)$ top boundary, the same $uS_Z^{\text{bot}}$ is identified as $(u\bar{f})\tilde{Z}'=0$, or equivalently, $\prod_m(m\bar{f}\tilde{Z}')^{u_m}=1$, which is a relation among the Ising terms of $\tilde{H}_{(X)}$.
Conversely, a relation $\prod_m(m\bar{f}\tilde{Z}')^{u_m}=1$ requires $u\bar{f}=0$.
Therefore, Eq.~(\ref{eq_Haah_local_dictionary}) maps the symmetry indicators of $\tilde{H}_{(Z)}$ one-to-one to the relations of $\tilde{H}_{(X)}$.

\emph{Symmetry indicators of $\tilde{H}_{(X)}$.}
By the same reason, the symmetry indicators of $\tilde{H}_{(X)}$ are $u\tilde{X}'$ with $u\in\ann_{\widehat{R}_\partial}(F)$;
they are identified from $uS_X^{\text{bot}}$, which under the $(Z)$ top boundary is identified as $(uf)\tilde{Z}=0$, or equivalently, $\prod_m(mf\tilde{Z})^{u_m}=1$, a relation among the Ising terms of $\tilde{H}_{(Z)}$.

\emph{Fractal structure.}
Since $R_\partial$ has no zero divisor, $uF=0$ or $u\bar{F}=0$ with $u\in R_\partial$ implies $u=0$, so every nonzero symmetry indicator above has infinite support.
Since $\overline{uF}=\bar u\bar F$, the solutions of $uF=0$ are the spatial inversions $\bar u$ of those of $u\bar F=0$, so it is enough to consider $u\bar F=0$.
Let $u_{i,j}$ be the coefficient of $x^iy^j$ in $u$.
For a monomial $t=x^ay^b$, the coefficient of $x^iy^j$ in $u\bar t$ is $u_{i+a,j+b}$; summing over the six monomials $1,\bx,\bx^2,\by,\by^2,\bx\by$ of $\bF$, the condition $u\bar F=0$ reads
\begin{equation}
	u_{i,j}+u_{i+1,j}+u_{i+2,j}+u_{i,j+1}+u_{i+1,j+1}+u_{i,j+2}=0\ \ ,\ \ \forall\, i,j.
\end{equation}
Solving it for $u_{i,j+2}$, each column $j+2$ of $u$ is determined by the two columns $j,j+1$; regarding $j$ as time, this is an order-2 linear cellular-automaton rule~\cite{zhang_li_ye_2024_hoca_spt_strange_correlators,huang_zhang_ye_2026_nonuniform_ca,zhang-ye-fusion-mobility2026}.
The simplest non-trivial solution is $u=\sum_{\text{all }m}m$, i.e., $u_{i,j}=1$ for all $i,j$.
It is inversion symmetric, and gives the global symmetries $\prod_{\text{all}}\tilde{X}$ of $\tilde{H}_{(Z)}$ and $\prod_{\text{all}}\tilde{X}'$ of $\tilde{H}_{(X)}$.

To show the fractal structure, for $w=\sum_mw_mm\in\widehat R_\partial$ and a positive integer $N$, let
\begin{equation}
	w^{[N]}\equiv\sum_mw_m\,m^N
\end{equation}
be $w$ dilated by $N$, i.e., $w$ with $x,y$ replaced by $x^N,y^N$; e.g., $\bar F^{[N]}=1+\bar x^N+\bar x^{2N}+\bar y^N+\bar y^{2N}+\bar x^N\bar y^N$.
Since $(mm')^N=m^Nm'^N$, dilation preserves products, $(vw)^{[N]}=v^{[N]}w^{[N]}$ for $v\in R_\partial$.
Moreover, $(a+b)^2=a^2+b^2$ over $\mathbb F_2$ gives $v^2=v^{[2]}$ for $v\in R_\partial$, and by iteration $v^N=v^{[N]}$ for $N=2^k$.
Let
\begin{equation}
	u^{(k)}\equiv\Big(\sum_{\text{all }m}m\Big)^{[N]}\bar F^{\,N-1}=\Big(\sum_{\text{all }m}m^N\Big)\bar F^{\,N-1}\ \ ,\ \ N=2^k,\ k=0,1,2,\cdots.
\end{equation}
Then $u^{(k)}\bar F=\big(\sum_{\text{all }m}m\big)^{[N]}\bar F^{N}=\big(\sum_{\text{all }m}m\big)^{[N]}\bar F^{[N]}=\big(\big(\sum_{\text{all }m}m\big)\bar F\big)^{[N]}=0$; $u^{(0)}$ is the simplest non-trivial solution above.
$u^{(k)}$ is the sum of all the translates of the finite shape $\bar F^{N-1}$ by $(N\mathbb Z)^2$, hence periodic with period $N$.
This shape is self-similar: $\bar F^{2N-1}=\bar F^{N-1}\,\bar F^{N}=\bar F^{N-1}\,\bar F^{[N]}$, i.e., the shape at level $k+1$ ($\bF^{2N-1}$) is the sum of six copies of the shape at level $k$ ($\bF^{N-1}$), translated by the six monomials of $\bar F^{[N]}$.
Therefore, the symmetry indicators $u^{(k)}\tilde X$ are built from self-similar fractal patterns $u^{(k)}$.

\subsection{\texorpdfstring{Global constraints}{Global constraints}}\label{app_Haah_global_constraints}
Under the $(Z)$ top boundary, the finite-support operators generate $\tilde{\mathcal{H}}_{(Z)}(L_z)$ from $|+\cdots+\rangle_{(Z)}$, with the flipping operators $\tilde Z(s)$, where $s\in I_{L_z}=(F,A^{L_z})$.
A symmetry indicator $\tilde{X}(u)$ acts on $|s\rangle_{(Z)}\equiv\tilde Z(s)|+\cdots+\rangle_{(Z)}$ as $(-1)^{\langle u,s\rangle}$.
It equals $1$ on the whole $\tilde{\mathcal{H}}_{(Z)}(L_z)$, i.e., it becomes a global constraint that cannot be flipped by any finite-support operator, if and only if $\langle u,s\rangle=0$ for all $s\in I_{L_z}$.
Since $\langle u,mF\rangle=(u\bar{F})_m$ and $\langle u,mA^{L_z}\rangle=(u\bar{A}^{L_z})_m$, this requires $u\in\ann_{\widehat{R}_\partial}(\bar{F})\cap\ann_{\widehat{R}_\partial}(\bar{A}^{L_z})$.
Below we show how these constraints arise in the sandwich and count them.

Let $\Gamma_{ij}=\sigma^{\text{bot}}(x^iy^j\mathcal G_2^{\text{bot}})$ be a translated bottom boundary gauge operator, identified as $\tilde X_{ij}$ under the $(Z)$ top boundary.
We first consider products of $\Gamma_{ij}$ that equal the identity.
Since $\sigma^{\text{bot}}(\mathcal G_2^{\text{bot}})=\bar z(\bar A Z_1+\bar B Z_2)$, this requires $\bar A s=\bar B s=0$.
The second condition gives $s_{ij}=t_{i-j}$, and the first becomes $t_{n-1}+t_n+t_{n+1}=0$.
Thus
\begin{equation}\label{eq_Haah_raw_G2_kernel}
 s_{ij}=t_{i-j},\qquad t_{n+3}=t_n,\qquad
 (t_0,t_1,t_2)=(\alpha,\beta,\alpha+\beta),
\end{equation}
with $\alpha,\beta\in\mathbb F_2$.
The two independent patterns with $(\alpha,\beta)=(1,0)$ and $(0,1)$ give two of the global constraints, present for every $L_z$.

Additional constraints arise when a product of $\Gamma_{ij}$ is a product of stabilizers.
Let $B_{ij,k}=\sigma_B(x^iy^jz^kS_Z)$.
The layer representatives in Appendix~\ref{appendix_layer_by_layer_calculations} are
\begin{equation}\label{eq_Haah_Z_stabilizer_layer_vectors}
 \sigma_B(z^kS_Z)=z^{k-1}(Z_1+\bar D Z_2)
 +z^k(\bar A Z_1+\bar B Z_2).
\end{equation}
To reproduce the bottom gauge product with exponent pattern $s$, the first layer of stabilizers must have pattern $\bar A s$.
Its $Z_2$ component matches when $\bar D\bar A s=\bar B s$, equivalently $\bar F s=0$, which is true since $s\in\ann_{\widehat R_\partial}(\bF)$ is needed for $s\tilde X$ to be a symmetry.
Cancellation in each successive layer then requires the stabilizer pattern $\bar A^{k+1}s$ in layer $k$.
The last layer needs to vanish since the top truncated $S_Z$ is not in the stabilizer group under $(Z)$ top boundary, and it vanishes when $\bar A^{L_z}s=0$; its other component vanishes as well, since $\bar B\bar A^{L_z-1}s=\bar D\bar A^{L_z}s$.
Reading the Pauli products site by site, we obtain
\begin{equation}\label{eq_Haah_completed_family_identification}
 \prod_{(i,j)\in\mathbb Z^2}\Gamma_{ij}^{s_{ij}}
 =\prod_{k=0}^{L_z-2}\prod_{(i,j)\in\mathbb Z^2}
 B_{ij,k}^{(\bar A^{k+1}s)_{ij}},\qquad
 \bar F s=\bar A^{L_z}s=0.
\end{equation}
For $L_z=1$, the right-hand side is the empty product, recovering the two identities above.
Fixing these possibly infinite bulk stabilizer products to $+1$ in the low-energy boundary description gives
\begin{equation}\label{eq_Haah_target_constraints}
 \prod_{(i,j)\in\mathbb Z^2}\tilde X_{ij}^{s_{ij}}=+1,
 \qquad
 s\in\ann_{\widehat R_\partial}(\bar F)
 \cap\ann_{\widehat R_\partial}(\bar A^{L_z}).
\end{equation}

To count these constraints, we count the linearly independent exponent patterns $s$ satisfying $\bar F s=\bar A^{L_z}s=0$: adding two patterns multiplies their Pauli products.
Denote this number by $N_{L_z}$.
For $L_z=1$, the conditions are equivalent to $\bar A s=\bar B s=0$, so $N_1=2$ by Eq.~(\ref{eq_Haah_raw_G2_kernel}).

For $L_z>1$, every allowed $s$ determines a pattern $r=\bar A s$ satisfying
\begin{equation}
 \bar F r=0,\qquad \bar A^{L_z-1}r=0.
\end{equation}
Thus $r$ is an allowed constraint pattern for a slab one layer thinner.
We can count the patterns $s$ by considering every such $r$ and solving $r=\bar A s$ with $\bar F s=0$ (which is required for $s\tilde X$ to commute with $F\tilde Z$).

For a fixed $r$, any two solutions differing by a pattern $\delta s$ satisfy $\bar A\delta s=\bar F\delta s=0$.
Since $\bar F=\bar A\bar D+\bar B$, these conditions are equivalent to $\bar A\delta s=\bar B\delta s=0$.
Hence $\delta s$ is a combination of the two patterns in Eq.~(\ref{eq_Haah_raw_G2_kernel}).
Once one solution exists, all solutions are obtained by adding these two patterns independently, giving exactly four choices of $s$ for the same $r$.

One solution for each allowed $r$ can be constructed by setting $s_{0,0}=s_{1,0}=0$ and using the following recurrences:
\begin{align}\label{eq_Haah_constraint_count_recursion}
 s_{i+2,j}+s_{i+1,j}+s_{i,j}&=r_{i,j+1},\notag\\
 s_{i,j+1}&=s_{i,j}+s_{i+1,j}+r_{i,j},\notag\\
 s_{i,j}&=s_{i+1,j+1}+r_{i+1,j}+r_{i,j+1}.
\end{align}
The first recurrence determines the entire row $j=0$ in both directions.
The second and third determine all rows with $j>0$ and $j<0$, respectively.
It can be checked that $s$ indeed satisfies $\bar A s=r$ and $\bar F s=0$ and $\bar A^{L_z}s=\bar A^{L_z-1}r=0$, so $s$ is indeed a solution.

Since $r$ is a constraint pattern for a slab one layer thinner, the two free binary parameters in choosing $s$ give two additional independent constraint patterns for each added layer.
Starting from two independent patterns at $L_z=1$, we obtain
$$
N_{L_z}=\dim_{\mathbb F_2}\!\left[
\ann_{\widehat R_\partial}(\bar F)
\cap\ann_{\widehat R_\partial}(\bar A^{L_z})
\right]=2L_z.
$$
These $2L_z$ independent global constraints consist of the two identities above and $2L_z-2$ additional constraints from bulk stabilizer products.
The number grows with $L_z$: the transport operators $m\mathcal T_{e_Z,z}$ of a thicker sandwich toggle fewer symmetry indicators, and the symmetry indicators that can no longer be toggled become global constraints.
In this sense, the identified 2d system under the $(Z)$ top boundary remembers the thickness of the type-II bulk.

    \section{Duality operators and physical sectors}
    \label{appendix_duality_operators}

    \subsection{Finite periodic systems}
    We first verify the KW operator of Eqs.~(\ref{eq_KW_intertwining})--(\ref{eq_KW_Gram}) on a periodic chain.
    We then apply the same calculation to the subsystem KW operator of the TFPIM in Eqs.~(\ref{eq_TFPIM_KW_intertwining})--(\ref{eq_TFPIM_KW_composition}) and check the X-cube FTH constraint, and finally to the QCM self-duality and the QCM--TFPIM duality.
    All twist indicators are fixed to $+1$ throughout this subsection.

    \subsubsection{KW duality on a periodic chain}
    We retain the notation of the main text: $\mathbf s$ labels a $Z$-basis configuration of the input chain, and $\mathbf t$ labels a $Z'$-basis configuration of the target chain.
    Both have $L$ binary entries, and all site indices are periodic.
    The matrix elements of Eq.~(\ref{eq_KW_explicit_form}) are
    \begin{equation}
        \langle\mathbf t|\mathsf D_{\rm KW}|\mathbf s\rangle
        =2^{-L/2}(-1)^{\sum_it_{i+1}(s_i+s_{i+1})}.
    \end{equation}
    Acting with $X_i$ on the input flips $s_i$, changing the exponent by $t_i+t_{i+1}$.
    Thus
    \begin{align}
        \langle\mathbf t|\mathsf D_{\rm KW}X_i|\mathbf s\rangle
        & =(-1)^{t_i+t_{i+1}}
        \langle\mathbf t|\mathsf D_{\rm KW}|\mathbf s\rangle\notag\\
        & =\langle\mathbf t|Z'_iZ'_{i+1}
        \mathsf D_{\rm KW}|\mathbf s\rangle.
    \end{align}
    Similarly, acting with $X'_{i+1}$ on the target flips $t_{i+1}$, changing the exponent by $s_i+s_{i+1}$.
    This gives
    \begin{align}
        \langle\mathbf t|X'_{i+1}\mathsf D_{\rm KW}|\mathbf s\rangle
        & =(-1)^{s_i+s_{i+1}}
        \langle\mathbf t|\mathsf D_{\rm KW}|\mathbf s\rangle\notag\\
        & =\langle\mathbf t|\mathsf D_{\rm KW}Z_iZ_{i+1}|\mathbf s\rangle.
    \end{align}
    Since these equalities hold for every pair of basis configurations, they prove the two intertwining relations in Eq.~(\ref{eq_KW_intertwining}).

    To compute $\mathsf D_{\rm KW}^\dagger\mathsf D_{\rm KW}$, let $\mathbf s'$ be another configuration of the input chain and denote $a_i=s_i+s'_i$.
    Using Eq.~(\ref{eq_KW_explicit_form}), direct calculation gives
    \begin{align}
        &\langle\mathbf s'|\mathsf D_{\rm KW}^\dagger\mathsf D_{\rm KW}|\mathbf s\rangle\notag\\
        &\quad=2^{-L}\sum_{\mathbf t}
        (-1)^{\sum_it_{i+1}(a_i+a_{i+1})}\notag\\
        &\quad=\prod_i\frac{1+(-1)^{a_i+a_{i+1}}}{2}.
        \label{eq_appendix_KW_input_Gram}
    \end{align}
    Here each sum of $t_i$ uses
    \begin{equation}
        \frac12\sum_{t_i=0}^1(-1)^{t_i c}
        =\begin{cases}1,&c=0,\\0,&c=1.\end{cases}
    \end{equation}
    Consequently, the matrix element is one precisely when $a_i=a_{i+1}$ at every bond, otherwise, it is a zero.
    On a connected periodic chain, $a_i$ must be zero for all $i$ or be one for all $i$.
    Equivalently, $s'_i=s_i$ for all $i$ or $s'_i=s_i+1$ for all $i$.
    Therefore
    \begin{align}
        \mathsf D_{\rm KW}^\dagger\mathsf D_{\rm KW}|\mathbf s\rangle&=\sum_{\mathbf s'}|\mathbf s'\rangle\langle\mathbf s'|\mathsf D_{\rm KW}^\dagger\mathsf D_{\rm KW}\sum_{\mathbf s''}|\mathbf s''\rangle\langle\mathbf s''|\mathbf s\rangle\notag\\
        &=|\mathbf s\rangle+|\mathbf s+\mathbf1\rangle=(1+U)|\mathbf s\rangle,
        \qquad U=\prod_iX_i.
    \end{align}

    Similarly, let $\mathbf t'$ be another target configuration and write $b_i=t_i+t'_i$, we have
    \begin{align}
        &\langle\mathbf t'|\mathsf D_{\rm KW}\mathsf D_{\rm KW}^\dagger
        |\mathbf t\rangle\notag\\
        &\quad=2^{-L}\sum_{\mathbf s}
        (-1)^{\sum_i b_{i+1}(s_i+s_{i+1})}\notag\\
        &\quad=\prod_i\frac{1+(-1)^{b_i+b_{i+1}}}{2}.
    \end{align}
    The same argument leaves only $b=0$ and $b=\mathbf1$, proving $\mathsf D_{\rm KW}\mathsf D_{\rm KW}^\dagger=1+U'$ with $U'=\prod_iX'_i$.
    This establishes both identities in Eq.~(\ref{eq_KW_Gram}).

    These identities determine the kernel and image.
    For an input state with $U|\psi_\pm\rangle=\pm|\psi_\pm\rangle$,
    \begin{equation}
        \|\mathsf D_{\rm KW}|\psi_\pm\rangle\|^2
        =\langle\psi_{\pm}|\mathsf D_{\rm KW}^\dagger\mathsf D_{\rm KW}|\psi_\pm\rangle=(1\pm1)\|\psi_\pm\|^2.
    \end{equation}
    Thus the kernel of $\mathsf D_{\rm KW}$ is exactly the $U=-1$ sector.
    Likewise, $\mathsf D_{\rm KW}^\dagger$ annihilates the $U'=-1$ sector, so the KW duality is indeed non-invertible.

    \subsubsection{Subsystem KW duality of the TFPIM}\label{app_subsec_D_square}

    We now construct the intertwiner of TFPIM, and show it is non-invertible.
    Let $\mathbf s\in\mathbb F_2^V$ and $\mathbf t\in\mathbb F_2^P$ label the input $\tilde Z$ basis under the $(ee)$ top boundary and the target $\tilde Z'$ basis under the $(mm)$ top boundary, respectively, and let $N=|V|=|P|=L_xL_y$.
    Define the incidence matrix $M$ (maps from $\mathbb F_2^V$ to $\mathbb F_2^P$) with $M_{pv}=1$ if $v\subset p$, and $0$ otherwise:
    
    \begin{equation}
    	M_{pv} =
    	\begin{cases}
    		1, & v\subset p,\\
    		0, & v\not\subset p,
    	\end{cases}
    \end{equation}
    so that
    
    \begin{equation}
    	(M\mathbf s)_p = \sum_v M_{pv}s_v = \sum_{v\subset p} s_v.
    \end{equation}
    The subsystem Kramers--Wannier duality operator is
    \begin{equation}
    	\mathsf D_\square = 2^{-N/2}\sum_{\mathbf s,\mathbf t}(-1)^{\mathbf t^T M\mathbf s}|\mathbf t\rangle\langle\mathbf s|,
    \end{equation}
    where the exponent is
    
    \begin{equation}
    	\mathbf t^T M\mathbf s = \sum_p t_p \sum_v M_{pv} s_v = \sum_p t_p \sum_{v\subset p} s_v,
    \end{equation}
    and the matrix elements are
    
    \begin{equation}
    	\langle\mathbf t|\mathsf D_\square|\mathbf s\rangle = 2^{-N/2}(-1)^{\mathbf t^T M\mathbf s} = 2^{-N/2}(-1)^{\sum_p t_p \sum_{v\subset p} s_v}.
    \end{equation}

    \paragraph{Intertwining relations. }

    Acting with $\tilde X_v$ on the input configuration flips the $v$-th component of $\mathbf s$, i.e., $\tilde X_v|\mathbf s\rangle = |\mathbf s+\mathbf 1_v\rangle$, where $(\mathbf 1_v)_{v'}=\delta_{vv'}$:
    
    \begin{align}
    	\langle\mathbf t|\mathsf D_\square\tilde X_v|\mathbf s\rangle
    	&= \langle\mathbf t|\mathsf D_\square|\mathbf s+\mathbf 1_v\rangle \notag\\
    	&= 2^{-N/2}(-1)^{\sum_p t_p \sum_{v'\subset p}(\mathbf s+\mathbf 1_v)_{v'}} \notag\\
    	&= 2^{-N/2}(-1)^{\sum_p t_p \sum_{v'\subset p}(s_{v'}+\delta_{vv'})}.
    \end{align}
    Since
    
    \begin{equation}
    	\sum_{v'\subset p}\delta_{vv'} =
    	\begin{cases}
    		1, & v\subset p,\\
    		0, & v\not\subset p,
    	\end{cases}
    \end{equation}
    so,
    
    \begin{equation}
    	\sum_p t_p \sum_{v'\subset p}\delta_{vv'} = \sum_{p\revsubset v} t_p,
    \end{equation}
    and we obtain
    
    \begin{align}
    	\langle\mathbf t|\mathsf D_\square\tilde X_v|\mathbf s\rangle
    	&= 2^{-N/2}(-1)^{\sum_p t_p \sum_{v'\subset p} s_{v'}}\cdot(-1)^{\sum_{p\revsubset v} t_p} \notag\\
    	&= (-1)^{\sum_{p\revsubset v} t_p}\langle\mathbf t|\mathsf D_\square|\mathbf s\rangle \notag\\
    	&= \langle\mathbf t|\prod_{p\revsubset v}\tilde Z'_p\mathsf D_\square|\mathbf s\rangle, \quad \forall\,|\mathbf t\rangle,|\mathbf s\rangle.
    \end{align}
    Conversely, acting with $\tilde X'_p$ on $\langle\mathbf t|$ flips the $p$-th component of $\mathbf t$, where $(\mathbf 1_p)_{p'}=\delta_{pp'}$:
    
    \begin{align}
    	\langle\mathbf t|\tilde X'_p\mathsf D_\square|\mathbf s\rangle
    	&= \langle\mathbf t+\mathbf 1_p|\mathsf D_\square|\mathbf s\rangle \notag\\
    	&= 2^{-N/2}(-1)^{\sum_{p'}(\mathbf t+\mathbf 1_p)_{p'}\sum_{v\subset p'} s_v} \notag\\
    	&= \langle\mathbf t|\mathsf D_\square|\mathbf s\rangle(-1)^{\sum_{p'}\delta_{pp'}\sum_{v\subset p'} s_v} \notag\\
    	&= \langle\mathbf t|\mathsf D_\square|\mathbf s\rangle(-1)^{\sum_{v\subset p} s_v} \notag\\
    	&= \langle\mathbf t|\mathsf D_\square\prod_{v\subset p}\tilde Z_v|\mathbf s\rangle, \quad \forall\,|\mathbf t\rangle,|\mathbf s\rangle.
    \end{align}
    So,
    
    \begin{equation}
    	\mathsf D_\square\tilde X_v = \prod_{p\revsubset v}\tilde Z'_p\mathsf D_\square, \qquad
    	\mathsf D_\square\prod_{v\subset p}\tilde Z_v = \tilde X'_p\mathsf D_\square.
    \end{equation}

    \paragraph{Composition $\mathsf D_\square^\dagger\mathsf D_\square$. }

    Now, consider $\mathbf s,\mathbf s'\in\mathbb F_2^V$,
    \begin{align}
    	\langle\mathbf s'|\mathsf D_\square^\dagger\mathsf D_\square|\mathbf s\rangle
    	&= 2^{-N}\langle\mathbf s'|\sum_{\mathbf s'',\mathbf t}(-1)^{\mathbf t^T M\mathbf s''}|\mathbf s''\rangle\langle\mathbf t|
    	\sum_{\mathbf s''',\mathbf t'}(-1)^{(\mathbf t')^T M\mathbf s'''}|\mathbf t'\rangle\langle\mathbf s'''|\mathbf s\rangle \notag\\
    	&= 2^{-N}\sum_{\mathbf t}(-1)^{\mathbf t^T M\mathbf s' + \mathbf t^T M\mathbf s} \notag\\
    	&= 2^{-N}\sum_{\mathbf t}(-1)^{\mathbf t^T M(\mathbf s'+\mathbf s)} \notag\\
    	&= 2^{-N}\sum_{\mathbf t}(-1)^{\sum_p t_p \sum_{v\subset p}(s'_v+s_v)} \notag\\
    	&= 2^{-N}\sum_{\mathbf t}\prod_p (-1)^{t_p \sum_{v\subset p}(s'_v+s_v)} \notag\\
    	&= 2^{-N}\prod_p\left(\sum_{t_p=0,1}(-1)^{t_p \sum_{v\subset p}(s'_v+s_v)}\right) \notag\\
    	&=
    	\begin{cases}
    		1, & \sum_{v\subset p}(s'_v+s_v)=0 \pmod 2,\ \forall p,\\
    		0, & \text{otherwise},
    	\end{cases}
    \end{align}
    where the exchange of $\sum_{\mathbf t}$ and $\prod_p$ uses the factorization identity [e.g., for two variables, $f_1(0)f_2(0)+f_1(0)f_2(1)+f_1(1)f_2(0)+f_1(1)f_2(1) = (f_1(0)+f_1(1))(f_2(0)+f_2(1))$].

    Let $a=\mathbf s'+\mathbf s$.
    Recall that
    
    \begin{equation}
    	M_{pv} =
    	\begin{cases}
    		1, & v\subset p,\\
    		0, & v\not\subset p,
    	\end{cases}
    \end{equation}
    and write
    
    \begin{equation}
    	\sum_{v\subset p}(s'_v+s_v) = \sum_{v\subset p} a_v = \sum_v M_{pv}a_v = (Ma)_p,
    \end{equation}
    so
    \begin{align}
    	\langle\mathbf s'|\mathsf D_\square^\dagger\mathsf D_\square|\mathbf s\rangle
    	&=
    	\begin{cases}
    		1, & (Ma)_p=0,\ \forall p,\\
    		0, & \text{otherwise}
    	\end{cases} \notag\\
    	&=
    	\begin{cases}
    		1, & Ma=0,\\
    		0, & Ma\ne 0.
    	\end{cases}
    \end{align}
    Thus $\langle\mathbf s'|\mathsf D_\square^\dagger\mathsf D_\square|\mathbf s\rangle=1$ iff $a\in\ker M$.

    In other words,
    \begin{align}
    	\mathsf D_\square^\dagger\mathsf D_\square|\mathbf s\rangle
    	&= \sum_{\mathbf s':\,M(\mathbf s+\mathbf s')=0}|\mathbf s'\rangle \notag\\
    	&= \sum_{a:\,Ma=0}|\mathbf s+a\rangle \notag\\
    	&= \sum_{a\in\ker M}\tilde X(a)|\mathbf s\rangle,
    	\label{eq_D_square_dagger_D_square_on_s}
    \end{align}
    where
    \begin{equation}
    	\tilde X(a)\equiv\prod_v\tilde X_v^{a_v}.
    \end{equation}
    Equation~(\ref{eq_D_square_dagger_D_square_on_s}) is true $\forall\,|\mathbf s\rangle$, so
    \begin{equation}
    	\mathsf D_\square^\dagger\mathsf D_\square = \sum_{a\in\ker M}\tilde X(a).
    \end{equation}
    Similarly,
    \begin{equation}
    	\mathsf D_\square\mathsf D_\square^\dagger = \sum_{b\in\ker M^T}\tilde X'(b),
    \end{equation}
    where
    \begin{equation}
    	\tilde X'(b)\equiv\prod_p(\tilde X'_p)^{b_p}.
    \end{equation}

    \paragraph{\texorpdfstring{Explicit structure of $\ker M$ and non-invertibility.}{Explicit structure of ker M and non-invertibility.}}

    Next, we figure out $\ker M$.
    For simplicity, here we label vertices by $(i,j)$ with periodic indices.
    The equation $Ma=0$ requires
    \begin{equation}
    	a_{ij}+a_{i-1,j}+a_{i,j-1}+a_{i-1,j-1}=0.
    \end{equation}
    In other words, $a_{ij}+a_{i-1,j}$ is independent of $j$, so
    \begin{equation}
    	a_{ij}+a_{i-1,j}+a_{i0}+a_{i-1,0}=0.
    	\label{eq_ker_M_row_diff}
    \end{equation}
    Adding Eq.~(\ref{eq_ker_M_row_diff}) from the first row to the $i$-th row gives
    \begin{align}
    	&\quad (a_{1j}+a_{0j}+a_{10}+a_{00})
    	+\cdots
    	+(a_{ij}+a_{i-1,j}+a_{i0}+a_{i-1,0}) = 0 \notag\\
    	&\iff a_{ij}+a_{0j}+a_{i0}+a_{00} = 0 \notag\\
    	&\iff a_{ij} = a_{i0}+a_{0j}+a_{00} \equiv r_i+c_j,
    \end{align}
    where $r_i\equiv a_{i0}$, $c_j\equiv a_{00}+a_{0j}$.
    Any $a\in\ker M$ can be represented by $a_{ij}=r_i+c_j$, while the representation $(r_i,c_j)$ has one redundancy:
    \begin{equation}
    	(r_i,c_j)\mapsto(r_i+1,c_j+1),\quad \forall\,i,j,
    \end{equation}
    since
    \begin{equation}
    	a_{ij}=r_i+c_j=(r_i+1)+(c_j+1).
    \end{equation}
    Thus
    \begin{equation}
    	\dim\ker M = L_x+L_y-1.
    \end{equation}
    This can serve as a completeness proof of the symmetry indicators illustrated in Fig.~\ref{fig_mm_ee_symmetries_option_A}.
    We choose the following set of generators of $\ker M$,
    \begin{equation}
    	h^{(0)},\dots,h^{(L_y-1)},\quad v^{(0)},\dots,v^{(L_x-2)},
    \end{equation}
    where
    \begin{equation}
    	(h^{(j_0)})_{ij} = \delta_{j,j_0},\qquad (v^{(i_0)})_{ij} = \delta_{i,i_0}
    \end{equation}
    correspond to the line-like symmetry indicators constructed in Fig.~\ref{fig_mm_ee_symmetries_option_A}.
    $v^{(L_x-1)}$ is redundant, thus excluded from the generator set.
    Similarly, $\ker M^T$ is spanned by the row and column patterns of plaquettes, and the corresponding line-like symmetry indicators of the plaquette TFPIM have one more redundancy, because their product over all rows is $\prod_p\tilde X'_p$, which equals $1$ due to the global constraint of the boundary theory under the $(mm)$ top boundary, inherited from the FTH.

    Denote the line-like symmetries corresponding to the chosen generators of $\ker M$ as $S_1,\dots,S_{L_x+L_y-1}$, we have
    \begin{equation}
    	\mathsf D_\square^\dagger\mathsf D_\square = \sum_{a\in\ker M}\tilde X(a) = \prod_{\ell=1}^{L_x+L_y-1}(1+S_\ell) = 2^{L_x+L_y-1}P_{\rm lines,+},
    \end{equation}
    where
    \begin{equation}
    	P_{\rm lines,+} = \prod_\ell\frac{1+S_\ell}{2}
    \end{equation}
    is the projector onto the common $+1$ eigenspace of the line-like symmetry indicators.

    The boundary theory of X-cube FTH under $(mm)$ top boundary has a Hilbert space with the global constraint,
    \begin{equation}
    	\tilde{\mathcal H}_{(mm)} = \left.\left(\bigotimes_p\mathbb C^2\right)\right|_{\prod_p\tilde X'_p=1}.
    \end{equation}
    Multiplying the first intertwining relation in Eq.~(\ref{eq_TFPIM_KW_intertwining}) over all plaquettes gives
    \begin{equation}
    	\prod_p\tilde X'_p\,\mathsf D_\square=\mathsf D_\square\prod_p\prod_{v\subset p}\tilde Z_v=\mathsf D_\square\prod_v\tilde Z_v^{\,4}=\mathsf D_\square,
    \end{equation}
    since each vertex is contained in four plaquettes; i.e., the image of $\mathsf D_\square$ lies in $\tilde{\mathcal H}_{(mm)}$.
    In other words, the relation that all the plaquette Ising terms of the vertex TFPIM multiply to $1$ is mapped to the global constraint $\prod_p\tilde X'_p=1$.
    Thus $\mathsf D_\square$ maps the boundary theory under the $(ee)$ top boundary into that under the $(mm)$ top boundary, without any further projection needed.
    Since $\mathsf D_\square^\dagger\mathsf D_\square$ is a projector (up to the factor $2^{L_x+L_y-1}$), and $\exists\,|\psi\rangle$, e.g.
    \begin{equation}
    	|\psi\rangle = \tilde Z_v\cdot|+\cdots+\rangle,\quad \text{s.t.}\quad P_{\rm lines,+}|\psi\rangle = 0,
    \end{equation}
    we have
    \begin{align}
    	\big|\big|\mathsf D_\square|\psi\rangle\big|\big|^2
    	&= \langle\psi|\mathsf D_\square^\dagger\mathsf D_\square|\psi\rangle \notag\\
    	&= 2^{L_x+L_y-1}\langle\psi|P_{\rm lines,+}|\psi\rangle = 0 \notag\\
    	&\iff \mathsf D_\square|\psi\rangle = 0.
    \end{align}
    In other words, $\mathsf D_\square$ is non-invertible.

    \paragraph{Symmetry absorption.}
    
    For any $a\in\ker M$ (so that $Ma=0$), acting with $\tilde X(a)$ on $|\mathbf s\rangle$ gives $\tilde X(a)|\mathbf s\rangle = |\mathbf s+a\rangle$, and therefore
    \begin{align}
    	\langle\mathbf t|\mathsf D_\square\tilde X(a)|\mathbf s\rangle
    	&= \langle\mathbf t|\mathsf D_\square|\mathbf s+a\rangle \notag\\
    	&= 2^{-N/2}(-1)^{\mathbf t^T M(\mathbf s+a)} \notag\\
    	&= 2^{-N/2}(-1)^{\mathbf t^T M\mathbf s}
    	= \langle\mathbf t|\mathsf D_\square|\mathbf s\rangle, \quad \forall\,|\mathbf t\rangle,|\mathbf s\rangle.
    \end{align}
    Thus, $\mathsf D_\square$ absorbs every symmetry indicator $\tilde X(a)$ with $a\in\ker M$:
    \begin{equation}
    	\mathsf D_\square\tilde X(a) = \mathsf D_\square, \quad \forall\,a\in\ker M
    	\quad\Longrightarrow\quad
    	\mathsf D_\square S_\ell = \mathsf D_\square, \quad \ell=1,\dots,L_x+L_y-1.
    \end{equation}
    Similarly, for any $b\in\ker M^T$ (so that $M^Tb=0$, i.e., $b^TM=0$), acting with $\tilde X'(b)$ on $\langle\mathbf t|$ gives $\langle\mathbf t|\tilde X'(b) = \langle\mathbf t+b|$, so
    \begin{align}
    	\langle\mathbf t|\tilde X'(b)\mathsf D_\square|\mathbf s\rangle
    	&= \langle\mathbf t+b|\mathsf D_\square|\mathbf s\rangle \notag\\
    	&= 2^{-N/2}(-1)^{(\mathbf t+b)^T M\mathbf s} \notag\\
    	&= 2^{-N/2}(-1)^{\mathbf t^T M\mathbf s}
    	= \langle\mathbf t|\mathsf D_\square|\mathbf s\rangle, \quad \forall\,|\mathbf t\rangle,|\mathbf s\rangle,
    \end{align}
    which gives
    \begin{equation}
    	\tilde X'(b)\mathsf D_\square = \mathsf D_\square, \quad \forall\,b\in\ker M^T
    	\quad\Longrightarrow\quad
    	S'_\ell\mathsf D_\square = \mathsf D_\square,
    \end{equation}
    where $S'_\ell$ are the line-like $\tilde X'$ symmetry indicators of the plaquette TFPIM under the $(mm)$ top boundary.
    Physically, the symmetry indicators on either side are intertwined by $\mathsf D_\square$ to relations, which are just $1$.

    \subsubsection{\texorpdfstring{QCM self-duality}{QCM self-duality}}\label{app_subsec_D_QCM}
    
    Denote the edges spanning $x(y)$-direction as $x(y)$-edge $e_x(e_y)$.
    Let $N=L_xL_y=|E_x|=|E_y|=|V|=|P|$, where
    \begin{equation}
    	E_x=\{x\text{-edge}\},\qquad E_y=\{y\text{-edge}\},\qquad V=\{\text{vertex}\},\qquad P=\{\text{plaquette}\}.
    \end{equation}
    Note that $B_{v,zx}^{\text{bot}}\sim B_{v,yz}^{\text{bot}}$ is identified as $\prod_{e_x\revsubset v}\tilde Z_{e_x}$ and $\prod_{e_y\revsubset v}\tilde Z'_{e_y}$ under $(em)$ and $(me)$ top boundaries, respectively.
    On the other hand, $A_c^{\text{bot}}$ is identified as $\prod_{e_x\subset p}\tilde X_{e_x}$ and $\prod_{e_y\subset p}\tilde X'_{e_y}$ under $(em)$ and $(me)$ top boundaries, respectively.
    Let $\mathbf s\in\mathbb F_2^{E_x}$ and $\mathbf t\in\mathbb F_2^{E_y}$ label the $Z$ eigenvalue configurations of $E_x$ QCM and $E_y$ QCM, respectively, and let $\mathbf u\in\mathbb F_2^V$ be the configurations of auxiliary vertex spins.
    With the incidence matrices $M_x$ (maps from $E_x$ to $V$) and $M_y$ (maps from $E_y$ to $V$) defined in Eq.~(\ref{eq_Mx_My_main}), the QCM self-duality operator induced by the X-cube FTH is
    \begin{equation}
    	\mathsf D_{\text{QCM}} = 2^{-N}\sum_{\mathbf s,\mathbf t}\sum_{\mathbf u\in\mathbb F_2^V}(-1)^{\mathbf u^T(M_x\mathbf s+M_y\mathbf t)}|\mathbf t\rangle\langle\mathbf s|.
    	\label{eq_D_QCM_def}
    \end{equation}
    We now verify its intertwining relation and its non-invertibility.

    From Eq.~(\ref{eq_Mx_My_main}) we can write
    \begin{equation}
    	(M_x\mathbf s)_v = \sum_{e_x}(M_x)_{v,e_x}s_{e_x} = \sum_{e_x\revsubset v} s_{e_x},
    	\qquad
    	(M_y\mathbf t)_v = \sum_{e_y\revsubset v} t_{e_y}.
    \end{equation}
    So,
    \begin{equation}
    	\mathbf u^T(M_x\mathbf s+M_y\mathbf t) = \sum_v u_v\left(\sum_{e_x\revsubset v}s_{e_x} + \sum_{e_y\revsubset v}t_{e_y}\right).
    \end{equation}
    The matrix elements of $\mathsf D_{\text{QCM}}$ are
    \begin{equation}
    	\langle\mathbf t|\mathsf D_{\text{QCM}}|\mathbf s\rangle
    	= 2^{-N}\sum_{\mathbf u\in\mathbb F_2^V}(-1)^{\mathbf u^T(M_x\mathbf s+M_y\mathbf t)}.
    	\label{eq_D_QCM_matrix_element}
    \end{equation}

    \paragraph{Intertwining relations.}

    Denote the unit of $\mathbb F_2^V$ as $\mathbf 1_v\in\mathbb F_2^V$, with $(\mathbf 1_v)_{v'}=\delta_{vv'}$.
    Shifting the summation variable $\mathbf u\mapsto\mathbf u+\mathbf 1_v$ leaves the sum invariant, so
    \begin{align}
    	\langle\mathbf t|\mathsf D_{\text{QCM}}|\mathbf s\rangle
    	&= 2^{-N}\sum_{\mathbf u\in\mathbb F_2^V}(-1)^{(\mathbf u+\mathbf 1_v)^T(M_x\mathbf s+M_y\mathbf t)} \notag\\
    	&= 2^{-N}\sum_{\mathbf u\in\mathbb F_2^V}(-1)^{\sum_{v'}(u_{v'}+\delta_{vv'})\left(\sum_{e_x\revsubset v'}s_{e_x} + \sum_{e_y\revsubset v'}t_{e_y}\right)} \notag\\
    	&= (-1)^{\sum_{e_x\revsubset v}s_{e_x} + \sum_{e_y\revsubset v}t_{e_y}}\cdot 2^{-N}\sum_{\mathbf u\in\mathbb F_2^V}(-1)^{\mathbf u^T(M_x\mathbf s+M_y\mathbf t)} \notag\\
    	&= (-1)^{\sum_{e_x\revsubset v}s_{e_x} + \sum_{e_y\revsubset v}t_{e_y}}\langle\mathbf t|\mathsf D_{\text{QCM}}|\mathbf s\rangle.
    \end{align}
    Putting the $s_{e_x}$ term on the left hand side, we have
    \begin{equation}
    	(-1)^{\sum_{e_x\revsubset v}s_{e_x}}\langle\mathbf t|\mathsf D_{\text{QCM}}|\mathbf s\rangle
    	= (-1)^{\sum_{e_y\revsubset v}t_{e_y}}\langle\mathbf t|\mathsf D_{\text{QCM}}|\mathbf s\rangle,
    \end{equation}
    which is
    \begin{equation}
    	\langle\mathbf t|\mathsf D_{\text{QCM}}\prod_{e_x\revsubset v}\tilde Z_{e_x}|\mathbf s\rangle
    	= \langle\mathbf t|\prod_{e_y\revsubset v}\tilde Z'_{e_y}\mathsf D_{\text{QCM}}|\mathbf s\rangle.
    	\label{eq_D_QCM_intertwine_Z_elem}
    \end{equation}
    Equation~(\ref{eq_D_QCM_intertwine_Z_elem}) holds for any $\mathbf t,\mathbf s$, so
    \begin{equation}
    	\mathsf D_{\text{QCM}}\prod_{e_x\revsubset v}\tilde Z_{e_x}
    	= \prod_{e_y\revsubset v}\tilde Z'_{e_y}\mathsf D_{\text{QCM}}.
    	\label{eq_D_QCM_intertwine_Z}
    \end{equation}

    Next, we consider the $X$ terms.
    Define the plaquette associated vectors $\mathbf 1_p^{(x)}\in\mathbb F_2^{E_x}$ and $\mathbf 1_p^{(y)}\in\mathbb F_2^{E_y}$ as
    \begin{equation}
    	(\mathbf 1_p^{(x)})_{e_x} =
    	\begin{cases}
    		1, & e_x\subset p,\\
    		0, & e_x\not\subset p,
    	\end{cases}
    	\qquad
    	(\mathbf 1_p^{(y)})_{e_y} =
    	\begin{cases}
    		1, & e_y\subset p,\\
    		0, & e_y\not\subset p.
    	\end{cases}
    \end{equation}
    So that we can write
    \begin{equation}
    	\prod_{e_x\subset p}\tilde X_{e_x}|\mathbf s\rangle = |\mathbf s+\mathbf 1_p^{(x)}\rangle,
    	\qquad
    	\prod_{e_y\subset p}\tilde X'_{e_y}|\mathbf t\rangle = |\mathbf t+\mathbf 1_p^{(y)}\rangle.
    \end{equation}
    Notice that each corner vertex $v\subset p$ of plaquette $p$ is attached to exactly one $x$-edge of $p$ and one $y$-edge of $p$, while any vertex $v\not\subset p$ is attached to neither.
    So,
    \begin{align}
    	(M_x\mathbf 1_p^{(x)})_v &= \sum_{e_x\revsubset v}(\mathbf 1_p^{(x)})_{e_x}
    	=
    	\begin{cases}
    		1, & v\subset p,\\
    		0, & v\not\subset p,
    	\end{cases} \notag\\
    	(M_y\mathbf 1_p^{(y)})_v &= \sum_{e_y\revsubset v}(\mathbf 1_p^{(y)})_{e_y}
    	=
    	\begin{cases}
    		1, & v\subset p,\\
    		0, & v\not\subset p,
    	\end{cases}
    \end{align}
    thus
    \begin{equation}
    	M_x\mathbf 1_p^{(x)} = M_y\mathbf 1_p^{(y)}.
    	\label{eq_Mx_1px_equals_My_1py}
    \end{equation}
    Using Eq.~(\ref{eq_Mx_1px_equals_My_1py}), we can write
    \begin{align}
    	\langle\mathbf t|\mathsf D_{\text{QCM}}\prod_{e_x\subset p}\tilde X_{e_x}|\mathbf s\rangle
    	&= \langle\mathbf t|\mathsf D_{\text{QCM}}|\mathbf s+\mathbf 1_p^{(x)}\rangle \notag\\
    	&= 2^{-N}\sum_{\mathbf u\in\mathbb F_2^V}(-1)^{\mathbf u^T\left(M_x(\mathbf s+\mathbf 1_p^{(x)})+M_y\mathbf t\right)} \notag\\
    	&= 2^{-N}\sum_{\mathbf u\in\mathbb F_2^V}(-1)^{\mathbf u^T(M_x\mathbf s+M_x\mathbf 1_p^{(x)}+M_y\mathbf t)} \notag\\
    	&= 2^{-N}\sum_{\mathbf u\in\mathbb F_2^V}(-1)^{\mathbf u^T\left(M_x\mathbf s+M_y(\mathbf t+\mathbf 1_p^{(y)})\right)} \notag\\
    	&= \langle\mathbf t+\mathbf 1_p^{(y)}|\mathsf D_{\text{QCM}}|\mathbf s\rangle \notag\\
    	&= \langle\mathbf t|\prod_{e_y\subset p}\tilde X'_{e_y}\mathsf D_{\text{QCM}}|\mathbf s\rangle,
    	\quad \forall\,|\mathbf t\rangle,|\mathbf s\rangle.
    \end{align}
    So,
    \begin{equation}
    	\mathsf D_{\text{QCM}}\prod_{e_x\subset p}\tilde X_{e_x}
    	= \prod_{e_y\subset p}\tilde X'_{e_y}\mathsf D_{\text{QCM}}.
    	\label{eq_D_QCM_intertwine_X}
    \end{equation}
    $\mathsf D_{\text{QCM}}$ indeed realizes the dictionary of QCM self-duality induced by the X-cube FTH.

    \paragraph{\texorpdfstring{Composition $\mathsf D_{\text{QCM}}^\dagger\mathsf D_{\text{QCM}}$.}{Composition of D_QCM-dagger and D_QCM.}}

    Now we show that $\mathsf D_{\text{QCM}}$ is non-invertible.
    Write
    \begin{equation}
    	\mathsf D_{\text{QCM}} = 2^{-N}\sum_{\mathbf s,\mathbf t}\left(\sum_{\mathbf u}(-1)^{\mathbf u^T(M_x\mathbf s+M_y\mathbf t)}\right)|\mathbf t\rangle\langle\mathbf s|.
    \end{equation}
    So,
    \begin{align}
    	\langle\mathbf s'|\mathsf D_{\text{QCM}}^\dagger\mathsf D_{\text{QCM}}|\mathbf s\rangle
    	&= 2^{-2N}\sum_{\mathbf t}\sum_{\mathbf u,\mathbf u'}
    	(-1)^{(\mathbf u')^T(M_x\mathbf s'+M_y\mathbf t)+\mathbf u^T(M_x\mathbf s+M_y\mathbf t)} \notag\\
    	&= 2^{-2N}\sum_{\mathbf u,\mathbf u'}(-1)^{(\mathbf u')^T M_x\mathbf s'+\mathbf u^T M_x\mathbf s}
    	\cdot\sum_{\mathbf t}(-1)^{((\mathbf u')^T+\mathbf u^T)M_y\mathbf t}.
    	\label{eq_D_QCM_gram_step1}
    \end{align}
    Write the exponent in the $\mathbf t$-sum as
    \begin{equation}
    	(\mathbf u'+\mathbf u)^T M_y\mathbf t
    	= \left(M_y^T(\mathbf u'+\mathbf u)\right)^T\mathbf t
    	= \sum_{e_y}\big(M_y^T(\mathbf u'+\mathbf u)\big)_{e_y}t_{e_y},
    \end{equation}
    so,
    \begin{align}
    	\sum_{\mathbf t}(-1)^{((\mathbf u')^T+\mathbf u^T)M_y\mathbf t}
    	&= \sum_{\mathbf t}(-1)^{\sum_{e_y}\big(M_y^T(\mathbf u'+\mathbf u)\big)_{e_y}t_{e_y}} \notag\\
    	&= \sum_{\mathbf t}\prod_{e_y}(-1)^{\big(M_y^T(\mathbf u'+\mathbf u)\big)_{e_y}t_{e_y}} \notag\\
    	&= \prod_{e_y}\sum_{t_{e_y}=0,1}(-1)^{\big(M_y^T(\mathbf u'+\mathbf u)\big)_{e_y}t_{e_y}}.
    \end{align}
    The equivalence between the second row and the third row uses the fact that for $K_i(t_i)$ in a ring, and any set $S$, $\sum_{\mathbf t\in S^N}\prod_{i=1}^N K_i(t_i) = \prod_{i=1}^N\sum_{t_i\in S}K_i(t_i)$.
    When $\big(M_y^T(\mathbf u'+\mathbf u)\big)_{e_y}=0\pmod 2$, $\sum_{t_{e_y}}(-1)^{\big(M_y^T(\mathbf u'+\mathbf u)\big)_{e_y}t_{e_y}}=2$, otherwise it is zero, so
    \begin{equation}
    	\sum_{\mathbf t}(-1)^{((\mathbf u')^T+\mathbf u^T)M_y\mathbf t}
    	=
    	\begin{cases}
    		2^N, & M_y^T(\mathbf u'+\mathbf u)=0,\\
    		0, & M_y^T(\mathbf u'+\mathbf u)\ne 0.
    	\end{cases}
    \end{equation}
    Thus
    \begin{equation}
    	\langle\mathbf s'|\mathsf D_{\text{QCM}}^\dagger\mathsf D_{\text{QCM}}|\mathbf s\rangle
    	= 2^{-N}\sum_{\substack{\mathbf u,\mathbf u':\\ \mathbf u'+\mathbf u\in\ker M_y^T}}(-1)^{(\mathbf u')^T M_x\mathbf s'+\mathbf u^T M_x\mathbf s}.
    \end{equation}
    Let $\mathbf b=\mathbf u'+\mathbf u$, so that $\mathbf u'=\mathbf u+\mathbf b$, we can write
    \begin{align}
    	\langle\mathbf s'|\mathsf D_{\text{QCM}}^\dagger\mathsf D_{\text{QCM}}|\mathbf s\rangle
    	&= 2^{-N}\sum_{\mathbf b\in\ker M_y^T}\sum_{\mathbf u}(-1)^{(\mathbf u+\mathbf b)^T M_x\mathbf s'+\mathbf u^T M_x\mathbf s} \notag\\
    	&= 2^{-N}\sum_{\mathbf b\in\ker M_y^T}(-1)^{\mathbf b^T M_x\mathbf s'}
    	\sum_{\mathbf u}(-1)^{\mathbf u^T M_x(\mathbf s'+\mathbf s)}.
    	\label{eq_D_QCM_gram_step2}
    \end{align}
    Notice that
    \begin{align}
    	\sum_{\mathbf u}(-1)^{\mathbf u^T M_x(\mathbf s'+\mathbf s)}
    	&= \sum_{\mathbf u}(-1)^{\sum_v u_v\big(M_x(\mathbf s'+\mathbf s)\big)_v} \notag\\
    	&= \sum_{\mathbf u}\prod_v(-1)^{u_v\big(M_x(\mathbf s'+\mathbf s)\big)_v} \notag\\
    	&= \prod_v\sum_{u_v=0,1}(-1)^{u_v\big(M_x(\mathbf s'+\mathbf s)\big)_v}.
    \end{align}
    When $\big(M_x(\mathbf s'+\mathbf s)\big)_v=0\pmod 2$, $\sum_{u_v}(-1)^{u_v\big(M_x(\mathbf s'+\mathbf s)\big)_v}=2$, otherwise it is zero, so
    \begin{equation}
    	\sum_{\mathbf u}(-1)^{\mathbf u^T M_x(\mathbf s'+\mathbf s)}
    	=
    	\begin{cases}
    		2^N, & M_x(\mathbf s'+\mathbf s)=0,\\
    		0, & M_x(\mathbf s'+\mathbf s)\ne 0.
    	\end{cases}
    \end{equation}
    Take it back to Eq.~(\ref{eq_D_QCM_gram_step2}), we get
    \begin{equation}
    	\langle\mathbf s'|\mathsf D_{\text{QCM}}^\dagger\mathsf D_{\text{QCM}}|\mathbf s\rangle
    	=
    	\begin{cases}
    		\displaystyle\sum_{\mathbf b\in\ker M_y^T}(-1)^{\mathbf b^T M_x\mathbf s}, & \mathbf s'+\mathbf s\in\ker M_x,\\[10pt]
    		0, & \mathbf s'+\mathbf s\notin\ker M_x,
    	\end{cases}
    \end{equation}
    where we have used $M_x\mathbf s'=M_x\mathbf s$ when $\mathbf s'+\mathbf s\in\ker M_x$.
    Using the above formula, we can write
    \begin{equation}
    	\mathsf D_{\text{QCM}}^\dagger\mathsf D_{\text{QCM}}|\mathbf s\rangle
    	= \sum_{\mathbf s':\,\mathbf s'+\mathbf s\in\ker M_x}\sum_{\mathbf b\in\ker M_y^T}(-1)^{\mathbf b^T M_x\mathbf s}|\mathbf s'\rangle.
    \end{equation}
    Let $\mathbf a=\mathbf s'+\mathbf s$, the above formula can be rewritten as
    \begin{align}
    	\mathsf D_{\text{QCM}}^\dagger\mathsf D_{\text{QCM}}|\mathbf s\rangle
    	&= \sum_{\mathbf a\in\ker M_x}\sum_{\mathbf b\in\ker M_y^T}(-1)^{\mathbf b^T M_x\mathbf s}|\mathbf s+\mathbf a\rangle \notag\\
    	&= \sum_{\mathbf a\in\ker M_x}\sum_{\mathbf b\in\ker M_y^T}\tilde X(\mathbf a)(-1)^{\mathbf b^T M_x\mathbf s}|\mathbf s\rangle,
    \end{align}
    where
    \begin{equation}
    	\tilde X(\mathbf a)\equiv\prod_{e_x\in E_x}\tilde X_{e_x}^{a_{e_x}},
    	\qquad
    	\tilde Z(\mathbf c)\equiv\prod_{e_x\in E_x}\tilde Z_{e_x}^{c_{e_x}}.
    \end{equation}
    Recall that $\tilde Z(\mathbf c)|\mathbf s\rangle = (-1)^{\mathbf c^T\mathbf s}|\mathbf s\rangle$, using $\mathbf b^T M_x\mathbf s = (M_x^T\mathbf b)^T\mathbf s$, we can write
    \begin{align}
    	\mathsf D_{\text{QCM}}^\dagger\mathsf D_{\text{QCM}}|\mathbf s\rangle
    	&= \sum_{\mathbf a\in\ker M_x}\sum_{\mathbf b\in\ker M_y^T}\tilde X(\mathbf a)(-1)^{(M_x^T\mathbf b)^T\mathbf s}|\mathbf s\rangle \notag\\
    	&= \sum_{\mathbf a\in\ker M_x}\sum_{\mathbf b\in\ker M_y^T}\tilde X(\mathbf a)\tilde Z(M_x^T\mathbf b)|\mathbf s\rangle.
    \end{align}
    The above formula holds for all $|\mathbf s\rangle$, so
    \begin{equation}
    	\mathsf D_{\text{QCM}}^\dagger\mathsf D_{\text{QCM}}
    	= \sum_{\mathbf a\in\ker M_x}\tilde X(\mathbf a)\sum_{\mathbf b\in\ker M_y^T}\tilde Z(M_x^T\mathbf b).
    	\label{eq_D_QCM_gram_general}
    \end{equation}

    \paragraph{\texorpdfstring{Explicit structure of $\ker M_x$, $\ker M_y^T$, and non-invertibility.}{Explicit structure of ker M_x, ker M_y^T, and non-invertibility.}}

    We label vertices by $v=[i,j]$, $x$-edges by $e_x=[i+1/2,j]$, $y$-edges by $[i,j+1/2]$ with $i\in\mathbb Z_{L_x}$ and $j\in\mathbb Z_{L_y}$.

    1.
    \textit{Solving $\mathbf a\in\ker M_x$:}
    The condition $(M_x\mathbf a)_{[i,j]}=0$ can be written as
    \begin{equation}
    	a_{[i-1/2,j]} + a_{[i+1/2,j]} = 0 \iff a_{[i-1/2,j]} = a_{[i+1/2,j]}, \quad \forall\,i,j.
    \end{equation}
    Thus $a_{[i+1/2,j]}$ is independent of $i$ along each $y=j$,
    \begin{equation}
    	a_{[i+1/2,j]} := r_j, \qquad r_j\in\mathbb F_2, \quad j=0,1,\dots,L_y-1.
    \end{equation}
    There is no relation among these $\mathbf a$, so $\dim\ker M_x = L_y$.
    The $L_y$ basis generators of $\ker M_x$ correspond to the $L_y$ $\tilde X$ line-like symmetry indicators $S_x^{(X)}(j)$ of the QCM under the $(em)$ top boundary, defined in Sec.~\ref{subsec_X_cube_em_me}.
    So the $\tilde X(\mathbf a)$ factor of $\mathsf D_{\text{QCM}}^\dagger\mathsf D_{\text{QCM}}$ is
    \begin{equation}
    	\sum_{\mathbf a\in\ker M_x}\tilde X(\mathbf a) = \prod_{j=0}^{L_y-1}\big(1+S_x^{(X)}(j)\big).
    	\label{eq_ker_Mx_product}
    \end{equation}

    2.
    \textit{Solving $\mathbf b\in\ker M_y^T$ and $M_x^T\mathbf b$:}
    For each $e_y=[i,j+1/2]$, its two attached vertices are $[i,j]$ and $[i,j+1]$.
    Thus $(M_y^T\mathbf b)_{[i,j+1/2]}=0$ can be written as
    \begin{equation}
    	b_{[i,j]} + b_{[i,j+1]} = 0 \iff b_{[i,j]} = b_{[i,j+1]}, \quad \forall\,i,j.
    \end{equation}
    So, $b_{[i,j]}$ is independent of $j$,
    \begin{equation}
    	b_{[i,j]} := c_i, \qquad c_i\in\mathbb F_2, \quad i=0,1,\dots,L_x-1.
    \end{equation}
    These $\mathbf b$ have no relation, $\dim\ker M_y^T = L_x$.
    However $M_x^T\mathbf b$ have relation, so are $\tilde Z(M_x^T\mathbf b)$.
    Since $(M_x^T)_{e_x,v}=(M_x)_{v,e_x}$,
    \begin{equation}
    	(M_x^T\mathbf b)_{e_x} = \sum_v (M_x^T)_{e_x,v}b_v = \sum_{v\subset e_x}b_v.
    \end{equation}
    For the $x$-edge $e_x=[i+1/2,j]$, it reads
    \begin{equation}
    	(M_x^T\mathbf b)_{[i+1/2,j]} = b_{[i,j]} + b_{[i+1,j]} = c_i + c_{i+1}.
    \end{equation}
    Notice that shifting $c_i\mapsto c'_i=c_i+1$ for all $i\in\mathbb Z_{L_x}$ at once leaves all $(M_x^T\mathbf b)_{[i+1/2,j]}$ invariant.
    Therefore, as $\mathbf b$ runs over the $2^{L_x}$ elements of $\ker M_y^T$, each image vector $M_x^T\mathbf b$ is hit $2$ times, $\dim(M_x^T\ker M_y^T) = L_x-1$.
    Choosing the $L_x$ basis vectors of $\ker M_y^T$, $(\mathbf b^{(i_0)})_{[i,j]}=\delta_{i,i_0}$, the operators $\tilde Z(M_x^T\mathbf b^{(i_0)})$ are the two-column $\tilde Z$-strings $S_y^{(Z)}(i_0)$ along $y$-direction, defined in Sec.~\ref{subsec_X_cube_em_me}, which satisfy the product relation $\prod_{i=0}^{L_x-1}S_y^{(Z)}(i)=1$.
    Therefore, the $\tilde Z$-factor of $\mathsf D_{\text{QCM}}^\dagger\mathsf D_{\text{QCM}}$ can be written as
    \begin{align}
    	\sum_{\mathbf b\in\ker M_y^T}\tilde Z(M_x^T\mathbf b)
    	&= \sum_{\{c_i=0,1\}}\tilde Z\left(\sum_{i=0}^{L_x-1}c_i M_x^T\mathbf b^{(i)}\right) \notag\\
    	&= \sum_{\{c_i=0,1\}}\prod_{i=0}^{L_x-1}\big(S_y^{(Z)}(i)\big)^{c_i} \notag\\
    	&= \prod_{i=0}^{L_x-1}\left(\sum_{c_i=0,1}\big(S_y^{(Z)}(i)\big)^{c_i}\right) \notag\\
    	&= \prod_{i=0}^{L_x-1}\big(1+S_y^{(Z)}(i)\big).
    \end{align}
    Using the relation $\prod_{i=0}^{L_x-1}S_y^{(Z)}(i)=1$, it can be further written as
    \begin{equation}
    	\sum_{\mathbf b\in\ker M_y^T}\tilde Z(M_x^T\mathbf b)
    	= 2\prod_{i=0}^{L_x-2}\big(1+S_y^{(Z)}(i)\big).
    	\label{eq_ker_MyT_product}
    \end{equation}
    Therefore,
    \begin{equation}
    	\mathsf D_{\text{QCM}}^\dagger\mathsf D_{\text{QCM}}
    	= 2\prod_{j=0}^{L_y-1}\big(1+S_x^{(X)}(j)\big)\prod_{i=0}^{L_x-2}\big(1+S_y^{(Z)}(i)\big)
    	= 2^{L_x+L_y}P_{\text{lines},+}^{\text{QCM}},
    \end{equation}
    where
    \begin{equation}
    	P_{\text{lines},+}^{\text{QCM}} \equiv \prod_{j=0}^{L_y-1}\frac{1+S_x^{(X)}(j)}{2}\prod_{i=0}^{L_x-2}\frac{1+S_y^{(Z)}(i)}{2}
    \end{equation}
    is the projector onto the common $+1$ eigenspace of line-like symmetry indicators $S_x^{(X)}(j)$ and $S_y^{(Z)}(i)$.

    For a state $|\psi\rangle$ of the $x$-edge QCM, with\footnote{The choice $j=1$ is arbitrary, it can be any other legal integer.}
    \begin{equation}
    	S_x^{(X)}(1)|\psi\rangle = -|\psi\rangle,
    \end{equation}
    we have
    \begin{align}
    	&\langle\psi|\mathsf D_{\text{QCM}}^\dagger\mathsf D_{\text{QCM}}|\psi\rangle = 2^{L_x+L_y}\langle\psi|P_{\text{lines},+}^{\text{QCM}}|\psi\rangle = 0 \notag\\
    	\iff\ &\big|\mathsf D_{\text{QCM}}|\psi\rangle\big| = 0 \iff \mathsf D_{\text{QCM}}|\psi\rangle = 0.
    \end{align}
    Such a $|\psi\rangle$ exists, e.g. $|\psi\rangle = \tilde Z_{[1/2,1]}|+\cdots+\rangle$.
    So, $\mathsf D_{\text{QCM}}$ is non-invertible.

    \paragraph{Symmetry absorption.}

    1.
    \textit{Absorption of $S_x^{(X)}(j)$ ($\mathbf a\in\ker M_x$):}
    for any $\mathbf a\in\ker M_x$ (so that $M_x\mathbf a=0$), using $\tilde X(\mathbf a)|\mathbf s\rangle = |\mathbf s+\mathbf a\rangle$, we have
    \begin{align}
    	\langle\mathbf t|\mathsf D_{\text{QCM}}\tilde X(\mathbf a)|\mathbf s\rangle
    	&= \langle\mathbf t|\mathsf D_{\text{QCM}}|\mathbf s+\mathbf a\rangle \notag\\
    	&= 2^{-N}\sum_{\mathbf u\in\mathbb F_2^V}(-1)^{\mathbf u^T(M_x(\mathbf s+\mathbf a)+M_y\mathbf t)} \notag\\
    	&= 2^{-N}\sum_{\mathbf u\in\mathbb F_2^V}(-1)^{\mathbf u^T(M_x\mathbf s+M_y\mathbf t)} \notag\\
    	&= \langle\mathbf t|\mathsf D_{\text{QCM}}|\mathbf s\rangle, \quad \forall\,|\mathbf t\rangle,|\mathbf s\rangle.
    \end{align}
    So,
    \begin{align}
    	&\mathsf D_{\text{QCM}}\tilde X(\mathbf a) = \mathsf D_{\text{QCM}}, \quad \forall\,\mathbf a\in\ker M_x \notag\\
    	\implies\ &\mathsf D_{\text{QCM}}S_x^{(X)}(j) = \mathsf D_{\text{QCM}}, \quad j=0,1,\dots,L_y-1.
    \end{align}

    2.
    \textit{Absorption of $S_y^{(Z)}(i)$ ($\mathbf b\in\ker M_y^T$):}
    for any $\mathbf b\in\ker M_y^T$ (so that $M_y^T\mathbf b=0$, $\mathbf b^T M_y\mathbf t = (M_y^T\mathbf b)^T\mathbf t = 0$), shifting the summation dummy index $\mathbf u\mapsto\mathbf u+\mathbf b$ in Eq.~(\ref{eq_D_QCM_matrix_element}) leaves the summation invariant, so
    \begin{align}
    	\langle\mathbf t|\mathsf D_{\text{QCM}}|\mathbf s\rangle
    	&= 2^{-N}\sum_{\mathbf u\in\mathbb F_2^V}(-1)^{(\mathbf u+\mathbf b)^T(M_x\mathbf s+M_y\mathbf t)} \notag\\
    	&= (-1)^{\mathbf b^T M_x\mathbf s}\cdot 2^{-N}\sum_{\mathbf u\in\mathbb F_2^V}(-1)^{\mathbf u^T(M_x\mathbf s+M_y\mathbf t)} \notag\\
    	&= (-1)^{(M_x^T\mathbf b)^T\mathbf s}\langle\mathbf t|\mathsf D_{\text{QCM}}|\mathbf s\rangle \notag\\
    	&= \langle\mathbf t|\mathsf D_{\text{QCM}}\tilde Z(M_x^T\mathbf b)|\mathbf s\rangle, \quad \forall\,|\mathbf t\rangle,|\mathbf s\rangle,
    \end{align}
    so
    \begin{align}
    	&\mathsf D_{\text{QCM}}\tilde Z(M_x^T\mathbf b) = \mathsf D_{\text{QCM}}, \quad \forall\,\mathbf b\in\ker M_y^T \notag\\
    	\implies\ &\mathsf D_{\text{QCM}}S_y^{(Z)}(i) = \mathsf D_{\text{QCM}}, \quad i=0,1,\dots,L_x-1.
    \end{align}
    Similarly, the symmetry indicators of the $y$-edge QCM can be absorbed by $\mathsf D_{\text{QCM}}$ on the left side.

    \subsubsection{\texorpdfstring{QCM--TFPIM duality}{QCM--TFPIM duality}}\label{app_subsec_D_QCM_TFPIM}
    
    Now, we illustrate the QCM--TFPIM duality induced by the X-cube FTH explicitly.
    Denote the $Z$-configurations of $x$-edge spins and vertex spins on the bottom boundary by $\mathbf s\in\mathbb F_2^{E_x}$ and $\mathbf t\in\mathbb F_2^V$, respectively.
    Recall the incidence matrix $M_x$ defined in Eq.~(\ref{eq_Mx_My_main}).
    The duality operator reads
    \begin{equation}
    	\mathsf D_{\text{QC}\to\text{TFPI}}
    	= 2^{-N/2}\sum_{\mathbf s\in\mathbb F_2^{E_x}}\sum_{\mathbf t\in\mathbb F_2^V}(-1)^{\mathbf t^T M_x\mathbf s}|\mathbf t\rangle\langle\mathbf s|.
    	\label{eq_D_QCM_to_TFPIM_def}
    \end{equation}
    Notice that
    \begin{align}
    	\mathbf t^T M_x\mathbf s
    	&= \sum_{v\in V}t_v(M_x\mathbf s)_v = \sum_{v\in V}t_v\sum_{e_x\in E_x}(M_x)_{v,e_x}s_{e_x} \notag\\
    	&= \sum_{v\in V}t_v\sum_{e_x\revsubset v}s_{e_x}.
    \end{align}
    So the matrix elements are
    \begin{equation}
    	\langle\mathbf t|\mathsf D_{\text{QC}\to\text{TFPI}}|\mathbf s\rangle
    	= 2^{-N/2}(-1)^{\sum_v t_v\sum_{e_x\revsubset v}s_{e_x}}.
    \end{equation}

    \paragraph{Intertwining relations.}

    First, denoting the unit vector $\mathbf 1_v\in\mathbb F_2^V$, $(\mathbf 1_v)_{v'}=\delta_{v,v'}$, we have
    \begin{align}
    	\langle\mathbf t|\tilde X'_v\mathsf D_{\text{QC}\to\text{TFPI}}|\mathbf s\rangle
    	&= \langle\mathbf t+\mathbf 1_v|\mathsf D_{\text{QC}\to\text{TFPI}}|\mathbf s\rangle \notag\\
    	&= 2^{-N/2}(-1)^{\sum_{v'}(t_{v'}+\delta_{v,v'})\sum_{e_x\revsubset v'}s_{e_x}} \notag\\
    	&= (-1)^{\sum_{e_x\revsubset v}s_{e_x}}\langle\mathbf t|\mathsf D_{\text{QC}\to\text{TFPI}}|\mathbf s\rangle \notag\\
    	&= \langle\mathbf t|\mathsf D_{\text{QC}\to\text{TFPI}}\prod_{e_x\revsubset v}\tilde Z_{e_x}|\mathbf s\rangle,
    \end{align}
    where $\tilde X'_v$ and $\tilde Z_{e_x}$ are vertex Pauli and $x$-edge Pauli, respectively.
    The above formula holds for all $|\mathbf s\rangle,|\mathbf t\rangle$, so
    \begin{equation}
    	\tilde X'_v\mathsf D_{\text{QC}\to\text{TFPI}} = \mathsf D_{\text{QC}\to\text{TFPI}}\prod_{e_x\revsubset v}\tilde Z_{e_x}.
    \end{equation}

    Second, acting $\prod_{e_x\subset p}\tilde X_{e_x}$ on $|\mathbf s\rangle$ shifts $\mathbf s\mapsto\mathbf s+\mathbf 1_p^{(x)}$, where $\mathbf 1_p^{(x)}\in\mathbb F_2^{E_x}$ is the vector with elements $(\mathbf 1_p^{(x)})_{e_x}=1$ iff $e_x\subset p$.
    Notice that
    \begin{equation}
    	\sum_{e_x\revsubset v}(\mathbf 1_p^{(x)})_{e_x} =
    	\begin{cases}
    		1, & v\subset p,\\
    		0, & v\not\subset p,
    	\end{cases}
    \end{equation}
    since iff $v\subset p$, there is a unique $e_x$ attached to $v$ and belonging to $p$ at the same time.
    So, we can write
    \begin{align}
    	\langle\mathbf t|\mathsf D_{\text{QC}\to\text{TFPI}}\prod_{e_x\subset p}\tilde X_{e_x}|\mathbf s\rangle
    	&= \langle\mathbf t|\mathsf D_{\text{QC}\to\text{TFPI}}|\mathbf s+\mathbf 1_p^{(x)}\rangle \notag\\
    	&= 2^{-N/2}(-1)^{\sum_v t_v\sum_{e_x\revsubset v}(\mathbf s+\mathbf 1_p^{(x)})_{e_x}} \notag\\
    	&= (-1)^{\sum_{v\subset p}t_v}\langle\mathbf t|\mathsf D_{\text{QC}\to\text{TFPI}}|\mathbf s\rangle \notag\\
    	&= \langle\mathbf t|\prod_{v\subset p}\tilde Z'_v\mathsf D_{\text{QC}\to\text{TFPI}}|\mathbf s\rangle,
    	\quad \forall\,|\mathbf t\rangle,|\mathbf s\rangle,
    \end{align}
    thus
    \begin{equation}
    	\mathsf D_{\text{QC}\to\text{TFPI}}\prod_{e_x\subset p}\tilde X_{e_x} = \prod_{v\subset p}\tilde Z'_v\mathsf D_{\text{QC}\to\text{TFPI}}.
    \end{equation}

    \paragraph{\texorpdfstring{Composition $\mathsf D_{\text{QC}\to\text{TFPI}}^\dagger\mathsf D_{\text{QC}\to\text{TFPI}}$ and non-invertibility.}{Composition and non-invertibility.}}

    By the same calculation as in Appendix~\ref{app_subsec_D_square} with $M$ replaced by $M_x$, we have
    \begin{equation}
    	\mathsf D_{\text{QC}\to\text{TFPI}}^\dagger\mathsf D_{\text{QC}\to\text{TFPI}}
    	= \sum_{\mathbf a\in\ker M_x}\tilde X(\mathbf a).
    \end{equation}
    Here $\tilde X(\mathbf a)\equiv\prod_{e_x}\tilde X_{e_x}^{a_{e_x}}$.
    Using Eq.~(\ref{eq_ker_Mx_product}), we have
    \begin{equation}
    	\mathsf D_{\text{QC}\to\text{TFPI}}^\dagger\mathsf D_{\text{QC}\to\text{TFPI}}
    	= \prod_{j=0}^{L_y-1}\big(1+S_x^{(X)}(j)\big) = 2^{L_y}P_{\text{rows},+},
    \end{equation}
    where $S_x^{(X)}(j)$ $=\prod_{i=0}^{L_x-1}\tilde X_{[i+1/2,j]}$ [see Eq.~(\ref{eq_QCM_line_symmetry_indicators})] are the $\tilde X$ product symmetry indicators along $x$-direction of the $x$-edge QCM, and $P_{\text{rows},+}\equiv\prod_{j=0}^{L_y-1}\frac{1+S_x^{(X)}(j)}{2}$ is the projector onto the common $+1$ eigenspace of those symmetry indicators.
    Note that the $\tilde Z$ product symmetry indicators $S_y^{(Z)}(i)$ along $y$-direction are not projected here.
    As in Appendix~\ref{app_subsec_D_QCM}, a state with $S_x^{(X)}(j)=-1$ for some $j$ is annihilated, so $\mathsf D_{\text{QC}\to\text{TFPI}}$ is non-invertible.

    \paragraph{Symmetry absorption and symmetry correspondence.}

    1.
    \textit{Absorption of $S_x^{(X)}(j)$ ($\mathbf a\in\ker M_x$):} for any $\mathbf a\in\ker M_x$ (so that $M_x\mathbf a=0$), acting $\tilde X(\mathbf a)$ on $|\mathbf s\rangle$ gives $\tilde X(\mathbf a)|\mathbf s\rangle = |\mathbf s+\mathbf a\rangle$, so
    \begin{align}
    	\langle\mathbf t|\mathsf D_{\text{QC}\to\text{TFPI}}\tilde X(\mathbf a)|\mathbf s\rangle
    	&= \langle\mathbf t|\mathsf D_{\text{QC}\to\text{TFPI}}|\mathbf s+\mathbf a\rangle \notag\\
    	&= 2^{-N/2}(-1)^{\mathbf t^T M_x(\mathbf s+\mathbf a)} \notag\\
    	&= 2^{-N/2}(-1)^{\mathbf t^T M_x\mathbf s}
    	= \langle\mathbf t|\mathsf D_{\text{QC}\to\text{TFPI}}|\mathbf s\rangle, \quad \forall\,|\mathbf t\rangle,|\mathbf s\rangle,
    \end{align}
    so
    \begin{align}
    	&\mathsf D_{\text{QC}\to\text{TFPI}}\tilde X(\mathbf a) = \mathsf D_{\text{QC}\to\text{TFPI}}, \quad \forall\,\mathbf a\in\ker M_x \notag\\
    	\implies\ &\mathsf D_{\text{QC}\to\text{TFPI}}S_x^{(X)}(j) = \mathsf D_{\text{QC}\to\text{TFPI}}, \quad j=0,1,\dots,L_y-1.
    \end{align}

    2.
    \textit{Absorption of $x$-direction line-like symmetries $S_x^{\prime(X)}(j)$ ($\mathbf b\in\ker M_x^T$):}
    for any $\mathbf b\in\mathbb F_2^V$, we have
    \begin{align}
    	\langle\mathbf t|\tilde X'(\mathbf b)\mathsf D_{\text{QC}\to\text{TFPI}}|\mathbf s\rangle
    	&= \langle\mathbf t+\mathbf b|\mathsf D_{\text{QC}\to\text{TFPI}}|\mathbf s\rangle \notag\\
    	&= 2^{-N/2}(-1)^{(\mathbf t+\mathbf b)^T M_x\mathbf s} \notag\\
    	&= (-1)^{(M_x^T\mathbf b)^T\mathbf s}\cdot\langle\mathbf t|\mathsf D_{\text{QC}\to\text{TFPI}}|\mathbf s\rangle \notag\\
    	&= \langle\mathbf t|\mathsf D_{\text{QC}\to\text{TFPI}}\tilde Z(M_x^T\mathbf b)|\mathbf s\rangle, \quad \forall\,|\mathbf t\rangle,|\mathbf s\rangle,
    \end{align}
    where $\tilde X'(\mathbf b)\equiv\prod_v(\tilde X'_v)^{b_v}$, so
    \begin{equation}
    	\tilde X'(\mathbf b)\mathsf D_{\text{QC}\to\text{TFPI}}
    	= \mathsf D_{\text{QC}\to\text{TFPI}}\tilde Z(M_x^T\mathbf b).
    	\label{eq_QCM_TFPIM_general_b}
    \end{equation}
    For $\mathbf b\in\ker M_x^T$, $M_x^T\mathbf b=0$, $\tilde Z(M_x^T\mathbf b)=1$, the above formula reduces to the symmetry absorption
    \begin{equation}
    	\tilde X'(\mathbf b)\mathsf D_{\text{QC}\to\text{TFPI}} = \mathsf D_{\text{QC}\to\text{TFPI}}.
    \end{equation}

    3.
    \textit{One-to-one correspondence of $y$-direction line-like symmetries ($\mathbf b\in\ker M_y^T$):}
    unlike the $x$-direction line-like symmetries, the $y$-direction line-like symmetries are not absorbed by $\mathsf D_{\text{QC}\to\text{TFPI}}$, but intertwined by $\mathsf D_{\text{QC}\to\text{TFPI}}$ to the $y$-direction line-like symmetries of vertex TFPIM.
    Take $\mathbf b\in\ker M_y^T$, $\tilde X'(\mathbf b)$ and $\tilde Z(M_x^T\mathbf b)$ are the products of $y$-direction line-like symmetries of vertex TFPIM and $x$-edge QCM, respectively, so Eq.~(\ref{eq_QCM_TFPIM_general_b}) becomes the intertwining relation between these symmetries.
    Denote the $x$-direction and $y$-direction line-like symmetry generators of vertex TFPIM as $S_x^{\prime(X)}(j)$ and $S_y^{\prime(X)}(i)$, respectively.
    The $y$-direction double-line symmetry generators $S_y^{(Z)}(i)$ of $x$-edge QCM have the relation $\prod_{i=0}^{L_x-1}S_y^{(Z)}(i)=1$, which corresponds to the relation $\prod_{i=0}^{L_x-1}S_y^{\prime(X)}(i)\prod_{j=0}^{L_y-1}S_x^{\prime(X)}(j)=1$ on the TFPIM side.
    Since both have one relation, the $y$-direction line-like symmetries of vertex TFPIM and $x$-edge QCM indeed correspond one-to-one.

    \section{\texorpdfstring{Details of Haah's code duality operator $\mathsf D_{\rm Haah}$}{Details of Haah's code duality operator D\_Haah}}\label{appendix_duality_operators_Haah}\label{app_Haah_duality_op}
    Recall $\mathsf D_{\rm Haah}$ in Eq.~(\ref{eq_D_Haah}), and $\tilde Z(s),\tilde X(m)$ represent the Pauli operators corresponding to $s\tilde Z$ and $m\tilde X$, and similarly for the primed operators.
    
    $\forall\text{ unit }m\in R_\partial$,
    \begin{align}
    	\mathsf D_{\rm Haah}\tilde Z(fm)
    	&= \sum_{t\in R_\partial}|t\rangle_{(X)}\,{}_{(Z)}\langle+\cdots+|\tilde Z(ft)\tilde Z(fm) \notag\\
    	&= \sum_{t\in R_\partial}|t\rangle_{(X)}\,{}_{(Z)}\langle+\cdots+|\tilde Z\big(f(t+m)\big) \notag\\
    	&= \sum_{t\in R_\partial}|t\rangle_{(X)}\,{}_{(Z)}\langle f(t+m)| \notag\\
    	&= \sum_{t\in R_\partial}|t-m\rangle_{(X)}\,{}_{(Z)}\langle ft| \notag\\
    	&= \tilde X'(m)\mathsf D_{\rm Haah}.
    \end{align}
    Let $m=x^i y^j$, $f=\sum_{ab}f_{ab}x^a y^b$, then
    \begin{equation}
    	\bar f m = \sum_{ab}f_{ab}x^{-a}y^{-b}m = \sum_{ab}f_{ab}x^{i-a}y^{j-b},
    \end{equation}
    \begin{equation}
    	\tilde Z'(\bar f m) = \prod_{ab}\big(\tilde Z'_{i-a,j-b}\big)^{f_{ab}}.
    \end{equation}
    For each individual $\tilde Z'$,
    \begin{equation}
    	\tilde Z'_{k,l}|t\rangle_{(X)} = (-1)^{t_{kl}}|t\rangle_{(X)}.
    \end{equation}
    So,
    \begin{align}
    	\tilde Z'(\bar f m)|t\rangle_{(X)}
    	&= \prod_{ab}\big(\tilde Z'_{i-a,j-b}\big)^{f_{ab}}|t\rangle_{(X)} \notag\\
    	&= \prod_{ab}(-1)^{f_{ab}t_{i-a,j-b}}|t\rangle_{(X)} \notag\\
    	&= (-1)^{\sum_{ab}f_{ab}t_{i-a,j-b}}|t\rangle_{(X)}.
    \end{align}
    Note that $\sum_{ab}f_{ab}t_{i-a,j-b} = (ft)_{ij}$:
    \begin{itemize}
    	\item suppose $f=\sum_{ab}f_{ab}x^a y^b$, $t=\sum_{kl}t_{kl}x^k y^l$, then
    	\begin{align}
    		ft &= \sum_{abkl}f_{ab}t_{kl}x^{a+k}y^{b+l} \notag\\
    		\iff (ft)_{ij} &= \sum_{ab}f_{ab}t_{i-a,j-b}.
    	\end{align}
    \end{itemize}
    Therefore,
    \begin{equation}
    	\tilde Z'(\bar f m)|t\rangle_{(X)} = (-1)^{(ft)_{ij}}|t\rangle_{(X)}.
    \end{equation}
    On the other hand,
    \begin{equation}
    	\tilde X(m)|ft\rangle_{(Z)} = (-1)^{(ft)_{ij}}|ft\rangle_{(Z)},
    \end{equation}
    so
    \begin{align}
    	\mathsf D_{\rm Haah}\tilde X(m)
    	&= \sum_t |t\rangle_{(X)}\,{}_{(Z)}\langle ft|\tilde X(m) \notag\\
    	&= \sum_t |t\rangle_{(X)}\,{}_{(Z)}\langle ft|(-1)^{(ft)_{ij}} \notag\\
    	&= \sum_t \tilde Z'(\bar f m)|t\rangle_{(X)}\,{}_{(Z)}\langle ft| \notag\\
    	&= \tilde Z'(\bar f m)\mathsf D_{\rm Haah}.
    \end{align}
    $\mathsf D_{\rm Haah}$ can be represented as
    \begin{equation}
    	\mathsf D_{\rm Haah}|s\rangle_{(Z)} =
    	\begin{cases}
    		|t\rangle_{(X)}, & s=ft,\\
    		0, & s\notin fR_\partial,
    	\end{cases}
    \end{equation}
    which is well defined since for any $t,t'\in R_\partial$, $ft=ft'$ implies $t=t'$.
    Note that $fR_\partial=(F)\subset I_{L_z}$.
    Since $A^{L_z}\notin fR_\partial$ (evaluating at $x=y=1$ gives $A^{L_z}(1,1)=1$, while $(ft)(1,1)=0$ for any $t$ since $F(1,1)=0$), $\mathsf D_{\rm Haah}|A^{L_z}\rangle_{(Z)} = 0$, the duality is non-invertible.

    \emph{Absorption of symmetries.}
    For $u\in\ann_{\widehat{R}_\partial}(\bar{F})$, $\langle u,ft\rangle=(u\bar{f}\,\bar{t})_0=0$ (see Eq.~\ref{pairing_def} for the pairing definition), so $\tilde{X}(u)|ft\rangle_{(Z)}=|ft\rangle_{(Z)}$, thus $\mathsf{D}_{\text{Haah}}\,\tilde{X}(u)=\sum_t|t\rangle_{(X)}\,{}_{(Z)}\langle ft|\tilde X(u)=\mathsf{D}_{\text{Haah}}$.

    \emph{Gram operators.}
    Since ${}_{(X)}\langle t|t'\rangle_{(X)}=\delta_{t,t'}$ and ${}_{(Z)}\langle ft|ft'\rangle_{(Z)}=\delta_{t,t'}$ (the states $|s\rangle_{(Z)}$ are orthonormal, and $ft=ft'$ implies $t=t'$), from Eq.~(\ref{eq_D_Haah}) we can write
    \begin{align}
    	\mathsf{D}_{\text{Haah}}^\dagger\mathsf{D}_{\text{Haah}}&=\sum_{t\in R_\partial}|ft\rangle_{(Z)}\,{}_{(Z)}\langle ft|, \\
    	\mathsf{D}_{\text{Haah}}\mathsf{D}_{\text{Haah}}^\dagger&=\sum_{t\in R_\partial}|t\rangle_{(X)}\,{}_{(X)}\langle t|=1_{\tilde{\mathcal{H}}_{(X)}}.
    \end{align}
    Thus $\mathsf{D}_{\text{Haah}}^\dagger\mathsf{D}_{\text{Haah}}$ is the projector onto $\operatorname{span}\{|s\rangle_{(Z)}:s\in fR_\partial\}$, which is the common $+1$ eigenspace of all the symmetry indicators $u\tilde{X}$, $u\in\ann_{\widehat{R}_\partial}(\bar{F})$, in $\tilde{\mathcal{H}}_{(Z)}(L_z)$.
    Indeed, every $|s\rangle_{(Z)}$ is a common eigenstate of them with eigenvalues $(-1)^{\langle u,s\rangle}$, so it suffices to show that $\langle u,s\rangle=0$ for all $u\in\ann_{\widehat{R}_\partial}(\bar{F})$ if and only if $s\in fR_\partial$.
    The ``if'' part is shown above.
    For the ``only if'' part, take $s\notin fR_\partial$, so that there exists a $\mathbb F_2$-linear functional on $R_\partial$ that vanishes on $fR_\partial$ but not on $s$, and since every $\mathbb F_2$-linear functional on $R_\partial$ can be written as $\langle u,\cdot\rangle$ for some $u\in\widehat{R}_\partial$, we get $\langle u,s\rangle=1\neq0$ and $(u\bar f)_m=\langle u,mf\rangle=0$ for all $m$, i.e., $u\in\ann_{\widehat{R}_\partial}(\bar{F})$. So, only if $s\in fR_\partial$, $\langle u,s\rangle=0$ for all $u\in\ann_{\widehat{R}_\partial}(\bF)$.

\end{document}